\documentclass[11pt, oneside]{Thesis} 

\graphicspath{{Pictures/}} 

\usepackage[square, numbers, comma, sort&compress]{natbib} 
\usepackage[utf8]{inputenc}
\usepackage[italian]{babel}
\usepackage{amsmath}
\usepackage{amsfonts}
\usepackage{amssymb}
\usepackage{amsthm}

\usepackage{graphicx}

\usepackage{siunitx}
\usepackage[version=4]{mhchem}
\usepackage{bigints}
\usepackage{placeins}

\usepackage{microtype}

\usepackage{braket}
\usepackage{mhchem}
\usepackage{subfigure}

\numberwithin{equation}{section}

\theoremstyle{plain}
\newtheorem{teorema}{Teorema}[section]
\newtheorem{proposizione}[teorema]{Proposizione}
\newtheorem{corollario}[teorema]{Corollario}

\theoremstyle{definition}
\newtheorem{definizione}{Definizione}[section]

\newtheorem{esempio}{Esempio}[section]

\theoremstyle{remark}
\newtheorem{osservazione}{Osservazione}[section]

\newcommand*{\diff}{\mathop{}\!\mathrm{d}}
\newcommand{\norm}[1]{\left\lVert#1\right\rVert}

\usepackage{hyperref}
\hypersetup{
	bookmarks=true,
	colorlinks=false,
	pdftitle={Simulazione Monte Carlo per il trasporto di cariche nel grafene},    
    pdfauthor={Giovanni Nastasi},     
    pdfsubject={Tesi di Laurea Magistrale},   
	linkcolor=red,          
    citecolor=green,        
    filecolor=magenta,      
    urlcolor=cyan    
} 
\title{\ttitle} 

\begin{document}

\frontmatter 

\setstretch{1.3} 

\fancyhead{} 
\rhead{\thepage} 
\lhead{} 

\pagestyle{fancy} 

\newcommand{\HRule}{\rule{\linewidth}{0.5mm}} 



\begin{titlepage}

\center 

\includegraphics[width=15em]{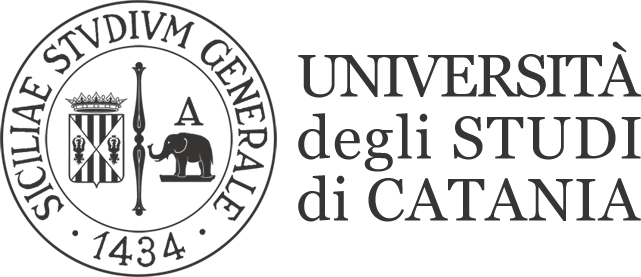}\\[1cm] 
 
\textsc{\Large Dipartimento di Matematica e Informatica}\\[0.5cm]
\textsc{\large Corso di Laurea Magistrale in Matematica}\\[0.5cm]

\rule[0.1cm]{\linewidth}{0.1mm}
\rule[0.5cm]{\linewidth}{0.6mm}

\vspace{10mm}



{ \huge \bfseries Simulazione Monte Carlo per il trasporto di cariche nel grafene}\\[0.4cm] 

\vspace{25mm}

\textsc{\large Tesi di Laurea Magistrale}\\[0.5cm] 

\vspace{75mm}
 
\begin{minipage}{0.4\textwidth}
\begin{flushleft} \large
\emph{Autore:}\\
Giovanni \textsc{Nastasi} 
\end{flushleft}
\end{minipage}
~
\begin{minipage}{0.4\textwidth}
\begin{flushright} \large
\emph{Relatore:} \\
Chiar.mo Prof.\\Vittorio \textsc{Romano} 
\end{flushright}
\end{minipage}\\[0.5cm]

\rule[0.1cm]{\linewidth}{0.1mm}\\
{\large{Anno Accademico 2015-2016}} 

\vfill 

\end{titlepage}

\pagestyle{fancy} 

\lhead{\emph{Indice}} 
\tableofcontents 






{
}




{
}




{


}





%

\chapter*{Introduzione}
\addcontentsline{toc}{chapter}{Introduzione}
\chaptermark{Introduzione.}

\lhead{Introduzione}


Grafene è il nome dato ad un singolo strato piano costituito solamente da atomi di carbonio, disposti in un reticolo bidimensionale a nido d'ape. Dal punto di vista teorico, il grafene è stato studiato a partire dal 1947. Quarant'anni dopo, si scoprirono le particolari proprietà elettrodinamiche di questo materiale, motivo per il quale cominciò il suo studio nell'ambito della fisica condensata. In un primo momento si pensava che il grafene non potesse esistere allo stato libero. Sorprendentemente, nel 2004 A.K.~Geim e K.S.~Novoselov, due studiosi dell'Università di Manchester, isolarono e caratterizzarono il grafene allo stato libero. Questo gli valse nel 2010 il premio Nobel per la fisica.

Il grafene è un materiale dalle innumerevoli proprietà quali l'elevata conducibilità termica ed elettrica, la trasparenza su tutto lo spettro luminoso, la resistenza sia meccanica che anche ad agenti fisici e chimici. Ciò colloca questo materiale in una posizione di rilievo nella ricerca in ambito tecnologico. In particolare, alcune prospettive di sviluppo riguardano la sostituzione del silicio con nuovi materiali a base di grafene, la costruzione di dispositivi elettronici completamente
flessibili, lo sviluppo di pannelli solari trasparenti, schermi, batterie e memorie RAM.

Lo spunto principale per questa tesi sperimentale nasce da un algoritmo proposto da V.~Romano, A.~Majorana e M.~Coco nel 2015 e descritto in \citep{ART:RoMaCo_JCP}, in cui si introduce un nuovo approccio al metodo di Simulazione Diretta Monte Carlo dove è pienamente rispettato il principio di esclusione di Pauli. Questo algoritmo è stato anche successivamente esteso al caso del grafene su substrato (cfr.~\citep{ART:CoMajRo_Ric_mat_2016}). Si è approfondita anche l'analisi degli effetti termici (cfr.~\citep{ART:CoMaRo_JCompTheorTrans}) e lo studio della mobilità (cfr.~\citep{ART:MajMaRo}) nel caso del grafene sospeso.

In questa tesi si svilupperanno tre tematiche principali. La prima sarà un approfondimento teorico sul trasporto di cariche nei semiconduttori e sulla sua applicazione al caso del grafene. La seconda tematica riguarderà la risoluzione delle equazioni di trasporto mediante la simulazione Monte Carlo. La terza costituirà la parte sperimentale e consisterà in primo luogo nel descrivere e mostrarne i risultati di un codice equivalente che lavora con tempi di esecuzione inferiori per le simulazioni, in secondo luogo si interverrà a livello modellistico nel caso del grafene su substrato, introducendo modelli più sofisticati per la scelta delle distanze delle impurezze. Infine si eseguiranno alcune simulazioni in cui saranno utilizzati substrati di diversa natura.

Uno degli obiettivi di questa tesi, come si è detto, consiste nel fornire una descrizione del modello semiclassico per il trasporto di cariche nei semiconduttori. Questo richiede alcuni richiami di meccanica classica e fisica atomica, esposti nel Capitolo~1. In tale capitolo inoltre si parlerà anche di alcuni aspetti quantistici nella descrizione atomica e delle proprietà degli operatori lineari negli spazi di Hilbert.

Nel Capitolo~2 si tratterà del modello semiclassico, il quale prevede che il trasporto di cariche avvenga supponendo che il moto tra due scattering successivi sia in accordo con la meccanica classica. Gli elementi necessari allo sviluppo di questo modello riguardano innanzi tutto alcune nozioni di meccanica quantistica, essenziali quando si descriveranno i fenomeni di scattering. Inoltre, risulterà fondamentale fornire anche elementi di struttura della materia allo stato solido. Tutto ciò convergerà nella deduzione dell'equazione di Boltzmann semiclassica e del relativo termine collisionale che descrive i fenomeni di scattering.

La simulazione Monte Carlo costituisce una delle tecniche più utilizzate per risolvere le equazioni appena menzionate. Nel Capitolo~3 si applicherà il modello cinetico di trasporto al caso del grafene, risolvendone le equazioni mediante le tecniche Monte Carlo. L'approccio considerato è quello presentato in \citep{ART:RoMaCo_JCP} in cui si utilizza un approccio al metodo di Simulazione Diretta Monte Carlo che rispetta il principio di esclusione di Pauli anche durante il free-flight degli elettroni. Questo algoritmo è stato anche esteso al caso del grafene su substrato (cfr.~\citep{ART:CoMajRo_Ric_mat_2016}) e prima della stesura di questa tesi, come si è detto, sono stati sviluppati con successo dei codici più rapidi per effettuare le simulazioni, sia nel caso del grafene sospeso che nel caso del grafene su substrato. Si concluderà il capitolo con alcuni risultati originali ottenuti scegliendo sia substrati diversi, sia differenti modelli per la distanza delle impurezze. Tali risultati saranno pubblicati in \citep{ART:CoMajNaRo}.

\clearpage 


\mainmatter 

\pagestyle{fancy} 



\chapter{Concetti preliminari} 

\label{Chapter1} 

\lhead{Capitolo 1. \emph{Concetti preliminari}} 



\section{Richiami di meccanica classica}

In questa tesi spesso si farà riferimento ad alcuni concetti di base della meccanica classica che permettono la descrizione completa del moto di alcuni sistemi fisici. Risulta quindi opportuno richiamare le basi della meccanica lagrangiana ed hamiltoniana e introdurre la meccanica statistica.

Per trattare i successivi paragrafi saranno presi come riferimento \citep{DISP:Benettin} e \citep{DISP:Carati}.

\subsection{Le equazioni di Lagrange}
Si consideri un sistema di $N$ punti materiali $P_1,\ldots,P_N$ di massa $m_1,\ldots,m_N$ e si introduca, per le loro coordinate cartesiane, la notazione compatta
\begin{equation}
\mathbf{w}=(w_1,\ldots,w_{3N})=(x_1,y_1,z_1,\ldots,x_N,y_N,z_N)\in\mathbb{R}^{3N}.
\end{equation}
Se il sistema è non vincolato allora è possibile esprimere le coordinate degli $N$ punti in funzione di $3N$ parametri o \textbf{coordinate libere} $q_1,\ldots,q_{3N}$ in modo tale che la rappresentazione parametrica
\begin{equation}
w_j=w_j(q_1,\ldots,q_{3N},t), \qquad j=1,\ldots,3N.
\end{equation}
sia un cambiamento di coordinate rispetto a quelle cartesiane e il sistema degli $N$ punti materiali si possa esprimere nella forma
\begin{equation}\label{EQ:CAP1:Forma_parametrica_P_i}
P_i=P_i(q_1,\ldots,q_{3N},t), \qquad i=1,\ldots,N.
\end{equation}
Introducendo la notazione abbreviata
\begin{equation}
\mathbf{q}=(q_1,\ldots,q_{3N}), \qquad \dot{\mathbf{q}}=(\dot{q}_1,\ldots,\dot{q}_{3N})
\end{equation}
le grandezze $\dot{q}_1,\ldots,\dot{q}_{3N}$ prendono il nome di \textbf{velocità generalizzate}. Tramite le \eqref{EQ:CAP1:Forma_parametrica_P_i} è possibile calcolare direttamente le velocità $\mathbf{v}_i$ di ciascun punto materiale
\begin{equation}
\mathbf{v}_i(\mathbf{q},\dot{\mathbf{q}},t)=\frac{\diff P_i}{\diff t} = \sum_{h=1}^{3N} \frac{\partial P_i}{\partial q_h}(\mathbf{q},t)\dot{q}_h + \frac{\partial P_i}{\partial t}(\mathbf{q},t).
\end{equation}
L'energia cinetica espressa in termini delle coordinate libere avrà la forma
\begin{equation}
T(\mathbf{q},\dot{\mathbf{q}},t)=\frac{1}{2}\sum_{i=1}^N m_i\mathbf{v}_i^2(\mathbf{q},\dot{\mathbf{q}},t).
\end{equation}
Supponendo che gli $N$ punti materiali siano soggetti ad altrettante forze attive posizionali $\mathbf{F}_i$ si ottiene il sistema di equazioni di Newton
\begin{equation}
m_i\dot{\mathbf{a}}_i=\mathbf{F}_i, \qquad i=1,\ldots,N.
\end{equation}
Un risultato importante è il seguente.
\begin{proposizione}
Per un sistema non vincolato di $N$ punti materiali, con energia cinetica $T=T(\mathbf{q},\dot{\mathbf{q}},t)$, soggetto a forze conservative derivanti dall'energia potenziale $V=V(\mathbf{q},t)$, le coordinate libere $q_1,\ldots,q_{3N}$ soddisfano le equazioni
\begin{equation}\label{EQ:CAP1:Lagrange}
\frac{\diff}{\diff t} \frac{\partial \mathcal{L}}{\partial \dot{q}_h} - \frac{\partial \mathcal{L}}{\partial q_h}=0, \qquad h=1,\ldots,3N,
\end{equation}
dove $\mathcal{L}(\mathbf{q},\dot{\mathbf{q}},t)$ è definita da
\begin{equation}
\mathcal{L}=T-V
\end{equation}
(per la dimostrazione consultare \citep{DISP:Benettin}).
\end{proposizione}
Le equazioni \eqref{EQ:CAP1:Lagrange} prendono il nome di \textbf{equazioni di Lagrange} e la funzione $\mathcal{L}$ è detta \textbf{funzione di Lagrange} del sistema o \textbf{lagrangiana}.

\subsection{Le equazioni di Hamilton}\label{PAR:CAP1:Eq_Hamilton}
Si consideri il sistema di punti materiali introdotto nel paragrafo precedente e si ponga il problema di passare dal sistema delle $3N$ equazioni del secondo ordine di Lagrange ad un sistema di $6N$ equazioni del primo ordine. Si considerano le equazioni di Lagrange \eqref{EQ:CAP1:Lagrange} e si introducono le nuove variabili $p_1,\ldots,p_{3N}$, dette \textbf{momenti coniugati} alle coordinate generalizzate $q_1,\ldots,q_{3N}$, mediante la definizione
\begin{equation}\label{EQ:CAP1:Momenti}
p_h = \frac{\partial L}{\partial \dot{q}_h} (\mathbf{q},\dot{\mathbf{q}},t), \qquad h=1,\ldots,3N.
\end{equation}
Allora le equazioni di Lagrange \eqref{EQ:CAP1:Lagrange} assumono la forma
\begin{equation}\label{EQ:CAP1:Lagrange_mod}
\dot{p}_h = \frac{\partial L}{\partial q_h} (\mathbf{q},\dot{\mathbf{q}},t), \qquad h=1,\ldots,3N.
\end{equation}
\`{E} possibile vedere le \eqref{EQ:CAP1:Momenti} come un sistema implicito di $3N$ equazioni per le $\dot{q}_h$. Inoltre questo sistema è invertibile rispetto a $\dot{q}_1,\ldots,\dot{q}_{3N}$ se la matrice $\left( \frac{\partial^2 L}{\partial \dot{q}_h \partial \dot{q}_k} \right)$ ha determinante non nullo. In tal caso è possibile esplicitare le $\dot{q}_h$ ottenendo le $n$ equazioni
\begin{equation}\label{EQ:CAP1:Momenti_implicite}
\dot{q}_h = \dot{q}_h (\mathbf{p},\mathbf{q},t), \qquad h=1,\ldots,{3N}.
\end{equation}
Le altre $3N$ equazioni sono le \eqref{EQ:CAP1:Lagrange_mod} che utilizzando le \eqref{EQ:CAP1:Momenti_implicite} assumono la forma
\begin{equation}
\dot{p}_h = \dot{p}_h (\mathbf{p},\mathbf{q},t), \qquad h=1,\ldots,{3N}.
\end{equation}
Le $6N$ equazioni così ottenute hanno una struttura particolare. Allo scopo di semplificare le scritture seguenti, sia $\mathbf{w}=(w_1,\ldots,w_{3N})\in\mathbb{R}^{3N}$ e si introduce la notazione compatta
\begin{equation}
\nabla_{\mathbf{w}}=\left( \frac{\partial}{\partial w_1},\ldots,\frac{\partial}{\partial w_{3N}} \right).
\end{equation}
A questo punto è possibile determinare il sistema di $6N$ equazioni cercato grazie alla seguente 
\begin{proposizione}
Si consideri la lagrangiana $\mathcal{L}(\mathbf{q},\dot{\mathbf{q}},t)$ associata al sistema di punti materiali $P_1,\ldots,P_N$. Se il determinante hessiano è non nullo, cioè $\det \left( \frac{\partial^2 L}{\partial \dot{q}_h \partial \dot{q}_k} \right) \neq 0$, allora il sistema
\begin{equation}\label{EQ:CAP1:Sistema_Lagrange}
\left\lbrace 
\begin{aligned}
\mathbf{p} & = \nabla_{\dot{\mathbf{q}}}\mathcal{L} \\
\dot{\mathbf{p}} & = \nabla_{\mathbf{q}}\mathcal{L}
\end{aligned}
\right. 
\end{equation}
è equivalente al sistema
\begin{equation}\label{EQ:CAP1:Sistema_Hamilton}
\left\lbrace 
\begin{aligned}
\dot{\mathbf{q}} & = \nabla_{\mathbf{p}}\mathcal{H} \\
\dot{\mathbf{p}} & = -\nabla_{\mathbf{q}}\mathcal{H}
\end{aligned}
\right.
\end{equation}
dove $\mathcal{H}(\mathbf{q},\mathbf{p},t)$ è la funzione definita da
\begin{equation}\label{EQ:CAP1:Hamiltoniana}
\mathcal{H}(\mathbf{q},\mathbf{p},t) = \left[ \mathbf{p} \cdot \dot{\mathbf{q}} - \mathcal{L}(\mathbf{q},\dot{\mathbf{q}},t) \right] _{\dot{\mathbf{q}}=\dot{\mathbf{q}}(\mathbf{q},\mathbf{p},t)},
\end{equation}
essendo la funzione $\dot{\mathbf{q}}(\mathbf{q},\mathbf{p},t)$ definita per inversione dalla prima delle \eqref{EQ:CAP1:Sistema_Lagrange}.\\
Inoltre si ha $\frac{\partial \mathcal{H}}{\partial t} = - \frac{\partial \mathcal{L}}{\partial t}$.
\end{proposizione}
\begin{proof}
Per il fatto che il determinante hessiano è non nullo allora l'inversione della prima delle  \eqref{EQ:CAP1:Sistema_Lagrange} è possibile e dunque la funzione $\mathcal{H}$ è ben definita. Differenziando il secondo membro della \eqref{EQ:CAP1:Hamiltoniana}, si ha
\begin{equation}
\diff \mathcal{H} = \mathbf{p} \cdot \diff \dot{\mathbf{q}} + \dot{\mathbf{q}} \cdot \diff \mathbf{p} - \nabla_{\mathbf{q}}\mathcal{L} \cdot \diff \mathbf{q} - \nabla_{\dot{\mathbf{q}}}\mathcal{L} \cdot \diff \dot{\mathbf{q}} - \frac{\partial \mathcal{L}}{\partial t} \diff t,
\end{equation}
in cui, per la prima delle \eqref{EQ:CAP1:Sistema_Lagrange}, si ha che i termini che moltiplicano $\diff \dot{\mathbf{q}}$ si cancellano. Inoltre differenziando il primo membro si ottiene
\begin{equation}
\diff \mathcal{H} = \nabla_{\mathbf{q}}\mathcal{H} \cdot \diff \mathbf{q} + \nabla_{\mathbf{p}}\mathcal{H} \cdot \diff \mathbf{p} + \frac{\partial \mathcal{H}}{\partial t} \diff t ,
\end{equation}
Allora per confronto si ha che
\begin{equation}
\dot{\mathbf{q}} = \nabla_{\mathbf{p}}\mathcal{H}, \qquad \nabla_{\mathbf{q}}\mathcal{L} = -\nabla_{\mathbf{q}}\mathcal{H}, \qquad \frac{\partial \mathcal{H}}{\partial t} = - \frac{\partial \mathcal{L}}{\partial t}.
\end{equation}
\end{proof}
La funzione $\mathcal{H}$ si dice \textbf{funzione di Hamilton} o \textbf{hamiltoniana} del sistema e le equazioni \eqref{EQ:CAP1:Sistema_Hamilton} si dicono \textbf{equazioni canoniche} o \textbf{di Hamilton}. In definitiva, nel caso in cui il determinante hessiano è non nullo si ha l'equivalenza tra le equazioni di Lagrange e quelle di Hamilton.

\subsection{Stato di un sistema classico}
Si consideri il sistema di coordinate introdotto nel Paragrafo \ref{PAR:CAP1:Eq_Hamilton}, si ponga $\mathbf{x}=(\mathbf{q},\mathbf{p})$, cioè il vettore di tutte le coordinate nello spazio delle fasi e si consideri il campo vettoriale $\mathbf{v} = \mathbf{v}(\mathbf{x})$, scaturito dalle \eqref{EQ:CAP1:Sistema_Hamilton} e definito da
\begin{equation}
\mathbf{v}(\mathbf{q},\mathbf{p}) = \left( \nabla_{\mathbf{p}}\mathcal{H} , -\nabla_{\mathbf{q}}\mathcal{H}\right).
\end{equation}
Allora le equazioni di Hamilton assumono la forma
\begin{equation}\label{EQ:CAP1:Equazioni_moto}
\dot{\mathbf{x}}=\mathbf{v}(\mathbf{x}),
\end{equation}
che è un sistema di $6N$ equazioni differenziali ordinarie del primo ordine. Assegnando la condizione iniziale $\mathbf{x}(t_0)=\mathbf{x}_0$ si ha che, sotto sufficienti ipotesi di regolarità, il sistema ammette una e una sola soluzione $\mathbf{x}(t)=(\mathbf{q}(t),\mathbf{p}(t))$. Inoltre si dice che $\mathbf{x}(t)$ è lo \textbf{stato} del sistema al tempo $t$.

\subsection{Lo spazio delle fasi}
Considerando il sistema di $N$ punti materiali non vincolati, il formalismo lagrangiano introduce uno spazio a $6N$ dimensioni, con coordinate locali $q_1,\ldots,q_{3N},\dot{q}_1,\ldots,\dot{q}_{3N}$, che è chiamato \textbf{spazio degli stati}. L'introduzione del formalismo hamiltoniano porta a considerare un secondo spazio a $6N$ dimensioni, con coordinate locali $q_1,\ldots,q_{3N},p_1,\ldots,p_{3N},$, che è detto \textbf{spazio delle fasi} e normalmente si indica con $\Gamma$.

L'utilità delle coordinate hamiltoniane rispetto a quelle lagrangiane sta nel fatto che le equazioni di Hamilton hanno una struttura simmetrica. Questo permette una riscrittura compatta delle equazioni di tali equazioni. Infatti, si considerino le equazioni di Hamilton nella forma data da \eqref{EQ:CAP1:Equazioni_moto}.
Allora è possibile ottenere il campo $\mathbf{v}$ nel modo seguente. Si consideri l'hamiltoniana $\mathcal{H}=\mathcal{H}(\mathbf{q},\mathbf{p},t)$ e si calcoli il gradiente
\begin{equation}
\nabla \mathcal{H} = \left( \frac{\partial \mathcal{H}}{\partial q_1},\ldots,\frac{\partial \mathcal{H}}{\partial q_{3N}},\frac{\partial \mathcal{H}}{\partial p_1},\ldots,\frac{\partial \mathcal{H}}{\partial p_{3N}} \right).
\end{equation}
Si vuole scambiare l'insieme delle prime $3N$ componenti con quello delle ultime $3N$ e introdurre un cambiamento di segno per uno dei due insiemi. Tutto ciò può essere ottenuto mediante la matrice $\mathbb{E}$, detta \textbf{matrice simplettica standard}, definita da
\begin{equation}
\mathbb{E} = \left(
\begin{array}{cc}
0_{3N} & \mathbb{I}_{3N} \\ 
-\mathbb{I}_{3N} & 0_{3N}
\end{array}
\right),
\end{equation}
dove $0_{3N}$ e $\mathbb{I}_{3N}$ sono rispettivamente la matrice nulla e la matrice identità in $\mathbb{R}^{3N}$. In definitiva, il campo vettoriale $\mathbf{v}= \mathbf{v} (\mathbf{x})$ che definisce la dinamica nello spazio delle fasi ha la struttura
\begin{equation}\label{EQ:CAP1:Struttura_v}
\mathbf{v} (\mathbf{x}) = \mathbb{E} \cdot \nabla \mathcal{H},
\end{equation}
ovvero in componenti
\begin{equation}
v_j = \sum_{k=1}^{3N} \mathbb{E}_{jk} \frac{\partial \mathcal{H}(\mathbf{x})}{\partial x_k}.
\end{equation}
Inoltre il campo vettoriale $\mathbf{v}= \mathbf{v} (\mathbf{x})$ è solenoidale, cioè soddisfa la condizione
\begin{equation}\label{EQ:CAP1:Cond_solenoidale}
\nabla \cdot \mathbf{v} =0
\end{equation}
Infatti
\begin{align*}
\nabla \cdot \mathbf{v} & = \sum_{j=1}^{6N} \frac{\partial v_j}{\partial x_j} =
\sum_{j=1}^{6N} \frac{\partial}{\partial x_j} \left( \sum_{k=1}^{6N} \mathbb{E}_{jk} \frac{\partial \mathcal{H}}{\partial x_k} \right) =
\sum_{j=1}^{6N} \sum_{k=1}^{6N} \mathbb{E}_{jk} \frac{\partial^2 \mathcal{H}}{\partial x_j \partial x_k}=\\
& = \sum_{j=1}^{3N} \sum_{k=3N+1}^{6N} 1 \cdot \frac{\partial^2 \mathcal{H}}{\partial x_j \partial x_k} + \sum_{j=3N+1}^{6N} \sum_{k=1}^{3N} (-1) \cdot \frac{\partial^2 \mathcal{H}}{\partial x_j \partial x_k} =\\
& = \sum_{j=1}^{3N} \sum_{\tilde{k}=1}^{3N} 1 \cdot \frac{\partial^2 \mathcal{H}}{\partial q_j \partial p_{\tilde{k}}} + \sum_{\tilde{j}=1}^{3N} \sum_{k=1}^{3N} (-1) \cdot \frac{\partial^2 \mathcal{H}}{\partial p_{\tilde{j}} \partial q_k} =\\
& = \sum_{j=1}^{3N} \sum_{k=1}^{3N} \left( \frac{\partial^2 \mathcal{H}}{\partial q_j \partial p_k} - \frac{\partial^2 \mathcal{H}}{\partial p_j \partial q_k} \right) =0,
\end{align*}
in cui l'ultima uguaglianza è ottenuta grazie al teorema di Schwartz.

\subsection{Le variabili dinamiche}
Le funzioni definite sullo spazio delle fasi sono chiamate \textbf{variabili dinamiche}.

Inoltre è possibile considerare variabili dinamiche dipendenti esplicitamente dal tempo, cioè funzioni della forma $f=f(\mathbf{x},t)$, con $\mathbf{x}=(\mathbf{q},\mathbf{p})$ coordinate nello spazio delle fasi. Infatti, sia data una variabile dinamica $f(\mathbf{x},t)$ e un moto $\mathbf{x}=\mathbf{x}(t)$, soluzione delle equazioni di moto \eqref{EQ:CAP1:Equazioni_moto} in corrispondenza del dato iniziale $\mathbf{x}(t_0)=\mathbf{x}_0$. Per la formula di derivazione delle funzioni composte, si ha
\begin{equation}\label{EQ:CAP1:Derivata_composta}
\dot{f} = \frac{\partial f}{\partial t} + \mathbf{v} \cdot \nabla f.
\end{equation}
Si ha allora la seguente
\begin{definizione}
Una variabile dinamica $f(\mathbf{q},\mathbf{p},t)$ è detta \textbf{costante del moto} rispetto ad una hamiltoniana $\mathcal{H}$ se
\begin{equation}\label{EQ:CAP1:Costante_moto}
\dot{f}=0
\end{equation}
per ogni soluzione delle equazioni di Hamilton con hamiltoniana $\mathcal{H}$ (ovvero per ogni dato iniziale $\mathbf{x}_0$).
\end{definizione}
Dalla \eqref{EQ:CAP1:Derivata_composta} e dalla \eqref{EQ:CAP1:Struttura_v} si ha il seguente
\begin{teorema}\label{TEOR:CAP1:Costante_moto_ver1}
Condizione necessaria e sufficiente affinché $f$ sia costante del moto sotto la dinamica indotta da una hamiltoniana $\mathcal{H}$ è che si abbia
\begin{equation}
\dot{f} = \frac{\partial f}{\partial t} + \mathbf{v} \cdot \nabla f \qquad \mbox{dove} \qquad \mathbf{v} (\mathbf{x}) = \mathbb{E} \cdot \nabla \mathcal{H}.
\end{equation}
\end{teorema}
\`{E} utile riscrivere esplicitamente la condizione \eqref{EQ:CAP1:Costante_moto} in termini delle coordinate $(\mathbf{q},\mathbf{p})$, ottenendo
\begin{equation}
\dot{f} = \frac{\partial f}{\partial t} + \sum_{i=1}^{3N} \left( \frac{\partial f}{\partial q_i} \dot{q}_i + \frac{\partial f}{\partial p_i} \dot{p}_i \right) = \frac{\partial f}{\partial t} + \sum_{i=1}^{3N} \left( \frac{\partial f}{\partial q_i} \frac{\partial \mathcal{H}}{\partial p_i} - \frac{\partial f}{\partial p_i} \frac{\partial \mathcal{H}}{\partial q_i} \right)
\end{equation}
\begin{definizione}
Per ogni coppia ordinata di variabili dinamiche $f$ e $g$ risulta definita una nuova variabile dinamica, chiamata \textbf{parentesi di Poisson} di $f$ e $g$ e denotata $\left\lbrace f,g \right\rbrace $, mediante la relazione
\begin{equation}
\left\lbrace f,g \right\rbrace = \nabla_{\mathbf{q}}f \cdot \nabla_{\mathbf{p}}g - \nabla_{\mathbf{p}}f \cdot \nabla_{\mathbf{q}}g \equiv \sum_{i=1}^{3N} \left( \frac{\partial f}{\partial q_i} \frac{\partial g}{\partial p_i} - \frac{\partial f}{\partial p_i} \frac{\partial g}{\partial q_i} \right).
\end{equation}
\end{definizione}
Le parentesi di Poisson soddisfano le quattro proprietà racchiuse nella seguente
\begin{proposizione}\label{PROP:CAP1:Parentesi_Poisson}
Siano $f$,$g$,$g_1$,$g_2$ e $h$ delle variabili dinamiche e $\alpha_1,\alpha_2 \in \mathbb{R}$. Sono valide le seguenti proprietà.
\begin{enumerate}
\item $\left\lbrace f,g \right\rbrace = - \left\lbrace g,f \right\rbrace$
\item $\left\lbrace f,\alpha_1 g_1 + \alpha_2 g_2 \right\rbrace = \alpha_1 \left\lbrace f,g_1 \right\rbrace + \alpha_2 \left\lbrace f,g_2 \right\rbrace$ \label{PTA:CAP1:Parentesi_Poisson_pta_2}
\item $\left\lbrace f,g_1 g_2 \right\rbrace = \left\lbrace f,g_1 \right\rbrace g_2 + g_1 \left\lbrace f,g_2 \right\rbrace$
\item $\left\lbrace f, \left\lbrace g,h \right\rbrace \right\rbrace + \left\lbrace g, \left\lbrace h,f \right\rbrace \right\rbrace + \left\lbrace h, \left\lbrace f,g \right\rbrace \right\rbrace =0$
\end{enumerate}
(per ulteriori dettagli consultare \citep{DISP:Carati}).
\end{proposizione}
\begin{osservazione}
Grazie a questa notazione la \eqref{EQ:CAP1:Derivata_composta} diventa
\begin{equation}
\dot{f} = \frac{\partial f}{\partial t} + \left\lbrace f,\mathcal{H} \right\rbrace
\end{equation}
e quindi per la condizione \eqref{EQ:CAP1:Costante_moto} si ha
\begin{equation}
\frac{\partial f}{\partial t} + \left\lbrace f,\mathcal{H} \right\rbrace =0.
\end{equation}
\end{osservazione}
Utilizzando le parentesi di Poisson, il Teorema \ref{TEOR:CAP1:Costante_moto_ver1} si può enunciare nel seguente modo equivalente.
\begin{teorema}\label{TEOR:CAP1:Costante_moto_ver2}
Una variabile dinamica $f$ è una costante del moto rispetto all'hamiltoniana $\mathcal{H}(\mathbf{q},\mathbf{p},t)$ se e solo se si ha
\begin{equation}
\frac{\partial f}{\partial t} + \left\lbrace f,\mathcal{H} \right\rbrace =0 \qquad \forall \mathbf{q},\mathbf{p},t.
\end{equation}
\end{teorema}

\subsection{La meccanica statistica}\label{PAR:CAP1:Meccanica_statistica}
Si consideri un sistema fisico costituito da un grande numero di punti materiali. In linea di principio è possibile conoscere lo stato del sistema al variare del tempo, che da qui in avanti sarà detto \textbf{stato microscopico}, in modo deterministico. Questo si ottiene assegnando l'hamiltoniana e lo stato iniziale, come già discusso nel Paragrafo . Tuttavia questa tecnica è impossibile da attuare nella pratica per l'enorme numero di equazioni differenziali che si dovrebbero risolvere. Si può passare allora ad una descrizione macroscopica, studiando solamente le caratteristiche globali del sistema e non quelle delle singole particelle. Le grandezze fisiche che entrano in gioco in questa descrizione sono raccolte in un tipo di stato detto \textbf{stato macroscopico}.

A titolo di esempio si consideri il problema della descrizione di una mole di una certa tipologia di gas. Per quanto riguarda la descrizione classica, cioè microscopica, assegnando l'hamiltoniana del sistema e lo stato iniziale di ciascuna particella del gas, è possibile conoscere le traiettorie di tutte le particelle che lo costituiscono. Inoltre è noto che una mole di una qualunque sostanza contiene un numero di particelle dell'ordine del numero di Avogadro,
\begin{equation}
N_A \simeq \si{\num{6.022140857(74)e23}\ \mole^{-1}}.
\end{equation}
Di conseguenza, nella pratica occorrerebbe assegnare e risolvere un numero spropositato di equazioni differenziali ordinarie, cioè dell'ordine di $\si{\num{e23}}$. Si passa allora alla descrizione termodinamica, cioè macroscopica, in cui lo stato del gas è definito assegnando due sole variabili ovvero la temperatura e il volume in cui è racchiuso.

La \textbf{meccanica statistica} nasce con lo scopo di creare una via di mezzo tra le due nozioni di stato, introducendo concetti di tipo probabilistico o statistico che porteranno alla definizione di una nuova tipologia di stato. \`{E} opportuno cominciare considerando un sistema meccanico hamiltoniano definito su una varietà dello spazio delle fasi $\Gamma$ (dove la varietà può essere anche $\mathbb{R}^{6N}$) sulla quale sia assegnata un'equazione differenziale in forma vettoriale
\begin{equation}
\dot{\mathbf{x}}=\mathbf{v}(\mathbf{x}),
\end{equation}
in cui $\mathbf{x}=(\mathbf{q},\mathbf{p})$ sono le coordinate hamiltoniane e $\mathbf{v}=\mathbf{v}(\mathbf{x})$ è un campo vettoriale assegnato sulla varietà che presenta la struttura hamiltoniana
\begin{equation}\label{EQ:CAP1:Struttura_v_stat}
\mathbf{v} = \mathbb{E} \cdot \nabla \mathcal{H},
\end{equation}
con una fissata hamiltoniana $\mathcal{H}$. 

Osservando che uno stesso stato macroscopico può essere rappresentato da un gran numero di stati microscopici, si considerino $\mathcal{N}$ copie di un particolare stato macroscopico, ciascuna descritta da un diverso stato microscopico, ovvero da un diverso punto rappresentativo dello spazio delle fasi. Sia $\diff \mathcal{N}(\mathbf{x},t)$ il numero di punti rappresentativi in un volume infinitesimo $\diff V = \prod_{i=1}^{3N} \diff q_i \diff p_i$ costruito attorno al punto $\mathbf{x}=(\mathbf{q},\mathbf{p})$. Allora è possibile definire la densità dello spazio delle fasi ponendo
\begin{equation}
\rho (\mathbf{x},t) \diff V = \lim_{\mathcal{N}\to +\infty} \frac{\diff \mathcal{N}(\mathbf{x},t)}{\mathcal{N}}.
\end{equation}
Si supponga che $\rho (\mathbf{x},t)$ sia, per ogni $t$, una densità di probabilità, cioè che siano soddisfatte le proprietà
\begin{equation}\label{EQ:CAP1:Proprieta_densita_mecc_stat}
\rho(\mathbf{x},t) \geq 0 \quad\forall \mathbf{x}\in\Gamma, \, \forall t, \qquad \int_{\Gamma} \! \rho(\mathbf{x},t) \,\diff V =1 \quad\forall t,
\end{equation}
e imponendo che la probabilità che il sistema si trovi al tempo $t$ nel dominio $D(t)$ sia data da
\begin{equation}\label{EQ:CAP1:Prob_mecc_stat}
p \left( \mathbf{x}\in D(t) \right) =\int_{D(t)} \! \rho(\mathbf{x},t) \,\diff V.
\end{equation}
Allora si dice che $\rho(\mathbf{x},t)$ rappresenta lo \textbf{stato statistico} del sistema al tempo $t$.

Una importante definizione che sarà utile nel seguito è la seguente. Sia $f=f(\mathbf{x},t)$ una variabile dinamica. Si definisce \textbf{valore medio} di $f$ la quantità
\begin{equation}
< f > = \int \! \rho(\mathbf{x},t)f(\mathbf{x},t) \,\diff V,
\end{equation}
che si può vedere come la generalizzazione della media pesata di tutti i possibili stati microscopici al tempo $t$. Inoltre, a partire dal valore medio si potrà calcolare la varianza e la distribuzione di probabilità di $f$ come usualmente accade nella teoria della probabilità.

Si vuole ottenere un'equazione che descriva l'evoluzione della densità $\rho(\mathbf{x},t)$. Allo scopo si introduca una densità di probabilità per i dati iniziali del sistema, cioè si assegni, per $t=t_0$, una funzione $\rho_0(\mathbf{x})$ soddisfacente le proprietà \eqref{EQ:CAP1:Proprieta_densita_mecc_stat} e \eqref{EQ:CAP1:Prob_mecc_stat}. Allora la probabilità che lo stato del sistema si trovi inizialmente in un arbitrario dominio $D_0\subseteq\Gamma$ è data da
\begin{equation}
p \left( \mathbf{x}\in D_0 \right) =\int_{D_0} \! \rho_0(\mathbf{x}) \,\diff V.
\end{equation}
Si supponga che il sistema sia descritto al variare del tempo $t$ da uno stato statistico, cioè da una densità di probabilità $\rho (\mathbf{x},t)$. Data la condizione di raccordo $\rho(\mathbf{x},0)=\rho_0(\mathbf{x})$, dalla \eqref{EQ:CAP1:Prob_mecc_stat}, essendo $D(t)$ l'evoluzione al tempo $t$ del dominio $D_0$, si ha la condizione di compatibilità
\begin{equation}\label{EQ:CAP1:Cond_comp}
\int_{D(t)} \! \rho(\mathbf{x},t) \,\diff V = \int_{D_0} \! \rho_0(\mathbf{x}) \,\diff V, \qquad \forall D_0, \, \forall t.
\end{equation}
Essendo il secondo membro della \eqref{EQ:CAP1:Cond_comp} una costante, fissati di volta in volta $D_0$ e $t$, si ha che la \eqref{EQ:CAP1:Cond_comp} è equivalente a
\begin{equation}
\frac{\diff}{\diff t} \int_{D(t)} \! \rho(\mathbf{x},t) \,\diff V = 0,
\end{equation}
cioè la densità $\rho$ è un invariante integrale. Inoltre tale condizione può essere espressa equivalentemente con una condizione di tipo differenziale. Cioè è possibile dimostrare (cfr.~\citep{DISP:Carati}) che la \eqref{EQ:CAP1:Cond_comp} implica che la densità deve soddisfare la seguente equazione differenziale
\begin{equation}
\frac{\partial \rho}{\partial t} + \nabla \cdot (\rho\mathbf{v}) =0,
\end{equation}
detta \textbf{equazione di continuità}. Inoltre dall'identità
\begin{equation}
\nabla \cdot (\rho \mathbf{v}) = \rho \nabla \cdot \mathbf{v} + \mathbf{v} \cdot \nabla \rho
\end{equation}
e dalla \eqref{EQ:CAP1:Struttura_v_stat} per cui si ha \eqref{EQ:CAP1:Cond_solenoidale}, segue che l'equazione di continuità assume la forma
\begin{equation}\label{EQ:CAP1:Liouville_2}
\frac{\partial \rho}{\partial t} + \mathbf{v} \cdot \nabla \rho =0,
\end{equation}
o equivalentemente
\begin{equation}\label{EQ:CAP1:Liouville}
\frac{\partial \rho}{\partial t} + \left\lbrace \rho,H \right\rbrace  =0.
\end{equation}
Grazie al Teorema \ref{TEOR:CAP1:Costante_moto_ver2} ciò è equivalente al fatto che la densità sia una costante del moto per l'hamiltoniana $\mathcal{H}$. Quando quest'ultima equazione è pensata in relazione agli stati statistici anziché alle variabili dinamiche prende il nome di \textbf{equazione di Liouville}.
\begin{osservazione}
Ritornando alle coordinate hamiltoniane, ricordando la posizione $\mathbf{x}=(\mathbf{q},\mathbf{p})$, si ha che $\mathbf{v}=(\dot{\mathbf{q}},\dot{\mathbf{p}})$ e di conseguenza l'equazione \eqref{EQ:CAP1:Liouville_2} si può scrivere
\begin{equation}\label{EQ:CAP1:Liouville_3}
\frac{\partial \rho}{\partial t} + (\mathbf{q},\mathbf{p}) \cdot \nabla \rho =0.
\end{equation}
Inoltre, ponendo
\begin{equation}
\nabla_{\mathbf{q}} \: \rho =\left( \frac{\partial \rho}{\partial q_1},\ldots,\frac{\partial \rho}{\partial q_n}\right) , \quad \nabla_{\mathbf{p}} \: \rho =\left( \frac{\partial \rho}{\partial p_1},\ldots,\frac{\partial \rho}{\partial p_n}\right)
\end{equation}
la \eqref{EQ:CAP1:Liouville_3} diventa
\begin{equation}\label{EQ:CAP1:Liouville_4}
\frac{\partial \rho}{\partial t} + \mathbf{q} \cdot \nabla_{\mathbf{q}} \: \rho + \mathbf{p} \cdot \nabla_{\mathbf{p}} \: \rho =0.
\end{equation}
\end{osservazione}


\section{Richiami di fisica atomica}
Tutta la materia è costituita da molecole che, a loro volta, sono composte da atomi. Il modello atomico più semplice consisteva nel considerare l'atomo come una piccola sfera rigida, indivisibile, avente diametro dell'ordine di $\si{\num{e-10}\ \metre}$, ed elettricamente neutra. Successivamente si cominciò a considerare l'atomo come formato da tre costituenti elementari: il \textbf{protone}, il \textbf{neutrone} e l'\textbf{elettrone}. Queste particelle sono elettricamente cariche (con diverse positività) e sono tenute insieme da forze di Coulomb. Oggi è noto che tali costituenti non sono elementari, ma sono a loro volta formati da particelle, tuttavia la visione appena descritta è sufficiente per gli scopi di questa tesi.

Nei successivi paragrafi saranno descritte alcune proprietà dei sistemi atomici, sarà illustrato il modello atomico classico e saranno descritti alcuni comportamenti che risultano anomali in questo contesto e che hanno avviato l'allontanamento dall'approccio classico per descrivere fenomeni che avvengono su scala atomica. 

Le informazioni dei paragrafi seguenti sono tratte da \citep{BOOK:French}, \citep{BOOK:Landau}, \citep{DISP:Anile} e \citep{BOOK:Mazzoldi}.

\subsection{La carica elementare}
Tra il 1831 e il 1834 M.~Faraday condusse una serie di esperimenti sull'elettrolisi. Tali esperimenti consistevano nel fare passare della corrente tra due elettrodi sospesi, immersi in varie soluzioni, e poi misurare la quantità di sostanze solide (in termini di massa) o gas (in termini di volume) liberati da ciascun elettrodo. Faraday scoprì che la quantità di ogni sostanza estratta è proporzionale alla quantità totale di carica elettrica passata attraverso la soluzione. In particolare, la quantità di carica trasportata da una mole (ovvero la quantità di sostanza che contiene lo stesso numero di molecole) è chiamata \textbf{costante di Faraday} e indicata con $F$, il cui valore è
\begin{equation}
F \simeq \si{\num{96485.33289(59)}\ \coulomb \ \mole^{-1}}.
\end{equation}
L'elettrone, come costituente di tutti i tipi di atomi, fu scoperto da J.J.~Thomson nel 1897. Grazie ad un esperimento dovuto a R.A.~Millikan si scoprì che la carica elettrica viene trasferita in multipli interi di una certa quantità di carica, detta \textbf{carica elementare} e denotata con $e$, il cui valore è
\begin{equation}
e \simeq \si{\num{1.6021766208(98)e-19}\ \coulomb}.
\end{equation}
Per questa ragione si dice che la carica elettrica è una grandezza quantizzata.

Inoltre, si supponga che $e$ sia, in valore assoluto, pari alla carica dell'elettrone e che essa sia la più piccola quantità di carica trasportata da un singolo atomo durante l'elettrolisi. Allora è possibile determinare il numero di atomi o di molecole in una mole di una certa sostanza calcolando il rapporto
\begin{equation}
N_A=\frac{F}{e}\simeq \si{\num{6.022140857(74)e23}\ \mole^{-1}},
\end{equation}
in cui la costante $N_A$ è il numero di Avogadro. Inoltre, la misura precisa del valore di $e$, insieme ad una relazione tra carica e massa dedotta da Thomson permise di calcolare il valore della \textbf{massa dell'elettrone}, indicata con $m_e$ e pari a
\begin{equation}
m_e \simeq \si{\num{9.10938356(11)e-31}\ \kilogram}.
\end{equation}

\subsection{Le linee spettrali}
Nella seconda metà del XIX secolo si affermò, attraverso una solida base sperimentale e supportata teoricamente dalle equazioni di Maxwell, la natura ondulatoria della radiazione elettromagnetica. Questo modello prevede che la radiazione elettromagnetica si propaghi sotto forma di onda, detta \textbf{onda elettromagnetica}. A ciascuna di essa si associa una frequenza $\nu$, una lunghezza d'onda $\lambda$ ed è noto che esse si propagano alla velocità $c$ pari a quella della luce nel vuoto, cioè vale
\begin{equation}
c = \si{\num{299792458}\ \metre \ \second^{-1}}.
\end{equation}
Queste tre grandezze sono legate dalla relazione $\nu = \frac{c}{\lambda}$. Inoltre, come per gli altri processi ondulatori, è possibile associare alla radiazione elettromagnetica una pulsazione $\omega=2\pi\nu$ e un vettore d'onda $\mathbf{k}$, avente modulo $\vert\mathbf{k}\vert=\frac{2\pi}{\lambda}$, come direzione e verso quelli di propagazione dell'onda.

\`{E} noto che le onde elettromagnetiche hanno origine dalla propagazione del campo elettrico e del campo magnetico. Esistono sorgenti macroscopiche di onde elettromagnetiche, come le antenne o le cariche elettriche accelerate, oppure esse possono avere origine anche da fenomeni interni all'atomo. Gli elettroni di singoli atomi liberi possono essere eccitati, cioè ricevere energia, attraverso vari meccanismi. Successivamente l'elettrone eccitato si diseccita ed emette energia sotto forma di radiazione elettromagnetica. Quindi è possibile costruire uno spettro d'emissione e si osserva che esso non è continuo, bensì è formato da righe, che prendono il nome di \textbf{linee spettrali}. A queste righe è possibile associare delle frequenze ben definite. Il primo a scoprire una regolarità tra le linee spettrali fu J.J.~Balmer nel 1885, mentre uno studio completo sullo spettro dell'idrogeno portò verso il 1890 alla formula empirica proposta da J.R.~Rydberg. Questo fatto gettò i presupposti per ipotizzare che l'energia radiante è quantizzata in unità, ciascuna delle quali sarà in seguito chiamata \textbf{fotone}.

\subsection{L'effetto fotoelettrico}
L'effetto fotoelettrico consiste nell'emissione di elettroni da parte di metalli dopo che essi vengono illuminati da un'onda elettromagnetica. Sperimentalmente si osserva che esiste un valore minimo della frequenza della luce usata, dipendente dal tipo di metallo, al di sotto del quale non sono emessi elettroni e inoltre al raggiungimento di tale soglia si ha immediatamente l'emissione. Questo comportamento è inspiegabile con la teoria classica, la quale prevede l'emissione di elettroni, indipendentemente dalla frequenza, dopo un tempo sufficiente a far assorbire agli elettroni l'energia necessaria per vincere l'energia di legame con gli atomi.

Nel 1905 A.~Einstein propose un modello per spiegare le proprietà dell'effetto fotoelettrico. L'idea di Einstein era che un'onda elettromagnetica monocromatica (cioè con frequenza $\nu$ definita) fosse in realtà costituita da un flusso di particelle elementari, dette quanti di luce e successivamente chiamate fotoni, ciascuno avente energia data dalla legge
\begin{equation}\label{EQ:CAP1:Rel_Planck_Einstein_nu}
E=h\nu,
\end{equation}
dove $h$ prende il nome di \textbf{costante di Planck} e vale
\begin{equation}
h = \si{\num{6.626070040(81)e-34}\ \joule \ \second}.
\end{equation}
Secondo l'idea di Einstein, il pacchetto di energia può essere o completamente assorbito oppure non assorbito affatto. Nel caso in cui l'energia del quanto sia uguale o superiore all'energia di legame dell'elettrone con il metallo $E_0$ (misurata sperimentalmente con l'effetto fotoelettrico), allora l'elettrone viene immediatamente espulso, trasformando in energia cinetica l'eventuale energia in eccesso. La soglia sulla frequenza osservata sperimentalmente si ottiene calcolando $\nu_0=\frac{E_0}{h}$. Tale ipotesi si rivelò in perfetto accordo con i dati sperimentali.

Si osservi che la \eqref{EQ:CAP1:Rel_Planck_Einstein_nu} prende il nome di \textbf{relazione di Planck-Einstein} e può essere espressa in termini di pulsazione scrivendo
\begin{equation}\label{EQ:CAP1:Rel_Planck_Einstein}
E=\hbar\omega,
\end{equation}
dove $\hbar$ è la \textbf{costante di Planck ridotta} che vale
\begin{equation}
\hbar=\frac{h}{2\pi} \simeq \si{\num{1.054571800(13)e-34}\ \joule \ \second}.
\end{equation}
Questa teoria pose le basi per introdurre nella fisica una nuova teoria, cioè quella che considera la luce dotata di una doppia natura, sia ondulatoria che corpuscolare.

\subsection{Il modello atomico di Rutherford}
Nel 1911 E.~Rutherford scoprì che tutta la carica positiva e la maggior parte della massa atomica è concentrata in uno spazio molto piccolo rispetto al volume totale occupato dall'atomo (meno di un decimo del raggio atomico). Questo spazio venne chiamato \textbf{nucleo} atomico. Esso si suppose costituito da un certo numero di protoni e neutroni. Inoltre egli ipotizzò che gli elettroni si muovessero attorno al nucleo, sotto l'azione elettrica attrattiva da esso esercitata, in accordo con la fisica classica. Questa descrizione contiene un paradosso. Infatti, la teoria di Maxwell afferma che particelle cariche accelerate devono emettere radiazione elettromagnetica e quindi perdere energia. Allora in particolare l'elettrone dovrebbe collassare sul nucleo, ma in realtà ciò non avviene.

Attualmente in chimica la composizione di un atomo è descritta da due numeri: il numero atomico $Z$ che indica il numero di protoni ed elettroni presenti nell'atomo; il numero di massa $A=Z+N$ che indica la somma tra il numero di protoni e il numero di neutroni che formano il nucleo dell'atomo. Secondo le ultime misure, la massa del protone è pari a 
\begin{equation}
m_p \simeq \si{\num{1.672621898(21)e-27}\ \kilogram},
\end{equation}
quella del neutrone è
\begin{equation}
m_n \simeq \si{\num{1.674927471(21)e-27}\ \kilogram}.
\end{equation}
Confrontando le masse delle particelle subatomiche si osserva quindi che la massa del protone è circa uguale a quella del neutrone, le quali sono circa $\si{\num{1836}}$ volte più grandi di quella dell'elettrone, motivo per il quale si considera la massa dell'atomo concentrata nel nucleo (che in effetti contiene oltre il $\si{\num{99.9}\percent}$ della massa dell'atomo). Riguardo alla carica si assume che la carica del protone sia positiva, pari a $+e$, la carica del neutrone sia nulla e quella dell'elettrone sia negativa, pari a $-e$. Infine, dal punto di vista delle dimensioni, quelle del protone e del neutrone sono dell'ordine di $\si{\num{e-15}}$, quella dell'elettrone è inferiore a $\si{\num{e-17}}$, motivo per il quale esso si può considerare puntiforme.

\subsection{Il modello atomico di Bohr-Sommerfeld}\label{PAR:CAP1:Modello_Bohr_Somm}
In accordo con la scoperta del nucleo, N.~Bohr propose nel 1913 il suo modello atomico sulla base dei seguenti due postulati.
\begin{enumerate}
\item Un atomo ha un numero di possibili ``stati stazionari''. In ognuno di questi stati gli elettroni compiono moti orbitali in accordo con le leggi della meccanica newtoniana, ma (contrariamente alle previsioni dell'elettromagnetismo classico) non irradiano fino a quando rimangono su orbite fisse.
\item Quando un atomo passa da uno stato stazionario ad un altro, corrispondente ad un cambiamento di orbita di uno degli elettroni dell'atomo, è emessa una radiazione sotto forma di fotone. L'energia del fotone è proprio la differenza di energia emessa tra lo stato iniziale e quello finale dell'atomo. La frequenza classica è collegata all'energia dalla relazione di Planck-Einstein \eqref{EQ:CAP1:Rel_Planck_Einstein}.
\end{enumerate}
Il concetto di stato stazionario consiste nell'ipotesi che esistano alcuni particolari stati, tra gli infiniti possibili, in cui gli elettroni possono muoversi senza emettere energia, cioè mantenendo costante la loro energia meccanica totale. Nel caso di un sistema costituito da un protone ed un elettrone (atomo di idrogeno), Bohr suppose che le orbite seguissero la cosiddetta regola di quantizzazione del momento angolare, in cui si suppone che il momento angolare dell'elettrone rispetto al nucleo debba essere un multiplo intero della quantità $\hbar=\frac{h}{2\pi}$, cioè
\begin{equation}\label{EQ:CAP1:Quant_momento_angolare}
m_e v r_n = n \hbar, \qquad n=1,2,3,\ldots
\end{equation}
dove $v$ è il modulo della velocità dell'elettrone e $r_n$ il raggio dell'orbita. La condizione classica di equilibrio si ottiene uguagliando in modulo la forza centripeta, dovuta alla rotazione dell'elettrone attorno al protone, alla forza di attrazione di Coulomb, cioè
\begin{equation}\label{EQ:CAP1:Cond_equilibrio_idrogeno}
m_e\frac{v^2}{r_n}=\frac{1}{4\pi\varepsilon_0}\frac{e^2}{r_n^2},
\end{equation}
in cui $\varepsilon_0$ è la costante dielettrica nel vuoto, che vale
\begin{equation}
\varepsilon_0 = \si{\num{8.85418781762e-12}\ \farad \metre^{-1}}.
\end{equation}
Utilizzando la \eqref{EQ:CAP1:Quant_momento_angolare} nella \eqref{EQ:CAP1:Cond_equilibrio_idrogeno}, si ottiene la seguente relazione per i raggi delle orbite
\begin{equation}
r_n = \frac{4\pi\varepsilon_0 \hbar^2}{m_e e^2}n^2 \simeq n^2 \ \si{\num{0.529e-10}\ \metre}.
\end{equation}
Inoltre è possibile calcolare l'energia totale corrispondente a ciascuna orbita come somma dell'energia cinetica dell'elettrone $E_k$ e dell'energia potenziale dell'elettrone nel campo del protone $U_e$, cioè
\begin{equation}
U_n=E_k+U_e = \frac{1}{2} \frac{e^2}{4\pi\varepsilon_0 r_n}-\frac{e^2}{4\pi\varepsilon_0 r_n} = -\frac{1}{2} \frac{e^2}{4\pi\varepsilon_0 r_n} \simeq -\frac{\si{\num{13.6}}}{n^2} \si{\electronvolt}.
\end{equation}
Anche l'energia è dunque quantizzata e si parla di livelli energetici dell'atomo di idrogeno. Si chiama \textbf{livello fondamentale} quello associato all'energia $U_1$, gli altri si chiamano \textbf{livelli eccitati}. Invece $n$ prende il nome di \textbf{numero quantico principale}.

Il modello di Bohr fu esteso da A.~Sommerfeld considerando la possibilità per l'elettrone di descrivere orbite ellittiche in cui il protone occupa uno dei fuochi. Inoltre egli impose che il momento angolare $\mathbf{L}$ dell'elettrone rispetto al protone e la sua proiezione lungo una determinata direzione fossero quantizzati in multipli interi di $\hbar$. In particolare, per ciascuna orbita, il modulo del momento angolare è dato da
\begin{equation}\label{EQ:CAP1:Quant_momento_angolare_L}
\vert\mathbf{L}\vert = (l+1)\hbar, \quad l=0,1,2,\ldots,n-1,
\end{equation}
in cui l'intero $l$ è chiamato \textbf{numero quantico azimutale}. Infine, per ciascuna orbita caratterizzata dal valore $l$, sono ammesse $2l+1$ inclinazioni del piano dell'orbita, in cui l'angolo $\alpha$ tra la normale al piano e una data direzione deve essere tale che
\begin{equation}\label{EQ:CAP1:Quant_proiezione_L}
\cos \alpha = \frac{m}{l+1}, \quad m=0,\pm 1, \pm 2, \ldots, \pm l,
\end{equation}
in cui l'intero $m$ si chiama \textbf{numero quantico magnetico}. Si osservi che ogni livello energetico risulta così suddiviso in ulteriori $2l+1$ livelli.

Il modello di Bohr-Sommerfeld spiegò molte caratteristiche dei fenomeni atomici osservati. Tuttavia questo modello non era coerente perché venivano applicati al moto dell'elettrone i risultati della meccanica classica ed introduceva opportunamente le ipotesi di quantizzazione senza una giustificazione teorica. Nacque così l'esigenza di un allontanamento radicale dai concetti della meccanica classica.

\subsection{La descrizione classica del moto di un elettrone libero}
Nella descrizione classica newtoniana, un elettrone è visto come un punto materiale dotato di massa $m_e$ e carica $-e$. Si tralasci il concetto di spin (il quale può essere spiegato solo attraverso la meccanica quantistica) e si consideri un elettrone isolato. \`{E} possibile associare all'elettrone un vettore posizione $\mathbf{x}\in\mathbb{R}^3$ e un vettore velocità $\mathbf{v}\in\mathbb{R}^3$, che in generale variano nel tempo $t$. In presenza di un campo elettrico $\mathbf{E}(\mathbf{x})$, sull'elettrone agirà una forza
\begin{equation}
\mathbf{F}=-e\mathbf{E},
\end{equation}
dove $q$ è la carica dell'elettrone in valore assoluto. Allora la seconda legge di Newton afferma che
\begin{equation}
m_e\dot{\mathbf{v}}=\mathbf{F}.
\end{equation}
Di conseguenza il moto dell'elettrone sarà descritto dalle equazioni
\begin{equation}
\left\lbrace
\begin{aligned}
\dot{\mathbf{x}}&=\mathbf{v}\\
m_e\dot{\mathbf{v}}&=-e\mathbf{E}
\end{aligned}
\right.
\end{equation}
Il moto è univocamente determinato assegnando lo stato del sistema all'istante iniziale $t_0$ tramite le condizioni
\begin{equation}
\left\lbrace
\begin{aligned}
\mathbf{x}(t_0)&=\mathbf{x}_0\\
\mathbf{v}(t_0)&=\mathbf{v}_0
\end{aligned}
\right.
\end{equation}
Poiché un campo elettrico stazionario proviene da un potenziale elettrico $V$, cioè
\begin{equation}
\mathbf{E}=-\nabla_{\mathbf{x}}V,
\end{equation}
allora l'energia totale $\mathcal{E}$ dell'elettrone può essere scritta come somma della sua energia cinetica e potenziale
\begin{equation}
\mathcal{E} = \frac{m_e\vert \mathbf{v}\vert^2}{2} -eV.
\end{equation}
Si osservi che ad ogni traiettoria è associato un valore dell'energia che si mantiene costante durante il moto. Infatti
\begin{equation}
\begin{aligned}
\dot{\mathcal{E}} & = m_e\mathbf{v}\cdot\dot{\mathbf{v}} -e\nabla_{\mathbf{x}}V\cdot\dot{\mathbf{x}}=\\
&=m_e\mathbf{v}\cdot\dot{\mathbf{v}} + e\mathbf{E}\cdot\mathbf{v}=\\
&=(m_e\dot{\mathbf{v}}-\mathbf{F})\cdot\mathbf{v}=0
\end{aligned}
\end{equation}
per cui, assegnando il valore dell'energia all'istante iniziale, si ha
\begin{equation}
\mathcal{E}(t)=\mathcal{E}(t_0)\equiv\frac{m_e\vert \mathbf{v}_0\vert^2}{2} -eV(\mathbf{x}_0) \qquad \forall t>t_0.
\end{equation}
Nella descrizione classica hamiltoniana si introduce la grandezza vettoriale $\mathbf{p}=m_e\mathbf{v}$ detta \textbf{quantità di moto} o \textbf{impulso}. Lo spazio di tutti gli stati $(\mathbf{x},\mathbf{p})\in\mathbb{R}^3\times\mathbb{R}^3$ è detto \textbf{spazio delle fasi}. Una particella elementare è rappresentata da un punto nello spazio delle fasi dotato di massa e carica. Si osservi che una traiettoria $\mathbf{x}=\mathbf{x}(t)$ nello spazio fisico individua una curva $(\mathbf{x}(t),m_e\dot{\mathbf{x}}(t))$ nello spazio delle fasi. Si è visto che ogni traiettoria determina un valore costante dell'energia totale della particella. Lo stesso rimane vero nello spazio delle fasi introducendo la funzione hamiltoniana
\begin{equation}\label{EQ:CAP1:Hamiltoniana_elettrone_libero}
\mathcal{H}(\mathbf{x},\mathbf{p})=\frac{\vert \mathbf{p}\vert^2}{2m_e}-eV(\mathbf{x}).
\end{equation}
Le equazioni classiche del moto di un elettrone saranno quindi
\begin{equation}
\left\lbrace
\begin{aligned}
\dot{\mathbf{x}} &= \nabla_{\mathbf{p}}\mathcal{H}\\
\dot{\mathbf{p}} &= -\nabla_{\mathbf{x}}\mathcal{H}
\end{aligned}
\right.
\end{equation}
Inoltre la funzione hamiltoniana è legata all'energia tramite la relazione
\begin{equation}\label{EQ:CAP1:Relazione_classica_energia}
\mathcal{E}(t)=\mathcal{H}(\mathbf{x}(t),\mathbf{p}(t))
\end{equation}
e, nel caso di forze conservative, l'energia è costante e non dipende dal tempo.


\section{Introduzione alla meccanica quantistica}
\`{E} noto che la fisica classica, rappresentata dalla meccanica newtoniana e dalle leggi di Maxwell dell'elettromagnetismo, opera molto bene in ambito macroscopico. Però quando si entra nel mondo atomico il solo utilizzo di queste teorie appare insufficiente per comprendere alcuni fenomeni che si presentano. L'insieme delle teorie fisiche che permettono di rimuovere tale mancanza prende il nome di \textbf{meccanica quantistica}.

\`{E} opportuno distinguere quando è necessario passare ad una descrizione quantistica. A tale scopo, si dice che un qualsiasi sistema fisico si comporta in modo quantistico quando l'\textbf{azione caratteristica} del sistema, cioè il numero ottenuto combinando opportunamente le dimensioni fisiche caratteristiche del sistema (Energia $\times$ tempo, o equivalentemente Quantità di moto $\times$ Lunghezza) è dell'ordine o inferiore al valore della costante di Planck $h$. La costante di Planck e il termine ``quantum'' vennero introdotti per la prima volta nel 1900 da M.~Planck nel suo lavoro sulla teoria del corpo nero.

Nei successivi paragrafi saranno illustrate alcune ulteriori caratteristiche dei sistemi atomici e alcuni concetti fondamentali della meccanica quantistica. Inoltre si introdurranno le onde materiali, il cui sviluppo permise di spiegare più accuratamente la dinamica dei sistemi atomici.

Per il contenuto dei successivi paragrafi si farà riferimento a \citep{BOOK:Mazzoldi}, \citep{BOOK:Landau}, \citep{BOOK:French} e \citep{BOOK:Dirac}.

\subsection{Il magnetismo atomico}\label{PAR:CAP1:Magnetismo_atomico}
Un ulteriore aspetto da considerare per completare la descrizione dell'atomo è lo studio delle proprietà magnetiche. Si consideri il modello di Bohr-Sommerfeld, illustrato nel Paragrafo \ref{PAR:CAP1:Modello_Bohr_Somm}. \`{E} noto dalla fisica classica (cfr.~\citep{BOOK:Mazzoldi}) che in un sistema di forze centrali, detto $\mathbf{L}$ il momento angolare del moto orbitale dell'elettrone e $\boldsymbol{\mu}$ il momento magnetico orbitale associato all'intensità di corrente scaturita dalla carica in moto dell'elettrone, si ha che $\boldsymbol{\mu}$ ed $\mathbf{L}$ sono legati dalla relazione
\begin{equation}
\boldsymbol{\mu} = - \frac{e}{2m_e}\mathbf{L}.
\end{equation}
Per il fatto che $\mathbf{L}$ è quantizzato secondo la relazione \eqref{EQ:CAP1:Quant_momento_angolare_L} allora anche il momento magnetico orbitale $\boldsymbol{\mu}$ è quantizzato secondo la relazione
\begin{equation}
\vert\boldsymbol{\mu}\vert = \frac{e\hbar}{2m_e}(l+1), \quad l=0,1,2,\ldots,n-1.
\end{equation}
Di conseguenza, per la \eqref{EQ:CAP1:Quant_proiezione_L}, anche la proiezione del momento magnetico orbitale nella direzione della normale al piano dell'orbita è quantizzata, sempre in unità $\frac{e\hbar}{2m_e}$.

Oltre al momento magnetico orbitale vi è un'altra causa di momento magnetico a livello atomico, che non ha equivalente in fisica classica. L'elettrone possiede un momento magnetico che appare proprio come una sua proprietà intrinseca ed è chiamato \textbf{spin}. Nel 1925 S.A.~Goudsmit e G.E.~Uhlenbeck postularono (cfr.~\citep{BOOK:French}) che lo spin di un elettrone è caratterizzato da un \textbf{numero quantico di spin} $s$ a due valori, cioè $s=\pm \hbar/2$, cui è associato un momento magnetico $\mathbf{S}$ di modulo pari a
\begin{equation}
\vert \mathbf{S} \vert = \pm \frac{\hbar}{2}.
\end{equation}
Seguendo questa teoria e per quanto osservato nel Paragrafo  \ref{PAR:CAP1:Modello_Bohr_Somm}, si ha che ogni livello energetico risulta dunque suddiviso in $2(2l+1)$ livelli. Si osservi infine che anche il protone ed il neutrone hanno uno spin e un momento magnetico, trascurabile nella trattazione del magnetismo atomico.

Questa teoria spiegò il risultato sperimentale dovuto a O.~Stern e W.~Gerlach, realizzato nel 1922, in cui si osservava che anche nel caso in cui l'elettrone si trova in un livello caratterizzato dal numero quantico azimutale $l=0$, come accade negli atomi di argento (\ce{Ag}), un fascio di atomi soggetti ad un campo magnetico disomogeneo si suddivide in più fasci.

\subsection{Principio di esclusione di Pauli}
I quattro numeri quantici $n$, $l$, $m$ ed $s$, descritti nei Paragrafi \ref{PAR:CAP1:Modello_Bohr_Somm} e \ref{PAR:CAP1:Magnetismo_atomico}, costituiscono lo \textbf{stato quantico} di un elettrone. Una legge fondamentale della meccanica quantistica che riguarda i numeri quantici è il cosiddetto \textbf{principio di esclusione di Pauli}, scoperto da W.~Pauli nel 1924. Esso afferma che in un atomo due elettroni non possono stare nello stesso stato quantico, cioè non possono avere tutti i numeri quantici uguali.

Di conseguenza, fissato $n$, è possibile ottenere il numero dei possibili stati quantici calcolando
\begin{equation}
\sum_{l=0}^{n-1} [2(2l+1)] = 2 \left[ 2\left( \sum_{l=0}^{n-1} l \right) +n \right] = 2 \left[ 2 \frac{(n-1)^2+(n-1)}{2} +n \right] = 2n^2.
\end{equation}

\subsection{Il modello nucleare a shell}\label{PAR:CAP1:Modello_shell}
Si definisce \textbf{shell}, una regione di spazio attorno al nucleo che può essere occupata da elettroni, corrispondente al numero quantico principale. Le shells definiscono la probabilità di trovare un elettrone in varie regioni di spazio relative al nucleo.

Le shells sono divise in \textbf{subshells} indicate con le lettere $s$, $p$, $d$ e $f$, ordinate per livelli di energia crescenti. Inoltre entro queste subshells, gli elettroni sono raggruppati in orbitali. Un \textbf{orbitale} è una regione dello spazio che contiene due elettroni ed ha uno specifico livello di energia quantizzato. La prima shell contiene un solo orbitale chiamato orbitale $1s$. La seconda shell contiene un orbitale $s$ e tre orbitali $p$ separati da angoli di $90^\circ$. Per questa ragione è possibile considerare gli orbitali $p$ orientati nelle direzioni degli assi $x$, $y$, e $z$. In definitiva, gli orbitali della seconda shell sono indicati con $2s$, $2p_x$, $2p_y$ e $2p_z$. Non saranno descritte le successive shells in quanto la configurazione elettronica degli elementi trattati in questa tesi si limita alla seconda shell.

La configurazione elettronica di un atomo è una descrizione degli orbitali occupati da ciascun elettrone. Ogni atomo ha un numero infinito di possibili configurazioni elettroniche. Per determinare quella a più bassa energia, detta \textbf{stato fondamentale}, si applicano le seguenti tre regole (cfr.~\cite{BOOK:Organic-Chemistry}).
\begin{enumerate}
\item (Principio dell'Aufbau) Gli orbitali si riempiono in ordine crescente di energia.
\item (Principio di esclusione di Pauli) Soltanto due elettroni possono occupare un orbitale ed i loro spin devono essere opposti.
\item (Regole di Hund) La prima parte afferma che, se orbitali di eguale energia sono disponibili ma non ci sono abbastanza elettroni tali da riempirli tutti completamente, allora un solo elettrone viene aggiunto a ciascuno di tali orbitali, prima di aggiungere un secondo elettrone ad uno qualsiasi di essi. La seconda parte afferma che, gli spin dei singoli elettroni negli orbitali di eguale energia devono essere allineati.
\end{enumerate}
Ricordando che gli elettroni hanno carica negativa, il riempimento parziale degli orbitali avviene in modo tale da minimizzare il più possibile la repulsione elettrostatica tra gli elettroni. Se gli elettroni sono disposti con una configurazione diversa allora la sua energia sarà più alta e si avrà uno \textbf{stato eccitato}.

Gli elettroni posizionati nelle shells più esterne sono coinvolti nella formazione dei legami chimici. Tali elettroni sono detti \textbf{elettroni di valenza} e il livello energetico in cui essi si trovano è detto \textbf{livello di valenza}. Gli elettroni che non partecipano ai legami chimici sono quelli che si trovano nelle shells più interne e sono detti elettroni di \textbf{core}. Inoltre gli atomi interagiscono in modo da riempire il livello di valenza nei modi seguenti.
\begin{enumerate}
\item Un atomo può diventare uno \textbf{ione} se perde o guadagna elettroni in modo da completare il livello più esterno. L'atomo risulterà carico positivamente se perde elettroni, carico negativamente se li guadagna. Ioni di cariche opposte si attraggono e tale attrazione genera i cosiddetti \textbf{legami ionici}.
\item Un atomo può condividere elettroni con uno o più atomi per riempire il livello più esterno. Questo tipo di legame è detto \textbf{legame covalente}.
\item I legami possono essere in parte ionici e in parte covalenti. Questi legami sono detti \textbf{legami covalenti polari}.
\end{enumerate}

\subsection{Gli orbitali atomici \texorpdfstring{$s$}{s} e \texorpdfstring{$p$}{p}}\label{PAR:CAP1:Orbitali_s_p}
Tutti gli orbitali di tipo $s$ hanno la forma di una sfera, centrata nel nucleo atomico. Una rappresentazione tridimensionale degli orbitali $1s$ e $2s$ è mostrata nella Figura \ref{FIG:CAP2:orbitali_s}.
\begin{figure}[h]
\centering
\includegraphics[width=0.5\columnwidth]{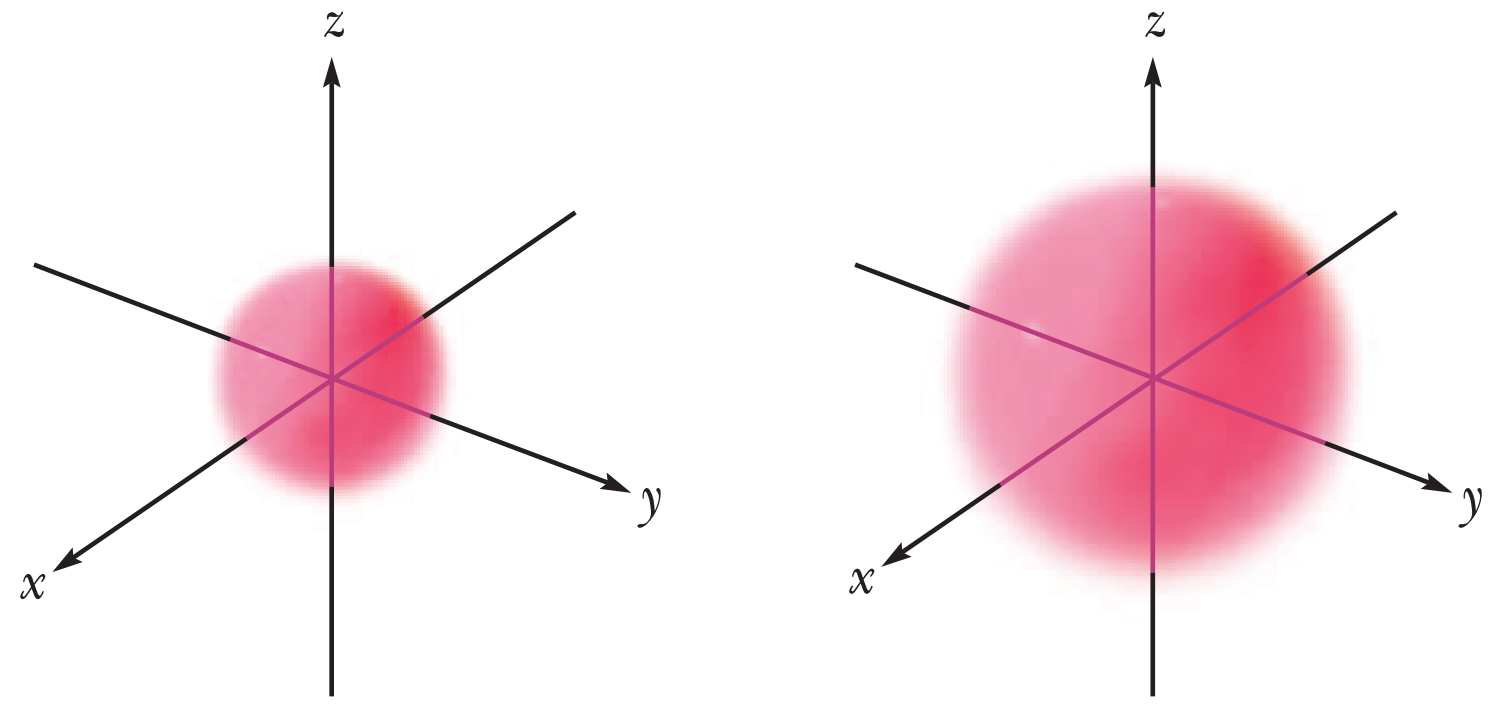}
\caption{Rappresentazione degli orbitali $1s$ (a sinistra) e $2s$ (a destra).}
\label{FIG:CAP2:orbitali_s}
\end{figure}
\FloatBarrier
Ogni orbitale $2p$ consiste di due lobi allineati con al centro il nucleo. Le direzioni dei tre orbitali di tipo $2p$ sono a due a due perpendicolari e sono indicati con $2p_x$, $2p_y$ e $2p_z$. Il segno della funzione d'onda è positivo in un lobo, zero nel nucleo e negativo nell'altro lobo. Questa rappresentazione è mostrata nella Figura \ref{FIG:CAP2:orbitali_p}.
\begin{figure}[h]
\centering
\includegraphics[width=\columnwidth]{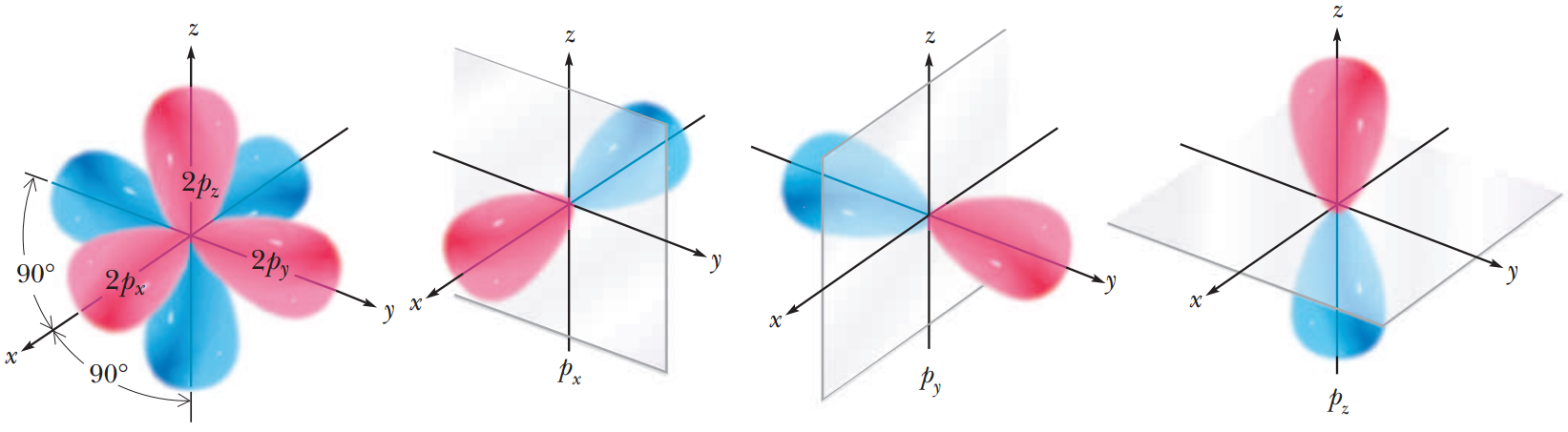}
\caption{Rappresentazione degli orbitali $2p$.}
\label{FIG:CAP2:orbitali_p}
\end{figure}
\FloatBarrier
Nei materiali gli atomi sono tenuti insieme da legami chimici, i quali sono creati dalla sovrapposizione degli orbitali di atomi adiacenti. Nel formare i legami covalenti, alcuni elementi come il carbonio (C), l'azoto (N) e l'ossigeno (O) usano gli orbitali atomici $2s$ e $2p$. Poiché gli orbitali $2p$ formano angoli di $90^\circ$ ci si aspetta che anche nei legami covalenti gli angoli tra gli atomi siano della stessa ampiezza. In realtà non è così e si trova che gli angoli nei legami covalenti sono di circa $109.5^\circ$ in molecole con un solo legame, di $120^\circ$ nelle molecole con due legami, di $180^\circ$ nelle molecole con tre legami. Secondo il modello di L.~Pauling gli orbitali atomici si possono combinare per formare nuovi orbitali atomici, detti \textbf{orbitali ibridi}, e il processo per formarli è detto \textbf{ibridizzazione}. Per gli orbitali ibridi si utilizza la seguente nomenclatura. Gli orbitali ibridi derivanti da un solo legame covalente tra due atomi sono detti $sp^3$, quelli derivanti da due legami sono detti $sp^2$ e infine quelli derivanti da tre legami sono detti $sp$.

Per quanto concerne questa tesi, si è interessati solamente agli orbitali ibridi di tipo $sp^2$, per cui adesso si tratteranno in dettaglio. Ciascuno di essi è dato dalla combinazione di un orbitale $2s$ e di due orbitali $2p$ e inoltre si trovano sempre a gruppi di tre (cfr.~\cite{BOOK:Organic-Chemistry}). Ciascun orbitale di tipo $sp^2$ consiste di due lobi, uno più grande ed uno più piccolo. Gli assi dei tre orbitali ibridi giacciono su un piano e sono diretti verso i vertici di un triangolo equilatero. In questo modo essi formano angoli di $120^\circ$. Il terzo orbitale di tipo $2p$ non è coinvolto nell'ibridizzazione e rimane con i suoi due lobi che si posizionano perpendicolarmente al piano degli orbitali ibridi $sp^2$. La Figura \ref{FIG:CAP2:orbitali_sp2} mostra quanto appena descritto. Infine, un atomo che possiede tre orbitali ibridi $sp^2$ e un solo orbitale di tipo $2p$ si dice essere ibridizzato $sp^2$.
\begin{figure}[h]
\centering
\includegraphics[width=0.8\columnwidth]{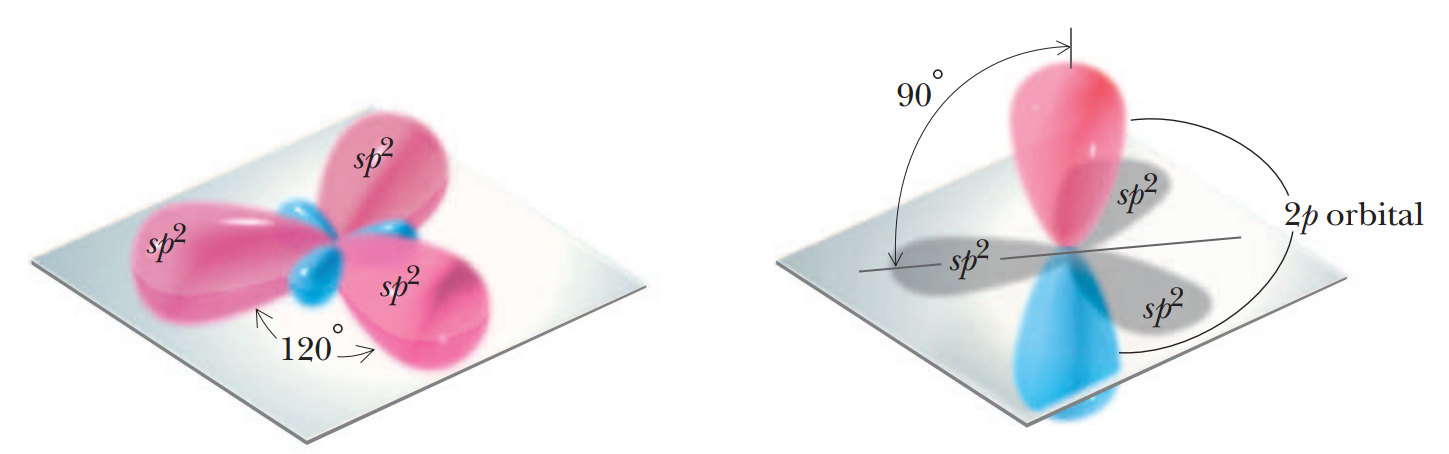}
\caption{Rappresentazione degli orbitali $sp^2$.}
\label{FIG:CAP2:orbitali_sp2}
\end{figure}
\FloatBarrier

\subsection{Le onde materiali}\label{PAR:CAP1:Onde_materiali}
Tra le onde luminose e i fotoni esiste una corrispondenza, nota come dualismo onda-corpuscolo. Infatti, un fotone è una particella che, pur essendo priva di massa, possiede un impulso $\mathbf{p}$ avente la stessa direzione e verso di quelli di propagazione dell'onda luminosa. Poiché dalla teoria della relatività ristretta vale la relazione $E=\vert \mathbf{p} \vert c$ (cfr.~\citep{DISP:Papa}) e utilizzando la \eqref{EQ:CAP1:Rel_Planck_Einstein_nu}, il modulo di $\mathbf{p}$ si ottiene calcolando
\begin{equation}
\vert\mathbf{p}\vert=\frac{E}{c}=\frac{h\nu}{c}=\frac{h}{\lambda}.
\end{equation}
Nel 1924 L.~de~Broglie ipotizzò che era possibile estendere il dualismo onda-corpuscolo anche alla materia imponendo che ogni particella di momento $\mathbf{p}$ avesse una lunghezza d'onda associata $\lambda$, data da
\begin{equation}
\lambda=\frac{h}{\vert\mathbf{p}\vert}.
\end{equation}
A partire da questa ipotesi, è possibile associare ad ogni particella elementare un'onda piana monocromatica, detta \textbf{onda materiale} che si propaga nella direzione di $\mathbf{p}$. Anche per un'onda materiale si definiscono il vettore d'onda $\mathbf{k}$ e la pulsazione $\omega$. Il vettore $\mathbf{k}$ ha direzione e verso di $\mathbf{p}$ e modulo pari a $\vert \mathbf{k} \vert=2\pi / \lambda$; la pulsazione è data dalla relazione di Planck-Einstein \eqref{EQ:CAP1:Rel_Planck_Einstein}, cioè $\omega=E/ \hbar$ e da cui è possibile ricavare la frequenza essendo, per definizione, $\omega=2\pi\nu$. Di conseguenza, il legame tra la quantità di moto della particella e il vettore d'onda dell'onda materiale associata è dato dalla relazione
\begin{equation}\label{EQ:CAP1:Legge_di_de_Broglie}
\mathbf{p}=\hbar \mathbf{k},
\end{equation}
nota come \textbf{legge di de~Broglie}. Quindi le onde di de~Broglie hanno la forma
\begin{equation}\label{EQ:CAP1:Onda_de_Broglie_k_omega}
Ae^{i(\mathbf{k}\cdot\mathbf{x}-t\omega)},
\end{equation}
dove $A$ è l'ampiezza dell'onda, oppure in termini di impulso e pulsazione
\begin{equation}\label{EQ:CAP1:Onda_de_Broglie_p_E}
Ae^{\frac{i}{\hbar}(\mathbf{p}\cdot\mathbf{x}-tE)}.
\end{equation}

Nel 1927, si ebbe la conferma sperimentale dell'esistenza di onde associate al comportamento di elettroni, grazie a due distinti esperimenti da parte di C.~Davisson e L.~Germer e, separatamente, G.P.~Thompson. Essi osservarono un fenomeno di diffrazione, tipico delle onde elettromagnetiche, ottenuto invece da fasci di elettroni separati su cristalli (per ulteriori informazioni sugli esperimenti appena citati si consulti \citep{BOOK:French}).

Nel 1926 E.~Schr\"{o}dinger modificò l'ipotesi di de~Broglie associando ad una particella materiale, non un'onda piana come la \eqref{EQ:CAP1:Onda_de_Broglie_p_E}, bensì un \textbf{pacchetto d'onde} costituito dalla sovrapposizione di onde piane con vettori d'onda vicini. In tal modo la perturbazione ha un'estensione limitata. Tuttavia, questa spiegazione non è corretta in quanto un pacchetto d'onde tende a dissiparsi molto rapidamente.

L'interpretazione fisica corrente è quella che fa uso del concetto di probabilità. Secondo questa teoria, lo stato di un sistema quantistico costituito da $N$ particelle è descritto da una funzione (in generale complessa), indicata solitamente con $\psi$ e chiamata \textbf{funzione d'onda}, dipendente dalle $3N$ coordinate (indicate in forma compatta con $\mathbf{x}\in\mathbb{R}^{3N}$) e dal tempo $t$. Inoltre, il modulo quadro di questa funzione fornisce la distribuzione delle probabilità dei valori delle coordinate. Questa interpretazione è detta \textbf{interpretazione di Copenaghen} e fu sviluppata tra il 1926 e il 1930 grazie ai lavori di M.~Born, W.~Heisenberg e N.~Bohr.

\subsection{Alcuni principi fondamentali}\label{PAR:CAP1:Principi_fond}
Un principio fondamentale della meccanica quantistica è detto \textbf{principio di sovrapposizione degli stati}. Esso consiste nell'ipotesi che fra gli stati esistano relazioni particolari tali che, se un sistema si trova in uno stato, esso possa venire considerato come facente parte contemporaneamente di due o più altri stati. Viceversa, due o più stati possono essere sovrapposti per formarne uno nuovo. Per comprendere meglio cosa comporta questo principio si considerino due stati $A$ e $B$, tali che esista un'osservazione che, effettuata su $A$ porti ad un risultato $a$, mentre, effettuata su $B$ conduca ad un risultato $b$. Allora, il risultato di un'osservazione eseguita sul sistema nello stato risultante dalla sovrapposizione di $A$ e $B$ sarà a volte $a$ ed a volte $b$, in accordo con una legge di probabilità che dipende dai pesi di $A$ e $B$. Questa probabilità di ottenere un certo risultato è tutto quello che la teoria quantistica permette di calcolare.

Nella meccanica quantistica il concetto di traiettoria della particella non esiste. Questo concetto si può esprimere mediante il \textbf{principio di indeterminazione}, scoperto da W.~Heisenberg nel 1927. Il punto di partenza è una proprietà che assume un processo di misura in meccanica quantistica, cioè il fatto che esso influisce sempre sull'elettrone, oggetto della misura. Inoltre, più precisa è la misura più forte è l'influenza da essa esercitata. Di conseguenza, mentre nella meccanica classica una particella possiede, ad ogni istante, posizione e velocità determinate, ciò non accade nella meccanica quantistica. Se, effettuando una misura, all'elettrone si assegna una posizione determinata, esso allora non ha, in generale, nessuna velocità determinata. Viceversa, se l'elettrone è dotato di una velocità determinata, allora esso non potrà avere una posizione determinata nello spazio. Pertanto si può affermare che la posizione e la velocità dell'elettrone sono grandezze esistenti non simultaneamente.

Si consideri un generico sistema atomico, composto da particelle che interagiscano secondo leggi di forza assegnate. Risulteranno possibili diversi moti di tali particelle e ciascuno dei quali si chiamerà \textbf{stato} del sistema. In meccanica classica, per descrivere completamente uno stato di un sistema fisico occorre assegnare la posizione e la velocità all'istante iniziale. In questo modo è possibile determinare le equazioni del moto e di conseguenza il comportamento del sistema in tutti gli istanti successivi. Nella meccanica quantistica una tale descrizione è impossibile perché le coordinate e le velocità non esistono simultaneamente. Di conseguenza, lo stato di un sistema quantistico è determinato da un numero minore di grandezze rispetto alla meccanica classica e ciò non può essere sufficiente a descrivere completamente il moto. Il compito della meccanica quantistica è soltanto quello di determinare la probabilità di avere un certo risultato di una misura.

Data una particella elementare sia $\mathbf{x}=(x_1,x_2,x_3)$ la sua posizione e $\mathbf{p}=(p_1,p_2,p_3)$ il suo impulso. Grazie ad un esperimento ideale sulla diffrazione degli elettroni (cfr.~), si determinano le seguenti relazioni
\begin{equation}\label{EQ:CAP1:Rel_Ind_Heisenberg}
\Delta x_i \cdot \Delta p_i \simeq h, \qquad i=1,2,3,
\end{equation}
dove $\Delta x_i$ e $\Delta p_i$ sono l'incertezza sulla misura della coordinata $i$-esima rispettivamente della posizione e dell'impulso, $h$ è la costante di Planck. Le \eqref{EQ:CAP1:Rel_Ind_Heisenberg} sono note come \textbf{Relazioni di Indeterminazione di Heisenberg} ed esprimono l'impossibilità di conoscere contemporaneamente con precisione arbitrariamente grande una coordinata della particella e la rispettiva componente dell'impulso.\\
In meccanica quantistica ad ogni particella elementare è associata un'onda detta \textbf{onda materiale} di vettore d'onda $\mathbf{k}$ e pulsazione $\omega$.


\section{Gli operatori lineari}
Per formalizzare matematicamente la meccanica quantistica è necessario tenere presenti alcune proprietà degli operatori lineari, i quali rivestono un ruolo fondamentale in tutta la teoria. Oltre alle definizioni e proprietà generali, si tratteranno alcune classi di operatori molto utilizzati e si concluderà descrivendone le proprietà spettrali.

Per i contenuti dei successivi paragrafi si farà riferimento a \citep{DISP:Fonte}.

\subsection{Prime definizioni}
\begin{definizione}
Un'applicazione lineare $A:X\longrightarrow Y$ è chiamata \textbf{operatore lineare} (o semplicemente \textbf{operatore}) da $X$ a $Y$. Si scriverà $Ax=y$, dove $x\in X$ e $y\in Y$. In particolare, se $X=Y$ si dice che $A$ è un operatore in $X$.
\end{definizione}
Si osservi che esistono operatori che possono essere definiti in tutto lo spazio $X$  ed operatori che possono essere definiti in un sottoinsieme di $X$. In ogni caso, l'insieme di definizione di un operatore $A$ è detto \textbf{dominio di definizione} dell'operatore $A$ e viene indicato con $D(A)$; l'insieme dei vettori $Ax$ con $x\in D(A)$ è detto \textbf{codominio} dell'operatore $A$ ed è indicato con $R(A)$.
Sull'insieme degli operatori lineari si possono definire le seguenti operazioni.
\begin{itemize}
\item Somma di due operatori
\begin{equation*}
(A+B)x=Ax+Bx, \qquad \forall x \in D(A) \cap D(B).
\end{equation*}
\item Prodotto di un numero per un operatore
\begin{equation*}
(\alpha A)x=\alpha (Ax), \qquad \forall \alpha \in \mathbb{C}, \; \forall x \in D(A).
\end{equation*}
\item Prodotto di due operatori
\begin{equation*}
(AB)x=A(Bx), \qquad \forall x \in D(B) \; \mbox{tale che} \; Bx \in D(A).
\end{equation*}
\end{itemize}
Si può verificare che l'insieme degli operatori lineari da $X$ a $Y$ costituisce uno spazio vettoriale su $\mathbb{C}$ con le operazioni di somma e prodotto per un numero sopra definite. \`{E} opportuno illustrare alcuni esempi notevoli.
\begin{esempio}
Le matrici $m \times n$ ad elementi complessi $a_{ij}$ con $i=1,2,\ldots,m$, $j=1,2,\ldots,n$ sono degli operatori lineari da $\mathbb{C}^n$ a $\mathbb{C}^m$. La legge di applicazione è
\begin{equation}
\sum_{j=1}^n a_{ij} \xi_j = \eta_i,
\end{equation}
dove $(\xi_1,\xi_2,\ldots,\xi_n)$, $(\eta_1,\eta_2,\ldots,\eta_m)$ sono, rispettivamente, elementi di $\mathbb{C}^n$ e $\mathbb{C}^m$.
\end{esempio}
\begin{esempio}
Una matrice infinita $A$ con elementi $a_{ij}\in\mathbb{C}$ con $i,j=1,2,\ldots$, può essere vista come un operatore lineare da un certo spazio $l^p$ ad un altro spazio $l^s$, $1 \leq p, s \leq \infty$ e con legge di applicazione
\begin{equation}
\sum_{j=1}^\infty a_{ij} \xi_j = \eta_i,
\end{equation}
dove $(\xi_1,\xi_2,\ldots)\in l^p$, $(\eta_1,\eta_2,\ldots)\in l^s$.
\end{esempio}
\begin{esempio}
Una funzione $k(x,y)$ a valori complessi, misurabile e definita per $x \in E$ ed $y \in F$, dove $E, F \subset \mathbb{R}^n$, definisce un operatore lineare da un certo spazio $L^p(F)$ ad un certo spazio $L^s(E)$, $1 \leq p, s \leq \infty$ con legge di applicazione
\begin{equation}
\int_{F} \! k(x,y)f(y) \, \diff y = g(x).
\end{equation}
Tale operatore è detto \textbf{operatore integrale} e la funzione $k(x,y)$ è chiamata \textbf{nucleo} dell'operatore.
\end{esempio}
\begin{esempio}
Un altro importante esempio è costituito dagli operatori differenziali, il più semplice dei quali è l'operatore
\begin{equation}
A\cdot = \frac{\diff}{\diff t}\cdot,
\end{equation}
che può essere visto come un operatore da $L^2([a,b])$ ad $L^2([a,b])$. Si osservi che tale operatore non può essere definito su tutto $L^2([a,b])$, perché questo spazio contiene anche funzioni non derivabili.
\end{esempio}

\subsection{Proprietà degli operatori lineari}
Gli operatori lineari hanno proprietà differenti a seconda delle ipotesi che possono essere effettuate sugli spazi su cui sono definiti. Si considerino in primo luogo gli operatori lineari definiti su spazi di dimensione finita. \`{E} possibile dimostrare i risultati che seguono.
\begin{proposizione}
Siano $X$ e $Y$ due $\mathbb{C}$-spazi vettoriali, aventi rispettivamente dimensione $n$ ed $m$. Allora ogni trasformazione lineare $\hat{A}:X\longrightarrow Y$ è rappresentata da
\begin{equation}
\sum_{j=1}^n a_{ij} \xi_j = \eta_i, \quad i=1,2,\ldots,m,
\end{equation}
dove $(\xi_1,\xi_2,\ldots,\xi_n)\in\mathbb{C}^n$, $(\eta_1,\eta_2,\ldots,\eta_m)\in\mathbb{C}^m$ ed i coefficienti $a_{ij}$ con $i=1,2,\ldots,m$, $j=1,2,\ldots,n$ sono univocamente determinati.
\end{proposizione}
Di conseguenza, la matrice $A=\left\lbrace a_{ij} \right\rbrace$ corrispondente ad una dato operatore $A$, per basi fissate, è unica e si chiama \textbf{rappresentazione matriciale} dell'operatore $A$.

Seguono alcune importanti definizioni che saranno utili nel seguito.
\begin{definizione}
Sia $A:X\longrightarrow Y$ un operatore iniettivo con $D(A)\subseteq X$, $R(A)\subseteq Y$. Allora, se $Ax=y$, l'operatore $A^{-1}$ che applicato ad $y$ dà $x$ è detto \textbf{operatore inverso} di $A$.
\end{definizione}
\`{E} possibile dimostrare che l'operatore inverso è un operatore lineare e che un operatore $A$ ammette inverso se e solo se l'equazione $Ax=0$ ha come unica soluzione $x=0$.

Un'altra proprietà degli operatori lineari è l'idempotenza.
\begin{definizione}
Sia $X$ uno spazio vettoriale. Un operatore $A:X\longrightarrow X$ si dice \textbf{idempotente} se $A^2=A$.
\end{definizione}

Gli operatori lineari si possono classificare in due famiglie: gli operatori limitati e gli operatori non limitati.
\begin{definizione}
Siano $(X,\norm{\cdot}_X)$ e $(Y,\norm{\cdot}_Y)$ due spazi di Banach e sia $A$ un operatore da $X$ a $Y$. Si dice che l'operatore $A$ è \textbf{limitato} se esiste un numero $M$ tale che
\begin{equation*}
\norm{Ax}_Y \leq M \norm{x}_X, \qquad \forall \in D(A).
\end{equation*}
\end{definizione}
Se $(X,\norm{\cdot}_X)$ e $(Y,\norm{\cdot}_Y)$ sono spazi di Banach allora l'insieme degli operatori da $X$ a $Y$ limitati, indicato con $B(X,Y)$ è uno spazio vettoriale. Un'altra proprietà importante degli operatori è la loro continuità.
\begin{definizione}
Siano $(X,\norm{\cdot}_X)$ e $(Y,\norm{\cdot}_Y)$ due spazi di Banach e sia $A:X\longrightarrow Y$ un operatore. Allora $A$ è detto \textbf{continuo nel punto} $X_0 \in D(A)$, se
\begin{equation}
\lim_{n \to \infty}{\norm{x_n-x_0}_X}=0
\end{equation}
implica
\begin{equation}
\lim_{n \to \infty}{\norm{Ax_n-Ax_0}_Y}=0.
\end{equation}
Inoltre, se $A$ risulta continuo in ogni punto del suo dominio sarà detto semplicemente \textbf{continuo}.
\end{definizione}
Si ha inoltre il seguente risultato.
\begin{teorema}
Siano $(X,\norm{\cdot}_X)$ e $(Y,\norm{\cdot}_Y)$ due spazi di Banach. Un operatore $A:X\longrightarrow Y$ è continuo se e solo se esso è limitato.
\end{teorema}
Una classe notevole di operatori lineari è quella costituita dai funzionali.
\begin{definizione}
Sia $X$ uno spazio normato. Un'applicazione lineare $f:X\longrightarrow\mathbb{C}$ è detta \textbf{funzionale lineare} su $X$.
\end{definizione}
I funzionali lineari sono particolari operatori lineari che hanno come spazio immagine $\mathbb{C}$. L'insieme di tutti i funzionali lineari e continui definiti su uno spazio normato $X$ forma uno spazio vettoriale normato e completo. Esso è detto \textbf{spazio duale} di $X$. Vi è un importante teorema di rappresentazione dei funzionali lineari e continui definiti su uno spazio di Hilbert.
\begin{teorema}[Riesz]\label{TEOR:CAP1:Riesz}
Sia $H$ uno spazio di Hilbert. Allora per ogni funzionale $f(x)$ lineare e continuo definito su tutto $H$ esiste un unico elemento $x_f$ tale che
\begin{equation}
f(x)=\left\langle x_f \vert x \right\rangle , \quad \forall x \in H,
\end{equation}
e inoltre $\norm{f}=\norm{x_f}$.
\end{teorema}
Si procede adesso definendo l'operatore aggiunto di un operatore limitato ed il conseguente concetto di operatore autoaggiunto. Successivamente si daranno le definizioni analoghe per il caso di operatori non limitati.

Sia $H$ uno spazio di Hilbert ed $A:H\longrightarrow H$ un operatore limitato, definito su tutto $H$. Fissato ad arbitrio un elemento $y\in H$, si definisce il funzionale $f$ su $H$ ponendo
\begin{equation}
f(x) = \left\langle y \vert Ax \right\rangle, \quad \forall x\in H.
\end{equation}
Si può verificare che il funzionale $f$ è lineare e continuo. Per il Teorema \ref{TEOR:CAP1:Riesz}, esiste un unico elemento $y^* \in H$ tale che
\begin{equation}
\left\langle y \vert Ax \right\rangle = \left\langle y^* \vert x \right\rangle, \quad \forall x\in H.
\end{equation}
Questa relazione stabilisce una corrispondenza $y \mapsto y^*$, definita su tutto $H$, che risulta essere lineare. Tale corrispondenza è chiamata \textbf{operatore aggiunto} di $A$ e si denota con $A^*$. Di conseguenza, osservando che $A^* y = y^*$, si ha la relazione
\begin{equation}
\left\langle y \vert Ax \right\rangle = \left\langle A^* y \vert x \right\rangle, \quad \forall x,y\in H,
\end{equation}
la quale permette di determinare esplicitamente $A^*$. Inoltre, se risulta $A^*=A$ allora l'operatore $A$ è detto \textbf{autoaggiunto}. Per completare, è opportuno enunciare il seguente criterio relativo agli operatori autoaggiunti.
\begin{teorema}
Condizione necessaria e sufficiente affinché un operatore limitato $A$ con $D(A) = H$ sia autoaggiunto è che $\left\langle x \vert Ax \right\rangle  $ sia reale $\forall x \in H$.
\end{teorema}

Sotto certe condizioni, si vedrà che è possibile definire l'aggiunto anche di un operatore non limitato. Infatti, sia $A:H \longrightarrow H$ un operatore non limitato con dominio $D(A)$ denso in $H$. Si consideri l'insieme $D^*$ dei vettori $y\in H$ per i quali esiste un vettore $y^* \in H$ tale che
\begin{equation}\label{EQ:CAP1:Def_autoaggiunto_1}
\left\langle y \mid Ax \right\rangle = \left\langle y^* \mid x \right\rangle, \qquad \forall x \in D(A).
\end{equation}
Si può dimostrare che il vettore $y^*$ è univocamente determinato da $y$. Di conseguenza, si ottiene una corrispondenza $y \mapsto y^*$. Inoltre si può verificare che tale corrispondenza è lineare. Quindi resta definito un operatore lineare $A^*$ con dominio $D(A^*)=D^*$. Anche in questo caso, l'operatore $A^*$ è detto \textbf{aggiunto} dell'operatore $A$ e, se risulta $A^*=A$, l'operatore $A$ si dice \textbf{autoaggiunto}.

\subsection{Gli operatori di proiezione}
Un'importante categoria di operatori limitati è costituita dagli operatori di proiezione. Per definirli, è essenziale richiamare il concetto di proiezione di un vettore su sottospazi complementari ortogonali di uno spazio di Hilbert. Si consideri in primo luogo uno spazio di Hilbert infinito-dimensionale.
\begin{definizione}
Sia $H$ uno spazio di Hilbert ad infinite dimensioni ed $M$ un suo sottospazio chiuso. Allora l'insieme
\begin{equation}
M^\perp = \left\lbrace x\in H \mid \left\langle x \vert y \right\rangle=0, \quad \forall y\in M  \right\rbrace 
\end{equation}
è chiamato \textbf{complemento ortogonale} del sottospazio $M$.
\end{definizione}
Si può dimostrare che $M^\perp$ è un sottospazio chiuso di $H$. Si ha il risultato seguente.
\begin{teorema}\label{TEOR:CAP1:proiezione_ortogonale}
Sia $H$ uno spazio di Hilbert ad infinite dimensioni ed $M$ un suo sottospazio chiuso. Allora, ogni vettore $x \in H$ è rappresentabile univocamente nella forma
\begin{equation}
x=h+h',
\end{equation}
dove $h\in M$ e $h' \in M^\perp$.
\end{teorema}
Sia $\left\lbrace e_i \right\rbrace_{i=1}^\infty $ un sistema ortonormale completo in $M$. Allora il vettore $h$ assume la forma seguente
\begin{equation}\label{EQ:CAP1:proiezione_ortogonale}
h = \sum_{i=1}^\infty e_i \left\langle e_i \vert x \right\rangle  .
\end{equation}
Tale vettore è detto \textbf{proiezione ortogonale} di $x$ su $M$. Analogo discorso vale per definire la proiezione ortogonale di $x$ su $M^\perp$. Inoltre, lo spazio $H$ si può decomporre nel modo seguente
\begin{equation}
H=M\oplus M^\perp,
\end{equation}
dove lo spazio $M^\perp$, in virtù del Teorema \ref{TEOR:CAP1:proiezione_ortogonale}, è univocamente determinato a partire da $M$. Le relazioni
\begin{equation}
P_M x = h, \quad P_{M^\perp} = h', \quad \forall x\in H
\end{equation} 
definiscono gli operatori $P_M$ e $P_{M^\perp}$, detti \textbf{operatori di proiezione ortogonali}, rispettivamente su $M$ e $M^\perp$.

Viceversa, dati due operatori $P_M$ e $P_{M^\perp}$ che siano lineari, limitati, autoaggiunti, idempotenti e tali che
\begin{equation}
P_M + P_{M^\perp} =I, \qquad P_M + P_{M^\perp} = 0,
\end{equation}
allora essi sono operatori di proiezione ortogonali, proiettando rispettivamente su $M=R(P_M)$ e $M^\perp = R(P_{M^\perp})$. Inoltre si ha la decomposizione $H=R(P_M) \oplus R(P_{M^\perp})$.
\begin{osservazione}
I risultati sopra provati valgono in particolare anche in spazi di Hilbert finito-dimensionali. Per questo, basta notare che in tali spazi ogni sottospazio $M$ è chiuso e per definire la proiezione ortogonale basta prendere un sistema ortonormale $\left\lbrace e_i \right\rbrace_{i=1}^k $, dove $k$ è la dimensione di $M$.
\end{osservazione}
Ora seguirà la definizione di proiezione di un vettore su sottospazi complementari non necessariamente ortogonali. Tale definizione può dunque essere applicata anche in spazi di Banach. Sia $X$ uno spazio di Banach e siano $M_1, M_2, \ldots, M_m$ sottospazi chiusi di $X$ tali che
\begin{equation}\label{EQ:CAP1:Decomposizione_operatori_proiezione}
X=M_1 \oplus M_2 \oplus \ldots \oplus M_m.
\end{equation}
Quindi $\forall x \in X$ può essere rappresentato univocamente come
\begin{equation}
x=x_1+x_2+\ldots +x_m, \qquad x_i\in M_i, \quad i=1,2,\ldots m.
\end{equation}
Le relazioni $P_i x = x_i$ definiscono degli operatori $P_i$, ciascuno dei quali è detto \textbf{operatore di proiezione} sul sottospazio $M_i$, $i = 1, 2, \ldots m$ lungo
\begin{equation}
M_1 \oplus M_2 \oplus \ldots \oplus M_{i-1} \oplus M_{i+1} \oplus \ldots \oplus M_m.
\end{equation}
Gli operatori di proiezione si dimostrano essere lineari e limitati. Si osservi che $P_i x =x$ se e solo se $x\in M_i$. Inoltre, si ha
\begin{equation}\label{EQ:CAP1:Cond_operatori_proiezione}
\sum_{i=1}^m P_i = I, \qquad P_i P_j = \delta_{ij} P_i.
\end{equation}
Viceversa, se $P_i$, $i=1,2,\ldots,m$ sono operatori limitati, definiti in tutto lo spazio e soddisfacenti la \eqref{EQ:CAP1:Cond_operatori_proiezione}, allora si ha la decomposizione \eqref{EQ:CAP1:Decomposizione_operatori_proiezione}, dove $M_i=R(P_i)$, $i=1,2,\ldots,m$ e ciascun $P_i$ è l'operatore di proiezione su $M_i$.

Infine, l'ultimo caso è il seguente. Sia $H$ uno spazio di Hilbert e si consideri la sua decomposizione
\begin{equation}
H=M_1 \oplus M_2 \oplus \ldots \oplus M_m,
\end{equation}
dove $M_i$, $i=1,2,\ldots,m$ sono sottospazi a due a due ortogonali. Allora gli operatori $P_i$, definiti come nel caso precedente, risultano essere ortogonali. Inoltre, ciascuno proietta su $M_i$ e soddisfano la condizione \eqref{EQ:CAP1:Cond_operatori_proiezione}. Infine, di questa proprietà vale anche il viceversa.

\subsection{Spettro di un operatore}
Una delle caratteristiche più importanti di un operatore è lo spettro. Spesso esso rappresenta il collegamento tra la teoria degli operatori e moltissimi problemi della fisica.
\begin{definizione}
Sia $X$ uno spazio di Banach ed $A$ un operatore in $X$. L'equazione in $\lambda$ ed in $x$
\begin{equation}\label{EQ:CAP1:Eq_agli_autovalori}
Ax=\lambda x,
\end{equation}
con $\lambda\in\mathbb{C}$ ed $x\in D(A)$, è chiamata \textbf{equazione agli autovalori}. Inoltre, il numero $\lambda$ è detto \textbf{autovalore} di $A$ se esiste un vettore $x\in D(A)$, differente dal vettore nullo, tale che la \eqref{EQ:CAP1:Eq_agli_autovalori} sia soddisfatta. L'elemento $x$ è chiamato \textbf{autovettore} di $A$ corrispondente all'autovalore $\lambda$.
\end{definizione}
\begin{osservazione}
Riscrivendo la \eqref{EQ:CAP1:Eq_agli_autovalori} nella forma equivalente
\begin{equation}
(A-\lambda I)x =0,
\end{equation}
è possibile osservare che $A$ ammette $\lambda$ come autovalore se e solo se l'operatore $A_\lambda := (A-\lambda I)$ non ammette inverso. Inoltre, se lo spazio è finito-dimensionale allora $A_\lambda^{-1}$ o non esiste oppure esiste e in tal caso è limitato. Negli spazi infinito-dimensionali si ha anche che l'operatore $A_\lambda^{-1}$, quando esiste, può anche essere non limitato. Da questo fatto seguono le seguenti definizioni.
\end{osservazione}
\begin{definizione}
Si dice che $\lambda$ è un \textbf{punto regolare} di $A$ se il corrispondente operatore $A_\lambda^{-1}$ esiste, è limitato ed è definito su tutto $X$. In questo caso, l'operatore $A_\lambda^{-1}$ si denota con $R(\lambda,A)$ ed è chiamato \textbf{risolvente} di $A$ nel punto $\lambda$. L'insieme dei punti regolari di $A$ è chiamato \textbf{insieme risolvente} di $A$ e si denota con $\rho(A)$. Infine, l'insieme $\sigma(A) = \mathbb{C}/\rho(A)$ è chiamato \textbf{spettro} di $A$.
\end{definizione}
Lo spettro di $A$ può essere suddiviso nei seguenti tre sottoinsiemi disgiunti $\sigma_d(A)$, $\sigma_c(A)$, $\sigma_r(A)$ in accordo alla seguente definizione.
\begin{definizione}
si chiama \textbf{spettro discreto} di $A$, e si indica con $\sigma_d(A)$, l'insieme dei punti $\lambda$ per cui $A_\lambda^{-1}$ non esiste. Si chiama \textbf{spettro continuo} di $A$, e si indica con $\sigma_c(A)$, l'insieme dei punti $\lambda$ per cui $A_\lambda^{-1}$ esiste non limitato con dominio denso in $X$. Infine, si chiama \textbf{spettro residuo} di $A$, e si indica con $\sigma_r(A)$, l'insieme dei punti $\lambda$ per cui $A_\lambda^{-1}$ esiste ma il suo dominio non è denso in $X$.
\end{definizione}

\subsection{Proprietà spettrali degli operatori in spazi di dimensione finita}
Sia $\hat{A}$ un operatore in uno spazio $X$ di dimensione finita $n$. In questo caso, lo spettro di $\hat{A}$ si riduce al solo spettro discreto, cioè solo ai suoi autovalori. Per determinarli, si consideri l'equazione agli autovalori
\begin{equation}
\hat{A}\hat{x}=\lambda \hat{x}
\end{equation}
e la sua rappresentazione matriciale
\begin{equation}
Ax=\lambda x
\end{equation}
in una certa base di $\mathbb{C}^n$. \`{E} noto che gli autovalori sono tutte e sole le radici dell'equazione secolare
\begin{equation}
\mbox{Det} (A-\lambda I)=0,
\end{equation}
che saranno indicate con $\lambda_1,\lambda_2,\ldots,\lambda_m$, con $m\leq n$.

Se l'operatore $\hat{A}$ ammette $n$ autovettori linearmente indipendenti $\psi_1,\ldots,\psi_n$, allora esso è diagonalizzabile e si può scrivere
\begin{equation}\label{EQ:CAP1:Rappresentazione_spettr_dim_finita}
A=\sum_{i=1}^n \lambda_i P'_i,
\end{equation}
dove $P'_i$, $i=1,2,\ldots,n$ sono le matrici $n\times n$
\begin{equation}
P'_1 =
\left( 
\begin{array}{cccc}
1 & 0 & \ldots & 0 \\ 
0 & 0 & \ldots & 0 \\ 
\vdots & \vdots & \ddots & \vdots \\ 
0 & 0 & \ldots & 0
\end{array}
\right)
\quad \ldots \quad
P'_n =
\left( 
\begin{array}{cccc}
0 & \ldots & 0 & 0 \\ 
0 & \ddots & \vdots & \vdots \\ 
0 & \ldots & 0 & 0 \\ 
0 & \ldots & 0 & 1
\end{array}
\right).
\end{equation}
La \eqref{EQ:CAP1:Rappresentazione_spettr_dim_finita} è detta \textbf{rappresentazione spettrale} della matrice $A$.

Se l'operatore $\hat{A}$ non è diagonalizzabile, esiste comunque una base dove esso è rappresentato in una forma abbastanza semplice, la cosiddetta forma canonica di Jordan. Essa è una matrice a blocchi diagonali, dove in ciascun blocco si può utilizzare una rappresentazione analoga al caso in cui l'operatore è diagonalizzabile.

\subsection{Proprietà spettrali degli operatori compatti}
Gli operatori compatti costituiscono un esempio notevole di operatori limitati. Tra le famiglie degli operatori in spazi infinito dimensionali, gli operatori compatti sono quelli che hanno proprietà abbastanza simili a quelle degli operatori in spazi finito dimensionali.
\begin{definizione}
Siano $X$ ed $Y$ spazi di Banach. Un operatore $A : X \longrightarrow Y$ si dice \textbf{compatto} se trasforma ogni insieme limitato di $X$ in un insieme di $Y$ la cui chiusura è compatta.
\end{definizione}
\begin{teorema}[Hilbert-Schmidt]\label{TEOR:CAP1:Hilbert_Schmidt}
Sia $A$ un operatore compatto ed autoaggiunto in $H$ con $D(A) = H$. Allora seguono i seguenti risultati.
\begin{enumerate}
\item Se $A$ ammette l'autovalore zero, allora $\sigma(A) = \sigma_d(A)$, ove $\sigma_d(A)$ è fatto da un numero finito di autovalori. Ciascun autovalore diverso da zero ha molteplicità finita, mentre l'autovalore zero ha molteplicità infinita.
\item Se $A$ non ammette l'autovalore zero, $\sigma(A) = \sigma_d(A)\cup\sigma_c(A)$, ove lo spettro discreto $\sigma_d(A)$ è fatto di una infinità numerabile di autovalori, ciascuno di molteplicità finita, ed aventi lo zero come unico punto di accumulazione, mentre lo spettro continuo $\sigma_c(A)$ è costituito da unico punto che è lo zero.
\item Ogni vettore $Ax$ può essere rappresentato come
\begin{equation}\label{EQ:CAP1:Rappresentazione_operatori_compatti}
Ax = \sum_{i} \lambda_i y_i \left\langle y_i \vert x \right\rangle ,
\end{equation}
dove $\lambda_i$ sono tutti gli autovalori di $A$ non nulli, ciascuno ripetuto un numero di volte pari alla sua molteplicità geometrica, ed $y_i$ sono gli autovettori ortonormali corrispondenti.
\end{enumerate}
\end{teorema}
\begin{corollario}\label{COR:CAP1:Hilbert_Schmidt}
In corrispondenza di ogni operatore compatto ed autoaggiunto in uno spazio di Hilbert $H$, si può determinare una base ortonormale in $H$ costituita da autovettori di $A$.
\end{corollario}
La dimostrazione del Teorema \ref{TEOR:CAP1:Hilbert_Schmidt} fornisce un algoritmo per determinare in sequenza gli autovalori $\lambda_1,\lambda_2,\ldots$ di $A$, soddisfacenti
\begin{equation}
\vert \lambda_1 \vert \geq \vert \lambda_2 \vert \geq \ldots
\end{equation}
ed ai quali, in virtù del Corollario \ref{COR:CAP1:Hilbert_Schmidt}, corrisponde una base di autovettori ortonormali. Di conseguenza, se ad un certo passo del procedimento si ottiene l'autovalore $\lambda_n=0$ allora si avrà $0=\lambda_{n+1}=\lambda_{n+2}=\ldots$. Inoltre, per la \eqref{EQ:CAP1:proiezione_ortogonale} e dalla definizione di operatore di proiezione, vale la relazione
\begin{equation}
P_i x = y_i \left\langle y_i \vert x \right\rangle,
\end{equation}
essendo $P_i$ l'operatore di proiezione ortogonale sul sottospazio sotteso da $y_i$. Si possono distinguere due casi. Se $A$ possiede l'autovalore $\lambda_n =0$, allora si hanno un numero finito di autovalori e la \eqref{EQ:CAP1:Rappresentazione_operatori_compatti} diventa
\begin{equation}
Ax = \sum_{i=1}^{n-1} \lambda_i P_i x, \quad \forall x \in H,
\end{equation}
da cui
\begin{equation}\label{EQ:CAP1:Rappresentazione_spettrale_1}
A = \sum_{i=1}^{n-1} \lambda_i P_i.
\end{equation}
Se $A$ non ammette l'autovalore zero, allora la \eqref{EQ:CAP1:Rappresentazione_operatori_compatti} diventa
\begin{equation}
Ax = \sum_{i=1}^{\infty} \lambda_i P_i x, \quad \forall x \in H,
\end{equation}
da cui
\begin{equation}\label{EQ:CAP1:Rappresentazione_spettrale_2}
A = \sum_{i=1}^{\infty} \lambda_i P_i.
\end{equation}
La \eqref{EQ:CAP1:Rappresentazione_spettrale_1} o la \eqref{EQ:CAP1:Rappresentazione_spettrale_2} è detta \textbf{rappresentazione spettrale} di $A$.

\subsection{Proprietà spettrali degli operatori autoaggiunti}
\`{E} possibile scrivere una rappresentazione spettrale anche per un generico operatore autoaggiunto. La derivazione della rappresentazione spettrale di un
operatore autoaggiunto è basata su un concetto di misura più generale rispetto a quelle considerate fino a qui.
\begin{definizione}
Sia $S$ un insieme non vuoto e $\mathcal{B}$ l'algebra di Borel di $S$. Allora si chiama \textbf{misura spettrale} $E(A)$ sullo spazio misurabile $(S, \mathcal{B})$ una applicazione $E$ che assegna ad ogni boreliano $A\in\mathcal{B}$ un operatore di proiezione ortogonale $E(A)$ in uno spazio di Hilbert $H$, con le proprietà:
\begin{enumerate}
\item $E(S)=I$
\item $E\left( \bigcup_{i=1}^\infty A_i \right) = \sum_{i=1}^\infty E(A_i), \qquad A_i\in\mathcal{B}, \quad A_i \cap A_j = \emptyset, \quad i\neq j, \quad i,j=1,2,\ldots$, dove
$\sum_{i=1}^\infty E(A_i)=s-\lim_{n\to\infty} \sum_{i=1}^n E(A_i)$.
\end{enumerate}
\end{definizione}
\begin{definizione}
Sia $H$ uno spazio si Hilbert ed $\left\lbrace E(\lambda) \right\rbrace $ una famiglia di operatori di proiezione ortogonali dipendente dal parametro reale $\lambda$, $-1 < \lambda < 1$. Siano $M(\lambda)$ i sottospazi di $H$ su cui ciascun operatore $E(\lambda)$ proietta. Allora la famiglia $\left\lbrace E(\lambda) \right\rbrace $ è chiamata \textbf{famiglia spettrale} se essa gode delle seguenti proprietà:
\begin{enumerate}
\item \`{E} non decrescente, cioè se $\lambda < \mu$ allora $M(\lambda) \subseteq M(\mu)$ e si scriverà $E(\lambda) \leq E(\mu)$.
\item \`{E} continua a destra, cioè $s-\lim_{\varepsilon\to 0^+} E(\lambda+\varepsilon)=:E(\lambda + 0)=E(\lambda)$.
\item $E(-\infty)=s-\lim_{\lambda\to -\infty} E(\lambda)=0, \quad E(+\infty)=s-\lim_{\lambda\to +\infty} E(\lambda)=I$.
\end{enumerate}
\end{definizione}
Sia $H$ uno spazio di Hilbert e $\left\lbrace E(\lambda) \right\rbrace $ una famiglia spettrale ad esso associata. Si consideri una partizione di $(-\infty,+\infty)$
\begin{equation}
-\infty < \ldots < \lambda_{-2} < \lambda_{-1} < \lambda_0 < \lambda_1 < \lambda_2 < \ldots < +\infty
\end{equation}
e, fissato $x \in H$, si costruisca la sommatoria
\begin{equation}
\sum_{i=-\infty}^{+\infty} \lambda_i' [E(\lambda_{i+1})-E(\lambda_i)]x, \quad \lambda_i < \lambda_i' \leq \lambda_{i+1}.
\end{equation}
\`{E} possibile dimostrare che esiste un sottoinsieme $D$ di $H$ per cui tale sommatoria converge nella norma di $H$ per ogni partizione di $(-\infty,+\infty)$. Allora si dimostra che esiste il seguente limite,
\begin{equation}
\lim_{\max_i \vert \lambda_{i+1}-\lambda_i \vert \to 0} \sum_{i=-\infty}^{+\infty} \lambda_i' [E(\lambda_{i+1})-E(\lambda_i)]x, \quad\forall x \in D.
\end{equation}
Di conseguenza, è possibile porre
\begin{equation}
\int_{-\infty}^{+\infty} \! \lambda \, \diff E(\lambda)x := \lim_{\max_i \vert \lambda_{i+1}-\lambda_i \vert \to 0} \sum_{i=-\infty}^{+\infty} \lambda_i' [E(\lambda_{i+1})-E(\lambda_i)]x.
\end{equation}
Poiché si prova che $D$ definisce una varietà lineare densa in $H$, allora questa relazione definisce in $D$ un operatore lineare, cioè è possibile associare ad ogni $x\in D$ l'elemento $Ax\in H$, dove
\begin{equation}\label{EQ:CAP1:Rappresentazione_spettrale_3}
A\cdot := \int_{-\infty}^{+\infty} \! \lambda \, \diff E(\lambda) \cdot,
\end{equation}
che è detta \textbf{rappresentazione spettrale} di $A$. Infine, si ha il seguente importante risultato.
\begin{teorema}[Teorema spettrale]\label{TEOR:CAP1:Spettrale}
Ogni operatore autoaggiunto $A$ ammette la rappresentazione \eqref{EQ:CAP1:Rappresentazione_spettrale_3}, dove la famiglia spettrale $\left\lbrace E(\lambda) \right\rbrace $ è individuata univocamente da $A$.
\end{teorema}

\subsection{L'operatore di Fourier}
Risulterà importante dare una rappresentazione delle funzioni di $L^1(\mathbb{R})$. Essa è nota come \emph{trasformata di Fourier} le cui proprietà hanno un ruolo importante nell'ambito della matematica applicata e saranno usate anche nei successivi paragrafi di questa tesi.
\begin{definizione}\label{DEF:CAP1:Trasf_Fourier}
Data una funzione $f\in L^1(\mathbb{R})$ si chiama \textbf{trasformata di Fourier} di $f$ e si indica con $\hat{f}$ la funzione definita $\forall \lambda \in \mathbb{R}$ da
\begin{equation}
\hat{f}(\lambda) = \frac{1}{\sqrt{2\pi}} \int_{-\infty}^{+\infty} \! f(t)e^{-i\lambda t} \, \diff t.
\end{equation}
\end{definizione}
Si può dimostrare che la funzione $f$ può essere scritta come
\begin{equation}
f(x) = \frac{1}{\sqrt{2\pi}} \int_{-\infty}^{+\infty} \! \hat{f}(\lambda)e^{i\lambda x} \, \diff \lambda.
\end{equation}
Questa rappresentazione è detta \textbf{trasformazione di Fourier}. Resta infine definito l'operatore integrale
\begin{equation}\label{EQ:CAP1:Oper_Fourier}
F[\cdot] = \frac{1}{\sqrt{2\pi}} \int_{-\infty}^{+\infty} \! \cdot e^{-i\lambda t} \, \diff t,
\end{equation}
chiamato \textbf{operatore di Fourier}. Dalle precedenti definizioni si deduce immediatamente che $\hat{f}=F[f]$.

La trasformata di Fourier gode di importanti proprietà che risulteranno utili nei prossimi paragrafi.
\begin{itemize}
\item \emph{Trasformata di Fourier di una derivata.}\\
Sia $f\in L^1(\mathbb{R})$ e $f'$ la sua derivata prima. Se $f'$ è continua e inoltre $f'\in L^1(\mathbb{R})$ allora si ha
\begin{equation}
F[f'] = i\lambda F[f].
\end{equation}
\item \emph{Il Teorema di Plancherel.}\\
Si vuole estendere la trasformata di Fourier per funzioni di $L^2(\mathbb{R})$. Infatti $L^2(\mathbb{R})$ non è un sottoinsieme di $L^1(\mathbb{R})$ perché la misura di Lebesgue di $\mathbb{R}$ è infinita (cfr.~\citep{BOOK:Rudin}). Però, se $f\in L^1(\mathbb{R}) \cap L^2(\mathbb{R})$ la Definizione \ref{DEF:CAP1:Trasf_Fourier} si può applicare e si dimostra che $\hat{f} \in L^2(\mathbb{R})$. Inoltre, si può dimostrare anche che se $f\in L^2(\mathbb{R})$ si ha $\lVert\hat{f}\rVert_2=\norm{f}_2$.
\end{itemize}

\chapter{Il modello semiclassico per il trasporto di cariche} 

\label{Chapter2} 

\lhead{Capitolo 2. \emph{Il modello semiclassico per il trasporto di cariche}} 



\section{Descrizione matematica della meccanica quantistica}
Partendo dai presupposti fisici descritti nei precedenti paragrafi, si passerà adesso alla descrizione matematica. In primo luogo si definirà lo spazio delle funzioni d'onda, poi si tratteranno le grandezze fisiche e le loro misure, infine si concluderà con la dinamica dei sistemi quantistici.

I testi di riferimento per i successivi paragrafi saranno \citep{DISP:Sacchetti}, \citep{BOOK:Landau}, \citep{BOOK:Dirac} e \citep{BOOK:Sakurai}.

\subsection{Lo spazio delle funzioni d'onda}\label{PAR:CAP1:Spazio_f_d_o}
Si consideri, per semplicità, il sistema quantistico di una particella, mobile nello spazio euclideo $\mathbb{R}^3$. Essa è descritta da una funzione, detta \textbf{funzione d'onda}, del tipo
\begin{equation}
\psi(\mathbf{x},t) : \mathbb{R}^3 \times \mathbb{R} \longrightarrow \mathbb{C},
\end{equation}
dove $\mathbf{x}=(x_1,x_2,x_3)\in\mathbb{R}^3$ è un punto dello spazio e $t\in\mathbb{R}$ è un istante di tempo. 

Considerato un insieme $A\subset\mathbb{R}^3$, nella interpretazione di Copenaghen, la probabilità di trovare la particella in $A$ all'istante $t$ è data da
\begin{equation}\label{EQ:CAP1:Prob_funz_d_onda_x}
P^t(A)=\int_A \! \rho_t(\mathbf{x}) \, \diff \mathbf{x},
\end{equation}
dove $\rho_t(\mathbf{x})= \vert \psi(\mathbf{x},t) \vert^2 = \psi^*(\mathbf{x},t)\psi(\mathbf{x},t)$ rappresenta la densità di probabilità all'istante $t$ e $\psi^*(\mathbf{x},t)$ è il complesso coniugato di $\psi(\mathbf{x},t)$. Allora, essendo $\rho_t(\mathbf{x})$ una densità di probabilità, essa deve soddisfare le condizioni
\begin{equation}\label{EQ:CAP1:Cond_funz_d_onda_x}
\vert \psi(\mathbf{x},t) \vert^2\geq 0 \quad \forall\mathbf{x}\in\mathbb{R}^3, \quad \forall t\in\mathbb{R}, \qquad \int_{\mathbb{R}^3} \! \vert \psi(\mathbf{x},t) \vert^2 \, \diff \mathbf{x} =1 \quad \forall t \in \mathbb{R}.
\end{equation}
Si osservi che per essere valide le \eqref{EQ:CAP1:Cond_funz_d_onda_x} e la \eqref{EQ:CAP1:Prob_funz_d_onda_x} è necessario richiedere che gli integrali presenti in esse esistano. Pertanto si richiede che la funzione d'onda $\psi(\cdot,t)$ sia quadrato-sommabile su $\mathbb{R}^3$, cioè deve essere $\psi(\cdot,t)\in L^2(\mathbb{R}^3)$ per ogni $t\in\mathbb{R}$.

Per descrivere matematicamente l'insieme delle funzioni d'onda di un dato sistema quantistico, è opportuno assumere le seguenti ipotesi.
\begin{enumerate}
\item La funzione d'onda nulla non è associata ad alcuno stato.
\item Se $\psi$ è la funzione d'onda associata ad un certo stato, allora tutte le funzioni d'onda del tipo $c\psi$ con $c\in\mathbb{C}$ corrispondono allo stesso stato.
\item Se un sistema quantistico può trovarsi negli stati descritti dalle funzioni d'onda $\psi_1$ e $\psi_2$ allora può trovarsi anche in tutti gli stati della forma $\lambda_1\psi_1+\lambda_2\psi_2$ con $\lambda_1,\lambda_2\in\mathbb{C}$.\label{EQ:CAP1:Principio_sovrapposizione} 
\end{enumerate}
Si osservi che l'ipotesi \ref{EQ:CAP1:Principio_sovrapposizione} esprime in termini di funzioni d'onda il principio di sovrapposizione degli stati, illustrato nel Paragrafo \ref{PAR:CAP1:Principi_fond}. Inoltre, in virtù dei precedenti tre postulati, l'insieme $H$ delle funzioni d'onda di un sistema quantistico assume la struttura di spazio vettoriale su $\mathbb{C}$. In particolare, per quanto visto prima esso sarà un sottospazio di $L^2(\mathbb{R}^3)$.

Di conseguenza $H$ è uno spazio di Hilbert complesso, la cui norma è indotta dal prodotto scalare di $L^2(\mathbb{R}^3)$
\begin{equation}\label{EQ:CAP1:Ptto_scalare_L2_f_o}
\left\langle \psi_1,\psi_2 \right\rangle = \int_{\mathbb{R}^3} \! \psi_1^*(\mathbf{x},t)\psi_2(\mathbf{x},t) \, \diff \mathbf{x}, \quad \forall \psi_1,\psi_2\in H.
\end{equation}
Allora la condizione di normalizzazione diventa
\begin{equation}
\norm{\psi(\cdot,t)} = 1, \quad \mbox{dove} \quad \norm{\psi} = \sqrt{\left\langle \psi,\psi \right\rangle}.
\end{equation}
Inoltre, a partire da questa interpretazione, è possibile affermare che ogni \textbf{stato} del sistema è rappresentato da un elemento $\psi\in H$ che soddisfi la condizione di normalizzazione.
\begin{osservazione}
Poiché le grandezze fisiche dipendono dalla funzione d'onda attraverso il suo modulo, si ha che la funzione d'onda $\psi$ e la funzione d'onda $\psi e^{i\alpha}$, con $\alpha\in\mathbb{R}$, definiscono lo stesso stato quantistico, cioè la funzione d'onda è sempre definita a meno di un fattore di fase del tipo $e^{i\alpha}$. Ciò spiega l'ipotesi 2. 
\end{osservazione}
\begin{osservazione}
Poiché vale il principio di sovrapposizione degli stati, si ha che tutte le equazioni cui soddisfano le funzioni d'onda devono necessariamente essere lineari rispetto alla funzione d'onda $\psi$.
\end{osservazione}
In meccanica quantistica non è possibile misurare $\mathbf{x}$ esattamente e la particella non è mai localizzata in un punto ben preciso. Quindi, in virtù dell'interpretazione probabilistica, la posizione $\mathbf{x}$ della particella è una variabile aleatoria la cui previsione è data da
\begin{equation}
E_{\psi}^t(\mathbf{x}) = \int_{\mathbb{R}^3} \! \mathbf{x} \vert \psi(\mathbf{x},t) \vert^2 \, \diff \mathbf{x}, \quad \forall t \in \mathbb{R}.
\end{equation}
e la varianza è data da
\begin{equation}
[\Delta_{\psi}^t(\mathbf{x})]^2 = E_{\psi}^t(\mathbf{x}^2) - [E_{\psi}^t(\mathbf{x})]^2, \quad \forall t \in \mathbb{R}.
\end{equation}

\subsection{Misure di una grandezza fisica}
Si consideri una grandezza fisica $f$ a valori reali, che in meccanica quantistica è detta anche \textbf{osservabile} (ad esempio la posizione, il momento, l'energia, etc.). Generalmente, in meccanica classica essa può assumere una distribuzione continua di valori. Come si vedrà in questo paragrafo, in meccanica quantistica, invece, i valori che l'osservabile può assumere non sono, in generale, distribuiti con continuità.

In meccanica quantistica ogni osservabile classica $f$ corrisponde ad un operatore lineare $F$ definito sullo spazio di Hilbert $H$. Allora, se il sistema è rappresentato dallo stato $\psi$, la previsione di $F$ è data da
\begin{equation}
E_\psi(F) = \left\langle \psi,F\psi \right\rangle .
\end{equation}
Inoltre, poiché $E_\psi(F)$ rappresenta una grandezza fisica allora deve essere un numero reale e di conseguenza l'operatore $F$ deve essere simmetrico. Analogamente, la varianza è data da
\begin{equation}
[\Delta_{\psi}(F)]^2 = E_{\psi}(F^2) - [E_{\psi}(F)]^2.
\end{equation}
Si osservi che, in modo equivalente, si può definire la varianza come
\begin{equation}
[\Delta_{\psi}(F)]^2 = \norm{ (F-E_\psi(F))\psi }^2.
\end{equation}
Allora $\Delta_{\psi}(F)=0$ se e solo se $\psi$ soddisfa l'equazione agli autovalori
\begin{equation}
F\psi=E_\psi(F)\psi,
\end{equation}
ovvero se e solo se $\psi$ è un autovettore di $F$ corrispondente all'autovalore $E_\psi(F)$ e ciò avviene se e solo se $E_\psi(F)$ appartiene allo spettro di $F$. Quindi risulta fondamentale conoscere la rappresentazione spettrale dell'operatore $F$, perché per mezzo di essa è possibile decomporre lo stato $\psi$ in un certo numero di stati, i cui corrispondenti operatori hanno varianza nulla, e i quali possono presentarsi con una probabilità ben determinata.

\subsection{La notazione bra-ket}
Una notazione molto usata in meccanica quantistica per descrivere gli stati è la notazione bra-ket. Essa è stata introdotta da P.~Dirac nel 1939. Sia $H$ uno spazio di Hilbert complesso su cui è definito il prodotto scalare $\left<\cdot,\cdot\right>$ e sia $\alpha\in H$ un suo elemento. Solitamente si prende $H=L^2(\mathbb{R}^3)$ e come prodotto scalare quello definito dalla \eqref{EQ:CAP1:Ptto_scalare_L2_f_o}. Da ora in avanti un generico stato del sistema sarà indicato con il simbolo $\ket{\alpha}$, che sarà chiamato \textbf{ket}. In virtù delle ipotesi formulate nel Paragrafo \ref{PAR:CAP1:Spazio_f_d_o}, per i ket valgono le seguenti proprietà.
\begin{enumerate}
\item $\ket{0}$ non è associato a nessuno stato.
\item $\ket{\alpha}$ e $c\ket{\alpha}$ con $c\in\mathbb{C}$ rappresentano lo stesso stato.
\item $a\ket{\alpha} + b\ket{\beta}\in H$ con $\ket{\alpha}, \ket{\beta} \in H$ e $a,b\in \mathbb{C}$.
\end{enumerate}
Sia $H^*$ lo spazio duale di $H$, cioè l'insieme dei funzionali lineari e continui da $H$ in $\mathbb{C}$. Esso è uno spazio lineare complesso. Fissato $\ket{\alpha} \in H$, sia $\ket{\beta}$ un generico elemento di $H$. \`{E} possibile definire, utilizzando il prodotto scalare di $H$, il funzionale lineare e continuo
\begin{equation*}
\varphi_\alpha : H \longrightarrow \mathbb{C} \quad \mbox{che associa} \quad \ket{\beta} \longmapsto \left< \ket{\alpha}, \ket{\beta} \right>.
\end{equation*}
Si è ottenuta allora la corrispondenza
\begin{equation}\label{EQ:CAP1:Corr_bra_ket}
H \longrightarrow H^* \quad \mbox{che associa} \quad \ket{\alpha} \longmapsto \varphi_\alpha.
\end{equation}
Viceversa, sia $\varphi\in H^*$. Allora, per il Teorema di rappresentazione di Riesz, si ha che esiste un unico elemento $\ket{ \alpha}\in H$ tale che
\begin{equation*}
\varphi \left( \ket{\beta} \right) = \left< \ket{\alpha}, \ket{\beta} \right>.
\end{equation*}
Allora la corrispondenza definita con la \eqref{EQ:CAP1:Corr_bra_ket} è biettiva. Inoltre, un generico elemento di $H^*$ è chiamato \textbf{bra} e si indica con la notazione $\bra{\beta}$. Allora, sottintendendo la dipendenza dal funzionale, si ha la corrispondenza biettiva
\begin{equation}
H \longleftrightarrow H^* \quad \mbox{che associa} \quad \ket{\alpha} \longleftrightarrow \bra{\alpha}.
\end{equation}
Si osservi che si è soliti indicare gli elementi tra loro corrispondenti con la stessa lettera. Infine, lo spazio $H^*$ si può rendere spazio di Hilbert introducendo in esso il prodotto scalare $\left< \cdot , \cdot \right>^*$ definito da
\begin{equation*}
\left< \bra{\alpha}, \bra{\beta} \right>^* = \left< \ket{\alpha} , \ket{\beta} \right>, \qquad \forall \bra{\alpha}, \bra{\beta}\in H^*.
\end{equation*}
\`{E} possibile definire il prodotto di un bra per un ket mediante la relazione
\begin{equation}
\braket{\alpha | \beta} = \left< \ket{\alpha}, \ket{\beta} \right>, \qquad \forall \bra{\alpha}\in H^*, \, \forall \ket{\beta}\in H.
\end{equation}
Sia $A$ un operatore lineare definito nello spazio dei ket $H$. Applicando l'operatore $A$ al ket $\ket{\alpha}$ si avrebbe $A\cdot \ket{\alpha}$. Da ora in avanti per indicare tale operazione si scriverà semplicemente $A\ket{\alpha}$. Il risultato è ancora un ket, cioè si ha $A\ket{\alpha} = \ket{\beta}$. Inoltre, per l'operatore $A$ definito sui ket si può definire la sua azione sui bra, che si indica con $\bra{\beta} A$, mediante la relazione
\begin{equation}
\left( \bra{\beta} A \right) \ket{\alpha} =  \bra{\beta} \left( A \ket{\alpha} \right), \qquad \forall  \ket{\alpha} \in H.
\end{equation}
Si osservi che il membro di destra risulta ben definito in quanto $A \ket{\alpha}$ è un ket e si è definito il prodotto di un bra per un ket. Nella notazione usuale le parentesi vengono rimosse scrivendo semplicemente $\braket{\beta | A | \alpha}$. Un'altra nozione importante è la seguente. Dato un operatore $A$, se esistono $a\in \mathbb{C}$ e $\ket{\alpha}\in H$ tali che
\begin{equation*}
A \ket{\alpha} = a \ket{\alpha}
\end{equation*}
allora si dice che $a$ è un \textbf{autovalore} di $A$ e che $\ket{\alpha}$ è un \textbf{autoket} di $A$. I corrispondenti degli autoket nello spazio dei bra sono detti \textbf{autobra}.

\subsection{Rappresentazione spettrale delle osservabili}
Sia $f$ un'osservabile a cui corrisponde l'operatore $F$ definito nello spazio di Hilbert $H$. Se $\sigma(F)$ è il suo spettro, allora, come specificato nel Paragrafo , vale la suddivisione
\begin{equation}
\sigma(F) = \sigma_d(F) \cup \sigma_c(F) \cup \sigma_r(F).
\end{equation}
Sia $\ket{\alpha}\in H$ un generico stato del sistema per il quale si vuole misurare l'osservabile $f$, cioè si vuole valutare $F\ket{\alpha}$. In base alle ipotesi sull'operatore $F$ e alla dimensione dello spazio $H$ ciascun ket avrà una particolare rappresentazione.

Nel caso in cui $F$ sia un operatore compatto ed autoaggiunto, per il Corollario \ref{COR:CAP1:Hilbert_Schmidt}, è possibile determinare una base ortonormale in $H$ costituita da autovettori di $F$, qui detti autoket. Tale base può avere un numero finito o numerabile di elementi, a seconda che $0\in\sigma(F)$ oppure no. Siano $f_n$, $n=0,1,2,\ldots$ gli autovalori associati ad $F$ e siano $\ket{\alpha_n}$, $n=0,1,2,\ldots$ gli elementi della corrispondente base di autoket, con l'indice $n$ eventualmente variabile in un insieme finito. Allora il generico ket $\ket{\alpha}$ si può scrivere come combinazione lineare di autoket, cioè
\begin{equation}
\ket{\alpha} = \sum_n c_n \ket{\alpha_n},
\end{equation}
dove i coefficienti $c_n$ sono dati da $c_n = \braket{\alpha_n | \alpha}$. Inoltre la proprietà di ortonormalità può essere espressa scrivendo
\begin{equation}
\braket{\alpha_n | \alpha_m} = \delta_{nm},
\end{equation}
dove $\delta_{nm}$ è il simbolo di Kronecker. Infine, sfruttando la rappresentazione spettrale di $F$, si ha
\begin{equation}
F \ket{\alpha} = \sum_n f_n \ket{\alpha_n} \braket{\alpha_n | \alpha}.
\end{equation}
Moltiplicando a sinistra la relazione precedente per $\bra{\alpha}$, si ottiene
\begin{align*}
\braket{\alpha | F | \alpha} & = \sum_n f_n \braket{\alpha | \alpha_n} \braket{\alpha_n | \alpha}=\\
& = \sum_n f_n \left\vert \braket{\alpha_n | \alpha} \right\vert^2.
\end{align*}
Alla relazione precedente si dà la seguente interpretazione fisica. Se gli $f_n$ sono i possibili valori che l'osservabile $f$ può assumere e la quantità $\left\vert \braket{\alpha_n | \alpha} \right\vert^2$ rappresenti la probabilità che il risultato della misura di $\ket{\alpha}$ sia $\ket{\alpha_n}$, allora $\braket{\alpha | F | \alpha}$ si può interpretare come il valore di aspettazione dell'operatore $F$ sullo stato $\ket{\alpha}$ e si pone $\braket{F}=\braket{\alpha | F | \alpha}$.

Si consideri adesso il caso in cui l'operatore $F$ sia autoaggiunto e abbia spettro continuo. Anche qui un generico ket $\ket{\alpha}$ può essere scritto come sovrapposizione di autoket nella forma
\begin{equation}
\ket{\alpha} = \int_{\sigma_c(F)} \! c(\lambda) \ket{\alpha_\lambda} \, \diff E(\lambda),
\end{equation}
dove $\lambda\in\sigma_c(F)$, $\left\lbrace E(\lambda)\right\rbrace $ è una famiglia spettrale e i coefficienti $c(\lambda)$ sono dati da $c(\lambda) = \braket{\alpha_\lambda | \alpha}$. Inoltre, in questo caso si impone che la proprietà di ortonormalità sia
\begin{equation}
\braket{\alpha_{\lambda '} | \alpha_\lambda} = \delta(\lambda - \lambda '),
\end{equation}
dove il termine al secondo membro è la delta di Dirac. Infine, sfruttando la \eqref{EQ:CAP1:Rappresentazione_spettrale_3}, si ottiene la rappresentazione spettrale di $F$
\begin{equation}
F\ket{\alpha} = \int_{\sigma_c(F)} \! \lambda \ket{\alpha_\lambda} \braket{\alpha_\lambda | \alpha} \, \diff E(\lambda).
\end{equation}
Analogamente al caso discreto, moltiplicando a sinistra la relazione precedente per $\bra{\alpha}$, si ottiene
\begin{align*}
\braket{\alpha | F | \alpha} & = \int_{\sigma_c(F)} \! \lambda \braket{\alpha | \alpha_\lambda} \braket{\alpha_\lambda | \alpha} \, \diff E(\lambda)=\\
& = \int_{\sigma_c(F)} \! \lambda \left\vert \braket{\alpha_\lambda | \alpha} \right\vert^2 \, \diff E(\lambda).
\end{align*}
Alla relazione precedente si dà la seguente interpretazione fisica. Se $\lambda$ è ogni possibile valore che l'osservabile $f$ può assumere e la quantità $\left\vert \braket{\alpha_\lambda | \alpha} \right\vert^2$ rappresenta la distribuzione di probabilità del risultato della misura di $\ket{\alpha}$, allora $\braket{\alpha | F | \alpha}$ si può interpretare come il valore di aspettazione dell'operatore $F$ sullo stato $\ket{\alpha}$ e si pone $\braket{F}=\braket{\alpha | F | \alpha}$.

Infine, per un generico operatore $F$ avente uno spettro in parte discreto e in parte continuo, un generico $\ket{\alpha}$ si scrive
\begin{equation}
\ket{\alpha} = \sum_n \braket{\alpha_n | \alpha} \ket{\alpha} + \int_{\sigma_c(F)} \! \braket{\alpha_\lambda | \alpha} \ket{\alpha_\lambda} \, \diff E(\lambda).
\end{equation}

\subsection{Le osservabili di posizione e di impulso}
Le grandezze fisiche di interesse per gli scopi di questa tesi sono principalmente la posizione, l'impulso e l'energia. Pensando queste grandezze come osservabili quantistiche, è possibile associare ad esse degli operatori lineari. Si consideri in primo luogo il caso unidimensionale. Indicando con $x$ l'operatore posizione, si vuole misurare la posizione occupata da una particella che si trova nello stato $\ket{\alpha}$. Si postula che gli autoket dell'operatore posizione che soddisfano l'equazione agli autovalori
\begin{equation}
x\ket{x'}=x'\ket{x'}
\end{equation}
formino una base. Allora è possibile rappresentare lo stato $\ket{\alpha}$ come
\begin{equation}
\ket{\alpha} = \int_{-\infty}^{+\infty} \! \ket{x'}\braket{x' | \alpha} \, \diff x'.
\end{equation}
Si vuole vedere qual è il significato fisico della quantità $\braket{x' | \alpha}$. Allora, calcolando il valore di aspettazione della posizione nello stato $\ket{\alpha}$ si ha
\begin{align*}
\braket{\alpha | x | \alpha} & = \int_{-\infty}^{+\infty} \! \braket{\alpha | x | x'}\braket{x' | \alpha} \, \diff x' = \\
& = \int_{-\infty}^{+\infty} \! x'\braket{\alpha | x'}\braket{x' | \alpha} \, \diff x' =\\
& = \int_{-\infty}^{+\infty} \! x' \left\vert \braket{x' | \alpha} \right\vert^2 \, \diff x'.
\end{align*}
Allora la quantità $\left\vert \braket{x' | \alpha} \right\vert^2$ rappresenta la distribuzione di probabilità della posizione della particella avente lo stato $\ket{\alpha}$. Di conseguenza $\braket{x' | \alpha}$ è proprio la funzione d'onda per lo stato fisico rappresentato da $\ket{\alpha}$, cioè
\begin{equation}
\psi_\alpha(x') = \braket{x' | \alpha},
\end{equation}
e inoltre la probabilità di trovare la particella in un certo intervallo $[a,b]$ è data da
\begin{equation}
P(a, b) = \int_a^b \! \left\vert \braket{x' | \alpha} \right\vert^2 \, \diff x'.
\end{equation}
Generalizzando in tre dimensioni si assume che gli autoket di posizione $\ket{\mathbf{x}'}$ formino un insieme completo. Il generico ket di stato di una particella si può scrivere come
\begin{equation}
\ket{\alpha} = \int \! \ket{\mathbf{x}'}\braket{\mathbf{x}' | \alpha} \, \diff \mathbf{x},
\end{equation}
dove $\ket{\mathbf{x}'}$ è un autoket simultaneo delle osservabili $x$, $y$ e $z$, cioè
\begin{equation}
\ket{\mathbf{x}'} = \ket{x',y',z'},
\end{equation}
ciascuna rispondente ad una delle equazioni agli autovalori
\begin{equation}
x\ket{x'}=x'\ket{x'}, \qquad y\ket{y'}=y'\ket{y'}, \qquad z\ket{z'}=z'\ket{z'}.
\end{equation}

Per quanto riguarda la deduzione dell'operatore impulso, si consideri in questa sede il caso in cui lo stato del sistema sia rappresentabile attraverso un'onda piana stazionaria, cioè del tipo
\begin{equation}
\psi(\mathbf{x})=e^{i\mathbf{k}\cdot\mathbf{x}}.
\end{equation}
Allora, si calcoli
\begin{equation}
\nabla_{\mathbf{x}} \psi(\mathbf{x}) =i\mathbf{k}\psi(\mathbf{x}) = i \frac{\hbar \mathbf{k}}{\hbar} \psi(\mathbf{x}) = i \frac{\mathbf{p}}{\hbar} \psi(\mathbf{x}).
\end{equation}
Supponendo che tale relazione sia valida in generale per un generico vettore di stato $\ket{\psi}$, si ha
\begin{equation}
\nabla_\mathbf{x} \ket{\psi} = i \frac{\mathbf{p}}{\hbar} \ket{\psi}.
\end{equation}
Di conseguenza si ha
\begin{equation}\label{EQ:CAP1:Operatore_impulso_3d}
\mathbf{p}\ket{\psi} = -i\hbar\nabla_{\mathbf{x}}\ket{\psi}.
\end{equation}
Considerando nuovamente il caso unidimensionale, gli autostati dell'operatore impulso soddisfano l'equazione agli autovalori
\begin{equation}\label{EQ:CAP1:Eq_autov_impulso}
p\ket{p'} = p'\ket{p'}
\end{equation}
e si suppone inoltre che tali autostati formino un insieme completo. Le funzioni d'onda corrispondenti agli autostati dell'operatore impulso, cioè le quantità $\braket{x' | p'}$, sono dette autofunzioni dell'operatore impulso nella rappresentazione delle coordinate. Nel caso unidimensionale la \eqref{EQ:CAP1:Operatore_impulso_3d} diventa
\begin{equation*}
p \ket{\psi} = -i\hbar \frac{\partial}{\partial x} \ket{\psi}
\end{equation*}
Scrivendo quest'ultima equazione per un autostato dell'operatore impulso $p'$ e moltiplicandola per $\bra{x'}$, si ha
\begin{equation*}
\braket{x' | p | p'} = -i\hbar \frac{\partial}{\partial x'} \braket{x' | p'}.
\end{equation*}
Utilizzando la relazione \eqref{EQ:CAP1:Eq_autov_impulso} si ottiene
\begin{equation}
-i\hbar \frac{\partial}{\partial x'} \braket{x' | p'} = p' \braket{x' | p'}.
\end{equation}
Essa è un'equazione differenziale ordinaria del primo ordine, la cui soluzione è
\begin{equation}
\psi_{p'}(x') = \braket{x' | p'} = Ne^{\frac{i}{\hbar}p'x'}.
\end{equation}
Per ottenere la costante $N$ si consideri la relazione di ortogonalità
\begin{equation}
\braket{p' | p''} = \delta (p'-p'').
\end{equation}
Allora si ha
\begin{align*}
\delta (p'-p'') & = \braket{p' | p''} = \\
& = \int_{-\infty}^{+\infty} \! \braket{p' | x'}\braket{x' | p''} \, \diff x' = \\
& = \vert N \vert^2 \int_{-\infty}^{+\infty} \! e^{-\frac{i}{\hbar}(p'-p'')x'} \, \diff x' = \\
& = \vert N \vert^2 2\pi\hbar \delta (p'-p''),
\end{align*}
dove l'ultimo passaggio si è ottenuto utilizzando la rappresentazione di Fourier della delta di Dirac, cioè
\begin{equation}
\delta(x) = \frac{1}{2\pi} \int_{-\infty}^{+\infty} \! e^{ikx} \, \diff k.
\end{equation}
Allora la relazione per la costante $N$ è
\begin{equation}
\vert N \vert^2 2\pi\hbar = 1.
\end{equation}
Scegliendo per convenzione $N$ reale e positivo, si ottiene
\begin{equation}
\psi_{p'}(x') = \braket{x' | p'} = \frac{1}{\sqrt{2\pi\hbar}} e^{\frac{i}{\hbar}p'x'}.
\end{equation}
Si osservi che un generico vettore di stato $\ket{\alpha}$ si può rappresentare mediante gli autostati dell'operatore impulso scrivendo
\begin{equation}
\ket{\alpha} = \int_{-\infty}^{+\infty} \! \ket{p'} \braket{p' | \alpha} \, \diff p',
\end{equation}
dove il generico coefficiente di questo sviluppo, cioè la funzione
\begin{equation}
\varphi_\alpha (p') = \braket{p' | \alpha},
\end{equation}
è detta funzione d'onda nella rappresentazione degli impulsi. Analogamente al caso della posizione, $\vert \varphi_\alpha (p') \vert^2$ rappresenta la distribuzione di probabilità dei valori dell'impulso. Adesso si calcoli $\braket{x' | \alpha}$ nella rappresentazione degli impulsi, cioè
\begin{align*}
\psi_\alpha (x') & = \braket{x' | \alpha} = \int_{-\infty}^{+\infty} \! \braket{x' | p'} \braket{p' | \alpha} \, \diff p' = \\
& = \frac{1}{\sqrt{2\pi\hbar}} \int_{-\infty}^{+\infty} \! e^{\frac{i}{\hbar}p'x'} \braket{p' | \alpha} \, \diff p' = \\
& = \frac{1}{\sqrt{2\pi\hbar}} \int_{-\infty}^{+\infty} \! e^{\frac{i}{\hbar}p'x'} \varphi_\alpha (p') \, \diff p'.
\end{align*}
Analogamente si può ottenere la trasformazione inversa, cioè
\begin{equation}
\varphi_\alpha (p')= \frac{1}{\sqrt{2\pi\hbar}} \int_{-\infty}^{+\infty} \! e^{-\frac{i}{\hbar}p'x'} \psi_\alpha (x') \, \diff x'.
\end{equation}
L'importante osservazione che ne segue è che le due rappresentazioni corrispondono matematicamente, rispettivamente, alla trasformata e all'antitrasformata di Fourier. Tali rappresentazioni possono essere scritte anche in tre dimensioni e presentano una forma analoga (cfr. \citep{BOOK:Sakurai}).

\subsection{Osservabili misurabili simultaneamente}
Si considerino due osservabili $f$ e $g$ e i due operatori associati $F$ e $G$. Se le due osservabili possono essere misurate simultaneamente allora esse devono insistere sugli stessi autostati $\ket{\psi_n}$. Allora è possibile calcolare il prodotto tra i due operatori $F$ e $G$ scrivendo
\begin{equation}
FG\ket{\psi} = FG\sum_n a_n \ket{\psi_n} = \sum_n f_n g_n a_n \ket{\psi_n}.
\end{equation}
Analogamente, si calcola
\begin{equation}
GF\ket{\psi} = GF\sum_n a_n \ket{\psi_n} = \sum_n g_n f_n a_n \ket{\psi_n}.
\end{equation}
\begin{definizione}
Si dice che due operatori $F$ e $G$ \textbf{commutano} se e solo se si ha
\begin{equation}
[F,G] = FG-GF=0.
\end{equation}
Il simbolo $[F,G]$ è chiamato \textbf{commutatore} degli operatori $F$ e $G$.
\end{definizione}
Ne segue che date due osservabili $f$ e $g$ e dati gli operatori associati $F$ e $G$ allora le due osservabili sono misurabili simultaneamente se e solo se i due operatori associati commutano tra loro.

Ad esempio, se $\mathbf{x}=(x_1,x_2,x_3)$ e $\mathbf{p}=(p_1,p_2,p_3)$, il commutatore di $x_i$ e $p_j$ è
\begin{align*}
[x_i,p_j] \ket{\psi} & = (x_ip_j-p_jx_i)\ket{\psi} = \\
& = -x_i i\hbar\frac{\partial}{\partial x_j}\ket{\psi} +i\hbar \frac{\partial}{\partial x_j} (x_i\ket{\psi})= \\
& = -x_i i\hbar\frac{\partial}{\partial x_j}\ket{\psi} + i\hbar \delta_{ij}\ket{\psi} + x_i i\hbar\frac{\partial}{\partial x_j}\ket{\psi} = \\
& = i\hbar \delta_{ij}\ket{\psi}.
\end{align*}
Allora, per $i\neq j$ si ha
\begin{equation}
[x_i,p_j] = 0,
\end{equation}
invece
\begin{equation}
[x_1,p_1]=[x_2,p_2]=[x_3,p_3]=i\hbar.
\end{equation}
Analogamente si dimostra che
\begin{equation}
[x_i,x_j] = 0, \qquad [p_i,p_j] = 0.
\end{equation}
Quindi le componenti della posizione e dell'impulso sono osservabili compatibili eccetto quando si vuole misurare simultaneamente la posizione e l'impulso per la stessa componente.

\subsection{Deduzione del principio di indeterminazione di Heisenberg}
Si consideri una particella in una dimensione e siano $\psi(x)$ e $\varphi(p)$ le rappresentazioni della sua funzione d'onda rispettivamente nello spazio delle coordinate e nello spazio degli impulsi. Interpretando $\vert \psi(x) \vert^2$ e $\vert \varphi(p) \vert^2$ come distribuzioni di probabilità, è possibile definire il \textbf{valore medio} della posizione e dell'impulso ponendo
\begin{equation}
\bar{x} = \int_{-\infty}^{+\infty} \! x \vert \psi(x) \vert^2 \, \diff x, \qquad
\bar{p} = \int_{-\infty}^{+\infty} \! p \vert \varphi(p) \vert^2 \, \diff p,
\end{equation}
e inoltre è possibile definire anche le quantità dette \textbf{incertezza sulla posizione} e \textbf{incertezza sull'impulso}, ponendo
\begin{equation}\label{EQ:CAP1:Heis_inc_rel}
(\Delta x)^2 = \frac{\int_{-\infty}^{+\infty} \! (x-\bar{x})^2 \vert \psi(x) \vert^2 \, \diff x}{\int_{-\infty}^{+\infty} \! \vert \psi(x) \vert^2 \, \diff x}, \qquad
(\Delta p)^2 = \frac{\int_{-\infty}^{+\infty} \! (p-\bar{p})^2 \vert \varphi(p) \vert^2 \, \diff p}{\int_{-\infty}^{+\infty} \! \vert \varphi(p) \vert^2 \, \diff p},
\end{equation}
si tratta cioè degli errori relativi sulla posizione e sull'impulso ottenuti dividendo la varianza della grandezza per la sua norma. Allora si può dimostrare il seguente risultato.
\begin{teorema}[Relazione di indeterminazione di Heisenberg]
Sia $\psi\in L^2(\mathbb{R})$. Allora
\begin{equation}
(\Delta x)^2(\Delta p)^2 \geq \frac{\hbar^2}{4},
\end{equation}
dove $(\Delta x)^2$ e $(\Delta p)^2$ sono definite nella \eqref{EQ:CAP1:Heis_inc_rel}.
\end{teorema}
\begin{proof}
Si assuma che $\psi$ è continua e regolare a tratti, e che le funzioni $x\psi(x)$ e $\psi'(x)$ siano in $L^2(\mathbb{R})$. Per prima cosa si procede dimostrando il caso $\bar{x}=\bar{p}=0$. Allo scopo, si consideri
\begin{equation}
\int_{-\infty}^{+\infty} \! x \psi^*(x)\psi ' (x) \,\diff x
\end{equation}
e si effettui una integrazione per parti, ottenendo
\begin{align*}
\int_{-\infty}^{+\infty} \! x \psi^*(x)\psi ' (x) \,\diff x & = 
\left[ x \psi^*(x) \psi(x) \right] _{-\infty}^{+\infty} -
\int_{-\infty}^{+\infty} \! \left( x\psi^*(x) \right) ' \psi(x)  \,\diff x=\\
& = \left[ x \vert \psi(x) \vert^2 \right] _{-\infty}^{+\infty} -
\int_{-\infty}^{+\infty} \! \vert \psi(x) \vert^2 \,\diff x -
\int_{-\infty}^{+\infty} \! x (\psi^*(x))'\psi(x) \,\diff x.
\end{align*}
Essendo $x \vert \psi(x) \vert^2 \in L^2(\mathbb{R})$, allora $\left[ x \vert \psi(x) \vert^2 \right] _{-\infty}^{+\infty} =0$ e infine si ha che
\begin{align*}
\int_{-\infty}^{+\infty} \! \vert \psi(x) \vert^2 \,\diff x & = 
- \int_{-\infty}^{+\infty} \! x \psi^*(x)\psi ' (x) \,\diff x 
- \int_{-\infty}^{+\infty} \! x (\psi^*(x))'\psi(x) \,\diff x =\\
& = -2 \mathrm{Re} \int_{-\infty}^{+\infty} \! x\psi^*(x)\psi'(x) \,\diff x.
\end{align*}
Elevando entrambi i membri al quadrato ed applicando la diseguaglianza di Cauchy-Schwarz si ottiene che
\begin{equation}\label{EQ:CAP1:Heis_rel_ind_1}
\left( \int_{-\infty}^{+\infty} \! \vert \psi(x) \vert^2 \,\diff x \right)^2 \leq 4 \left( \int_{-\infty}^{+\infty} \! x^2\vert \psi(x) \vert^2 \,\diff x \right) \left( \int_{-\infty}^{+\infty} \! \vert \psi '(x) \vert^2 \,\diff x \right). 
\end{equation}
Infine, per il teorema di Plancherel, si ha
\begin{equation}\label{EQ:CAP1:Heis_rel_ind_2}
\int_{-\infty}^{+\infty} \! \vert \psi(x) \vert^2 \,\diff x = \frac{1}{2\pi} \int_{-\infty}^{+\infty} \! \vert \hat{\psi}(p) \vert^2 \,\diff p = \frac{1}{2\pi} \int_{-\infty}^{+\infty} \! \vert \varphi(p) \vert^2 \,\diff p.
\end{equation}
Allora applicando questa relazione a $\psi'(x)$, si ottiene
\begin{equation}\label{EQ:CAP1:Heis_rel_ind_3}
\begin{aligned}
\int_{-\infty}^{+\infty} \! \vert \psi '(x) \vert^2 \,\diff x & = \frac{1}{2\pi} \int_{-\infty}^{+\infty} \! \vert \hat{[\psi']}(p) \vert^2 \, \diff p = \frac{1}{2\pi} \int_{-\infty}^{+\infty} \! \frac{p^2}{\hbar^2} \vert \hat{\psi}(p) \vert^2 \, \diff p = \\
&=\frac{1}{2\pi\hbar^2} \int_{-\infty}^{+\infty} \! \vert \varphi(p) \vert^2 \, \diff p
\end{aligned}
\end{equation}
avendo usato la relazione $\hat{[\psi']}=\frac{i}{\hbar}p\hat{\psi}$. Allora sostituendo la \eqref{EQ:CAP1:Heis_rel_ind_2} e la \eqref{EQ:CAP1:Heis_rel_ind_3} nella \eqref{EQ:CAP1:Heis_rel_ind_1}, si ottiene
\begin{align*}
\left(\int_{-\infty}^{+\infty} \! \vert \psi(x) \vert^2 \,\diff x \right) & \left( \frac{1}{2\pi} \int_{-\infty}^{+\infty} \! \vert \varphi(p) \vert^2 \,\diff p \right) \\
& \leq 4 \left( \int_{-\infty}^{+\infty} \! x^2\vert \psi(x) \vert^2 \,\diff x \right) \left( \frac{1}{2\pi\hbar^2} \int_{-\infty}^{+\infty} \! \vert \varphi(p) \vert^2 \, \diff p \right),
\end{align*}
da cui segue la tesi. Il caso generale può essere ricondotto a quello appena trattato attraverso un cambiamento di variabile. Infatti ponendo
\begin{equation}
\Psi(x) = e^{-i\bar{p}x}\psi(x+\bar{x}),
\end{equation}
si ottiene che $\Delta_\Psi x = \Delta x$ e che $\Delta_{\hat{\Psi}} p =\Delta p$ (cfr. \citep{BOOK:Folland_Fourier}).
\end{proof}

\subsection{L'equazione di Schr\"{o}dinger}
La funzione d'onda si può ottenere come soluzione di un'equazione differenziale alle derivate parziali. In questo paragrafo si ricaverà tale equazione in un caso particolare, ovvero quello di un elettrone libero, soggetto ad un potenziale costante $V_0$. In meccanica classica la sua dinamica è determinata dalla funzione hamiltoniana \eqref{EQ:CAP1:Hamiltoniana_elettrone_libero} che scritta per il potenziale $V_0$ diviene
\begin{equation}\label{EQ:CAP1:Hamiltoniana_elettrone_libero_V_cost}
\mathcal{H}(\mathbf{x},\mathbf{p})=\frac{\vert \mathbf{p}\vert^2}{2m_e}-eV_0.
\end{equation}
Come già visto nel Paragrafo \ref{PAR:CAP1:Onde_materiali}, all'elettrone è possibile associare un'onda piana avente vettore d'onda $\mathbf{k}$ e pulsazione $\omega$, per cui la funzione d'onda sarà
\begin{equation}\label{EQ:CAP1:Funzione_d_onda}
\psi(\mathbf{x},t) = e^{i(\mathbf{k}\cdot\mathbf{x}-\omega t)}.
\end{equation}
Derivando la \eqref{EQ:CAP1:Funzione_d_onda} rispetto a $t$ e calcolando il laplaciano rispetto a $\mathbf{x}$ si ottengono le seguenti due equazioni
\begin{equation}\label{EQ:CAP1:Funzione_d_onda_Dt}
\frac{\partial \psi}{\partial t} = -i\omega\psi,
\end{equation}
\begin{equation}\label{EQ:CAP1:Funzione_d_onda_Dx}
\Delta_{\mathbf{x}} \psi = -\vert \mathbf{k} \vert^2 \psi.
\end{equation}
Utilizzando la relazione di Planck-Einstein \eqref{EQ:CAP1:Rel_Planck_Einstein} nella \eqref{EQ:CAP1:Funzione_d_onda_Dt}, si ottiene la seguente equazione
\begin{equation}\label{EQ:CAP1:Funzione_d_onda_e_P_E}
i\hbar\frac{\partial \psi}{\partial t} = E\psi,
\end{equation}
la quale permette di legare l'energia di un elettrone libero alla derivata temporale della funzione d'onda. Utilizzando la legge di de Broglie \eqref{EQ:CAP1:Legge_di_de_Broglie} nella \eqref{EQ:CAP1:Funzione_d_onda_Dx}, si ha
\begin{equation}
\Delta_{\mathbf{x}} \psi = - \frac{\vert \mathbf{p} \vert^2}{\hbar^2} \psi,
\end{equation}
da cui, ricavando $\vert \mathbf{p} \vert^2$ dalla \eqref{EQ:CAP1:Hamiltoniana_elettrone_libero_V_cost} e sostituendo, si ottiene
\begin{equation}\label{EQ:CAP1:Funzione_d_onda_Dx_mod}
-\frac{\hbar}{2m_e} \Delta_{\mathbf{x}}\psi-eV_0\psi=\mathcal{H}(\mathbf{x},\mathbf{p})\psi.
\end{equation}
Introducendo l'operatore hamiltoniano definito come
\begin{equation}
H=-\frac{\hbar}{2m_e} \Delta_{\mathbf{x}}-eV_0
\end{equation}
la \eqref{EQ:CAP1:Funzione_d_onda_Dx_mod} si può scrivere come
\begin{equation}\label{EQ:CAP1:Legame_hamiltoniana_op_hamiltoniano}
H\psi=\mathcal{H}(\mathbf{x},\mathbf{p})\psi.
\end{equation}
la funzione hamiltoniana classica non ha un significato diretto in meccanica quantistica perché, per quanto visto nel Paragrafo \ref{PAR:CAP1:Principi_fond}, non è possibile conoscere contemporaneamente con precisione le coordinate di $\mathbf{x}$ e di $\mathbf{p}$. Relativamente al caso dell'elettrone libero, questa relazione esprime un legame tra la funzione hamiltoniana classica $\mathcal{H}$ e l'operatore $H$, che continua ad essere ben definito.

Utilizzando la relazione \eqref{EQ:CAP1:Relazione_classica_energia} e poi la \eqref{EQ:CAP1:Legame_hamiltoniana_op_hamiltoniano} nella \eqref{EQ:CAP1:Funzione_d_onda_e_P_E}, si ottiene
\begin{equation}\label{EQ:CAP1:Equazione_di_Schroedinger}
i\hbar\frac{\partial \psi}{\partial t} = H\psi,
\end{equation}
che è nota come \textbf{equazione di Schr\"{o}dinger}. Essa è stata ricavata prima in un caso particolare. In generale, viene postulata la sua validità per descrivere il comportamento di una particella elementare oppure di un sistema quantistico. Inoltre, per la descrizione completa, è necessario imporre delle condizioni iniziali
\begin{equation}\label{EQ:CAP1:Equazione_di_Schroedinger_Cond_ini}
\psi(0,\mathbf{x})=\psi\ped{ini}(\mathbf{x}).
\end{equation}

\subsection{Il caso stazionario}
In questo paragrafo si tratterà lo studio dell'equazione di Schr\"{o}dinger nel caso stazionario, cioè in cui l'hamiltoniana non dipenda esplicitamente dal tempo. Si consideri l'equazione \eqref{EQ:CAP1:Equazione_di_Schroedinger} e si cerchino soluzioni del tipo
\begin{equation}\label{EQ:CAP1:Funzione_d_onda_variabili_separate}
\psi(\mathbf{x},t)=a(t)\phi(\mathbf{x}).
\end{equation}
Sostituendo questa espressione nella \eqref{EQ:CAP1:Equazione_di_Schroedinger} e osservando che l'hamiltoniana non dipenda esplicitamente dal tempo, si ottiene
\begin{equation}
i \hbar \frac{\partial a(t)}{\partial t} \phi(\mathbf{x}) = a(t)H\phi(\mathbf{x}),
\end{equation}
da cui si ottiene
\begin{equation}
\frac{i \hbar}{a(t)} \frac{\partial a(t)}{\partial t} \phi(\mathbf{x}) = \frac{H\phi(\mathbf{x})}{\phi(\mathbf{x})} \equiv E,
\end{equation}
per una certa costante $E$. Si osservi che conoscendo $E$ è possibile subito ricavare $a(t)$, fissando la condizione iniziale $a(0)=a_0$, con $a_0$ costante, ottenendo
\begin{equation}
a(t)= a_0 \exp \frac{Et}{i\hbar}.
\end{equation}
Inoltre la costante $E$ è calcolata insieme alla $\phi$ nel senso che esse sono rispettivamente autovalore ed autofunzione dell'operatore $H$, cioè si ha
\begin{equation}
H\phi=E\phi,
\end{equation}
la quale prende il nome di \textbf{equazione di Schr\"{o}dinger stazionaria}. Infine bisogna imporre la condizione iniziale \eqref{EQ:CAP1:Equazione_di_Schroedinger_Cond_ini}. Si osservi che in generale sarà
\begin{equation}
\psi\ped{ini}(\mathbf{x}) \neq a_0 \psi(\mathbf{x}).
\end{equation}
Supponendo che l'operatore $H$ ammette un insieme completo di autostati, è possibile sostituire la \eqref{EQ:CAP1:Funzione_d_onda_variabili_separate} con
\begin{equation}
\psi(\mathbf{x},t)= \sum_{\lambda\in\Lambda} a_\lambda (t) \psi_\lambda (\mathbf{x}).
\end{equation}
Ragionando come sopra, si può scrivere
\begin{equation}
\psi(\mathbf{x},t) = \sum_{\lambda\in\Lambda} a_{\lambda,\mbox{ini}} \exp \frac{E_\lambda t}{i\hbar} \phi_\lambda (\mathbf{x}),
\end{equation}
dove
\begin{equation}
H\phi_\lambda = E_\lambda \phi_\lambda,
\end{equation}
ed i coefficienti $a_{\lambda,\mbox{ini}}$ sono ricavati in modo che valga l'espressione
\begin{equation}
\psi\ped{ini} (\mathbf{x}) = \sum_{\lambda\in\Lambda} a_{\lambda,\mbox{ini}} \phi_\lambda (\mathbf{x}).
\end{equation}
\`{E} evidente che gli indici quantici hanno una grande importanza, per cui spesso si scrive l'equazione di Schr\"{o}dinger stazionaria nella forma
\begin{equation}
H\psi_\lambda = E_\lambda \psi_\lambda,
\end{equation}
specificando gli indici quantici $\lambda\in\Lambda$.


\section{Struttura della materia allo stato solido}
Nei solidi le posizioni di equilibrio degli atomi che compongono un materiale possono essere disposte o meno secondo una struttura regolare. Se la struttura è regolare, conoscerne a fondo le caratteristiche è di fondamentale importanza per la dinamica elettronica.

Il contenuto dei successivi paragrafi è tratto da \citep{DISP:Anile}, \citep{BOOK:Bassani}, \citep{BOOK:Ashcroft} e \citep{BOOK:Kittel}.

\subsection{Descrizione quantistica della materia allo stato solido}
Nella materia allo stato solido, gli elettroni si considerano vincolati ai nuclei atomici. Per studiarne il moto è opportuno specificare l'hamiltoniana $\mathcal{H}$ del sistema fisico costituito da elettroni e nuclei atomici. Per gli scopi di questa tesi, un materiale solido si può pensare come costituito da nuclei atomici ed elettroni. I nuclei si possono supporre indivisibili, dotati di massa $M$, carica $+Ze$, spin e momento magnetico trascurabili; gli elettroni si possono considerare dotati di massa $m$, carica $-e$, spin $\frac{1}{2}$ e momento magnetico $\mu_0$.

Sia $\mathbf{r}$ il vettore contenente le coordinate di posizione e di spin di tutti gli elettroni del sistema, $\mathbf{R}$ il vettore delle coordinate di posizione e di spin di tutti i nuclei. Allora la funzione d'onda degli elettroni e dei nuclei è del tipo $\psi(\mathbf{r},\mathbf{R},t)$ ed è determinata dall'equazione di Schr\"{o}dinger scritta per l'operatore hamiltoniano $H$ associato ad $\mathcal{H}$. \`{E} possibile decomporre l'hamiltoniana come somma algebrica di tre contributi:
\begin{equation}
\mathcal{H}=\mathcal{H}_e + \mathcal{H}_{eN} + \mathcal{H}_N,
\end{equation}
dove $\mathcal{H}_e$ è l'hamiltoniana relativa ai soli elettroni, $\mathcal{H}_{eN}$ l'hamiltoniana relativa alle interazioni tra elettroni e nuclei, $\mathcal{H}_N$ l'hamiltoniana relativa ai soli nuclei. Esplicitamente si ha
\begin{equation}
\mathcal{H}_e = \sum_i \frac{\vert \mathbf{p}_i \vert^2}{2m_e} +\frac{1}{2} \sum_{i\neq j} \frac{q^2}{\vert \mathbf{r}_i-\mathbf{r}_j \vert},
\end{equation}
\begin{equation}
\mathcal{H}_{eN} = -\sum_{i,I} \frac{Z_I q^2}{\vert \mathbf{r}_i - \mathbf{R}_I \vert},
\end{equation}
\begin{equation}
\mathcal{H}_N = \sum_I \frac{\vert \mathbf{p}_I \vert^2}{2m_I} + \frac{1}{2} \sum_{I\neq J} \frac{Z_I Z_J q^2}{\vert \mathbf{R}_I - \mathbf{R}_J \vert},
\end{equation}
dove $m_e$ è la massa dell'elettrone, $m_I$ e $Z_I$ sono rispettivamente la massa e il numero atomico nel nucleo contrassegnato con $I$.

Poiché l'hamiltoniana $\mathcal{H}$ non dipende esplicitamente dal tempo, ci si può ridurre a risolvere l'equazione di Schr\"{o}dinger stazionaria
\begin{equation}\label{EQ:CAP1:Eq_Schroedinger_stazionaria_autovalori}
H \psi_{n,v} = E_{n,v} \psi_{n,v},
\end{equation}
dove gli indici $n$ e $v$ si riferiscono rispettivamente agli stati degli elettroni e dei nuclei e l'operatore hamiltoniano $H$ si considera decomponibile in
\begin{equation}
H= H_e + H_{eN} + H_N,
\end{equation}
dove $H_e$, $H_{eN}$ e $H_N$ sono gli operatori hamiltoniani associati rispettivamente alle hamiltoniane $\mathcal{H}_e$, $\mathcal{H}_{eN}$ e $\mathcal{H}_N$.

\subsection{Approssimazione di Born-Oppenheimer}
I nuclei sono più massivi e meno mobili degli elettroni, pertanto è possibile decomporre il moto del sistema nelle oscillazioni dei nuclei intorno alle configurazioni di equilibrio e nel moto degli elettroni rispetto ai nuclei. Questa idea prende il nome di \textbf{approssimazione di Born-Oppenheimer}. Più precisamente, si suppone che la funzione d'onda globale $\psi_{n,v}$ si possa decomporre come
\begin{equation}\label{EQ:CAP1:Approssimazione_B_O_1}
\psi_{n,v} (\mathbf{r},\mathbf{R}) = \psi_n (\mathbf{r},\mathbf{R}) F_{n,v} (\mathbf{R}),
\end{equation}
dove $\psi_n (\mathbf{r},\mathbf{R})$ è la funzione d'onda che descrive il sistema degli elettroni e $F_{n,v} (\mathbf{R})$ è la funzione d'onda che descrive il sistema dei nuclei. Inoltre si suppone che l'azione, sulla funzione d'onda, dell'operatore hamiltoniano relativo ai nuclei può essere approssimata come
\begin{equation}\label{EQ:CAP1:Approssimazione_B_O_2}
H_N (\psi_n F_{n,v}) \approx \psi_n H_N (F_{n,v}).
\end{equation}
Infine, osservando che $F_{n,v}$ non dipende esplicitamente da $\mathbf{r}$, si ha anche
\begin{equation}\label{EQ:CAP1:Approssimazione_B_O_3}
H_e (\psi_n F_{n,v}) = F_{n,v} H_e (\psi_n).
\end{equation}
Allora considerando la \eqref{EQ:CAP1:Eq_Schroedinger_stazionaria_autovalori}, sostituendo la \eqref{EQ:CAP1:Approssimazione_B_O_1} e utilizzando le \eqref{EQ:CAP1:Approssimazione_B_O_2} e \eqref{EQ:CAP1:Approssimazione_B_O_3}, si ottiene
\begin{equation}
F_{n,v} H_e (\psi_n) + H_{eN}(\psi_n F_{n,v}) + \psi_n H_N (F_{n,v}) = E_{n,v} (\psi_n F_{n,v}).
\end{equation}
Dividendo per $\psi_n F_{n,v}$ si ha l'equazione
\begin{equation}\label{EQ:CAP1:Approssimazione_B_O_4}
\frac{H_e (\psi_n)}{\psi_n} + H_{eN} + \frac{H_N (F_{n,v})}{F_{n,v}} = E_{n,v}.
\end{equation}
I primi due termini devono essere uguali ad una costante indipendente da $\mathbf{r}$ perché sommati al terzo termine, che dipende solo da $\mathbf{R}$, devono essere uguali ad una costante. Di conseguenza si ha
\begin{equation}
\frac{H_e (\psi_n)}{\psi_n} + H_{eN} = E_n (\mathbf{R}),
\end{equation}
da cui, sostituendo nella \eqref{EQ:CAP1:Approssimazione_B_O_4}, si ricava
\begin{equation}
\frac{H_N (F_{n,v})}{F_{n,v}} + E_n (\mathbf{R}) = E_{n,v}.
\end{equation}
Si ottengono così due equazioni accoppiate che scritte esplicitamente diventano
\begin{equation}\label{EQ:CAP1:Sistema_B_O}
\left\lbrace 
\begin{aligned}
& (H_e + H_{eN}(\mathbf{r},\mathbf{R}))\psi_n = E_n(\mathbf{R})\psi_n \\
& (H_N + E_n(\mathbf{R})) F_{n,v} = E_{n,v} F_{n,v}
\end{aligned}
\right.
\end{equation}
Quindi, dalla prima equazione, per ogni autostato elettronico (contrassegnato dall'indice $n$), è possibile determinare gli autostati dei nuclei (contrassegnati dall'indice $v$) attraverso la seconda equazione.

Si vuole decomporre l'equazione \eqref{EQ:CAP1:Sistema_B_O}$_1$ in una parte statica, che descrive l'interazione degli elettroni con il reticolo dei nuclei, e una parte dinamica, che descrive gli effetti di interazione degli elettroni con le vibrazioni del reticolo (fononi). Cioè si scrive
\begin{equation}
H_e + H_{eN}(\mathbf{r},\mathbf{R}) = (H_e + H_{eN}(\mathbf{r},\mathbf{R}_0)) + (H_{eN}(\mathbf{r},\mathbf{R}) - H_{eN}(\mathbf{r},\mathbf{R}_0)),
\end{equation}
dove $\mathbf{R}_0$ è lo stato fondamentale del sistema dei nuclei, relativo all'energia più bassa. \`{E} possibile verificare che lo stato fondamentale $\mathbf{R}_0$ corrisponde ad un minimo per il potenziale
\begin{equation}
V_n(\mathbf{R}) = \frac{1}{2} \sum_{I\neq J} \frac{Z_I Z_J q^2}{\vert \mathbf{R}_I - \mathbf{R}_J \vert} + E_n (\mathbf{R}).
\end{equation}
Nota la funzione $E_n (\mathbf{R})$, è possibile determinare lo stato fondamentale
dei nuclei. Di conseguenza, il moto dei nuclei sarà costituito da piccole oscillazioni intorno allo stato fondamentale determinato dal potenziale $V_n$.

Si consideri l'equazione statica per gli elettroni, cioè
\begin{equation}
(H_e + H_{eN}(\mathbf{r},\mathbf{R}_0)) \psi_n = E_n (\mathbf{R}_0) \psi_n.
\end{equation}
Una ulteriore approssimazione, detta degli elettroni indipendenti, consiste nel supporre che la funzione degli elettroni si possa scrivere come prodotto delle funzioni d'onda di un solo elettrone, cioè
\begin{equation}
\psi_n (\mathbf{r},\mathbf{R}_0) = \prod_i \psi_n^i (\mathbf{r}^i,\mathbf{R}_0).
\end{equation}
In questa approssimazione si considera trascurabile l'effetto dell'interazione tra elettrone ed elettrone. Ciascuna funzione d'onda $\psi_n^i$ deve soddisfare un'equazione di Schr\"{o}dinger del tipo
\begin{equation}
-\frac{\hbar^2}{2m}\Delta_{\mathbf{r}^i} \psi_n^i (\mathbf{r}^i,\mathbf{R}_0) + U(\mathbf{r}^i,\mathbf{R}_0)\psi_n^i(\mathbf{r}^i,\mathbf{R}_0) = E_n^i \psi_n^i (\mathbf{r}^i,\mathbf{R}_0),
\end{equation}
dove il potenziale $U$ dipende dalla posizione dei nuclei nello stato fondamentale.

\subsection{Le strutture cristalline}
Gli atomi che compongono la materia allo stato solido soddisfano due proprietà fondamentali: le distanze interatomiche sono molto piccole (dell'ordine di $\si{\num{e-10} \metre}$) e le posizioni di equilibrio degli atomi sono fissate. A seconda di come sono disposte tali posizioni di equilibrio si può avere o meno una struttura regolare nel solido. Più precisamente, fissato un nucleo atomico come origine, è possibile definire la funzione $g(x)$ come la probabilità di trovare un altro nucleo a distanza $x$ dall'origine. Tale funzione è detta \textbf{funzione di correlazione a coppie} e in un solido essa presenta dei massimi anche per valori di $x$ grandi. In base a questa funzione, i solidi possono trovarsi in tre stati: cristallino, policristallino e amorfo. In un \textbf{cristallo} la funzione $g(x)$ presenta dei picchi discreti e la funzione non dipende dalla posizione del nucleo fissato; in un \textbf{policristallo} si ha che esso può essere pensato come un aggregato di tanti cristalli orientati in modo diverso e in questo caso la funzione di correlazione a coppie appare a tratti come quella di un cristallo ed essa cambia quando si passa da un cristallo all'altro; in un \textbf{solido amorfo}, infine, la funzione di correlazione a coppie presenta dei picchi accentuati solo per piccoli valori di $x$, per poi oscillare intorno ad un valore medio per valori più grandi di $x$ (cfr. \citep{BOOK:Bassani}).

I cristalli godono di proprietà di regolarità e simmetria che consentono uno studio teorico approfondito. Da adesso in poi si tratterà solamente la struttura dei cristalli, in quanto gli altri tipi di solidi esulano dagli scopi di questa tesi. Una proprietà fondamentale nella modellizzazione dei cristalli è la seguente: fissato un punto $\mathbf{r}$, tutti i punti associati ad esso mediante specifici vettori di traslazione $\mathbf{r}_n$ sono ad esso equivalenti, cioè indistinguibili fisicamente. In formule, se
\begin{equation*}
\mathbf{r}' = \mathbf{r} + \mathbf{r}_n,
\end{equation*}
allora si può dire che $\mathbf{r}'\equiv\mathbf{r}$ e tutte le osservabili fisiche hanno lo stesso valore in $\mathbf{r}$ e in $\mathbf{r}'$. Il concetto di cristallo può essere formalizzato matematicamente mediante la nozione di reticolo cristallino (cfr. \citep{DISP:Anile}).
\begin{definizione}
Si chiama \textbf{reticolo cristallino} un sottoinsieme numerabile di $\mathbb{R}^d$, con $d=1,2,3$, generato da $d$ vettori indipendenti $\mathbf{a}_1,\ldots,\mathbf{a}_d$, cioè
\begin{equation}
R=\left\lbrace \mathbf{a}\in\mathbb{R}^d \mid \mathbf{a}=n_1\mathbf{a}_1+\ldots+n_d\mathbf{a}_d, \quad n_1,\ldots,n_d\in\mathbb{Z} \right\rbrace .
\end{equation}
Inoltre l'insieme dei vettori $\left\lbrace \mathbf{a}_1,\ldots,\mathbf{a}_d \right\rbrace $ è detto \textbf{base} di $R$, mentre i vettori $\mathbf{a}_i$ per $i=1,\ldots,d$ sono detti \textbf{vettori primitivi} e si dice che essi generano $R$.
\end{definizione}
Questa definizione si può interpretare nel modo seguente. Dato il vettore $\mathbf{a} = \sum_{i=1}^d n_i \mathbf{a}_i$ e fissato un punto del cristallo, partendo da esso e spostandosi di $n_i$ passi di lunghezza $a_i$ nella direzione di $\mathbf{a}_i$ per $i=1,\ldots,d$ si troverà un altro punto del cristallo equivalente al primo nel senso spiegato in precedenza. Al variare di $n_i \in \mathbb{Z}$ per $i=1,\ldots,d$ si determineranno tutti i punti del cristallo equivalenti tra loro (cfr. \citep{BOOK:Kittel}). Inoltre si osservi che è possibile scegliere la base del reticolo in infiniti modi equivalenti. Infatti basta prendere $d$ vettori $\mathbf{a}'_1,\ldots,\mathbf{a}'_d$ con
\begin{equation}
\mathbf{a}'_i = \sum_{j=1}^d m_{ij} \mathbf{a}_j,
\end{equation}
dove $\det(M)=1$ con $M=(m_ij)$. Spesso si sceglie una base formata da vettori di lunghezza minima.
\begin{definizione}
Fissato un punto di un reticolo cristallino, i punti del reticolo ad esso più vicini sono detti \textbf{primi vicini}. Il numero di primi vicini si dice \textbf{numero di coordinazione}.
\end{definizione}

Le proprietà di equivalenza valgono in tutto lo spazio quindi il reticolo cristallino deve essere infinito, però ovviamente i cristalli sono finiti. Pertanto è necessario assumere che la maggior parte dei punti sarà sufficientemente distante dalla superficie. Frequentemente, però, i reticoli cristallini si considerano finiti non tanto perché gli effetti dovuti alla superficie sono importanti, ma semplicemente in quanto la descrizione risulta più semplice (cfr. \citep{BOOK:Ashcroft}).

Riprendendo la modellizzazione risultano utili le seguenti definizioni.
\begin{definizione}
Si chiama \textbf{cella primitiva} (o \textbf{elementare}) di un reticolo $R$ in $\mathbb{R}^d$, un qualunque sottoinsieme $\mathcal{D}_R$ di $\mathbb{R}^d$ contenente un unico elemento di $R$ (di solito l'origine) e tale che i suoi traslati formino una partizione di $\mathbb{R}^d$.
\end{definizione}
Si osservi che in uno stesso reticolo esistono infinite scelte di celle primitive.
\begin{definizione}
Si chiama \textbf{cella unitaria} la cella primitiva definita dai vettori di base di lunghezza minima, cioè esplicitamente
\begin{equation*}
\mathcal{D} = \left\lbrace \mathbf{x} = \alpha_1\mathbf{a}_1+\ldots+\alpha_d\mathbf{a}_d, \, \mbox{con} \, \alpha_1,\ldots,\alpha_d\in [0,1[ \right\rbrace
\end{equation*}
inoltre le dimensioni della cella unitaria sono dette \textbf{costanti di reticolo}.
\end{definizione}
\begin{definizione}
Si chiama \textbf{cella primitiva di Wigner-Seitz} la regione $\mathcal{D}_{WS}$ intorno all'origine i cui punti sono più vicini all'origine che ad ogni altro punto del reticolo, cioè in formule
\begin{equation*}
\mathcal{D}_{WS} = \left\lbrace \mathbf{x}\in\mathbb{R}^d \mid \vert \mathbf{x} \vert \leq \vert \mathbf{x}+\mathbf{a} \vert, \, \forall \mathbf{a}\in R \right\rbrace.
\end{equation*}
\end{definizione}
Per chiarire la precedente definizione si osservi la Figura \ref{FIG:CAP1:Wigner_Seitz}.
\begin{figure}[ht]
\centering
\includegraphics[width=0.5\columnwidth]{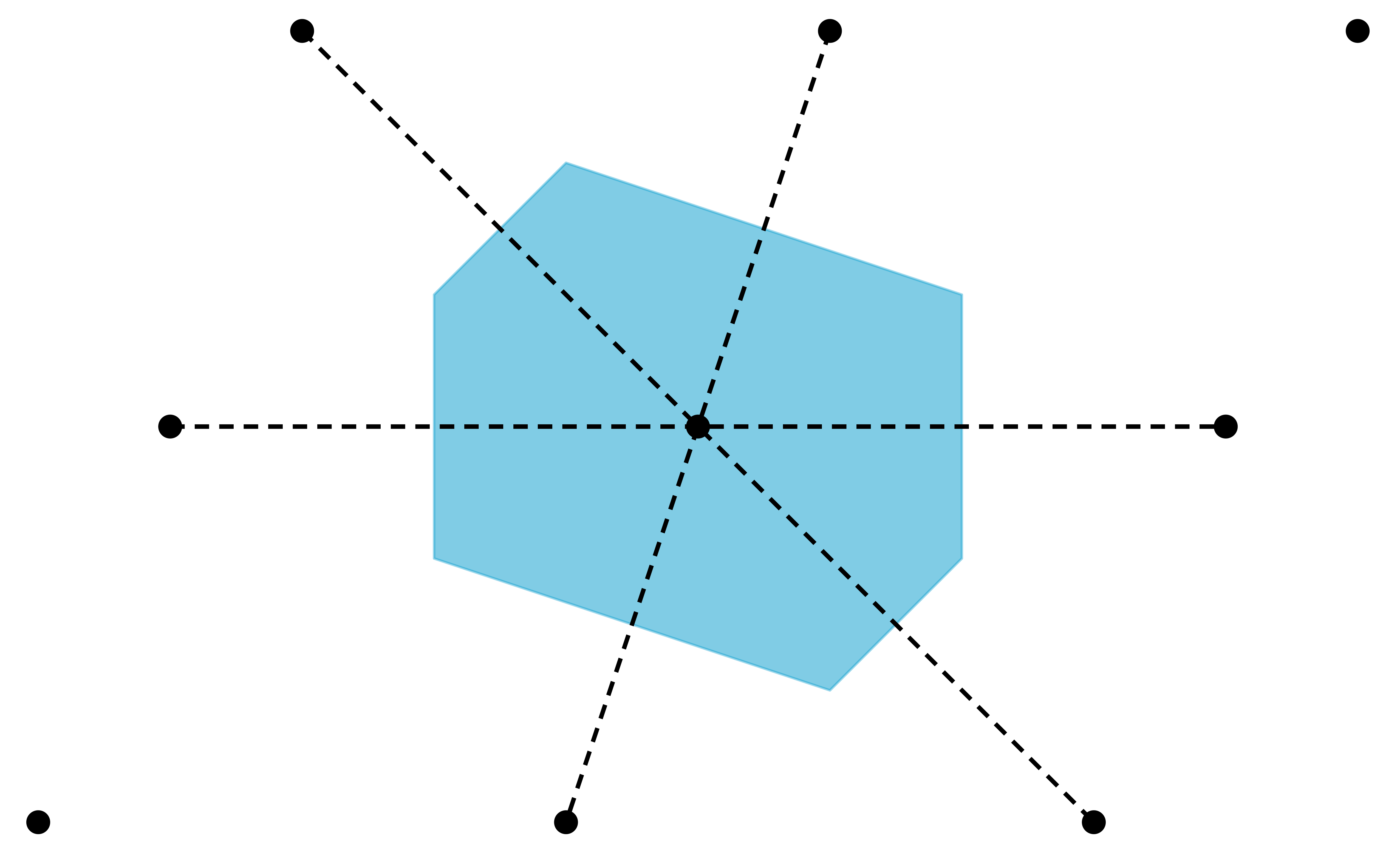}
\caption{Rappresentazione della cella primitiva Wigner-Seitz per un reticolo cristallino bidimensionale (cfr. \citep{BOOK:Ashcroft}).}
\label{FIG:CAP1:Wigner_Seitz}
\end{figure}
\FloatBarrier
I reticoli cristallini vengono classificati in base alle loro proprietà di traslazione. Cioè si consideri una traslazione $T_\mathbf{a}:\mathbb{R}^d\rightarrow\mathbb{R}^d$ definita da
\begin{equation*}
T_\mathbf{a} (\mathbf{x}) = \mathbf{x}+\mathbf{a}.
\end{equation*}
Allora si ha $T_\mathbf{a} (R)=R, \quad \forall \mathbf{a}\in R$ e viceversa, cioè un reticolo è invariante per traslazione se e solo se viene traslato rispetto ad un suo elemento. \`{E} possibile dimostrare che esiste solo un numero finito di reticoli di traslazione, detti \textbf{reticoli di Bravais}. Nel caso unidimensionale esiste un solo reticolo di Bravais, nel caso bidimensionale ne esistono 5 mentre nel caso tridimensionale ne esistono 14.

Per la trattazione dei prossimi capitoli sarà utile approfondire i reticoli di Bravais bidimensionali, per cui è opportuno darne una descrizione più dettagliata. Nel caso bidimensionale una cella elementare è costituita da due vettori, che saranno indicati con $\mathbf{a}_1$, $\mathbf{a}_2$ e dall'angolo che essi formano, che sarà indicato con $\varphi$. In base alle diverse scelte dei vettori primitivi e dell'angolo da essi formato, si dimostra che esistono 5 reticoli di Bravais bidimensionali: obliquo, rettangolare, rettangolare centrato, quadrato ed esagonale. La descrizione di questi reticoli è riassunta nella Tabella \ref{TAB:CAP1:Bravais_2D}.

Sia $R$ un reticolo di Bravais di base $\left\lbrace \mathbf{a}_1,\ldots,\mathbf{a}_d \right\rbrace$. Si consideri una funzione d'onda del tipo
\begin{equation}\label{EQ:CAP1:Funz_onda_Bravais}
\psi(\mathbf{x}) = e^{i\mathbf{k}\cdot\mathbf{x}},
\end{equation}
dove $\mathbf{k}$ è un generico vettore d'onda. In generale $\psi(\mathbf{x})$ non ha la periodicità di $R$, però imponendo questa proprietà si ha
\begin{equation*}
e^{i\mathbf{k}\cdot(\mathbf{x}+\mathbf{r})} = e^{i\mathbf{k}\cdot\mathbf{x}},
\end{equation*}
per ogni $\mathbf{r}\in R$. Inoltre, semplificando si ottiene
\begin{equation}\label{EQ:CAP1:Funz_onda_Bravais_pta}
e^{i\mathbf{k}\cdot \mathbf{r}} = 1,
\end{equation}
per ogni $\mathbf{r}\in R$. Allora le funzioni d'onda del tipo \eqref{EQ:CAP1:Funz_onda_Bravais} e soddisfacenti la \eqref{EQ:CAP1:Funz_onda_Bravais_pta} hanno la stessa periodicità di $R$. Si ha la definizione seguente.
\begin{definizione}
Dato un reticolo di traslazione $R$ di base $\left\lbrace \mathbf{a}_1,\ldots,\mathbf{a}_d \right\rbrace$, si chiama \textbf{reticolo inverso} il reticolo $G$ di base $\left\lbrace \mathbf{b}_1,\ldots,\mathbf{b}_d \right\rbrace$ con questi vettori definiti dalla relazione
\begin{equation}\label{EQ:CAP1:Relazione_reticolo_inverso}
\mathbf{b}_i\cdot\mathbf{a}_j = 2\pi\delta_{ij},
\end{equation}
dove $\delta_{ij}$ è il simbolo di Kronecker.
\end{definizione}
Osservando che $\mathbf{r}=n_1\mathbf{a}_1+\ldots+n_d\mathbf{a}_d$ allora se i $\mathbf{b}_i$ soddisfano la relazione \eqref{EQ:CAP1:Relazione_reticolo_inverso}, la \eqref{EQ:CAP1:Funz_onda_Bravais_pta} è automaticamente soddisfatta. Inoltre $\mathbf{b}_1,\ldots,\mathbf{b}_d$ sono $d$ vettori ciascuno dei quali è ortogonale a tutti gli elementi tranne uno della base di $R$. Pertanto $\left\lbrace \mathbf{b}_1,\ldots,\mathbf{b}_d \right\rbrace$ è una base di $\mathbb{R}^d$ e l'insieme
\begin{equation*}
G = \left\lbrace \mathbf{x}\in\mathbb{R}^d \mid \mathbf{x} = m_1\mathbf{b}_1+m_d\mathbf{b}_d, \, \mbox{con} \, m_1,\ldots,m_d\in\mathbb{Z} \right\rbrace
\end{equation*}
definisce un nuovo reticolo cristallino. Inoltre se $R$ è un reticolo di Bravais allora anche $G$ lo è.
\FloatBarrier
\begin{table}[!ht]
\centering
\begin{tabular}{|c|c|c|c|}
\hline
Reticolo cristallino & Sistema & Cella elementare & Parametri \\
\hline
\includegraphics[width=0.3\columnwidth]{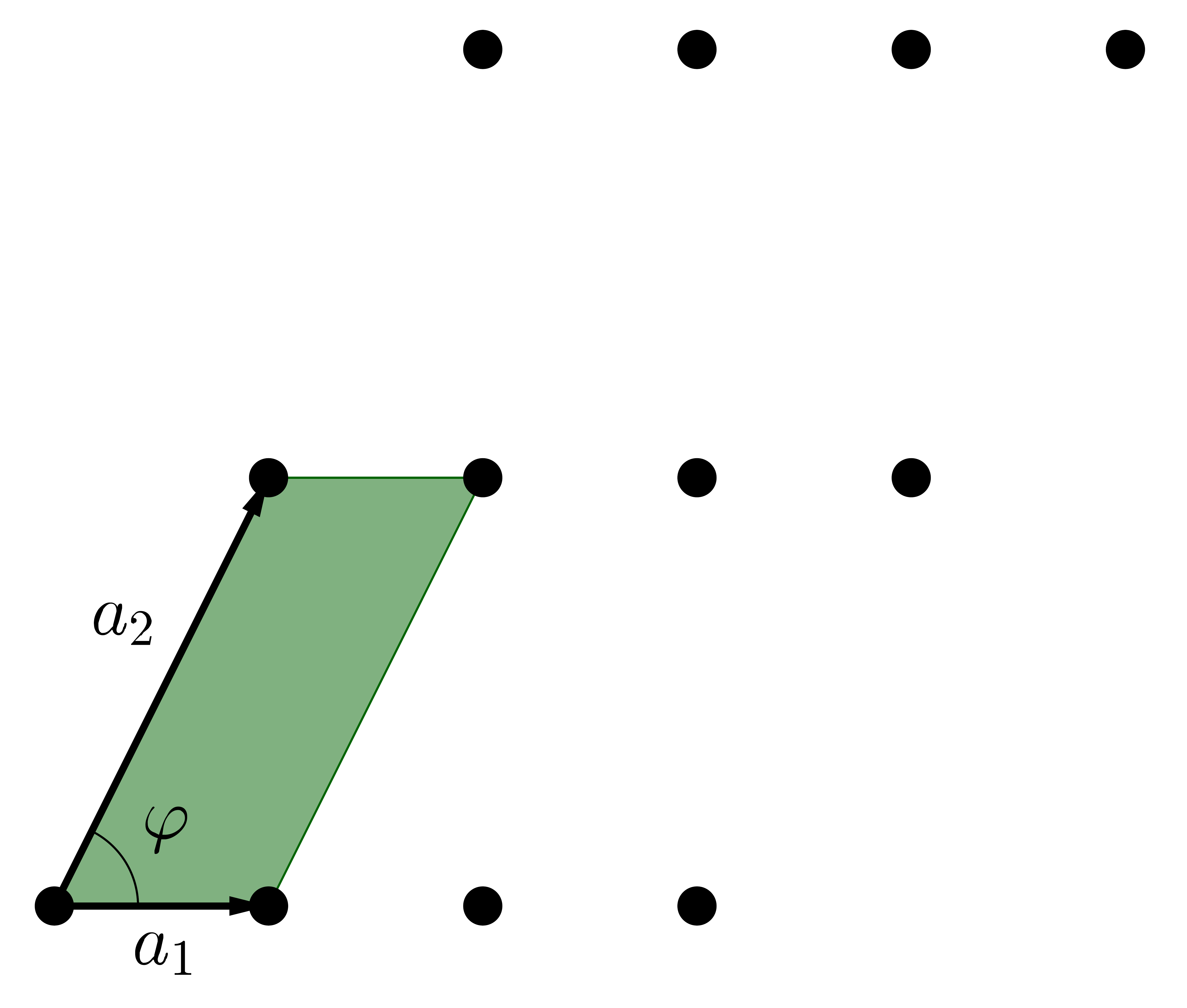} & obliquo & parallelogrammo & $a_1 \neq a_2, \quad \varphi\neq 90^\circ$ \\
\hline
\includegraphics[width=0.25\columnwidth]{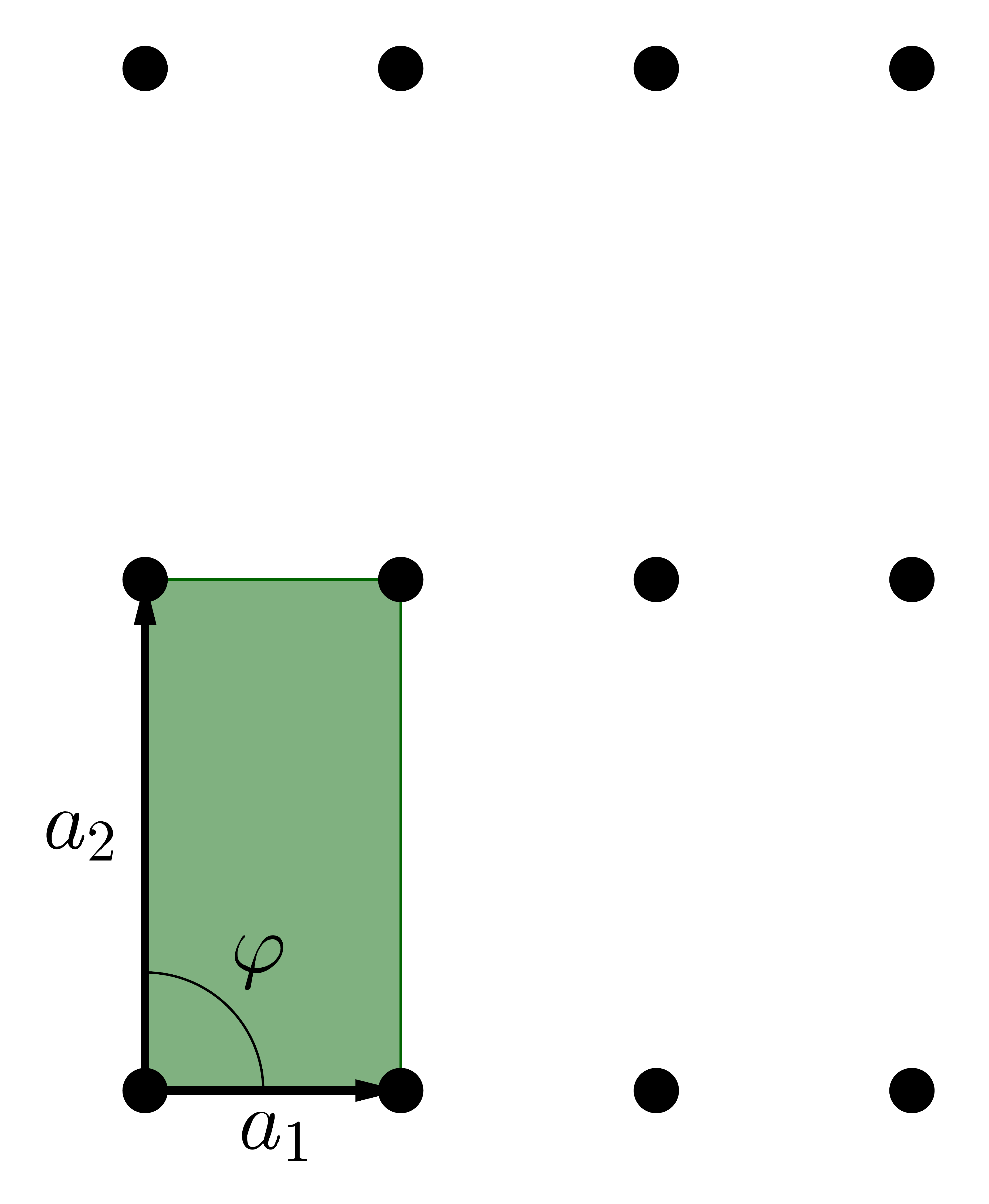} & rettangolare & rettangolo & $a_1 \neq a_2, \quad \varphi = 90^\circ$ \\
\hline
\includegraphics[width=0.35\columnwidth]{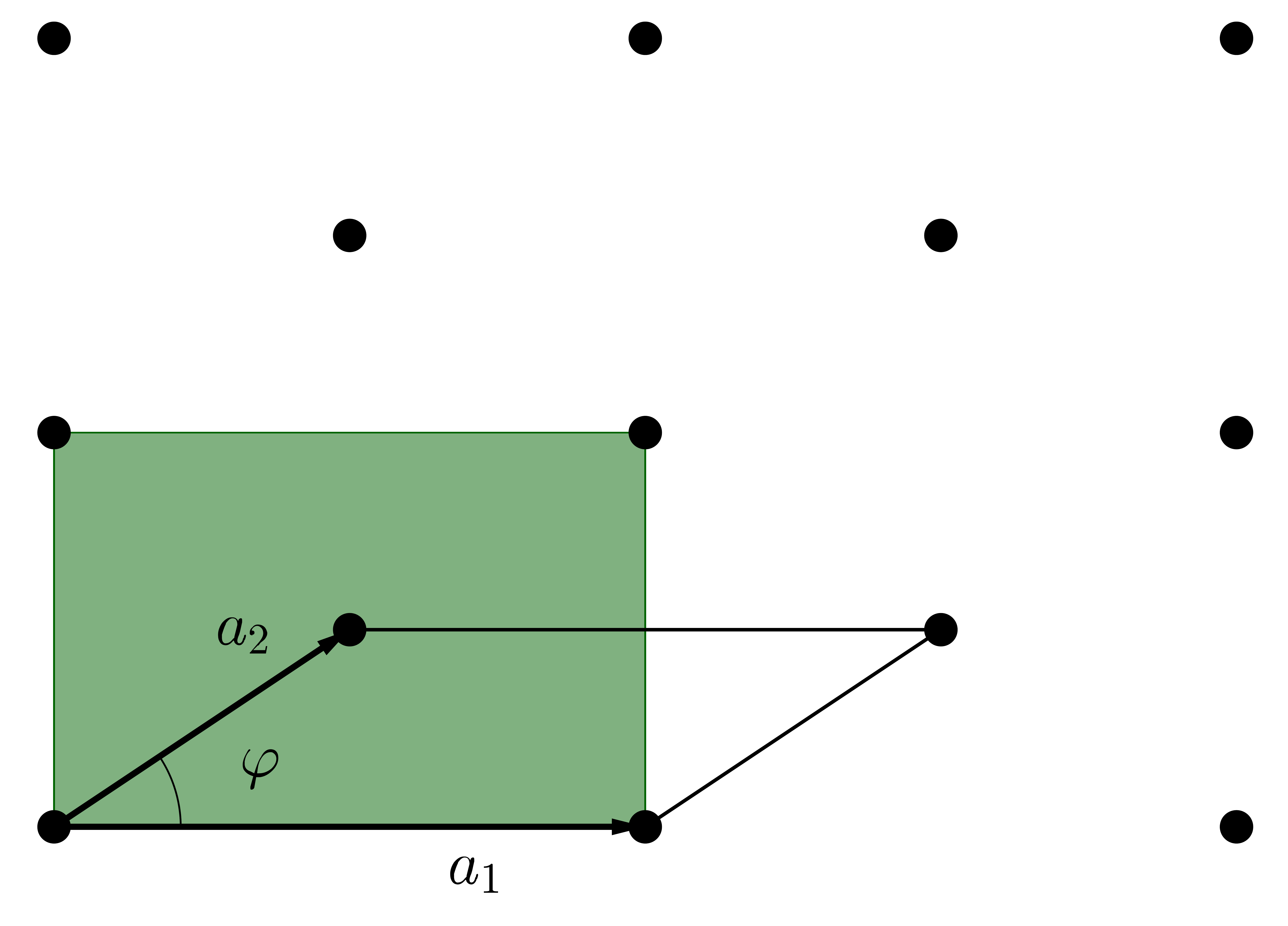} & rettangolare & parallelogrammo & $a_1 \neq a_2, \quad \varphi\neq 90^\circ$ \\
\hline
\includegraphics[width=0.3\columnwidth]{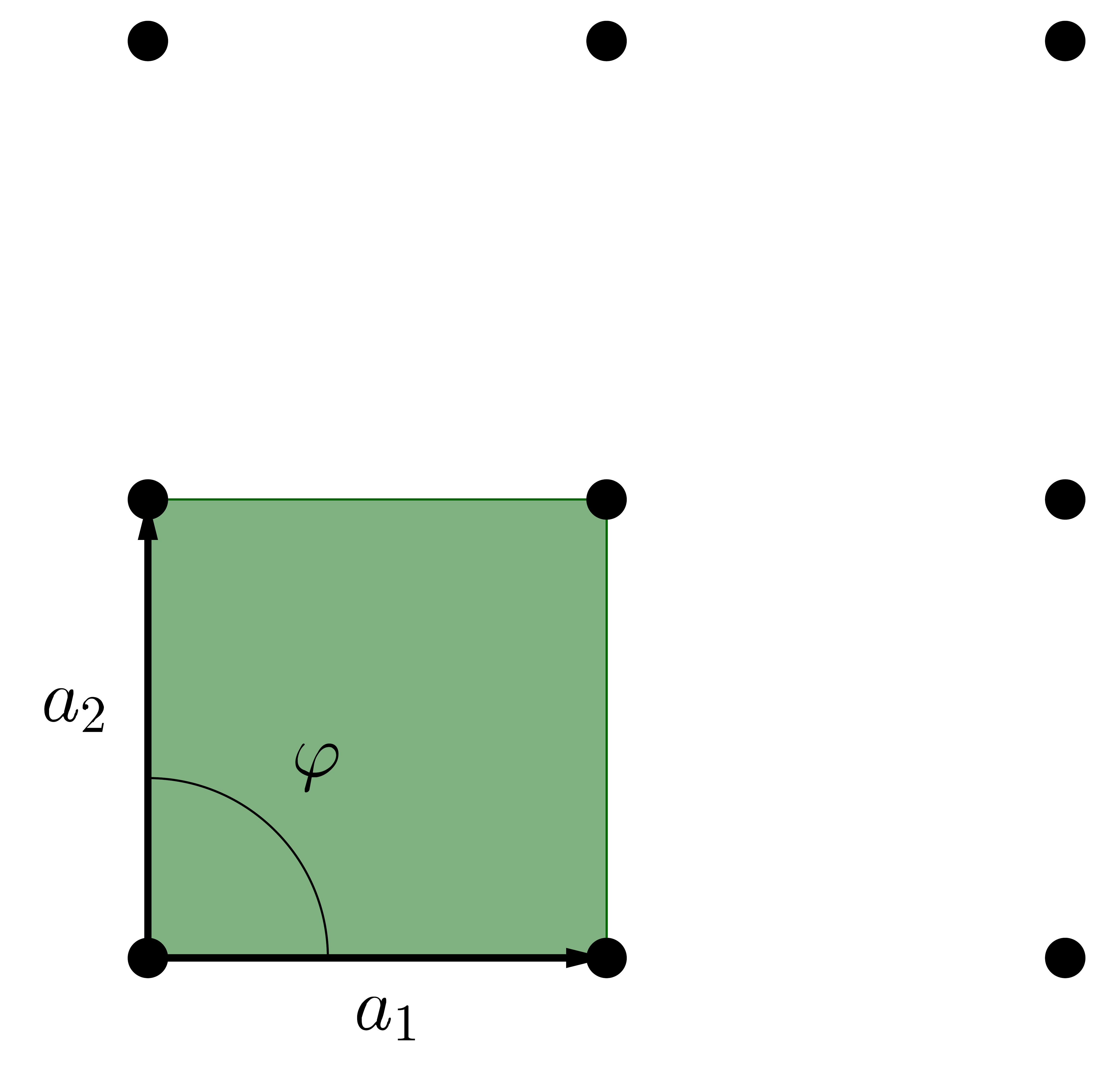} & quadrato & quadrato & $a_1 = a_2, \quad \varphi = 90^\circ$ \\
\hline
\includegraphics[width=0.35\columnwidth]{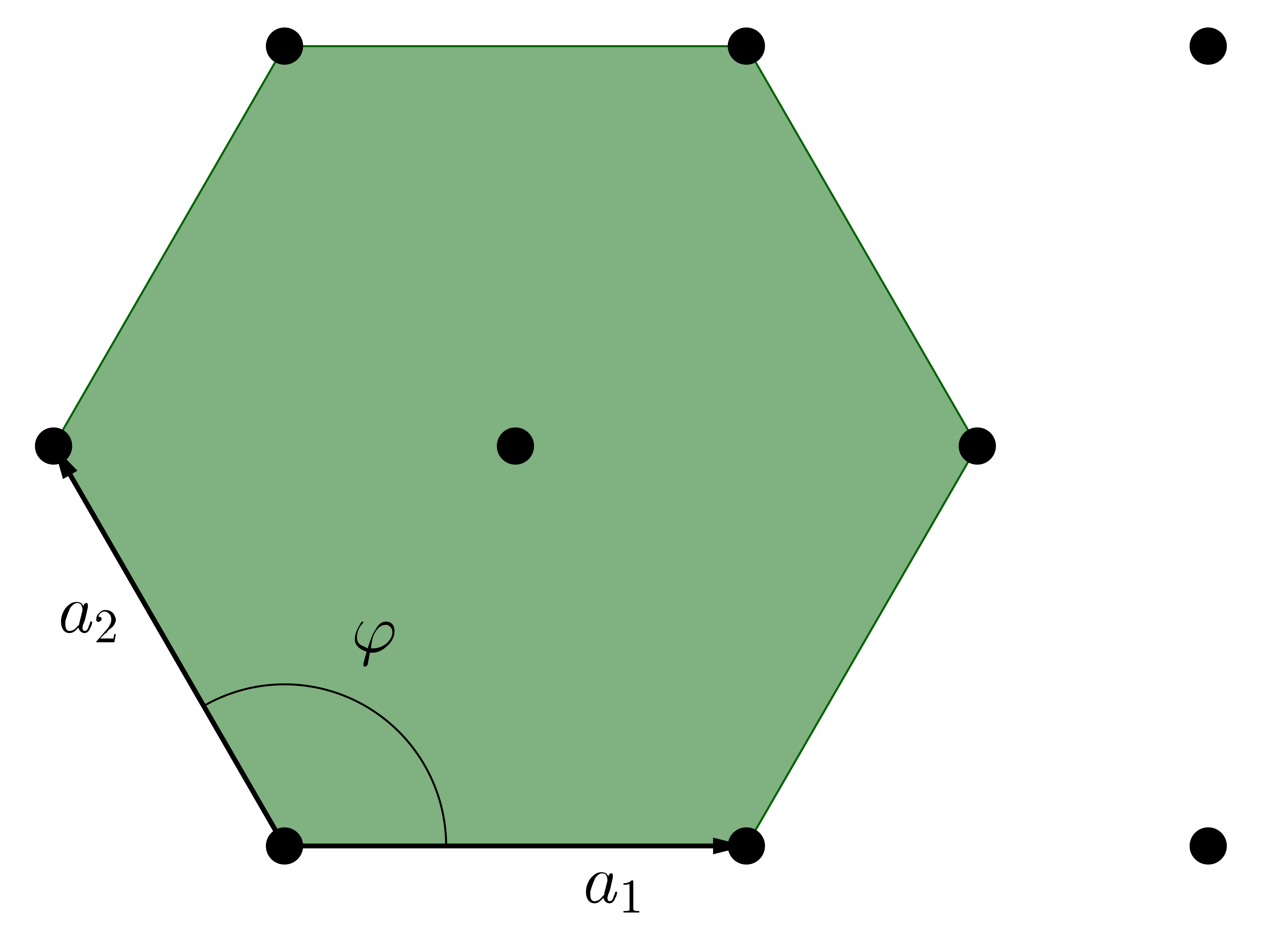} & esagonale & rombo & $a_1 = a_2, \quad \varphi = 120^\circ$ \\
\hline
\end{tabular}
\caption{Reticoli di Bravais bidimensionali. Le parti evidenziate rappresentano le celle elementari nei vari casi.}
\label{TAB:CAP1:Bravais_2D}
\end{table}
\FloatBarrier
\newpage
\begin{definizione}
Si chiama \textbf{prima zona di Brillouin} la cella primitiva di Wigner-Seitz del reticolo inverso, cioè l'insieme
\begin{equation}
\mathcal{B}=\left\lbrace \mathbf{k}\in\mathbb{R}^3 \mid \vert \mathbf{k} \vert \leq \vert \mathbf{k}+\mathbf{b} \vert \quad \forall \mathbf{b}\in G \right\rbrace.
\end{equation}
Il centro della prima zona di Brillouin è detto \textbf{punto} $\Gamma$.
\end{definizione}

\subsection{Moto di un elettrone in presenza di un potenziale periodico}\label{PAR:CAP2:Moto_elettrone}
I nuclei atomici collocati nei punti di un reticolo $R$ generano un potenziale elettrostatico $U_R$, che è periodico con la stessa periodicità del reticolo cristallino, cioè
\begin{equation}
U_R (\mathbf{x}) = U_R (\mathbf{x} + \mathbf{a} ) \qquad \forall \mathbf{a} \in R.
\end{equation}
Prima di illustrare un'importante conseguenza dell'introduzione del potenziale periodico è opportuno dare la seguente
\begin{definizione}
Data una funzione d'onda $\psi(\mathbf{x})$ e un reticolo cristallino $R$, si dice che $\psi$ è \textbf{quasi-periodica} rispetto a $R$ se
\begin{equation}
\psi(\mathbf{x}) = u(\mathbf{x})\exp(i\mathbf{k}\cdot \mathbf{x}),
\end{equation}
per $\mathbf{k}\in\mathbb{R}^3$ fissato e con $u$ soddisfacente la relazione $u(\mathbf{x})=u(\mathbf{x}+\mathbf{a}) \quad \forall \mathbf{a}\in R$.
\end{definizione}
Inoltre vale il seguente
\begin{teorema}[di Bloch]
Le autofunzioni di una hamiltoniana con un potenziale periodico rispetto ad un reticolo $R$, possono essere scelte in modo tale da essere quasi-periodiche rispetto a $R$.
\end{teorema}
\begin{proof}
Consultare \citep{DISP:Anile}.
\end{proof}
In virtù del Teorema di Bloch è possibile indicizzare le autofunzioni (e quindi gli autovalori) di $H$ mediante i valori di $\mathbf{k}$. Di conseguenza, la funzione $\psi_\mathbf{k}$ e il livello energetico $E_\mathbf{k}$ sono, al variare di $\mathbf{k}\in\mathbb{R}^3$, soluzione del seguente problema agli autovalori in una cella primitiva $\mathcal{D}$ del reticolo cristallino $R$
\begin{equation}\label{EQ:CAP1:Probl_autov_reticolo}
\left\lbrace 
\begin{aligned}
& -\frac{\hbar}{2m_e}\Delta_\mathbf{x} \psi_\mathbf{k} + U_R \psi_\mathbf{k} = E_\mathbf{k} \psi_\mathbf{k} \quad \mbox{in} \; \mathcal{D} \\
& \psi_\mathbf{k} (\mathbf{x}+\mathbf{a}) = e^{i\mathbf{k}\cdot\mathbf{a}} \psi_\mathbf{k}(\mathbf{x}), \quad \forall \mathbf{x},\mathbf{x}+\mathbf{a}\in\partial\mathcal{D}.
\end{aligned}
\right.
\end{equation}
Inoltre, la funzione d'onda $\psi_\mathbf{k}$ e lo stato energetico $E_\mathbf{k}$, viste come funzioni di $\mathbf{k}$, sono periodiche rispetto al reticolo inverso $G$, cioè
\begin{equation}
\psi_{\mathbf{k}+\mathbf{b}} = \psi_\mathbf{k}, \quad E_{\mathbf{k}+\mathbf{b}} = E_\mathbf{k}, \quad \forall \mathbf{b}\in G.
\end{equation}
Quindi, essendo $\psi_\mathbf{k}$ periodica rispetto a $G$, è sufficiente far variare $\mathbf{k}$ in una cella primitiva del reticolo inverso. Solitamente si prende $\mathbf{k}\in\mathcal{B}$.

\subsection{La struttura elettronica a bande}
Nei cristalli gli elettroni che si trovano nei livelli energetici più bassi li occupano stabilmente, perché le loro funzioni d'onda non raggiungono gli atomi vicini in modo significativo. La stessa cosa non si può dire per gli elettroni che si trovano nei livelli più esterni. In particolare, per studiare quelli che contribuiscono ai legami chimici è necessario considerare la presenza degli atomi vicini. Un modo per avere un'idea qualitativa della struttura elettronica a bande dei semiconduttori è detto metodo \emph{tight-binding}. Esso consiste nel considerare un cristallo dotato della stessa struttura di quello in esame ma avente una costante di reticolo arbitrariamente grande. Tale sistema può essere descritto da un insieme di atomi isolati e quindi le funzioni d'onda rappresentano esattamente gli stati di tutti gli elettroni. Facendo decrescere la costante di reticolo, raggiunta una certa distanza, gli elettroni più esterni iniziano a sentire la presenza degli atomi più vicini. Allora le funzioni d'onda si sovrappongono e formano stati che si estendono su tutto il cristallo. L'effetto di questa sovrapposizione consiste nel fatto che gli orbitali atomici di pari energia si suddividono in livelli differenti. Partendo dagli stati più esterni i singoli livelli diventano bande di energia strette che crescono e alla fine si sovrappongono, come mostrato in Figura \ref{FIG:CAP2:tight_binding}.
\begin{figure}[ht]
\centering
\includegraphics[width=0.4\columnwidth]{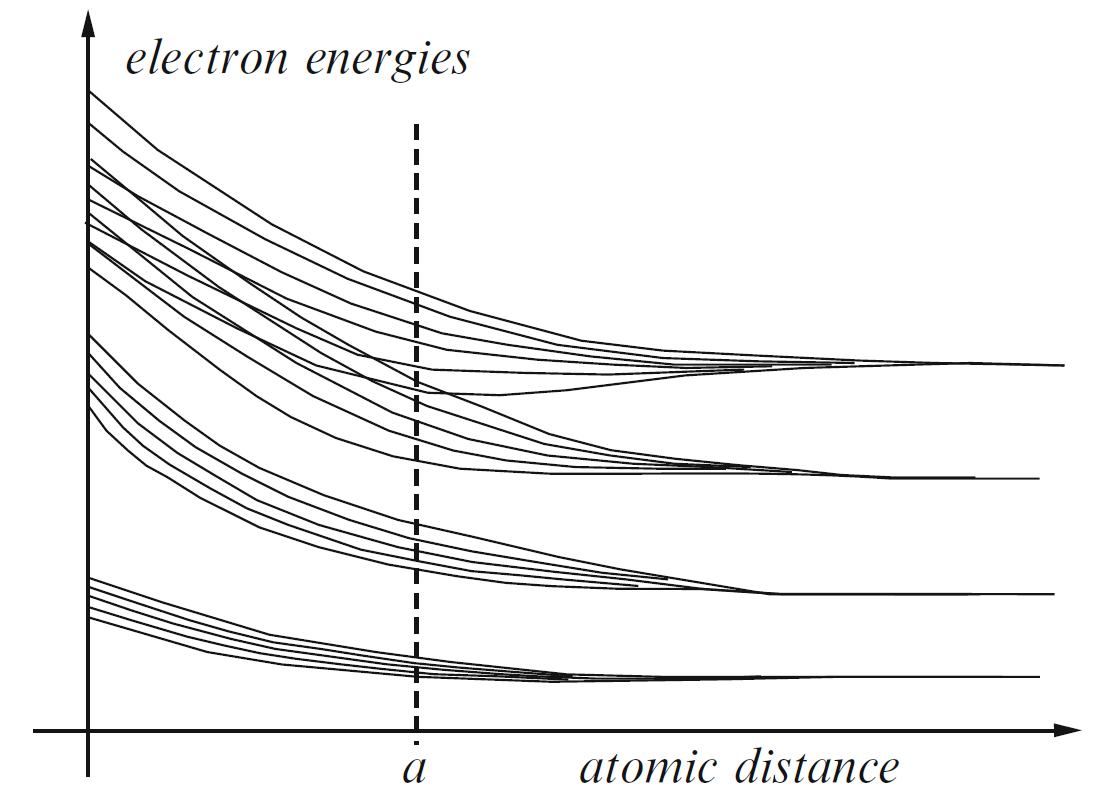}
\caption{Rappresentazione schematica della formazione delle bande di energia. La costante di reticolo è indicata con $a$.}
\label{FIG:CAP2:tight_binding}
\end{figure}
\FloatBarrier
Si consideri il problema agli autovalori \eqref{EQ:CAP1:Probl_autov_reticolo}. Fissato $\mathbf{k}\in\mathcal{B}$, la \eqref{EQ:CAP1:Probl_autov_reticolo}$_1$ è un'equazione ellittica definita in un insieme compatto di $\mathbb{R}^3$. Allora è possibile dimostrare (cfr.~\citep{DISP:Anile}) che esiste una successione di autovalori $E_{n,\mathbf{k}}$, con relative autofunzioni $\psi_{n,\mathbf{k}}$, con $n \in\mathbb{N}$, che soddisfano il problema \eqref{EQ:CAP1:Probl_autov_reticolo}. Si osservi che $E_{n,\mathbf{k}}$ rappresenta l'energia totale, cinetica più potenziale, dell'elettrone che si trova nell'autostato $\psi_{n,\mathbf{k}}$. Le funzioni $k \rightarrow E_{n,\mathbf{k}}$ sono dette \textbf{bande energetiche}. Inoltre tali funzioni sono numerate in modo tale che, fissato $\mathbf{k}$, la banda di numero più piccolo ha energia più bassa. In tal modo con la notazione $\mathcal{E}_n(\mathbf{k})$ è indicata l'$n$-esima banda energetica.

Allora, ricordando il legame tra energia e pulsazione stabilito dalla relazione di Planck-Einstein, per ogni $\mathbf{k}$ fissato è possibile determinare in corrispondenza di $\mathcal{E}_n(\mathbf{k})$ una pulsazione
\begin{equation}
\omega_n(\mathbf{k})=\frac{1}{\hbar}\mathcal{E}_n(\mathbf{k}).
\end{equation}
Quindi la $n$-esima banda energetica $\mathcal{E}_n(\mathbf{k})$, può anche essere interpretata come relazione di dispersione di un pacchetto d'onda associato alla sovrapposizione di onde di Bloch $e^{i\mathbf{k}\cdot\mathbf{x}}u_{n,\mathbf{k}}$. La corrispondente velocità di gruppo è definita dalla relazione
\begin{equation}\label{EQ:CAP1:Vel_gruppo_elettroni}
\mathbf{v}_n(\mathbf{k})=\frac{1}{\hbar} \nabla_\mathbf{k} \mathcal{E}_n(\mathbf{k})
\end{equation}
e rappresenta la velocità del pacchetto d'onda elettronico.

Vi possono essere valori di energia che non sono assunti da nessuna delle funzioni $\mathcal{E}_n(\mathbf{k})$, dette \textbf{energie proibite} e sono tali che non possono essere raggiunte da alcun elettrone che si muove nel cristallo e costituiscono i cosiddetti \textbf{gap energetici}. Allo zero assoluto gli elettroni occuperanno gli stati disponibili a più bassa energia nelle bande di energia del cristallo. L'energia del più alto livello energetico completamente occupato è detto \textbf{livello di Fermi} del materiale ed è indicato con $\varepsilon_F$. Si possono presentare due situazioni ben diverse tra di loro.
\begin{enumerate}
\item Un certo numero di bande possono essere completamente riempite, mentre tutte le altre rimangono completamente vuote. In questo caso, si chiama \textbf{banda di valenza}, e si indica con $\mathcal{E}_v(\mathbf{k})$, la più alta banda riempita; si chiama \textbf{banda di conduzione}, e si indica con $\mathcal{E}_c(\mathbf{k})$, la più bassa banda vuota. Si chiama \textbf{gap energetico} tra le due bande, e si indica con $E_g$, la quantità
\begin{equation}
E_g = \min_{\mathbf{k}\in\mathcal{B}} \mathcal{E}_c(\mathbf{k}) - \max_{\mathbf{k}\in\mathcal{B}} \mathcal{E}_v(\mathbf{k}).
\end{equation}
A seconda di quanto sia grande questo gap energetico, i cristalli sono detti \textbf{isolanti} o \textbf{semiconduttori}. Tipicamente $E_g$ varia da $1$ a parecchi $\si{\electronvolt}$ nel caso degli isolanti ed è compreso tra $0.1$ e $\si{\num{0.5}\electronvolt}$ per i semiconduttori.
\item Un certo numero di bande sono solo parzialmente riempite. In questo caso si dice che il cristallo è un \textbf{conduttore}.
\end{enumerate}
Se la temperatura non è nulla, un certo numero di elettroni, per eccitazione termica, passerà nella banda di conduzione e lascerà anche un numero uguale di stati non occupati nella banda di valenza. Questo avviene più agevolmente nei semiconduttori, perché il gap energetico tra la banda di valenza e quella di conduzione non è molto grande. Sotto l'azione di un campo elettrico esterno, gli elettroni passati alla banda di conduzione, hanno a disposizione un gran numero di stati liberi nei quali si possono spostare e pertanto possono partecipare alla conduzione elettrica. Questi elettroni lasciano degli stati vuoti nella banda di valenza, i quali costituiscono degli stati liberi in cui si possono spostare gli elettroni di valenza. Allora è possibile studiare la dinamica di questi stati vuoti in modo analogo alla dinamica degli elettroni passati alla banda di conduzione. Tali stati si pensano come se fossero occupati da particelle di carica positiva, dette \textbf{lacune}, che si muovono in verso opposto a quello degli elettroni.

Nelle regioni attorno ai minimi della banda di conduzione, dette valli, o attorno ai massimi della banda di valenza, la funzione $\mathcal{E}(\mathbf{k})$ può essere approssimata da una funzione quadratica in $\mathbf{k}$ (cfr.~\citep{BOOK:Jacoboni}). Questa approssimazione delle bande è detta \textbf{approssimazione parabolica}.

\subsection{Vibrazioni reticolari e fononi}
I nuclei degli atomi che costituiscono il reticolo non sono fissi, ma oscillano attorno alle loro posizioni di equilibrio, che corrispondono al minimo dell'energia potenziale del sistema. Per comprendere meglio questo concetto si consideri, ad esempio, un reticolo cristallino unidimensionale e lo si studi dal punto di vista classico. Esso si può vedere come una catena lineare di punti materiali aventi uguale massa e collegati da molle, rappresentanti i legami chimici, e posti a distanza pari alla costante del reticolo. Si suppone per semplicità che il moto possa avvenire solo nella direzione della catena e che essa formi un grande anello in modo da imporre condizioni al bordo periodiche (il problema si riporta al caso della corda vibrante in un intervallo finito). Si può dimostrare dalla meccanica classica (cfr.~\citep{BOOK:Jacoboni}) che le soluzioni delle equazioni di moto sono onde piane e imponendo le condizioni periodiche si dimostra che i vettori d'onda possibili sono del tipo
\begin{equation}
k_n = \frac{2\pi}{L}n, \qquad n=0,\pm 1, \pm 2, \ldots,
\end{equation}
dove $L$ è la lunghezza della catena. Questi vettori d'onda definiscono i modi vibrazionali del sistema. Quanto appena scritto serve solamente a dare un'idea di come vibri un reticolo cristallino, tuttavia descrivere dettagliatamente i modi vibrazionali di un sistema complesso non è necessario per la stesura di questa tesi e si rimanda a \citep{BOOK:Jacoboni}. Si conclude questa breve parentesi con una definizione. Si chiama \textbf{densità degli stati}, il numero di vettori d'onda per unità di lunghezza, cioè nel caso unidimensionale precedente
\begin{equation}
g(k) = \frac{L}{2\pi}.
\end{equation}
Il trasporto degli elettroni è fortemente influenzato dalle vibrazioni reticolari. Inoltre, supponendo che la loro ampiezza è piccola rispetto alle distanze tra i nuclei, le vibrazioni reticolari sono descritte come oscillazioni armoniche indipendenti. In meccanica quantistica i valori di energia assunti da ogni modo vibrante sono in numero discreto, cioè
\begin{equation}
\varepsilon_n = (n+1/2)\hbar\omega, \qquad n=0,1,2,\ldots,
\end{equation}
dove $\omega$ è la frequenza associata all'oscillatore armonico. Allora i modi corrispondenti a questi valori di energia discreti sono detti \textbf{fononi}. Un fonone è quindi un quanto vibrazionale del reticolo e rappresenta l'analogo di quanto visto per la quantizzazione dei campi elettromagnetici che invece dà origine ai fotoni. Se l'energia è $\varepsilon_n$ allora si dice che vi sono $n$ fononi corrispondenti. Quando un fonone è sottoposto a qualche interazione con il reticolo allora il suo stato viene modificato e si dice che avviene un fenomeno di \textbf{scattering}. In particolare, se lo stato del fonone considerato ha energia $\varepsilon_n$, può avvenire una transizione o allo stato con energia $\varepsilon_{n-1}$ oppure allo stato con energia $\varepsilon_{n+1}$. Nel primo caso si dice che è avvenuto un \textbf{assorbimento}, nel secondo caso si dice che è avvenuta una \textbf{emissione}.

Nei cristalli si possono avere parecchi tipi di vibrazioni reticolari, che vengono però raggruppate in due grandi classi: le onde ottiche e le onde acustiche. La relazione di dispersione delle onde acustiche è lineare, cioè
\begin{equation}
\hbar \omega\ped{ac} (\vert \mathbf{k} \vert) \simeq c_s \hbar \vert \mathbf{k} \vert,
\end{equation}
dove $c_s$ è la velocità del suono nel cristallo. Inoltre l'energia di queste onde è piccola. Le onde ottiche hanno una relazione di dispersione costante, cioè
\begin{equation}
\hbar \omega\ped{op} (\vert \mathbf{k} \vert) \simeq \hbar \omega\ped{op} = \mbox{costante}.
\end{equation}
Inoltre i valori dell'energia sono molto più grandi rispetto a quelli delle onde acustiche. Poiché l'energia di queste vibrazioni è quantizzata, esse danno origine ai fononi, che di conseguenza si distinguono in \textbf{fononi ottici} e \textbf{fononi acustici}. Poiché la relazione di dispersione è in generale periodica e il suo periodo coincide con l'ampiezza della prima zona di Brillouin (cfr.~\citep{BOOK:Jacoboni}) allora i vettori d'onda si considerano variabili nella prima zona di Brillouin.

\subsection{Densità degli stati}\label{PAR:CAP1:Densita_stati}
Si consideri un reticolo cristallino tridimensionale finito di vettori primitivi $\mathbf{a}_1$, $\mathbf{a}_2$ e $\mathbf{a}_3$. Sia $N_i$ il numero di vertici del reticolo nella direzione di $\mathbf{a}_i$. Allora il numero totale di vertici è $N=N_1 N_2 N_3$ che corrisponde al numero di celle primitive del reticolo. Poiché il potenziale del reticolo è periodico, per il Teorema di Bloch, le autofunzioni di una certa hamiltoniana assumono la forma
\begin{equation}
\psi_{\mathbf{k}} (\mathbf{x})= u(\mathbf{x})e^{i\mathbf{k}\cdot\mathbf{x}},
\end{equation}
dove $u(\mathbf{x})$ è una funzione periodica di periodo pari a quello del potenziale. Imponendo le condizioni al bordo
\begin{equation}
\psi_{\mathbf{k}} (\mathbf{0}) = \psi_{\mathbf{k}} (N_i\mathbf{a}_i), \qquad i=1,2,3,
\end{equation}
si può scegliere $u(\mathbf{x})=\psi_{\mathbf{k}} (\mathbf{x})$ e di conseguenza si ottiene
\begin{equation}
e^{i\mathbf{k}\cdot\mathbf{x}}=1.
\end{equation}
Quindi $\mathbf{k}$ deve essere un vettore del reticolo reciproco, cioè si può scrivere nella forma
\begin{equation}
\mathbf{k} = x_1\mathbf{b}_1+x_2\mathbf{b}_2+x_3\mathbf{b}_3, \qquad x_i=m_i/N_i, \quad m_1,m_2,m_3\in\mathbb{Z}.
\end{equation}
Allora l'incremento minimo si ha per $m_1=m_2=m_3=1$ che corrisponde ad un volume $\delta\mathbf{k}$ nello spazio reciproco,
\begin{equation}
\delta\mathbf{k}=\frac{\mathbf{b}_1}{N_1}\cdot\frac{\mathbf{b}_2}{N_2}\times \frac{\mathbf{b}_3}{N_3} = \frac{1}{N}\mbox{mis}(\mathcal{B})=\frac{1}{N}\frac{(2\pi)^3}{V_c} = \frac{(2\pi)^3}{V},
\end{equation}
dove $V$ è il volume del reticolo cristallino e $V_c$ è il volume di una cella primitiva. Allora la densità degli stati nel caso tridimensionale è
\begin{equation}
g(\mathbf{k}) = \frac{V}{(2\pi)^3}.
\end{equation}
Nei cristalli reali, la distanza tra due vettori d'onda è molto piccola, ovvero la densità degli stati è molto grande. Sia $f(\mathbf{k})$ una funzione regolare che si vuole sommare al variare di $\mathbf{k}$ in $\mathcal{B}$. Si può assumere che in un piccolo intervallo $\Delta \mathbf{k}$ la funzione è costante e il numero di stati è dato da $g(\mathbf{k})\Delta\mathbf{k}$, cioè si ha
\begin{equation}
\sum_\mathbf{k} f(\mathbf{k}) = \sum_j f(\mathbf{k}_j)g(\mathbf{k})\Delta\mathbf{k} = \sum_j f(\mathbf{k}_j)\frac{V}{(2\pi)^3}\Delta\mathbf{k}\approx \frac{V}{(2\pi)^3}\int_{\mathcal{B}} \! f(\mathbf{k}) \, \diff^3 \mathbf{k},
\end{equation}
dove l'ultimo passaggio si ottiene quando $\Delta\mathbf{k}$ diventa sufficientemente piccolo.

Si consideri un sistema di elettroni o di fononi all'interno di un reticolo cristallino. Se $f(\mathcal{E}(\mathbf{k}))$ rappresenta il numero di occupazione degli stati con energia $\mathcal{E}(\mathbf{k})$ e se si conosce l'espressione di $\mathcal{E}(\mathbf{k})$, allora invertendo tale espressione si può calcolare il numero totale di elettroni o fononi, cioè
\begin{equation}
N = \int \! f(\mathcal{E}(\mathbf{k}))\mathcal{N}(\mathcal{E}) \, \diff \mathcal{E},
\end{equation}
avendo effettuato un cambiamento di variabile nell'integrale, per cui $\mathcal{N}(\mathcal{E})$ rappresenta ancora la densità degli stati. La distribuzione del numero di occupazione degli stati è diversa a seconda che si consideri un sistema di fononi o di elettroni. 

Nel caso dei fononi, il numero $n_{\boldsymbol{\xi}}$ dei possibili stati che un fonone di vettore d'onda $\boldsymbol{\xi}$ può assumere segue la \textbf{distribuzione di Bose-Einstein}, cioè
\begin{equation}
n_{\boldsymbol{\xi}} = \frac{1}{e^{\hbar\omega_{\boldsymbol{\xi}}/k_B T}-1},
\end{equation}
dove $\omega_{\boldsymbol{\xi}}$ è la frequenza del fonone, $T$ è la temperatura del reticolo cristallino e $k_B$ è la costante di Boltzmann, pari a
\begin{equation}
k_B = \si{\num{1.38064852(79)e-23} \joule\per\kelvin}.
\end{equation}
Si osservi che la distribuzione degli stati dipende dall'energia del fonone, che è rappresentata dalla quantità $\hbar\omega_{\boldsymbol{\xi}}$.

Invece per quanto riguarda gli elettroni, si assume che la distribuzione del numero di occupazione degli stati con energia $\varepsilon(\mathbf{k})$ segue la \textbf{statistica di Fermi-Dirac}, cioè
\begin{equation}
f_{FD}(\varepsilon(\mathbf{k})) = \frac{1}{1+\exp\left( \frac{\varepsilon(\mathbf{k})-\varepsilon_F}{k_B T} \right)},
\end{equation}
dove $\varepsilon_F$ è il livello di Fermi, $k_B$ è la costante di Boltzmann e $T$ è la temperatura del reticolo cristallino.


\section{I modelli cinetici di trasporto}
Nei successivi paragrafi si costruirà un modello per descrivere il moto di un certo numero di elettroni soggetti ad un potenziale esterno, utilizzando un approccio di tipo statistico. Dopo aver illustrato l'approssimazione semiclassica si procederà costruendo l'equazione semiclassica di Liouville da cui si dedurrà l'equazione semiclassica di Vlasov, con un metodo analogo a quello usato per i gas rarefatti. Infine, si dedurrà l'equazione semiclassica di Boltzmann, specificandone il termine collisionale ottenuto dalla descrizione dei meccanismi di scattering ed in virtù della regola d'oro di Fermi.

Per il contenuto dei successivi paragrafi si farà riferimento a \citep{DISP:Anile}, \citep{BOOK:Kardar}, \citep{BOOK:Jacoboni}, \citep{BOOK:Jacoboni_Lugli} e \citep{BOOK:Lundstrom}.

\subsection{L'approssimazione semiclassica}
L'approssimazione semiclassica consiste nel considerare gli elettroni e le lacune come se fossero particelle puntiformi. Un elettrone che si trova nella $n$-esima banda energetica $\mathcal{E}_n(\mathbf{k})$ è descritto come una particella puntiforme che si muove con una velocità  $\mathbf{v}_n(\mathbf{k})$, data dalla \eqref{EQ:CAP1:Vel_gruppo_elettroni}. Inoltre si suppone che esso si muova come se fosse un elettrone libero con impulso pari all'impulso del cristallo. Per descriverne il moto occorre scrivere le equazioni di Hamilton per la coppia di variabili coniugate $(\mathbf{x},\mathbf{p})=(\mathbf{x},\hbar\mathbf{k})$, corrispondenti all'hamiltoniana
\begin{equation}
\mathcal{H}(\mathbf{x},\mathbf{p})=\mathcal{E}_n\left( \frac{\mathbf{p}}{\hbar} \right) -eV(\mathbf{x}),
\end{equation}
dove $V$ è un potenziale elettrico esterno che agisce sul sistema. Allora, le equazioni del moto sono date da
\begin{equation}
\dot{\mathbf{x}} = \mathbf{v}_n(\mathbf{k}), \qquad \dot{\mathbf{k}} = \frac{e}{\hbar} \nabla_{\mathbf{x}} V(\mathbf{x}).
\end{equation}
Per quanto riguarda le lacune, occorre osservare che il comportamento di una lacuna nella $n$-esima banda è quello di una particella con carica elettrica positiva e con relazione di dispersione data da $-\mathcal{E}_n(\mathbf{k})$. In particolare, applicando le precedenti considerazioni alle bande di conduzione e di valenza, si ha
\begin{equation}
\dot{\mathbf{x}} = \mathbf{v}_c(\mathbf{k}), \qquad \dot{\mathbf{k}} = \frac{e}{\hbar} \nabla_{\mathbf{x}} V(\mathbf{x}), \qquad \mathbf{v}_c(\mathbf{k})=\frac{1}{\hbar}\nabla_{\mathbf{k}} E^c(\mathbf{k}),
\end{equation}
per la banda di conduzione, e
\begin{equation}
\dot{\mathbf{x}} = \mathbf{v}_v(\mathbf{k}), \qquad \dot{\mathbf{k}} = -\frac{e}{\hbar} \nabla_{\mathbf{x}} V(\mathbf{x}), \qquad \mathbf{v}_c(\mathbf{k})=-\frac{1}{\hbar}\nabla_{\mathbf{k}} E^v(\mathbf{k}),
\end{equation}
per la banda di valenza. L'approssimazione semiclassica ha significato fisico se i campi esterni applicati al solido variano lentamente e non sono molto intensi. Poiché nelle applicazioni tecnologiche, le bande di energia più elevata sono scarsamente popolate, solitamente si studia il trasporto di carica limitatamente alla banda di conduzione più bassa in energia.

\subsection{Equazione semiclassica di Liouville}
Il trasporto di carica nei solidi è dovuto ad un numero $N$ di elettroni il cui moto può essere studiato risolvendo l'equazione di Schr\"{o}dinger per la funzione d'onda del sistema degli $N$ elettroni. Utilizzando invece l'approssimazione semiclassica, gli elettroni si considerano come $N$ particelle puntiformi. Si indichi con $\mathbf{x}_i$ e con $\mathbf{k}_i$ rispettivamente la posizione e il vettore d'onda della particella $i$-esima. Supponendo che le forze $\mathbf{F}_i$ siano indipendenti da tutti i $\mathbf{k}_i$ e possano essere rappresentate da un campo avente potenziale $U(\mathbf{x}_1, \ldots , \mathbf{x}_N)$, il moto è descritto da un sistema classico
con hamiltoniana
\begin{equation}
\mathcal{H}(\mathbf{x}_1,\ldots,\mathbf{x}_N,\hbar\mathbf{k}_1,\ldots,\hbar\mathbf{k}_N)=\sum_{i=1}^N \mathcal{E}_i(\mathbf{k}_i)+U(\mathbf{x}_1, \ldots , \mathbf{x}_N),
\end{equation}
dove $\mathcal{E}_i(\mathbf{k}_i)$ rappresenta la banda di conduzione più bassa in energia relativa alla particella $i$-esima. Le equazioni di tale sistema sono quindi
\begin{equation}\label{EQ:CAP1:Sist_semicl_Hamilton}
\dot{\mathbf{x}}_i = \frac{1}{\hbar}\nabla_{\mathbf{k}_i}\mathcal{H} \equiv \mathbf{v}_i, \qquad \dot{\mathbf{k}}_i = -\frac{1}{\hbar}\nabla_{\mathbf{x}_i}\mathcal{H} \equiv \frac{1}{\hbar}\mathbf{F}_i.
\end{equation}
Inoltre occorre aggiungere le condizioni iniziali
\begin{equation}
\mathbf{x}_i(0)=\mathbf{x}_{i,0}, \qquad \mathbf{k}_i(0)=\mathbf{k}_{i,0}.
\end{equation}
Poiché il numero $N$ è molto grande, questo problema presenta delle difficoltà a livello computazionale. Pertanto si ricorre alla meccanica statistica, le cui basi sono state già introdotte nel Paragrafo \ref{PAR:CAP1:Meccanica_statistica}. Sia
\begin{equation}
P_N(\mathbf{x}_1,\ldots,\mathbf{x}_N,\mathbf{k}_1,\ldots,\mathbf{k}_N,t)
\end{equation}
la densità di probabilità congiunta che, per ogni istante $t\in\mathbb{R}$, definisce la distribuzione di probabilità degli elettroni nello spazio delle fasi $\mathcal{F}=\mathbb{R}^{3N}\times\mathcal{B}^N\subset\mathbb{R}^{6N}$. Si introduca la notazione compatta
\begin{equation}
\mathbf{x}=(\mathbf{x}_1,\ldots,\mathbf{x}_N), \qquad \mathbf{k}=(\mathbf{k}_1,\ldots,\mathbf{k}_N).
\end{equation}
Allora, procedendo come nel Paragrafo \ref{PAR:CAP1:Meccanica_statistica}, è possibile trovare l'equazione che descrive l'evoluzione temporale della densità di probabilità congiunta, cioè dalla \eqref{EQ:CAP1:Liouville_4} si ha
\begin{equation}
\frac{\partial P_N}{\partial t} + \dot{\mathbf{x}} \cdot \nabla_{\mathbf{x}} \: P_N + \dot{\mathbf{k}} \cdot \nabla_{\mathbf{k}} \: P_N=0.
\end{equation}
Infine, utilizzando le \eqref{EQ:CAP1:Sist_semicl_Hamilton}, si ottiene
\begin{equation}
\frac{\partial P_N}{\partial t} + \sum_{i=1}^N \mathbf{v}(\mathbf{k}_i) \cdot \nabla_{\mathbf{x}_i} P_N + \frac{1}{\hbar} \sum_{i=1}^N \mathbf{F}_i \cdot \nabla_{\mathbf{k}_i} P_N=0,
\end{equation}
detta \textbf{equazione semiclassica di Liouville}. Essa può essere scritta anche utilizzando parentesi di Poisson scrivendo
\begin{equation}\label{EQ:CAP1:Eq_semicl_Liouville_Poisson}
\frac{\partial P_N}{\partial t} + \left\lbrace P_N,\mathcal{H} \right\rbrace =0.
\end{equation}

\subsection{La gerarchia BBGKY}
Si supponga che le forze esterne possono essere espresse come la somma di un campo elettrico esterno e del potenziale di interazione a due particelle, cioè
\begin{equation}\label{EQ:CAP1:Potenziale}
\mathcal{U}(\mathbf{x}) = \sum_{i=1}^N \mathcal{U}\ap{ext}(\mathbf{x}_i) + \frac{1}{2} \sum_{\substack{i,j=1\\ i \neq j}}^N \mathcal{U}\ap{int}(\mathbf{x}_i,\mathbf{x}_j),
\end{equation}
con il potenziale di interazione simmetrico, cioè
\begin{equation}
\mathcal{U}\ap{int}(\mathbf{x}_i,\mathbf{x}_j)=\mathcal{U}\ap{int}(\mathbf{x}_j,\mathbf{x}_i).
\end{equation}
Di conseguenza, le forze esterne esercitate sul sistema per mezzo del potenziale $\mathcal{U}\ap{ext}(\mathbf{x}_i)$ sono date da
\begin{equation}\label{EQ:CAP1:Forze_esterne}
\mathbf{F}\ap{ext}(\mathbf{x}_i)=-\nabla_{\mathbf{x}_i}\mathcal{U}\ap{ext}(\mathbf{x}_i), \quad \forall i=1,\ldots,N
\end{equation}
e le forze di interazione tra due particelle, esercitate sul sistema per mezzo del potenziale $\mathcal{U}\ap{int}(\mathbf{x}_i,\mathbf{x}_j)$, sono data da
\begin{equation}\label{EQ:CAP1:Forze_interne}
\mathbf{F}\ap{int}(\mathbf{x}_i,\mathbf{x}_j)=-\nabla_{\mathbf{x}_{ij}}\mathcal{U}\ap{int}(\mathbf{x}_i,\mathbf{x}_j), \quad \forall i,j=1,\ldots,N, \quad i\neq j,
\end{equation}
dove $\mathbf{x}_{ij}=(\mathbf{x}_i,\mathbf{x}_j)$. Questa scelta permette di dimostrare che la densità $P_N$ è indipendente dalla numerazione delle particelle, cioè esse risultano \textbf{indistinguibili} l'una dall'altra, ovvero
\begin{equation}\label{EQ:CAP1:Particelle_indistinguibili}
P_N(\mathbf{x}_1,\ldots,\mathbf{x}_N,\mathbf{k}_1,\ldots,\mathbf{k}_N,t)=P_N(\mathbf{x}_{\pi(1)},\ldots,\mathbf{x}_{\pi(N)},\mathbf{k}_{\pi(1)},\ldots,\mathbf{k}_{\pi(N)},t),
\end{equation}
per ogni permutazione $\pi$ di $\left\lbrace 1,\ldots,N \right\rbrace $, con $\mathbf{x}_i\in\mathbb{R}^3$, $\mathbf{k}_i\in\mathcal{B}$. Inoltre se la \eqref{EQ:CAP1:Particelle_indistinguibili} è verificata per $t=0$ allora essa è verificata per ogni $t$. \`{E} possibile introdurre la densità di probabilità congiunta per un sottoinsieme di $d$ particelle estratte da un insieme $N$ di particelle nel modo seguente
\begin{equation}\label{EQ:CAP1:Densita_d_particelle}
\begin{aligned}
P_N^{(d)}(\mathbf{x}_1,\ldots,\mathbf{x}_d, & \mathbf{k}_1,\ldots,\mathbf{k}_d,t)=\\
&=\int_{(\mathbb{R}^3\times\mathcal{B})^{N-d}} \! P_N(\mathbf{x}_1,\ldots,\mathbf{x}_N,\mathbf{k}_1,\ldots,\mathbf{k}_N,t) \, \diff \mathbf{x}_{d+1} \diff \mathbf{k}_{d+1} \cdots \diff \mathbf{x}_{N} \diff \mathbf{k}_{N}.
\end{aligned}
\end{equation}
Si vogliono determinare le leggi di evoluzione nel tempo delle $P_N^{(d)}$, con $1\leq d\leq N-1$. Esse si ottengono integrando l'equazione di Liouville rispetto a $3(N-d)$ variabili di posizioni e vettori d'onda. Per fare ciò si consideri la seguente hamiltoniana
\begin{equation}
\mathcal{H}(\mathbf{x},\hbar\mathbf{k})=\sum_{i=1}^N\frac{\vert \mathbf{p}_i \vert^2}{2m_e} +\mathcal{U}(\mathbf{x})
\end{equation}
dove $\mathcal{U}(\mathbf{x})$ dato dalla \eqref{EQ:CAP1:Potenziale}, $m_e$ è la massa dell'elettrone e ricordando che $\mathbf{p}_i=\hbar\mathbf{k}_i$. Per valutare l'evoluzione nel tempo delle $P_N^{(d)}$, è opportuno suddividere l'hamiltoniana in
\begin{equation}
\mathcal{H}=\mathcal{H}_d+\mathcal{H}_{N-d}+\mathcal{H}',
\end{equation}
dove $\mathcal{H}_d$ e $\mathcal{H}_{N-d}$ includono le interazioni tra ogni gruppo di particelle, essendo
\begin{equation}\label{EQ:CAP1:H_d}
\mathcal{H}_d = \sum_{n=1}^d \left[ \frac{\vert \mathbf{p}_n \vert^2}{2m} + \mathcal{U}\ap{ext}\left( \mathbf{x}_n \right)  \right] + \frac{1}{2} \sum_{\substack{n,m=1\\ n \neq m}}^d \mathcal{U}\ap{int} \left( \mathbf{x}_n,\mathbf{x}_m \right),
\end{equation}
\begin{equation}\label{EQ:CAP1:H_N_d}
\mathcal{H}_{N-s} = \sum_{i=d+1}^N \left[ \frac{\vert \mathbf{p}_i \vert^2}{2m} + \mathcal{U}\ap{ext}\left( \mathbf{x}_i \right)  \right] + \frac{1}{2}  \sum_{\substack{i,j=d+1\\ i \neq j}}^N \mathcal{U}\ap{int} \left( \mathbf{x}_i,\mathbf{x}_j \right),
\end{equation}
invece le interazioni tra le molecole sono contenute in
\begin{equation}\label{EQ:CAP1:H_primo}
\mathcal{H}'= \sum_{n=1}^d \sum_{i=d+1}^N \mathcal{U}\ap{int} \left( \mathbf{x}_n,\mathbf{x}_i \right).
\end{equation}
Dalla \eqref{EQ:CAP1:Densita_d_particelle} e dalla \eqref{EQ:CAP1:Eq_semicl_Liouville_Poisson} si ha che l'evoluzione nel tempo delle $P_N^{(d)}$ può essere ottenuta calcolando
\begin{equation}\label{EQ:CAP1:Evoluzione_P_N_d}
\frac{\partial P_N^{(d)}}{\partial t} = \int \! \frac{\partial P_N}{\partial t} \, \prod_{i=d+1}^N \diff V_i = - \int \! \left\lbrace P_N,\mathcal{H}_d+\mathcal{H}_{N-d}+\mathcal{H}' \right\rbrace  \, \prod_{i=d+1}^N \diff V_i,
\end{equation}
in cui le parentesi di Poisson all'ultimo membro possono essere calcolate separatamente in virtù della proprietà \ref{PTA:CAP1:Parentesi_Poisson_pta_2} della Proposizione \ref{PROP:CAP1:Parentesi_Poisson}.

Poiché l'integrazione non riguarda le prime $d$ coordinate allora è possibile scambiare il segno di integrale e le parentesi di Poisson nel primo termine che si viene a creare applicando la proprietà \ref{PTA:CAP1:Parentesi_Poisson_pta_2} della Proposizione \ref{PROP:CAP1:Parentesi_Poisson} nella \eqref{EQ:CAP1:Evoluzione_P_N_d}, cioè si ha
\begin{equation}\label{EQ:CAP1:Evoluzione_H_d}
\int \! \left\lbrace P_N,\mathcal{H}_d \right\rbrace  \, \prod_{i=d+1}^N \diff V_i = \left\lbrace \left( \int \! P_N \, \prod_{i=d+1}^N \diff V_i \right),\mathcal{H}_d  \right\rbrace = \left\lbrace  P_N^{(d)},\mathcal{H}_d\right\rbrace .
\end{equation}
Per quanto riguarda il secondo termine della \eqref{EQ:CAP1:Evoluzione_P_N_d}, scrivendo esplicitamente le parentesi di Poisson e integrando successivamente per parti, si ottiene
\begin{equation}\label{EQ:CAP1:Evoluzione_H_N_d}
-\int \! \left\lbrace P_N,\mathcal{H}_{N-d} \right\rbrace  \, \prod_{i=d+1}^N \diff V_i = 0.
\end{equation}
Infine, per quanto riguarda il terzo termine della \eqref{EQ:CAP1:Evoluzione_P_N_d}, scrivendo esplicitamente le parentesi di Poisson, suddividendo la sommatoria, integrando per parti e semplificando si ha
\begin{equation}\label{EQ:CAP1:Evoluzione_H_primo}
-\int \! \left\lbrace P_N,\mathcal{H}' \right\rbrace \, \prod_{i=d+1}^N \diff V_i = (N-d) \sum_{n=1}^d \int \! \diff V_{d+1} \nabla_{\mathbf{x}_n} \mathcal{U}\ap{int}(\mathbf{x}_n,\mathbf{x}_{d+1}) \cdot \nabla_{\mathbf{p}_n} P_N^{(d+1)}
\end{equation}
(ulteriori dettagli su come ottenere \eqref{EQ:CAP1:Evoluzione_H_N_d} e \eqref{EQ:CAP1:Evoluzione_H_primo} sono presenti in \citep{BOOK:Kardar}).
Allora, sommando \eqref{EQ:CAP1:Evoluzione_H_d}, \eqref{EQ:CAP1:Evoluzione_H_N_d} e \eqref{EQ:CAP1:Evoluzione_H_primo} e sostituendo nella \eqref{EQ:CAP1:Evoluzione_P_N_d}, si ha
\begin{equation}\label{EQ:CAP1:Gerarchia_BBGKY_P_N}
\frac{\partial P_N^{(d)}}{\partial t} - \left\lbrace \mathcal{H}_d,P_N^{(d)} \right\rbrace = (N-d) \sum_{n=1}^d \int \! \diff V_{d+1} \nabla_{\mathbf{x}_n} \mathcal{U}\ap{int}(\mathbf{x}_n,\mathbf{x}_{d+1}) \cdot \nabla_{\mathbf{p}_n} P_N^{(d+1)}.
\end{equation}
Infine, scrivendo esplicitamente le parentesi di Poisson e ricordando che $\mathbf{p}_i=\hbar\mathbf{k}_i$ per ogni $i=1,\ldots,N$, si ottiene
\begin{equation}\label{EQ:CAP1:Poisson_H_d}
\left\lbrace \mathcal{H}_d,P_N^{(d)} \right\rbrace = \frac{1}{\hbar} \sum_{i=1}^d \left( \nabla_{\mathbf{x}_i} \mathcal{H}_d \cdot \nabla_{\mathbf{k}_i} P_N^{(d)} + \nabla_{\mathbf{k}_i} \mathcal{H}_d \cdot \nabla_{\mathbf{x}_i} P_N^{(d)} \right),
\end{equation}
in cui la sommatoria è estesa per $i=1,\ldots,d$ in quanto le funzioni che in essa compaiono sono nelle variabili $\mathbf{x}_1,\ldots,\mathbf{x}_d,\mathbf{k}_1,\ldots,\mathbf{k}_d$.
 
Per il calcolo del $\nabla_{\mathbf{x}_i} \mathcal{H}_d$ si osservi che nelle sommatorie della \eqref{EQ:CAP1:H_d} sopravvivono soltanto i termini che contengono $\mathbf{x}_i$ e inoltre si elimina il termine dipendente dai $\mathbf{p}_n$, ottenendo
\begin{equation}
\begin{aligned}
\nabla_{\mathbf{x}_i} \mathcal{H}_d & = \nabla_{\mathbf{x}_i} \mathcal{U}\ap{ext} (\mathbf{x}_i) +\frac{1}{2} \sum_{\substack{m=1\\ m \neq i}}^d  \nabla_{\mathbf{x}_i} \mathcal{U}\ap{int}(\mathbf{x}_i,\mathbf{x}_m) + \frac{1}{2} \sum_{\substack{n=1\\ n \neq i}}^d  \nabla_{\mathbf{x}_i} \mathcal{U}\ap{int}(\mathbf{x}_n,\mathbf{x}_i) =\\
 & = \nabla_{\mathbf{x}_i} \mathcal{U}\ap{ext} (\mathbf{x}_i) + \sum_{\substack{m=1\\ m \neq i}}^d  \nabla_{\mathbf{x}_i} \mathcal{U}\ap{int}(\mathbf{x}_i,\mathbf{x}_m)
\end{aligned}
\end{equation}
in cui l'ultima uguaglianza è ottenuta sfruttando la simmetria di $\mathcal{U}\ap{int}$. Inoltre in virtù delle \eqref{EQ:CAP1:Forze_esterne} e \eqref{EQ:CAP1:Forze_interne}, si ha
\begin{equation}\label{EQ:CAP1:Grad_x_H_d}
\nabla_{\mathbf{x}_i} \mathcal{H}_d = -\mathbf{F}\ap{ext}(\mathbf{x}_i) - \sum_{\substack{m=1\\ m \neq i}}^d  \mathbf{F}\ap{int}(\mathbf{x}_i,\mathbf{x}_m).
\end{equation}
Per quanto riguarda il calcolo del $\nabla_{\mathbf{k}_i} \mathcal{H}_d$ si osservi che nella \eqref{EQ:CAP1:H_d} sopravvive solo il termine contenente $\mathbf{p}_i$, da cui, ricordando che $\mathbf{p}_i=\hbar\mathbf{k}_i$ e che $\mathbf{p}_i=m_e\mathbf{v}(\mathbf{k}_i)$, si ottiene
\begin{equation}\label{EQ:CAP1:Grad_k_H_d}
\nabla_{\mathbf{k}_i} \mathcal{H}_d = \nabla_{\mathbf{k}_i} \frac{\vert \hbar \mathbf{k}_i \vert^2}{2m_e} = \frac{\hbar^2}{m_e} \mathbf{k}_i = \hbar \frac{\mathbf{p}_i}{m_e} = \hbar \mathbf{v}(\mathbf{k}_i).
\end{equation}
Sostituendo le \eqref{EQ:CAP1:Grad_x_H_d} e \eqref{EQ:CAP1:Grad_k_H_d} nella \eqref{EQ:CAP1:Poisson_H_d} si ha
\begin{equation}\label{EQ:CAP1:Poisson_P_N_mod}
\begin{aligned}
\left\lbrace \mathcal{H}_d,P_N^{(d)} \right\rbrace = & - \frac{1}{\hbar} \sum_{i=1}^d \mathbf{F}\ap{ext}(\mathbf{x}_i) \cdot \nabla_{\mathbf{k}_i} P_N^{(d)} -\frac{1}{\hbar} \sum_{i=1}^d \sum_{\substack{j=1\\ j \neq i}}^d  \mathbf{F}\ap{int}(\mathbf{x}_i,\mathbf{x}_j)\cdot\nabla_{\mathbf{k}_i} P_N^{(d)}\\
&- \sum_{i=1}^d \mathbf{v}(\mathbf{k}_i) \cdot \nabla_{\mathbf{x}_i} P_N^{(d)}.
\end{aligned}
\end{equation}
Inoltre si consideri il secondo membro della \eqref{EQ:CAP1:Gerarchia_BBGKY_P_N}, osservando che
\begin{equation}
\nabla_{\mathbf{x}_n} \mathcal{U}\ap{int}(\mathbf{x}_n,\mathbf{x}_{d+1})=-\mathbf{F}\ap{int}(\mathbf{x}_n,\mathbf{x}_{d+1})
\end{equation}
e ricordando che $\mathbf{p}_n=\hbar\mathbf{k}_n$ si ottiene
\begin{equation}\label{EQ:CAP1:Gerarchia_BBGKY_mod}
\begin{aligned}
(N-d) \sum_{n=1}^d \int \! \diff V_{d+1} & \nabla_{\mathbf{x}_n} \mathcal{U}\ap{int}(\mathbf{x}_n,\mathbf{x}_{d+1}) \cdot \nabla_{\mathbf{p}_n} P_N^{(d+1)} =\\
&-\frac{1}{\hbar} (N-d) \sum_{n=1}^d \int \! \diff V_{d+1} \mathbf{F}\ap{int}(\mathbf{x}_n,\mathbf{x}_{d+1}) \cdot \nabla_{\mathbf{k}_n} P_N^{(d+1)}
\end{aligned}
\end{equation}
In definitiva, sostituendo la \eqref{EQ:CAP1:Poisson_P_N_mod} e la \eqref{EQ:CAP1:Gerarchia_BBGKY_mod} nella \eqref{EQ:CAP1:Gerarchia_BBGKY_P_N}, si ha
\begin{equation}\label{EQ:CAP1:Gerarchia_BBGKY}
\begin{aligned}
\frac{\partial P_N^{(d)}}{\partial t} & + \frac{1}{\hbar} \sum_{i=1}^d \mathbf{F}\ap{ext}(\mathbf{x}_i) \cdot \nabla_{\mathbf{k}_i} P_N^{(d)} + \frac{1}{\hbar} \sum_{i=1}^d \sum_{\substack{j=1\\ j \neq i}}^d  \mathbf{F}\ap{int}(\mathbf{x}_i,\mathbf{x}_j)\cdot\nabla_{\mathbf{k}_i} P_N^{(d)} \\
& + \sum_{i=1}^d \mathbf{v}(\mathbf{k}_i) \cdot \nabla_{\mathbf{x}_i} P_N^{(d)} \\
& + \frac{1}{\hbar} (N-d) \sum_{i=1}^d \int \! \diff V_{d+1} \mathbf{F}\ap{int}(\mathbf{x}_i,\mathbf{x}_{d+1}) \cdot \nabla_{\mathbf{k}_i} P_N^{(d+1)} =0.
\end{aligned}
\end{equation}
La \eqref{EQ:CAP1:Gerarchia_BBGKY} è detta \textbf{gerarchia BBGKY}, attribuita agli scienziati Bogoliubov, Born, Green, Kirkwood e Yvon che per primi l'hanno studiata. Si tratta di una gerarchia di equazioni nel senso che l'equazione per la $P_N^{(d)}$ dipende dalla successiva, cioè dalla $P_N^{(d+1)}$.

\subsection{L'equazione semiclassica di Vlasov}
A partire dalle equazioni della gerarchia \eqref{EQ:CAP1:Gerarchia_BBGKY} si vuole ricavare un'equazione per la densità di probabilità di una particella, cioè la $P_N^{(1)}$, estratta da un insieme molto grande di particelle, facendo cioè tendere $N\rightarrow\infty$. Questa funzione ha una interpretazione fisica molto importante, poiché permette di calcolare la densità del numero di elettroni. Si consideri lo spazio delle fasi $\mathcal{F}=\mathbb{R}^3\times\mathcal{B}$. Se vi sono $N$ elettroni, la densità del numero di elettroni nel volumetto $\diff \mathbf{x}\diff \mathbf{k}$ centrato in $(\mathbf{x}, \mathbf{k})$ è data da
\begin{equation}
N P_N^{(1)}(\mathbf{x},\mathbf{k},t)\diff \mathbf{x}\diff \mathbf{k}.
\end{equation}
Quindi la densità $n$ del numero di elettroni per unità di volume si ottiene integrando sull'insieme di tutti i possibili valori di $\mathbf{k}$. Ponendo $f_N(\mathbf{x},\mathbf{k},t)=NP_N^{(1)}(\mathbf{x},\mathbf{k},t)$, si ha
\begin{equation}\label{EQ:CAP1:Densita_elettroni}
n(\mathbf{x},t) = \int_{\mathcal{B}} \!  f_N(\mathbf{x},\mathbf{k},t) \, \diff \mathbf{k}.
\end{equation}
Analogamente, la velocità media $\mathbf{u}$ è data da
\begin{equation}
\begin{aligned}
\mathbf{u}(\mathbf{x},t)&=\frac{\int_{\mathcal{B}} \! \mathbf{v}(\mathbf{k}) N P_N^{(1)}(\mathbf{x},\mathbf{k},t) \, \diff \mathbf{k}}{\int_{\mathcal{B}} \!  N P_N^{(1)}(\mathbf{x},\mathbf{k},t) \, \diff \mathbf{k}}=\\
&= \frac{1}{n(\mathbf{x},t)}\int_{\mathcal{B}} \! \mathbf{v}(\mathbf{k}) N P_N^{(1)}(\mathbf{x},\mathbf{k},t) \, \diff \mathbf{k}.
\end{aligned}
\end{equation}
Invece, la funzione densità $P_N^{(2)}$ fornisce informazioni sull'interazione tra due particelle estratte da un insieme di $N$, nello spazio delle fasi $\mathcal{F}^2$. Per proseguire la trattazione risulterà di fondamentale importanza supporre che le particelle siano \textbf{indipendenti}, cioè
\begin{equation}\label{EQ:CAP1:Densita_2_part_indipendenti}
P_N^{(2)}(\mathbf{x},\mathbf{y},\mathbf{k},\mathbf{k}',t)=P_N^{(1)}(\mathbf{x},\mathbf{k},t)P_N^{(1)}(\mathbf{y},\mathbf{k}',t).
\end{equation}
Inoltre sarebbe opportuno che dopo il passaggio al limite continui ad avere significato l'espressione \eqref{EQ:CAP1:Densita_elettroni}. Per ricavare l'equazione per la densità di probabilità di una particella, si procede scrivendo la gerarchia BBGKY per $d=1$, ottenendo
\begin{equation}\label{EQ:CAP1:Densita_2_particelle}
\begin{aligned}
\frac{\partial P_N^{(1)}}{\partial t} & + \frac{1}{\hbar} \mathbf{F}\ap{ext}(\mathbf{x}) \cdot \nabla_{\mathbf{k}} P_N^{(1)} + \mathbf{v}(\mathbf{k}) \cdot \nabla_{\mathbf{x}} P_N^{(1)} \\
& + \frac{1}{\hbar} (N-1) \iint \! \mathbf{F}\ap{int}(\mathbf{x},\mathbf{y}) \cdot \nabla_{\mathbf{k}} P_N^{(2)} \, \diff \mathbf{y} \diff \mathbf{k}'=0.
\end{aligned}
\end{equation}
Considerando l'ultimo termine della \eqref{EQ:CAP1:Densita_2_particelle} e sostituendo la \eqref{EQ:CAP1:Densita_2_part_indipendenti}, si ha
\begin{equation}\label{EQ:CAP1:Densita_2_particelle_term_fin}
\begin{aligned}
\frac{1}{\hbar} (N-1) & \iint \! \mathbf{F}\ap{int}(\mathbf{x},\mathbf{y}) \cdot \nabla_{\mathbf{k}} P_N^{(2)} \, \diff \mathbf{y} \diff \mathbf{k}'=\\
&=\frac{1}{\hbar} (N-1) \iint \! \mathbf{F}\ap{int}(\mathbf{x},\mathbf{y}) \cdot \nabla_{\mathbf{k}} \left(  P_N^{(1)}(\mathbf{x},\mathbf{k},t)P_N^{(1)}(\mathbf{y},\mathbf{k}',t) \right)  \, \diff \mathbf{y} \diff \mathbf{k}'=\\
&= \frac{1}{\hbar} (N-1) \left[  \int \! \mathbf{F}\ap{int}(\mathbf{x},\mathbf{y}) \left(  \int \! P_N^{(1)}(\mathbf{y},\mathbf{k}',t) \, \diff \mathbf{k}' \right) \, \diff \mathbf{y} \right]  \cdot \nabla_{\mathbf{k}} P_N^{(1)}(\mathbf{x},\mathbf{k},t)=\\
&= \frac{1}{\hbar} \left( 1-\frac{1}{N} \right)  \left[  \int \! \mathbf{F}\ap{int}(\mathbf{x},\mathbf{y}) n(\mathbf{y},t) \, \diff \mathbf{y} \right]  \cdot \nabla_{\mathbf{k}} P_N^{(1)}(\mathbf{x},\mathbf{k},t).
\end{aligned}
\end{equation}
Inoltre, sostituendo la \eqref{EQ:CAP1:Densita_2_particelle_term_fin} nella \eqref{EQ:CAP1:Densita_2_particelle} e moltiplicando tutta l'espressione ottenuta per $N$, si ha
\begin{equation}\label{EQ:CAP1:Densita_2_f_N}
\begin{aligned}
\frac{\partial f_N}{\partial t} & + \frac{1}{\hbar} \mathbf{F}\ap{ext}(\mathbf{x}) \cdot \nabla_{\mathbf{k}} f_N + \mathbf{v}(\mathbf{k}) \cdot \nabla_{\mathbf{x}} f_N \\
& + \frac{1}{\hbar} \left( 1-\frac{1}{N} \right)  \left[  \int \! \mathbf{F}\ap{int}(\mathbf{x},\mathbf{y}) n(\mathbf{y},t) \, \diff \mathbf{y} \right]  \cdot \nabla_{\mathbf{k}} f_N=0.
\end{aligned}
\end{equation}
Supponendo che esista il limite per $N\rightarrow\infty$ di $f_N$ e ponendo
\begin{equation}
f(\mathbf{x},\mathbf{k},t)=\lim_{n\to\infty} f_N(\mathbf{x},\mathbf{k},t)
\end{equation}
è possibile passare al limite la \eqref{EQ:CAP1:Densita_2_f_N}, ottenendo
\begin{equation}\label{EQ:CAP1:Eq_Vlasov}
\frac{\partial f}{\partial t} + \mathbf{v}(\mathbf{k}) \cdot \nabla_{\mathbf{x}} f + \frac{1}{\hbar} \mathbf{F}\ap{eff}(\mathbf{x},t) \cdot \nabla_{\mathbf{k}} f=0,
\end{equation}
dove la forza efficace $\mathbf{F}\ap{eff}$ è definita come
\begin{equation}\label{EQ:CAP1:F_eff}
\mathbf{F}\ap{eff}(\mathbf{x},t) = \mathbf{F}\ap{ext}(\mathbf{x}) +  \int_{\mathbb{R}^3} \! \mathbf{F}\ap{int}(\mathbf{x},\mathbf{y}) n(\mathbf{y},t) \, \diff \mathbf{y},
\end{equation}
con
\begin{equation}
n(\mathbf{x},t) = \int_{\mathcal{B}} \! f(\mathbf{x},\mathbf{k},t) \, \diff \mathbf{k}.
\end{equation}
La funzione $f$ si chiama funzione di distribuzione e la \eqref{EQ:CAP1:Eq_Vlasov} è detta \textbf{equazione semiclassica di Vlasov}. Essa vale per un sistema di particelle semiclassiche soggette ad un potenziale esterno e ad un potenziale di interazione binaria. Per distinguere tra elettroni e lacune è necessario specificare questi potenziali.

\subsection{Il sistema di Vlasov-Poisson}
Si consideri l'equazione \eqref{EQ:CAP1:F_eff} e si supponga che le particelle siano elettroni. Se agisce sul sistema una distribuzione di carica esterna, cioè non trasportata dagli elettroni, indicata con $\rho\ap{ext}$, allora essa induce un campo elettrico $\mathbf{E}\ap{ext}$, legato alla forza esterna esercitata sul sistema $\mathbf{F}\ap{ext}$ per mezzo della relazione
\begin{equation}
\mathbf{F}\ap{ext} = -e\mathbf{E}\ap{ext}.
\end{equation}
Inoltre il legame tra la densità di carica $\rho\ap{ext}$ e il campo elettrico da essa indotto $\mathbf{E}\ap{ext}$ è dato dalla legge di Gauss per il campo elettrico
\begin{equation}
\nabla_{\mathbf{x}} \cdot (\varepsilon_s \mathbf{E}\ap{ext} ) = \rho\ap{ext},
\end{equation}
dove $\varepsilon_s$ è la costante dielettrica del semiconduttore.

Per quanto riguarda la forza $\mathbf{F}\ap{int}$ di interazione elettrone-elettrone, essa è data dalla legge di Coulomb,
\begin{equation}
\mathbf{F}\ap{int}(\mathbf{x},\mathbf{y}) = -\frac{e^2}{4\pi\varepsilon_s} \frac{\mathbf{x}-\mathbf{y}}{\vert \mathbf{x}-\mathbf{y} \vert^3}.
\end{equation}
Il campo elettrico $\mathbf{E}\ap{int}(\mathbf{x},t)$ cui è soggetto un elettrone che occupa la posizione $\mathbf{x}$ al tempo $t$ si ottiene calcolando la forza totale che agisce su di esso per effetto delle interazioni con gli altri elettroni del sistema, cioè
\begin{equation}
\int_{\mathbb{R}^3} \! \mathbf{F}\ap{int}(\mathbf{x},\mathbf{y})n(\mathbf{y},t) \, \diff \mathbf{y},
\end{equation}
e dividendo per la carica dell'elettrone. Quindi si ha
\begin{equation}
\begin{aligned}
\mathbf{E}\ap{int}(\mathbf{x},t) & = \frac{1}{-e} \int_{\mathbb{R}^3} \! \mathbf{F}\ap{int}(\mathbf{x},\mathbf{y})n(\mathbf{y},t) \, \diff \mathbf{y} =\\
&=\frac{e}{4\pi\varepsilon_s} \int_{\mathbb{R}^3} \! \frac{\mathbf{x}-\mathbf{y}}{\vert \mathbf{x}-\mathbf{y} \vert^3} n(\mathbf{y},t) \, \diff \mathbf{y}.
\end{aligned}
\end{equation}
Poiché $n(\mathbf{x},t)$ rappresenta in questo caso la densità del numero di elettroni per unità di volume, la densità di carica per unità di volume è data da
\begin{equation}
-e n(\mathbf{x},t).
\end{equation}
Allora, la legge di Gauss per il campo elettrico assume la forma
\begin{equation}
\nabla_{\mathbf{x}} \cdot (\varepsilon_s \mathbf{E}\ap{int} ) = -e n(\mathbf{x},t).
\end{equation}
Inoltre, ponendo
\begin{equation}
\mathbf{E}\ap{eff}=\mathbf{E}\ap{ext}+\mathbf{E}\ap{int},
\end{equation}
si ha che
\begin{equation}
\mathbf{F}\ap{eff}=-e\mathbf{E}\ap{eff}.
\end{equation}
Allora, il campo elettrico efficace soddisfa la relazione
\begin{equation}
\nabla_{\mathbf{x}} \cdot (\varepsilon_s \mathbf{E}\ap{eff} ) = \rho\ap{ext} -e n(\mathbf{x},t).
\end{equation}
Infine, trascurando effetti magnetici, è possibile definire il potenziale elettrico efficace, $V\ap{eff}$, scrivendo
\begin{equation}
\mathbf{E}\ap{eff}(\mathbf{x},t)=-\nabla_{\mathbf{x}} V\ap{eff}(\mathbf{x},t).
\end{equation}
Allora si ottiene la seguente equazione di Poisson per il potenziale efficace,
\begin{equation}\label{EQ:CAP1:Eq_Poisson}
-\nabla_{\mathbf{x}} \cdot (\varepsilon_s\nabla_{\mathbf{x}} V\ap{eff})=\rho\ap{ext} -e n.
\end{equation}
Da questo punto in poi si ometterà il suffisso $\ap{eff}$ e si parlerà semplicemente di campo elettrico e potenziale elettrico. Accoppiando la \eqref{EQ:CAP1:Eq_Poisson} con la \eqref{EQ:CAP1:Eq_Vlasov}, si ottiene il seguente sistema
\begin{equation}\label{EQ:CAP1:Sist_Vlasov_Poisson}
\left\lbrace
\begin{aligned}
&\frac{\partial f}{\partial t} + \mathbf{v}(\mathbf{k}) \cdot \nabla_{\mathbf{x}} f - \frac{e}{\hbar} \mathbf{E}(\mathbf{x},t) \cdot \nabla_{\mathbf{k}} f=0\\
&-\nabla_{\mathbf{x}} \cdot (\varepsilon_s\nabla_{\mathbf{x}} V)=\rho\ap{ext} -e n
\end{aligned}
\right.
\end{equation}
detto \textbf{sistema di Vlasov-Poisson}, in cui vale la relazione $\mathbf{E}(\mathbf{x},t)=-\nabla_{\mathbf{x}} V(\mathbf{x},t)$.

\`{E} possibile ripetere quanto fin qui esposto per gli elettroni anche per le lacune, sostituendo $-e$ con $e$. In tal caso, indicando con $f_v$ la funzione di distribuzione delle lacune, il sistema di Vlasov-Poisson assume la forma
\begin{equation}
\left\lbrace
\begin{aligned}
&\frac{\partial f_h}{\partial t} + \mathbf{v}_h(\mathbf{k}) \cdot \nabla_{\mathbf{x}} f_h + \frac{e}{\hbar} \mathbf{E}(\mathbf{x},t) \cdot \nabla_{\mathbf{k}} f_h=0\\
&-\nabla_{\mathbf{x}} \cdot (\varepsilon_s\nabla_{\mathbf{x}} V)=\rho\ap{ext} +e p
\end{aligned}
\right.
\end{equation}
dove
\begin{equation}
p(\mathbf{x},t)=\int_{\mathcal{B}} \! f_v(\mathbf{x},\mathbf{k},t) \, \diff \mathbf{k}.
\end{equation}
Infine, nel caso di sistemi bipolari, cioè in cui intervengono sia elettroni che lacune, è possibile scrivere la densità $\rho\ap{ext}$ come
\begin{equation}
\rho\ap{ext}(\mathbf{x})=e(N_D(\mathbf{x})-N_A(\mathbf{x})),
\end{equation}
dove $N_D$ e $N_A$ sono rispettivamente le densità degli ioni \textbf{donori}, cioè ioni del reticolo cristallino che possono donare elettroni, e degli ioni \textbf{accettori}, cioè ioni del reticolo cristallino che possono accettare elettroni. Allora l'equazione di Poisson per il potenziale elettrico diventa
\begin{equation}
\nabla_{\mathbf{x}} \cdot (\varepsilon_s(\mathbf{x})\nabla_{\mathbf{x}} V(\mathbf{x},t))=-e(N_D(\mathbf{x})-N_A(\mathbf{x})-n(\mathbf{x},t)+p(\mathbf{x},t)),
\end{equation}

\subsection{Le interazioni elettroniche}
Il moto degli elettroni in presenza di un potenziale periodico, come quello definito da un cristallo, è descritto dalle funzioni d'onda di Bloch, come visto nel Paragrafo \ref{PAR:CAP2:Moto_elettrone}. Le onde di Bloch si muovono liberamente all'interno del cristallo. Occasionalmente l'elettrone incontra una perturbazione causata da una vibrazione reticolare oppure dalla presenza di impurità o difetti nel materiale. Quando ciò accade si dice che avviene uno \textbf{scattering}. Esso produce una variazione dell'energia e dell'impulso dell'elettrone (cfr.~\citep{BOOK:Lundstrom}). Inoltre, una transizione può essere:
\begin{itemize}
\item \textbf{intra-band}, se la variazione di energia lascia l'elettrone nella stessa banda;
\item \textbf{inter-band}, se la variazione di energia porta l'elettrone in un'altra banda.
\end{itemize}
Gli scattering possono essere classificati in base alle valli in cui si trovano gli stati iniziale e finale. Una transizione può essere:
\begin{itemize}
\item \textbf{intra-valley}, se gli stati iniziale e finale stanno nella stessa valle;
\item \textbf{inter-valley} se gli stati iniziale e finale stanno in valli diverse.
\end{itemize}
Merita una descrizione più approfondita il caso dei fononi. Lo scattering elettrone-fonone si può ricondurre a due fenomeni elementari: l'assorbimento di un fonone, e l'emissione di un fonone. 

Quando avviene un assorbimento, un fonone di vettore d'onda $\xi$ ed energia $\hbar\omega$ viene assorbito da un elettrone di pseudo-vettore d'onda $\mathbf{k}$ ed energia $\mathcal{E}_c(\mathbf{k})$, dove con $\mathcal{E}_c$ si indica l'espressione dell'energia nella banda di conduzione. Dopo l'interazione si ha
\begin{equation}
\mathbf{k}'=\mathbf{k}+\xi+\mathbf{g}, \qquad \mathcal{E}_c(\mathbf{k}')=\mathcal{E}_c(\mathbf{k})+\hbar\omega,
\end{equation}
dove $\mathbf{g}$ è un vettore del reticolo inverso tale che $\mathbf{k}'\in\mathcal{B}$. In questo processo si ha dunque la conservazione dell'impulso, a meno di un vettore del reticolo inverso, e dell'energia totale.

Quando avviene un'emissione un elettrone di pseudo-vettore d'onda $\mathbf{k}$ ed energia $\mathcal{E}_c(\mathbf{k})$ modifica il suo pseudo-vettore d'onda in $\mathbf{k}'$ e la sua energia in $\mathcal{E}_c(\mathbf{k}')$ dando origine anche ad un fonone di vettore d'onda $\xi$ ed energia $\hbar\omega$. Si hanno quindi le relazioni
\begin{equation}
\mathbf{k}'=\mathbf{k}-\xi+\mathbf{g}, \qquad \mathcal{E}_c(\mathbf{k}')=\mathcal{E}_c(\mathbf{k})-\hbar\omega.
\end{equation}

\subsection{L'equazione semiclassica di Boltzmann}\label{PAR:CAP1:Eq_Boltz_semi}
La descrizione presentata fino a questo punto è valida se il cristallo è perfettamente periodico. Nei semiconduttori reali la periodicità è distrutta da varie cause: dislocazioni nel cristallo (cioè linee o piani delle unità del cristallo la cui posizione è spostata rispetto a quella ideale), atomi mancanti, atomi che occupano regioni che nel caso ideale sarebbero vuote, stress meccanici, drogaggio con impurezze e vibrazioni termiche del reticolo. I difetti che si presentano maggiormente sono gli ultimi due. Nella trattazione che seguirà sarà considerato soltanto l'ultimo dei difetti prima menzionati e saranno trascurati gli altri. Le vibrazioni termiche del reticolo sono rappresentate come delle particelle fittizie che perturbano il potenziale del reticolo, quindi le interazioni si possono vedere come scattering tra fononi. Formalmente questi effetti si introducono nel modello inserendo un membro di destra non nullo nell'equazione di Vlasov \eqref{EQ:CAP1:Sist_Vlasov_Poisson}$_1$ che descrive gli effetti di scattering, cioè si ha
\begin{equation}\label{EQ:CAP1:Eq_Boltzmann}
\frac{\partial f}{\partial t} + \mathbf{v}(\mathbf{k}) \cdot \nabla_{\mathbf{x}} f -\frac{e}{\hbar} \mathbf{E} \cdot \nabla_{\mathbf{k}} f=\mathcal{C}[f],
\end{equation}
dove $\mathcal{C}[f]$ è detto \textbf{termine collisionale per gli elettroni}. La \eqref{EQ:CAP1:Eq_Boltzmann} prende il nome di \textbf{equazione semiclassica di Boltzmann per gli elettroni nella banda di conduzione dei semiconduttori}. Analogamente, per le lacune si ottiene
\begin{equation}
\frac{\partial f_h}{\partial t} + \mathbf{v}_h(\mathbf{k}) \cdot \nabla_{\mathbf{x}} f_h +\frac{e}{\hbar} \mathbf{E} \cdot \nabla_{\mathbf{k}} f_h=\mathcal{C}_h[f_h],
\end{equation}
dove $\mathcal{C}[f]$ è detto \textbf{termine collisionale per le lacune}.

Il termine collisionale $\mathcal{C}[f]$ nella \eqref{EQ:CAP1:Eq_Boltzmann} deve essere interpretato come un termine che modifica la distribuzione $f$ per effetto degli scattering. Per mezzo di essi il valore di $f(\mathbf{x},\mathbf{k},t)$ può variare in due modi: una particella che si trova nello stato $(\mathbf{x},\mathbf{k}',t)$ assumerà dopo lo scattering lo stato $(\mathbf{x},\mathbf{k},t)$; una particella che si trova nello stato $(\mathbf{x},\mathbf{k},t)$ occuperà dopo lo scattering lo stato $(\mathbf{x},\mathbf{k}',t)$. Nel primo caso si avrà un incremento del numero di particelle per unità di volume dello spazio delle fasi, nel secondo caso una diminuzione. Si osservi che gli effetti degli scattering intervengono modificando solamente i vettori d'onda e non la posizione. Allora la variazione netta del valore di $f(\mathbf{x},\mathbf{k},t)$ si ottiene integrando su tutti i possibili valori di $\mathbf{k}'$. Inoltre, è possibile separare $\mathcal{C}[f]$ in due termini, cioè si ha
\begin{equation}
\mathcal{C}[f] = \int_\mathcal{B} \! \Psi\ap{in} - \Psi\ap{out} \, \diff \mathbf{k}',
\end{equation}
dove $\Psi\ap{in}$ rappresenta il termine di guadagno, o \textbf{gain}, e $\Psi\ap{out}$ rappresenta il termine di perdita, o \textbf{loss}.

Si analizzeranno adesso i due termini separatamente. Si consideri il termine $\Psi\ap{in}$. Si vuole determinare la distribuzione di probabilità che una particella raggiunga lo stato $(\mathbf{x},\mathbf{k},t)$ supposto che essa si trovi nello stato $(\mathbf{x},\mathbf{k}',t)$. Affinché ciò avvenga è necessario e sufficiente che si verifichino contemporaneamente i seguenti tre eventi: la particella deve occupare lo stato di vettore d'onda $\mathbf{k}'$; deve avvenire uno scattering dallo stato di vettore d'onda $\mathbf{k}'$ a quello di vettore d'onda $\mathbf{k}$; lo stato di vettore d'onda $\mathbf{k}$ deve essere libero, in accordo con il principio di esclusione di Pauli. Questi eventi si verificano seguendo delle distribuzioni di probabilità, che sono rispettivamente $f(\mathbf{x},\mathbf{k}',t)$, $P(\mathbf{k}',\mathbf{k})$ e $1-f(\mathbf{x},\mathbf{k},t)$, dove la distribuzione $P$ sarà determinata nel Paragrafo \ref{PAR:CAP1:Regola_Fermi}. Allora, il termine di gain assume la forma
\begin{equation}
\Psi\ap{in} = P(\mathbf{k}',\mathbf{k})f(\mathbf{x},\mathbf{k}',t)(1-f(\mathbf{x},\mathbf{k},t)).
\end{equation}
Analogamente, considerando il termine $\Psi\ap{in}$, occorre seguire il procedimento precedente con la differenza che lo stato di partenza è questa volta $(\mathbf{x},\mathbf{k},t)$ e lo stato di arrivo è $(\mathbf{x},\mathbf{k}',t)$. Di conseguenza, il termine di loss si può scrivere come
\begin{equation}
\Psi\ap{out} = P(\mathbf{k},\mathbf{k}')f(\mathbf{x},\mathbf{k},t)(1-f(\mathbf{x},\mathbf{k}',t)).
\end{equation}
In definitiva, il termine collisionale assume la forma
\begin{equation}
\mathcal{C}[f] = \int_\mathcal{B} \! P(\mathbf{k}',\mathbf{k})f(\mathbf{k}')(1-f(\mathbf{k})) - P(\mathbf{k},\mathbf{k}')f(\mathbf{k})(1-f(\mathbf{k}')) \, \diff \mathbf{k}',
\end{equation}
in cui per semplicità di scrittura è stato posto $f(\mathbf{k})=f(\mathbf{x},\mathbf{k},t)$ e $f(\mathbf{k}')=f(\mathbf{x},\mathbf{k}',t)$.

\subsection{Il metodo perturbativo}
Solamente pochi problemi quantistici hanno una soluzione esatta. In molti casi è necessario ricorrere ad una soluzione approssimata. Uno dei metodi più usati è il \textbf{metodo perturbativo}. Si consideri l'hamiltoniana $\mathcal{H}$ del sistema, a cui è associato un operatore hamiltoniano $H$. Si supponga che esso possa essere scritto nella forma
\begin{equation}\label{EQ:CAP1:Metodo_perturbativo}
H = H_0 + \lambda V',
\end{equation}
dove $H_0$ è un operatore autoaggiunto, indipendente dal tempo e di cui si conosce lo spettro discreto, $V'$ è un operatore, detto \textbf{perturbazione}, e $\lambda$ è un parametro chiamato \textbf{costante di accoppiamento}. Si possono distinguere due casi: quello in cui l'operatore $V'$ sia indipendente dal tempo e il caso in cui ne sia dipendente.

Si consideri il primo caso. Indicando con $\varepsilon_{0,m}$ gli autovalori di $V'$ e con $\left\vert\phi_{0,m}\right\rangle$ i suoi autovettori corrispondenti, si ha che essi soddisfano l'equazione agli autovalori
\begin{equation}\label{EQ:CAP1:EQ_autov_metod_pertur}
H_0 \left\vert \phi_{0,m} \right\rangle = \varepsilon_{0,m} \left\vert \phi_{0,m} \right\rangle,
\end{equation}
in cui gli autovettori soddisfano la seguente relazione di ortonormalità
\begin{equation}\label{EQ:CAP1:Rel_orton_autovett}
\left\langle \phi_{0,n}\vert\phi_{0,m} \right\rangle =\delta_{nm}.
\end{equation}
L'obiettivo da raggiungere sarà determinare gli autovalori $\varepsilon_m$ e gli autovettori $\left\vert\phi_m\right\rangle$ dell'operatore hamiltoniano totale $H$, i quali devono soddisfare l'equazione agli autovalori
\begin{equation}\label{EQ:CAP1:Eq_autoval_ham_tot}
H \left\vert\phi_m\right\rangle = \varepsilon_m \left\vert\phi_m\right\rangle.
\end{equation}
Si procede supponendo di poter espandere le quantità incognite in potenze di $\lambda$, cioè si ha
\begin{equation}\label{EQ:CAP1:Espansioni_lambda}
\left\vert\phi_m\right\rangle = \left\vert\phi_m^{(0)}\right\rangle +\lambda \left\vert\phi_m^{(1)}\right\rangle + \lambda^2 \left\vert\phi_m^{(2)}\right\rangle +\ldots,
\qquad
\varepsilon_m = \varepsilon_m^{(0)} + \lambda \varepsilon_m^{(1)} + \lambda^2 \varepsilon_m^{(2)}+\ldots.
\end{equation}
Successivamente, si sostituiscono le espressioni \eqref{EQ:CAP1:Espansioni_lambda} nella \eqref{EQ:CAP1:Eq_autoval_ham_tot}, ottenendo un'eguaglianza tra due polinomi in $\lambda$, i cui coefficienti devono uguagliarsi per ogni potenza di $\lambda$. Allora, tenendo conto di \eqref{EQ:CAP1:Metodo_perturbativo}, si hanno le seguenti relazioni
\begin{equation}\label{EQ:CAP1:Eq_metodo_perturb}
\begin{aligned}
&\left[ H_0-\varepsilon_m^{(0)} \right] \left\vert \phi_m^{(0)} \right\rangle =0\\
&\left[ H_0-\varepsilon_m^{(0)} \right] \left\vert \phi_m^{(1)} \right\rangle = \left[ \varepsilon_m^{(1)} -V' \right] \left\vert \phi_m^{(0)} \right\rangle \\
&\left[ H_0-\varepsilon_m^{(0)} \right] \left\vert \phi_m^{(2)} \right\rangle = \left[ \varepsilon_m^{(1)} -V' \right] \left\vert \phi_m^{(1)} \right\rangle + \varepsilon_m^{(2)} \left\vert \phi_m^{(0)} \right\rangle  \\
&\ldots
\end{aligned}
\end{equation}
\`{E} possibile rappresentare i vettori delle equazioni precedenti nella base di autovettori relativa all'operatore $H_0$, cioè si ha
\begin{equation}\label{EQ:CAP1:Rapp_vett_metodo_perturb}
\left\vert \phi_m^{(i)} \right\rangle = \sum_n C_n^{(i)}(m) \left\vert \phi_{0,m} \right\rangle.
\end{equation}
Sostituendo tali relazioni nelle equazioni \eqref{EQ:CAP1:Eq_metodo_perturb}, si ottiene
\begin{equation}\label{EQ:CAP1:Eq_metodo_perturb_2}
\begin{aligned}
&\left[ H_0-\varepsilon_m^{(0)} \right] \sum_n C_n^{(0)}(m) \left\vert \phi_{0,n} \right\rangle =0\\
&\left[ H_0-\varepsilon_m^{(0)} \right] \sum_n C_n^{(1)}(m) \left\vert \phi_{0,n} \right\rangle = \left[ \varepsilon_m^{(1)} -V' \right] \sum_n C_n^{(0)}(m) \left\vert \phi_{0,n} \right\rangle \\
&\left[ H_0-\varepsilon_m^{(0)} \right] \sum_n C_n^{(2)}(m) \left\vert \phi_{0,n} \right\rangle = \left[ \varepsilon_m^{(1)} -V' \right] \sum_n C_n^{(1)}(m) \left\vert \phi_{0,n} \right\rangle + \varepsilon_m^{(2)} \sum_n C_n^{(0)}(m) \left\vert \phi_{0,n} \right\rangle  \\
&\ldots
\end{aligned}
\end{equation}
Si osservi che all'ordine zero gli autovalori dell'operatore $H$ coincidono con quelli dell'operatore $H_0$. Allora la prima delle \eqref{EQ:CAP1:Eq_metodo_perturb_2}, diventa
\begin{equation}
\sum_n \left[ \varepsilon_{0,n}-\varepsilon_{0,m} \right]  C_n^{(0)}(m) \left\vert \phi_{0,n} \right\rangle =0.
\end{equation}
Se lo spettro è non degenere, ovvero se $\varepsilon_{0,n}\neq\varepsilon_{0,m}$ per $n\neq m$, allora tutti i coefficienti $C_n^{(0)}(m)$ devono annullarsi per $n\neq m$, invece i coefficienti $C_m^{(0)}(m)$ restano indeterminati e si considerano uguali a 1 per normalizzazione. Allora, in questo caso, anche gli autovettori dell'operatore $H$ coincidono con quelli dell'operatore $H_0$. In definitiva si ha
\begin{equation}\label{EQ:CAP1:Rel_metod_pert_ord_zero}
\varepsilon_m^{(0)} = \varepsilon_{0,m}, \qquad C_n^{(0)}(m)=\delta_{nm}, \qquad \left\vert \phi_m^{(0)} \right\rangle = \left\vert \phi_{0,m} \right\rangle.
\end{equation}
Si consideri ancora il caso in cui lo spettro dell'operatore $H_0$ sia non degenere. Con lo scopo di calcolare la correzione sugli autovalori e autovettori relativa al primo ordine in $\lambda$, si consideri la seconda equazione delle \eqref{EQ:CAP1:Eq_metodo_perturb_2} e la si moltiplichi a sinistra per $\left\langle \phi_{0,m} \right\vert$, allora si ha
\begin{equation}
\left\langle \phi_{0,m} \right\vert \left[ H_0-\varepsilon_m^{(0)} \right] \sum_n C_n^{(1)}(m) \left\vert \phi_{0,n} \right\rangle = \left\langle \phi_{0,m} \right\vert \left[ \varepsilon_m^{(1)} -V' \right] \sum_n C_n^{(0)}(m) \left\vert \phi_{0,m} \right\rangle,
\end{equation}
inoltre, utilizzando la \eqref{EQ:CAP1:Rel_orton_autovett}, si ottiene
\begin{equation}
\left\langle \phi_{0,m} \right\vert \left[ H_0-\varepsilon_m^{(0)} \right] C_m^{(1)}(m) \left\vert \phi_{0,m} \right\rangle = \left\langle \phi_{0,m} \right\vert \left[ \varepsilon_m^{(1)} -V' \right] C_m^{(0)}(m) \left\vert \phi_{0,m} \right\rangle.
\end{equation}
Si osservi che, in virtù della \eqref{EQ:CAP1:EQ_autov_metod_pertur}, il primo membro di quest'ultima equazione è nullo. Allora, dividendo per $C_m^{(0)}(m)$, si ha
\begin{equation}\label{EQ:CAP1:Espress_autov_metod_pert}
\varepsilon_m^{(1)} = \left\langle \phi_{0,m} \right\vert V' \left\vert \phi_{0,m} \right\rangle.
\end{equation}
Questa relazione permette di calcolare la correzione al primo ordine perturbativo degli autovalori $\varepsilon_{0,m}$. Mentre per determinare le correzioni al primo ordine perturbativo degli autovettori, si consideri ancora la seconda equazione delle \eqref{EQ:CAP1:Eq_metodo_perturb_2} e la si moltiplichi a sinistra questa volta per $\left\langle \phi_{0,n} \right\vert$, con $m\neq n$, ottenendo
\begin{equation}
\left\langle \phi_{0,n} \right\vert \left[ H_0-\varepsilon_m^{(0)} \right] \sum_n C_n^{(1)}(m) \left\vert \phi_{0,n} \right\rangle = \left\langle \phi_{0,m} \right\vert \left[ \varepsilon_m^{(1)} -V' \right] \sum_n C_n^{(0)}(m) \left\vert \phi_{0,n} \right\rangle,
\end{equation}
da cui in virtù della \eqref{EQ:CAP1:Rel_orton_autovett}, si ha
\begin{equation}
C_n^{(1)}(m) \left[  \left\langle \phi_{0,n} \right\vert H_0 \left\vert \phi_{0,n} \right\rangle -\varepsilon_m^{(0)} \right]  = C_n^{(0)}(m) \left[ \varepsilon_m^{(1)}- \left\langle \phi_{0,m} \right\vert V' \left\vert \phi_{0,n} \right\rangle \right].
\end{equation}
Allora, utilizzando la \eqref{EQ:CAP1:EQ_autov_metod_pertur} e la \eqref{EQ:CAP1:Espress_autov_metod_pert}, si ottiene
\begin{equation}
C_n^{(1)}(m) \left[  \varepsilon_{0,n} -\varepsilon_m^{(0)} \right]  = C_n^{(0)}(m) \left[ \varepsilon_m^{(1)}- \varepsilon_n^{(1)} \right],
\end{equation}
e, infine, ricordando le \eqref{EQ:CAP1:Rel_metod_pert_ord_zero}, si ha
\begin{equation}
C_n^{(1)}(m) \left[  \varepsilon_{0,n} -\varepsilon_{0,m} \right]  = \left[ \varepsilon_m^{(1)}- \varepsilon_n^{(1)} \right],
\end{equation}
da cui si ricavano le espressioni per i coefficienti, cioè
\begin{equation}
C_n^{(1)}(m) = \frac{\varepsilon_m^{(1)}- \varepsilon_n^{(1)}}{\varepsilon_{0,n} -\varepsilon_{0,m}}.
\end{equation}
Di conseguenza, ricordando la \eqref{EQ:CAP1:Rapp_vett_metodo_perturb}, si ottengono le espressioni per le correzioni al primo ordine perturbativo degli autovettori, ovvero
\begin{equation}
\left\vert \phi_m^{(1)} \right\rangle = \sum_n \frac{\varepsilon_m^{(1)}- \varepsilon_n^{(1)}}{\varepsilon_{0,n} -\varepsilon_{0,m}} \left\vert \phi_{0,m} \right\rangle.
\end{equation}
Procedendo in modo analogo si ottengono per ricorsione le espressioni per le correzioni degli autovalori e degli autovettori fino all'ordine perturbativo voluto. Si osservi che questa procedura è valida nel caso in cui gli autovalori sono tutti distinti. Per quanto concerne il caso degenere, si consulti \citep{BOOK:Schiff}.

L'altro casi si ha quando la perturbazione è dipendente dal tempo. Supponendo di conoscere, come il primo caso, gli autovalori $\varepsilon_{0,n}$ e gli autovettori $\psi_{0,n}$ dell'operatore hamiltoniano imperturbato $H_0$, bisogna cercare la soluzione dell'equazione di Schr\"{o}dinger dipendente dal tempo
\begin{equation}
i\hbar \frac{\partial}{\partial t} \left\vert \Psi(t) \right\rangle = H \left\vert \Psi(t) \right\rangle.
\end{equation}
Ricordando che uno stato quantico è rappresentato a meno di un fattore di fase, è opportuno rappresentare gli autostati dell'operatore hamiltoniano imperturbato $H_0$ specificando il fattore di fase, cioè
\begin{equation}
\left\vert \Phi_{0,n}(t) \right\rangle = \left\vert \phi_{0,n} \right\rangle e^{-i\omega_{0,n}(t-t_0)}, \qquad \omega_{0,n}=\frac{\varepsilon_{0,n}}{\hbar}.
\end{equation}
Allora è possibile rappresentare un generico stato $\left\vert \Psi(t) \right\rangle$ in questa base scrivendo
\begin{equation}\label{EQ:CAP1:Generico_stato_perturb_dip_t}
\left\vert \Psi(t) \right\rangle = \sum_n a_n(t)\left\vert \phi_{0,n} \right\rangle e^{-i\omega_{0,n}(t-t_0)}.
\end{equation}
Allora, ricordando la \eqref{EQ:CAP1:EQ_autov_metod_pertur}, l'equazione di Schr\"{o}dinger diventa
\begin{equation}
\begin{aligned}
& i\hbar \sum_n \dot{a}_n(t)\left\vert \phi_{0,n} \right\rangle e^{-i\omega_{0,n}(t-t_0)} + i\hbar \sum_n a_n(t)\left\vert \phi_{0,n} \right\rangle (-i\omega_{0,n}) e^{-i\omega_{0,n}(t-t_0)} = \\
& \sum_n a_n(t) \varepsilon_{0,n} \left\vert \phi_{0,n} \right\rangle e^{-i\omega_{0,n}(t-t_0)} +  \sum_n H'(t) a_n(t) \left\vert \phi_{0,n} \right\rangle e^{-i\omega_{0,n}(t-t_0)}
\end{aligned}
\end{equation}
Osservando che, poiché $\omega_{0,n}=\varepsilon_{0,n} / \hbar$, nell'equazione precedente il secondo termine del primo membro è uguale al primo termine del secondo membro, si ottiene
\begin{equation}
i\hbar \sum_n \dot{a}_n(t)\left\vert \phi_{0,n} \right\rangle e^{-i\omega_{0,n}(t-t_0)} = \sum_n H'(t) a_n(t) \left\vert \phi_{0,n} \right\rangle e^{-i\omega_{0,n}(t-t_0)}.
\end{equation}
Moltiplicando quest'ultima equazione a sinistra per $\left\langle \phi_{0,m} \right\vert e^{i\omega_{0,m}(t-t_0)}$ e ricordando che $\left\langle \phi_{0,m} \vert \phi_{0,n} \right\rangle = \delta_{mn} $, si ha
\begin{equation}
i\hbar \, \dot{a}_m(t) = \sum_n a_n(t) \left\langle \phi_{0,m} \right\vert H'(t) \left\vert \phi_{0,n} \right\rangle e^{i(\omega_{0,m}-\omega_{0,n})(t-t_0)}.
\end{equation}
Allora, ponendo
\begin{equation}
\omega_{mn} = \omega_{0,m}-\omega_{0,n} \qquad \mbox{e} \qquad H'_{mn}(t)=\left\langle \phi_{0,m} \right\vert H'(t) \left\vert \phi_{0,n} \right\rangle,
\end{equation}
si ottiene
\begin{equation}\label{EQ:CAP1:Eq_coeffi_met_pert_dip_t}
\dot{a}_m(t) = \frac{1}{i\hbar} \sum_n H'_{mn}(t) a_n(t) e^{i\omega_{mn}(t-t_0)}.
\end{equation}
A questo punto si consideri l'espansione dei coefficienti che determinano la perturbazione, cioè
\begin{equation}\label{EQ:CAP1:Esp_coeff_met_per_dip_t}
a_m(t) = a_m^{(0)}(t) + \lambda a_m^{(1)}(t) + \lambda^2 a_m^{(2)}(t)+\ldots.
\end{equation}
Allora, sostituendo tale espansione nella \eqref{EQ:CAP1:Eq_coeffi_met_pert_dip_t} e osservando che ponendo
\begin{equation}
V'_{mn}(t) = \left\langle \phi_{0,m} \right\vert V'(t) \left\vert \phi_{0,n} \right\rangle
\end{equation}
si ottiene che $H'_{mn}(t)=\lambda V'_{mn}(t)$, allora la \eqref{EQ:CAP1:Eq_coeffi_met_pert_dip_t} diventa
\begin{equation}
\begin{aligned}
\dot{a}_m^{(0)}(t)+\lambda\dot{a}_m^{(1)} & +\lambda^2\dot{a}_m^{(2)}(t)+\ldots =\\ 
& \frac{1}{i\hbar}\lambda \sum_n V'_{mn}(t) a_n^{(0)}(t) e^{i\omega_{mn}(t-t_0)} + \frac{1}{i\hbar}\lambda^2 \sum_n V'_{mn}(t) a_n^{(1)}(t) e^{i\omega_{mn}(t-t_0)} + \ldots.
\end{aligned}
\end{equation}
Quindi si è determinata un'eguaglianza tra due polinomi in $\lambda$. Di conseguenza, separando i vari ordini, si ottiene
\begin{equation}\label{EQ:CAP1:Svilup_ord_perturb_dip_t}
\dot{a}_m^{(0)}(t)=0, \qquad \dot{a}_m^{(s+1)}(t)=\frac{1}{i\hbar} \sum_n V'_{mn}(t) a_n^{(s)}(t) e^{i\omega_{mn}(t-t_0)}.
\end{equation}
Dalla prima di queste equazioni, si ottiene
\begin{equation}\label{EQ:CAP1:Eq_ordine_pert_zero_dip_t}
a_m^{(0)}(t)= \mbox{cost} = a_m^{(0)}(t_0),
\end{equation}
e le soluzioni per gli altri ordini si ottengono per ricorrenza. Arrestando lo sviluppo \eqref{EQ:CAP1:Svilup_ord_perturb_dip_t} al primo ordine perturbativo, si ha
\begin{equation}\label{EQ:CAP1:Svilup_arrestato_ord_perturb_dip_t}
a_m^{(0)}(t) = a_m^{(0)}(t_0), \qquad \dot{a}_m^{(1)}(t)=\frac{1}{i\hbar} \sum_n V'_{mn}(t) a_n^{(0)}(t_0) e^{i\omega_{mn}(t-t_0)}.
\end{equation}
Inoltre, supponendo che all'istante iniziale $t_0$ il sistema si trovi nello stato $\left\vert \phi_{0,i}\right\rangle$, si vuole conoscere la probabilità che al tempo $t$ il sistema si trovi nello stato $\left\vert \phi_{0,f}\right\rangle$. Si osservi che entrambi sono autostati dell'operatore hamiltoniano imperturbato $H_0$. Inoltre, poiché il sistema al tempo $t_0$ si trova nello stato $\left\vert \phi_{0,i}\right\rangle$, allora per la \eqref{EQ:CAP1:Generico_stato_perturb_dip_t}, si deve avere
\begin{equation}\label{EQ:CAP1:Rel_coeff_pert_dip_t_1}
a_n(t_0)=\delta_{in},
\end{equation}
da cui avendo supposto che all'istante $t_0$ il sistema fosse imperturbato, si ha
\begin{equation}\label{EQ:CAP1:Rel_coeff_pert_dip_t}
a_n^{(0)}(t_0)=\delta_{in}.
\end{equation}
Allora, per lo stato finale al tempo $t$, dalla \eqref{EQ:CAP1:Rel_coeff_pert_dip_t_1} si ottiene la condizione iniziale
\begin{equation}\label{EQ:CAP1:Cond_ini_met_per_dip_t}
a_f(t_0)=0.
\end{equation}
Invece, considerando la seconda delle \eqref{EQ:CAP1:Svilup_arrestato_ord_perturb_dip_t} e sostituendo la relazione \eqref{EQ:CAP1:Rel_coeff_pert_dip_t}, per lo stato finale al tempo $t$, si ottiene
\begin{equation}
\begin{aligned}
\dot{a}_f^{(1)}(t) & = \frac{1}{i\hbar} \sum_n V'_{fn}(t) \delta_{in} e^{i\omega_{fn}(t-t_0)} \\
& = \frac{1}{i\hbar} V'_{fi}(t) e^{i\omega_{fi}(t-t_0)}.
\end{aligned}
\end{equation}
Infine, considerando lo sviluppo \eqref{EQ:CAP1:Esp_coeff_met_per_dip_t} arrestato al primo ordine perturbativo e derivando rispetto a $t$, si ricava l'equazione approssimata
\begin{equation}
\begin{aligned}
\dot{a}_f(t) & \approx \dot{a}_f^{(0)}(t)+\lambda\dot{a}_f^{(1)}(t) \\
& = 0 + \lambda\frac{1}{i\hbar} V'_{fi}(t) e^{i\omega_{fi}(t-t_0)} \\
& = \frac{1}{i\hbar} H'_{fi}(t) e^{i\omega_{fi}(t-t_0)}.
\end{aligned}
\end{equation}
Integrando quest'ultima equazione con la condizione iniziale \eqref{EQ:CAP1:Cond_ini_met_per_dip_t}, si ha la relazione approssimata
\begin{equation}\label{EQ:CAP1:Rel_a_i_f_met_pert_dip_t}
a_{i\to f}(t) \approx \frac{1}{i\hbar} \int_{t_0}^t \! H'_{fi}(t') e^{i\omega_{fi}(t'-t_0)} \, \diff t',
\end{equation}
in cui con il simbolo $i\to f$ si è voluta specificare la dipendenza da $i$. Poiché sia lo stato iniziale che quello finale sono autostati, allora la probabilità di transizione $P_{i \to f}(t)$ dall'autostato $\left\vert \phi_{0,i}\right\rangle$ all'autostato $\left\vert \phi_{0,f}\right\rangle$ al tempo $t$ si calcola scrivendo
\begin{equation}
\begin{aligned}
P_{i \to f}(t) = \left\vert \left\langle \phi_{0,f} \vert \phi_{0,i} \right\rangle \right\vert^2 = \left\vert a_{i\to f}(t) \right\vert^2 & \approx \left\vert \frac{1}{i\hbar} \int_{t_0}^t \! H'_{fi}(t') e^{i\omega_{fi}(t'-t_0)} \, \diff t' \right\vert^2 \\
& = \left\vert \frac{1}{i\hbar} e^{-i\omega_{fi}t_0} \int_{t_0}^t \! H'_{fi}(t') e^{i\omega_{fi}t'} \, \diff t' \right\vert^2 \\
& = \left\vert \frac{1}{i\hbar} \int_{t_0}^t \! H'_{fi}(t') e^{i\omega_{fi}t'} \, \diff t' \right\vert^2,
\end{aligned}
\end{equation}
dove l'ultimo passaggio si ottiene in quanto il fattore costante $e^{-i\omega_{fi}t_0}$ è un numero immaginario puro e quindi il suo modulo quadro vale 1. Inoltre, valgono le relazioni
\begin{equation}\label{EQ:CAP1:Omega_fi}
H'_{fi}(t)=\left\langle \phi_{0,f} \right\vert H'(t) \left\vert \phi_{0,i} \right\rangle  \qquad \mbox{e} \qquad \omega_{fi} = \frac{\varepsilon_{0,f}-\varepsilon_{0,i}}{\hbar}.
\end{equation}
Quindi in definitiva si ha
\begin{equation}\label{EQ:CAP1:Prob_trans}
P_{i \to f}(t) \approx \left\vert \frac{1}{i\hbar} \int_{t_0}^t \! \left\langle \phi_{0,f} \right\vert H'(t') \left\vert \phi_{0,i} \right\rangle e^{i\omega_{fi}t'} \, \diff t' \right\vert^2,
\end{equation}
con $\omega_{fi} = \frac{\varepsilon_{0,f}-\varepsilon_{0,i}}{\hbar}$.

\subsection{La regola d'oro di Fermi}\label{PAR:CAP1:Regola_Fermi}
Per risolvere completamente il problema di trovare la distribuzione di probabilità per un cambiamento di stato di un sistema quantistico, utilizzando il metodo perturbativo, è necessario assegnare nella \eqref{EQ:CAP1:Prob_trans} l'operatore $H'(t)$. Poiché qualsiasi perturbazione può essere espressa come sovrapposizione di perturbazioni armoniche, queste ultime rivestono un ruolo molto importante. Nel caso della perturbazione armonica si sceglie $H'(t)$ della forma
\begin{equation}
H'(t) = F e^{-i \omega t} + F^\dagger e^{i \omega t},
\end{equation}
dove $F$ è l'operatore integrale di Fourier definito dalla \eqref{EQ:CAP1:Oper_Fourier} e $F^\dagger$ è l'operatore hermitiano aggiunto di $F$. Si osservi che $F$ e $F^\dagger$ sono operatori indipendenti dal tempo. Gli elementi di matrice si calcolano scrivendo
\begin{equation}
H'_{fi}(t) = \left\langle \phi_{0,f} \right\vert H'(t) \left\vert \phi_{0,i} \right\rangle = \left\langle \phi_{0,f} \right\vert F e^{-i \omega t} + F^\dagger e^{i \omega t} \left\vert \phi_{0,i} \right\rangle = F_{fi} e^{-i \omega t} + F_{if}^\dagger e^{i \omega t}.
\end{equation}
Sostituendo quest'ultima formula nella \eqref{EQ:CAP1:Rel_a_i_f_met_pert_dip_t}, si ottiene
\begin{equation}
\begin{aligned}
a_{i\to f}(t) & \approx \frac{1}{i\hbar} e^{-i\omega_{fi}t_0} \int_{t_0}^t \! \left(  F_{fi} e^{-i \omega t'} + F_{if}^\dagger e^{i \omega t'} \right)  e^{i\omega_{fi}t'} \, \diff t'\\
& = \frac{1}{i\hbar} e^{-i\omega_{fi}t_0} \left[ F_{fi} \int_{t_0}^t \! e^{i (\omega_{fi}-\omega) t'} \, \diff t' + F_{if}^\dagger \int_{t_0}^t \! e^{i (\omega_{fi}+\omega) t'} \, \diff t' \right]\\
& =  \frac{1}{i\hbar} e^{-i\omega_{fi}t_0} \left[ F_{fi} \frac{e^{i(\omega_{fi}-\omega)t}-e^{i(\omega_{fi}-\omega)t_0}}{i(\omega_{fi}-\omega)} + F_{if}^\dagger \frac{e^{i(\omega_{fi}+\omega)t}-e^{i(\omega_{fi}+\omega)t_0}}{i(\omega_{fi}+\omega)} \right].
\end{aligned}
\end{equation}
Si osservi che uno dei due termini diventa dominante quando il denominatore si annulla, cioè quando
\begin{equation}
\omega_{fi}\mp\omega\approx 0.
\end{equation}
Sostituendo nell'equazione precedente la seconda delle \eqref{EQ:CAP1:Omega_fi}, si ha
\begin{equation}
\varepsilon_{0,f}\approx\varepsilon_{0,i}\pm\hbar\omega.
\end{equation}
Ciò accade se nel cambiamento di stato avviene un assorbimento oppure un'emissione di un quanto di energia. Di conseguenza, è dominante il primo termine se avviene un assorbimento, lo è il secondo se avviene un'emissione.

Si può calcolare la probabilità di trovare il sistema al tempo $t$ nello stato finale $\ket{\phi_{0,f}}$ dopo un assorbimento, ovvero
\begin{equation}\label{EQ:CAP1:Fermi_ass}
\begin{aligned}
\left\vert a_f\ap{(ass)}(t) \right\vert^2 & = \frac{1}{\hbar^2} \left\vert F_{fi} \right\vert^2 \frac{\left\vert e^{i(\omega_{fi}-\omega)t} - e^{i(\omega_{fi}-\omega)t_0} \right\vert^2}{(\omega_{fi}-\omega)^2}\\
& = \frac{2}{\hbar^2} \left\vert F_{fi} \right\vert^2 \frac{1-\cos\left(\left( \omega_{fi}-\omega \right)(t-t_0)\right)}{(\omega_{fi}-\omega)^2} \\
& = \frac{1}{\hbar^2} \left\vert F_{fi} \right\vert^2 \frac{\sin^2\left[ \frac{\omega_{fi}-\omega}{2} \Delta t \right]}{\left( \frac{\omega_{fi}-\omega}{2} \right)^2},
\end{aligned}
\end{equation}
dove $\Delta t = t-t_0$.

Si consideri la funzione
\begin{equation}
f_t(\alpha) =
\begin{cases}
\frac{\sin^2(\alpha t)}{\pi \alpha^2 t} &\mbox{se } \alpha\neq 0\\
\frac{t}{\pi} &\mbox{se } \alpha = 0
\end{cases}
\end{equation}
nella variabile $\alpha$, dove $t$ è un parametro reale. La funzione risulta continua in $\mathbb{R}$ e inoltre
\begin{equation}\label{EQ:CAP1:Fermi_app_1}
\lim_{t\to +\infty} f_t(\alpha) =
\begin{cases}
0 & \mbox{se } \alpha \neq 0\\
+\infty & \mbox{se } \alpha = 0
\end{cases}
\end{equation}
Si vuole calcolare l'integrale su $\mathbb{R}$ di $f_t(\alpha)$, cioè
\begin{equation}
\int_{-\infty}^{+\infty} \! f_t(\alpha) \, \diff \alpha = \int_{-\infty}^{+\infty} \! \frac{\sin^2(\alpha t)}{\pi \alpha^2 t} \, \diff \alpha = \frac{1}{\pi} \int_{-\infty}^{+\infty} \! \frac{\sin^2 \xi}{\xi^2} \, \diff \xi.
\end{equation}
Per calcolare l'ultimo integrale si consideri la funzione
\begin{equation}
\chi_{[-1,1]}(x) =
\begin{cases}
1 & \mbox{se } x\in [-1,1]\\
0 & \mbox{altrimenti}
\end{cases}
\end{equation}
e si calcoli la sua trasformata di Fourier $\hat{\chi}_{[-1,1]}(y)$, cioè
\begin{equation}
\begin{aligned}
\hat{\chi}_{[-1,1]}(y) & = \int_{-\infty}^{+\infty} \! \chi_{[-1,1]}(x)e^{-2\pi i x y} \, \diff x = \int_{-1}^{1} \! e^{-2\pi i x y} \, \diff x = \left[ \frac{e^{-2\pi i x y}}{-2\pi i y} \right]_{x=-1}^{x=1} \\
& = \frac{e^{-2\pi i y}-e^{2\pi i y}}{-2\pi i y} = \frac{\sin (2\pi y)}{\pi y}.
\end{aligned}
\end{equation}
Poiché si ha ovviamente $\chi_{[-1,1]}(x)\in L^2(\mathbb{R})$, applicando il Teorema di Plancherel, si ottiene
\begin{equation}
\begin{aligned}
2 & = \int_{-\infty}^{+\infty} \! \left\vert \chi_{[-1,1]}(x) \right\vert^2 \, \diff x = \int_{-\infty}^{+\infty} \! \left\vert \hat{\chi}_{[-1,1]}(y) \right\vert^2 \, \diff y = 4 \int_{-\infty}^{+\infty} \! \left\vert \frac{\sin (2\pi y)}{2\pi y} \right\vert^2 \, \diff y \\
& = 4 \int_{-\infty}^{+\infty} \! \left\vert \frac{\sin \xi}{\xi} \right\vert^2 \frac{1}{2\pi} \, \diff \xi = \frac{2}{\pi} \int_{-\infty}^{+\infty} \! \left\vert \frac{\sin \xi}{\xi} \right\vert^2 \, \diff \xi
\end{aligned}
\end{equation}
e quindi in definitiva si ha
\begin{equation}\label{EQ:CAP1:Fermi_app_2}
\int_{-\infty}^{+\infty} \! f_t(\alpha) \, \diff \alpha = 1.
\end{equation}
Di conseguenza, per valori di $\Delta t$ sufficientemente grandi, dalla \eqref{EQ:CAP1:Fermi_app_1} e dalla \eqref{EQ:CAP1:Fermi_app_2} vale l'approssimazione
\begin{equation}
f_t(\alpha) \approx \delta (\alpha),
\end{equation}
dove $\delta$ è la delta di Dirac.

Allora approssimando la \eqref{EQ:CAP1:Fermi_ass} si ha
\begin{equation}
\left\vert a_f\ap{(ass)} \right\vert^2 = \frac{\pi}{\hbar^2} \vert F_{fi} \vert^2 \Delta t \delta \left( \frac{\omega_{fi}-\omega}{2} \right) = \frac{2\pi}{\hbar} \vert F_{fi} \vert^2 \Delta t \delta (\varepsilon_f - \varepsilon_i -\hbar \omega ).
\end{equation}
Inoltre, la probabilità di transizione per unità di tempo, chiamata anche \textbf{transition rate}, è data da
\begin{equation}
P(i\rightarrow f) = \frac{2\pi}{\hbar} \vert F_{fi} \vert^2 \delta (\varepsilon_f - \varepsilon_i -\hbar \omega ),
\end{equation}
che è detta \textbf{regola d'oro di Fermi}. Nel caso in cui lo stato finale faccia parte di un continuo di stati di cui si conosce la densità in energia $g(\varepsilon)$, si può considerare la probabilità di transizione verso uno degli stati aventi una certa energia e la regola d'oro di Fermi assume la forma
\begin{equation}
P(i\rightarrow f) = \frac{2\pi}{\hbar} \vert F_{fi} \vert^2 \delta (\varepsilon_f - \varepsilon_i -\hbar \omega )g(\varepsilon).
\end{equation}
Con un procedimento analogo è possibile calcolare la probabilità di transizione per unità di tempo dopo un'emissione, ovvero
\begin{equation}
P(i\rightarrow f) = \frac{2\pi}{\hbar} \vert F^\dagger_{fi} \vert^2 \delta (\varepsilon_f - \varepsilon_i +\hbar \omega )g(\varepsilon).
\end{equation}


\subsection{Il nucleo collisionale}\label{PAR:CAP1:Nucleo_coll}
Si è visto nel Paragrafo \ref{PAR:CAP1:Eq_Boltz_semi} che il termine collisionale $\mathcal{C}[f]$ assume la forma 
\begin{equation}
\mathcal{C}[f] = \int_\mathcal{B} \! P(\mathbf{k}',\mathbf{k})f(\mathbf{k}')(1-f(\mathbf{k})) - P(\mathbf{k},\mathbf{k}')f(\mathbf{k})(1-f(\mathbf{k}')) \, \diff \mathbf{k}',
\end{equation}
in cui per semplicità di scrittura è stato posto $f(\mathbf{k})=f(\mathbf{x},\mathbf{k},t)$ e $f(\mathbf{k}')=f(\mathbf{x},\mathbf{k}',t)$. Occorre dunque scrivere in modo esplicito la probabilità di transizione $P(\mathbf{k},\mathbf{k}')$. Sapendo che nel caso dei fononi la densità in energia degli stati è data dalla distribuzione di Bose-Einstein $n_{\boldsymbol{\xi}}$, dove $\boldsymbol{\xi}$ indica il vettore d'onda del tipo di fonone considerato, e in virtù della regola d'oro di Fermi si può ottenere (cfr.~\citep{DISP:Anile},\citep{BOOK:Jacoboni},\citep{BOOK:Lundstrom}) la seguente forma generale
\begin{equation}
P(\mathbf{k},\mathbf{k}') = \mathcal{G}(\mathbf{k},\mathbf{k}')[(n_{\boldsymbol{\xi}}+1)]\delta(\mathcal{E}'-\mathcal{E}+\hbar\omega_{\boldsymbol{\xi}})+n_{\boldsymbol{\xi}}\delta(\mathcal{E}'-\mathcal{E}-\hbar\omega_{\boldsymbol{\xi}}),
\end{equation}
dove $\mathcal{G}(\mathbf{k},\mathbf{k}')$ è detto \textbf{fattore di sovrapposizione} e dipende dalla struttura a bande e dal particolare tipo di interazione. Inoltre gode delle proprietà
\begin{equation}
\mathcal{G}(\mathbf{k},\mathbf{k}')=\mathcal{G}(\mathbf{k}',\mathbf{k}) \quad \mbox{e} \quad \mathcal{G}(\mathbf{k},\mathbf{k}')\geq 0.
\end{equation}
Infine si osservi che se $\mathcal{E}'=\mathcal{E}+\hbar\omega$ si ha un assorbimento di energia, se $\mathcal{E}'=\mathcal{E}-\hbar\omega$ si ha emissione di energia e se $\mathcal{E}'=\mathcal{E}$ lo scattering è elastico.


\chapter{La simulazione Monte Carlo per il grafene} 

\label{Chapter3} 

\lhead{Capitolo 3. \emph{La simulazione Monte Carlo per il grafene}} 



\section{La struttura del grafene}
Il grafene è un materiale dalle innumerevoli proprietà fisiche (cfr.~\citep{ART:Castro-Neto}, \citep{ART:GRAPH2}, \citep{ART:Geim-Novoselov}). \`{E} noto che gli elettroni nel grafene si comportano come particelle relativistiche prive di massa (fermioni di Dirac). Questo fatto conferisce delle proprietà molto particolari a questo materiale come l'effetto Hall quantistico anomalo e l'assenza di localizzazione. Ulteriori proprietà del grafene sono l'elevata mobilità elettronica a temperatura ambiente ($\si{\num{250000}.\centi\metre^2\per\volt\second}$), l'eccezionale conduttività termica ($\si{\num{5000}.\watt\metre^{-1}\kelvin^{-1}}$) e le superiori proprietà meccaniche con un modulo di Young pari a $\si{\num{1}.\tera\pascal}$.

Al fine di applicare il modello descritto nel Capitolo~2 al caso del grafene, è di fondamentale importanza descrivere la struttura del suo reticolo cristallino e delle bande energetiche.

I successivi paragrafi sono tratti da \citep{THESIS:Lichtenberger} e \citep{ART:Castro-Neto}.

\subsection{La configurazione elettronica}
Il grafene è costituito da un singolo strato piano di atomi di carbonio, disposti in un reticolo bidimensionale a nido d'ape (cfr. \cite{ART:Geim-Novoselov}). Esso è un costituente fondamentale per i materiali di altre dimensionalità derivanti dalla grafite: può essere avvolto nei fullereni (0D), arrotolato nei nanotubi (1D) o accatastato nella grafite (3D) (vedi Figura \ref{FIG:CAP2:graphene}).
\begin{figure}[ht]
\centering
\includegraphics[width=0.8\columnwidth]{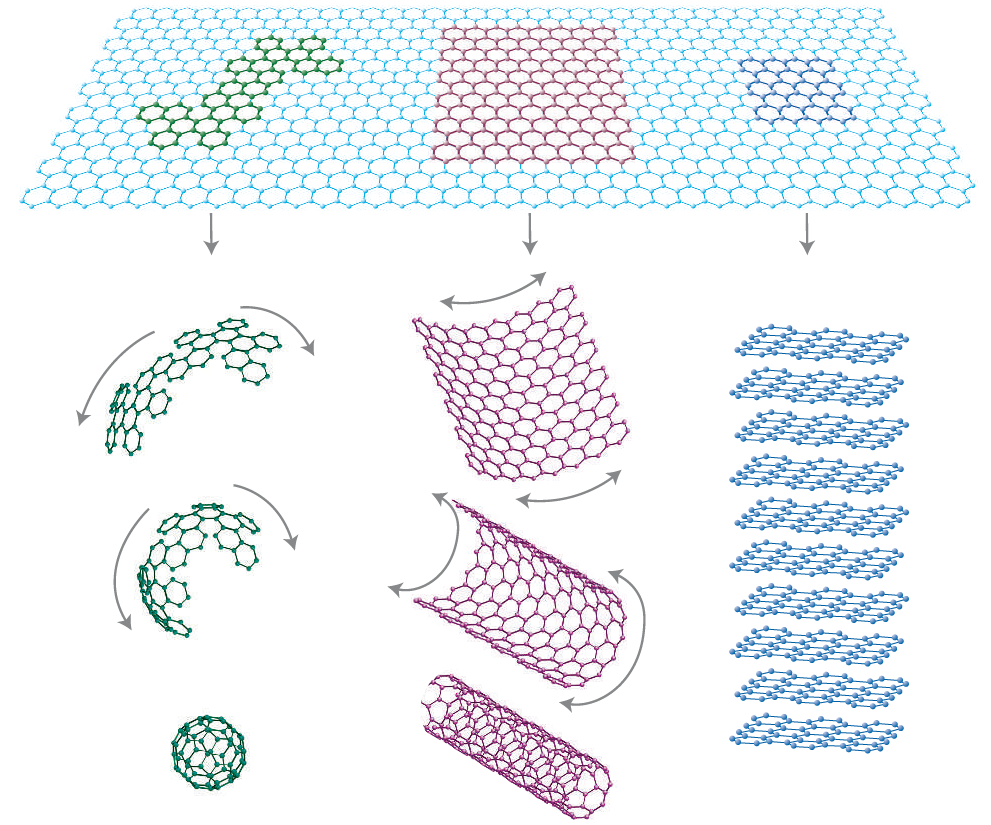}
\caption{In alto, monostrato di grafene (2D). In basso a partire da sinistra, fullerene (0D), nanotubo di carbonio (1D) e grafite (3D).}
\label{FIG:CAP2:graphene}
\end{figure}

Il grafene è un cristallo costituito da atomi di carbonio (C). Il carbonio è un elemento chimico dotato di 6 protoni e 6 elettroni. In virtù delle regole enunciate nel Paragrafo \ref{PAR:CAP1:Modello_shell}, la configurazione elettronica del carbonio presenta dunque gli orbitali $1s$ e $2s$ completamente riempiti e due degli orbitali $2p$ contenenti un solo elettrone. In forma compatta, lo stato fondamentale del carbonio è $1s^2 \, 2s^2 \, 2p^1$.

Il carbonio presenta 2 elettroni di core e 4 di valenza. Nel grafene ciascun atomo di carbonio possiede un legame covalente con altri due atomi di carbonio. Ogni atomo presenta quindi tre orbitali ibridi di tipo $sp^2$, giacenti sullo stesso piano e formanti angoli di $120^\circ$. Il quarto elettrone è invece descritto da un orbitale $2p_z$ il cui asse è perpendicolare al piano degli orbitali ibridi e dà luogo alle bande (cfr.~\citep{ART:Castro-Neto}).

\subsection{Il reticolo cristallino a nido d'ape}
I reticoli di Bravais non esauriscono la descrizione dei cristalli. Infatti esistono casi in cui i punti del reticolo non si ottengono per traslazione applicando i vettori primitivi a partire dai punti vicini. Il grafene è descritto da un reticolo bidimensionale di atomi di carbonio disposti a nido d'ape. Un reticolo a nido d'ape non è un reticolo di Bravais. Infatti la configurazione dei vertici del reticolo non si può ottenere solamente per traslazione. Dalla Figura \ref{FIG:CAP2:Reticolo_nido_d_ape}, si evince che i punti $A$ e $B$ non sono equivalenti ma per ottenere uno dall'altro occorre un ribaltamento di $180^\circ$.
\begin{figure}[ht]
\centering
\includegraphics[width=0.4\columnwidth]{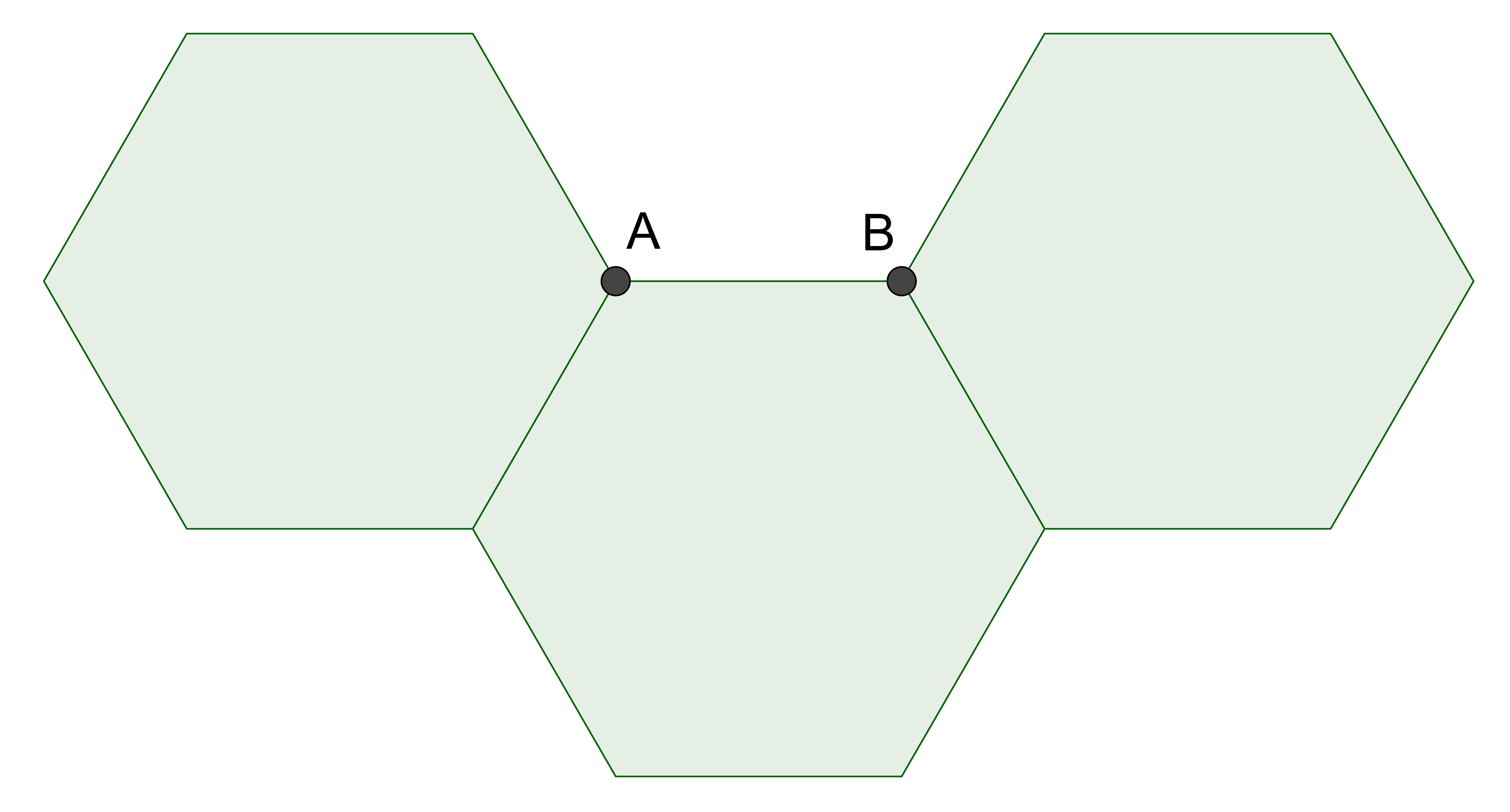}
\caption{Alcune celle di un reticolo a nido d'ape.}
\label{FIG:CAP2:Reticolo_nido_d_ape}
\end{figure}
\`{E} possibile comunque procedere ad una descrizione di questa tipologia di reticoli utilizzando la nozione di \textbf{reticolo con base}. In questo caso l'unità elementare a partire dalla quale si considera la traslazione non è il singolo punto, bensì un insieme di punti. In particolare, il reticolo a nido d'ape può essere rappresentato come un reticolo bidimensionale di Bravais esagonale con una base formata da due punti (Figura \ref{FIG:CAP2:Reticolo_nido_d_ape_2}).
\begin{figure}[ht]
\centering
\includegraphics[width=0.45\columnwidth]{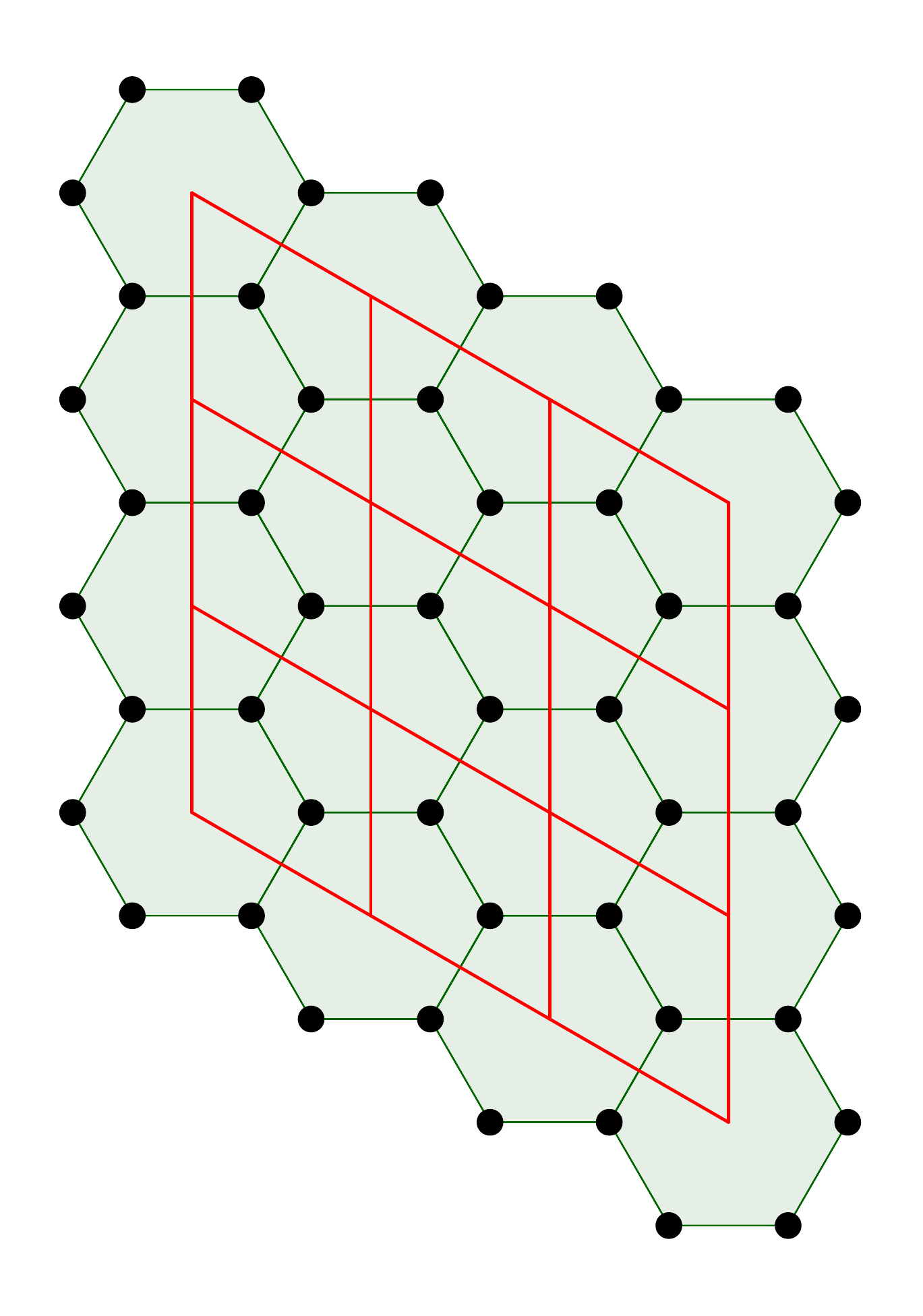}
\caption{Reticolo a nido d'ape in cui ciascuna cella primitiva (parallelogrammi) contiene i due punti di base.}
\label{FIG:CAP2:Reticolo_nido_d_ape_2}
\end{figure}
Infatti esso può essere pensato come composto da due reticoli triangolari compenetranti, indicati con le lettere $A$ e $B$. Le posizioni degli atomi di uno di questi reticoli si trovano nei centri dei triangoli definiti dall'altro reticolo (Figura \ref{FIG:CAP2:Reticolo_nido_d_ape_3}).
\begin{figure}[ht]
\centering
\includegraphics[width=0.6\columnwidth]{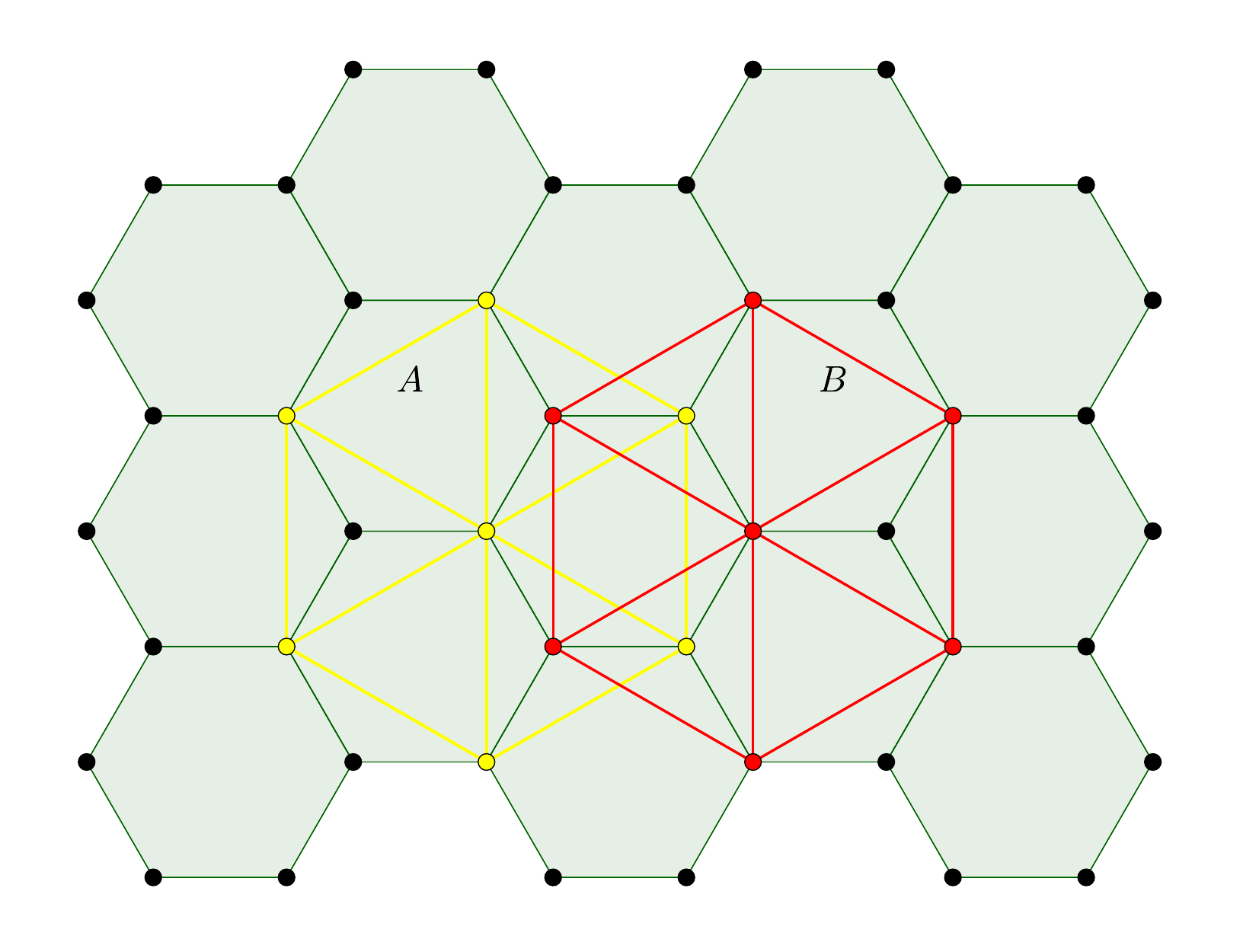}
\caption{Struttura del reticolo a nido d'ape.}
\label{FIG:CAP2:Reticolo_nido_d_ape_3}
\end{figure}
Inoltre ogni atomo di carbonio è circondato solamente da atomi di carbonio che appartengono all'altro reticolo. I vettori primitivi sono
\begin{equation}
\mathbf{a}_1 = \frac{a}{2} (3,\sqrt{3}), \qquad \mathbf{a}_2 = \frac{a}{2} (3,-\sqrt{3}),
\end{equation}
dove $a\approx 1.42 \si{\angstrom}$ è la distanza tra due atomi di carbonio. Poiché la base è biatomica questa non è la costante di reticolo, che invece è $a\sqrt{3} \approx 2.46 \si{\angstrom}$. Si può calcolare che vi sono approssimativamente $\si{\num{3.8d15}}$ atomi di carbonio per $\si{\square\centi\meter}$ e che la densità del grafene è pari a $\rho = \si{\num{7.63d-11}\kilogram.\centi\meter\squared}$. Relativamente al sottoreticolo $A$, i primi vicini si trovano tramite i vettori
\begin{equation}
\boldsymbol{\delta}_1 = \frac{a}{2}(1,\sqrt{3}), \qquad \boldsymbol{\delta}_2 = \frac{a}{2}(1,-\sqrt{3}), \qquad \boldsymbol{\delta}_3 = -a(1,0). 
\end{equation}
Invece per ottenere quelli del sottoreticolo $B$ è sufficiente considerare gli opposti, cioè
\begin{equation}
\boldsymbol{\gamma}_1 = -\frac{a}{2}(1,\sqrt{3}), \qquad \boldsymbol{\gamma}_2 = \frac{a}{2}(-1,\sqrt{3}), \qquad \boldsymbol{\gamma}_3 = a(1,0). 
\end{equation}
La Figura \ref{FIG:CAP2:Reticolo_nido_d_ape_4} riassume le caratteristiche sopraelencate.
\begin{figure}[ht]
\centering
\includegraphics[width=0.7\columnwidth]{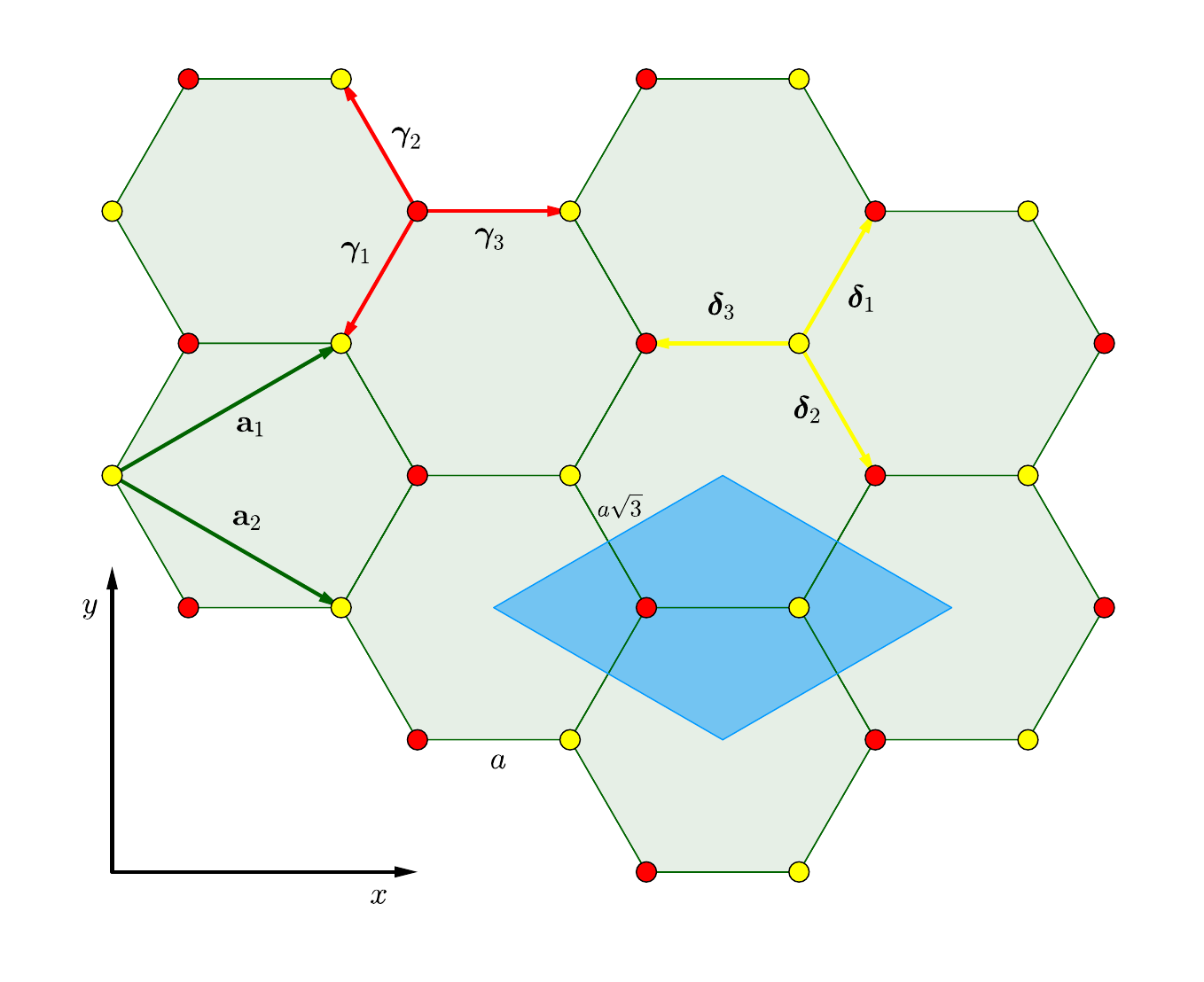}
\caption{Struttura del reticolo a nido d'ape in cui i sottoreticoli $A$ e $B$ sono indicati rispettivamente in giallo e rosso. Inoltre sono rappresentati con lo stesso colore i rispettivi primi vicini. Infine la cella unitaria è evidenziata in azzurro.}
\label{FIG:CAP2:Reticolo_nido_d_ape_4}
\end{figure}
Dalla relazione \eqref{EQ:CAP1:Relazione_reticolo_inverso}, con semplici calcoli si ottiene che la base del reticolo inverso è costituita dai vettori
\begin{equation}
\mathbf{b}_1 = \frac{2\pi}{3a}(1,\sqrt{3}), \qquad \mathbf{b}_2 = \frac{2\pi}{3a}(1,-\sqrt{3}).
\end{equation}
Inoltre si calcola facilmente che $\mathbf{b}_1$ e $\mathbf{b}_2$ verificano le proprietà del reticolo di Bravais esagonale (Tabella \ref{TAB:CAP1:Bravais_2D}). Pertanto il reticolo inverso è un reticolo di Bravais esagonale.

Costruendo le zone di Brillouin si osserva che il reticolo inverso presenta la stessa struttura a nido d'ape del reticolo cristallino, ma ruotata di $90^\circ$ e con dimensioni differenti. Inoltre il punto $\Gamma$ si trova al centro degli esagoni costituenti il reticolo.
Con semplici nozioni di base di geometria piana si costruisce la prima zona di Brillouin ed è possibile calcolare le coordinate dei suoi vertici. Solitamente se ne individuano soltanto due consecutivi che si indicano con $\mathbf{K}$ e $\mathbf{K}'$ e sono detti \textbf{punti di Dirac}, si veda la Figura \ref{FIG:CAP2:Brillouin_zone}, le cui coordinate sono
\begin{equation}
\mathbf{K} = \left( \frac{2\pi}{3a}, \frac{2\pi}{3\sqrt{3}a} \right), \qquad \mathbf{K}' = \left( \frac{2\pi}{3a}, -\frac{2\pi}{3\sqrt{3}a} \right).
\end{equation}
Si può verificare che ogni altro vertice è equivalente ad uno dei punti di Dirac in quanto è possibile collegarli con un vettore del reticolo inverso.
\begin{figure}[ht]
\centering
\includegraphics[width=0.45\columnwidth]{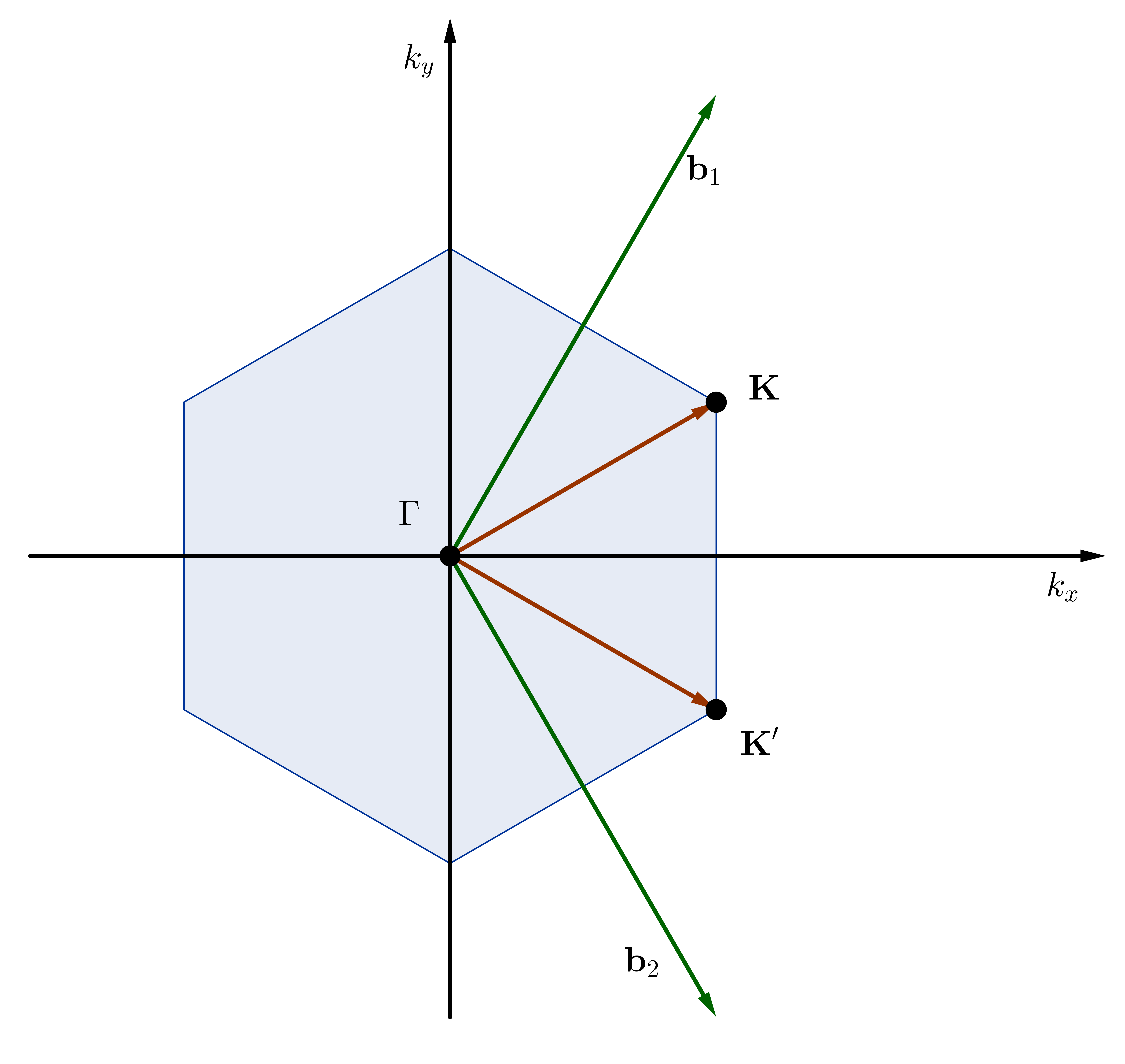}
\caption{Rappresentazione grafica del reticolo inverso. La prima zona di Brillouin è evidenziata in blu.}
\label{FIG:CAP2:Brillouin_zone}
\end{figure}
\FloatBarrier
Solitamente le valli in prossimità di $\mathbf{K}$ e $\mathbf{K}'$ sono trattate come equivalenti ad una sola valle.

\subsection{La struttura elettronica a bande nel grafene}\label{PAR:CAP2:Struttura_bande_grafene}
Per ottenere le bande energetiche nel caso del grafene, si scrive l'hamiltoniana relativa agli elettroni utilizzando il metodo tight-binding. L'approssimazione che si effettua consiste nel supporre che gli elettroni possono passare solo agli atomi primi vicini oppure a quelli secondi vicini (crf. \citep{ART:Castro-Neto}). Le bande di energia che si ottengono hanno la forma
\begin{equation}
E_{\pm}(\mathbf{k}) = \pm \gamma_1 \sqrt{3+f(\mathbf{k})}-\gamma_2 f(\mathbf{k}),
\end{equation}
con
\begin{equation}
f(\mathbf{k})=2\cos \left( \sqrt{3}k_y a \right) +4\cos \left( \frac{\sqrt{3}}{2}k_y a \right)\cos\left( \frac{3}{2}k_x a \right),
\end{equation}
dove $\gamma_1 \approx \si{\num{2.8}.\electronvolt}$ è l'energia di salto tra due primi vicini, cioè tra sottoreticoli diversi e $\gamma_2 \approx \si{\num{0.1}.\electronvolt}$ è l'energia di salto tra due secondi vicini, cioè nello stesso sottoreticolo. Le coordinate di $\mathbf{k}$ sono espresse rispetto al punto $\Gamma$. Gli stati elettronici dati da $E_{+}(\mathbf{k})$ formano la banda $\pi^*$, quelli dati da $E_{-}(\mathbf{k})$ formano la banda $\pi$. Un semplice calcolo permette di stabilire che $E_{+}(\mathbf{k})=E_{-}(\mathbf{k})$ in $\mathbf{K}$ e $\mathbf{K}'$. Questo significa che le due bande si toccano esattamente in corrispondenza dei punti di Dirac. Trascurando l'energia di salto tra secondi vicini (cioè ponendo $\gamma_2=0$), le due bande risultano simmetriche, cioè si ha
\begin{equation}
E_{+}(\mathbf{k})=-E_{-}(\mathbf{k}), \qquad E_{\pm}(\mathbf{K})=E_{\pm}(\mathbf{K}')=0.
\end{equation}
Si può osservare che le bande hanno forma conica in prossimità dei punti di Dirac, quindi una buona approssimazione per la relazione di dispersione è
\begin{equation}
E_{\pm}(\mathbf{k}) \approx \pm\hbar v_F \vert \mathbf{k} \vert + \mathcal{O} (\vert \mathbf{k} \vert^2),
\end{equation}
dove $v_F$ è detta velocità di Fermi ed è data teoricamente da
\begin{equation}
v_F = \frac{3\gamma_1 a}{2\hbar},
\end{equation}
Il cui valore ottenuto sperimentalmente è $v_F \approx \si{\num{1e6}\metre\per\second}$. La Figura \ref{FIG:CAP2:band_structure} mostra graficamente quanto sopra esposto.
\begin{figure}[ht]
\centering
\includegraphics[width=0.6\columnwidth]{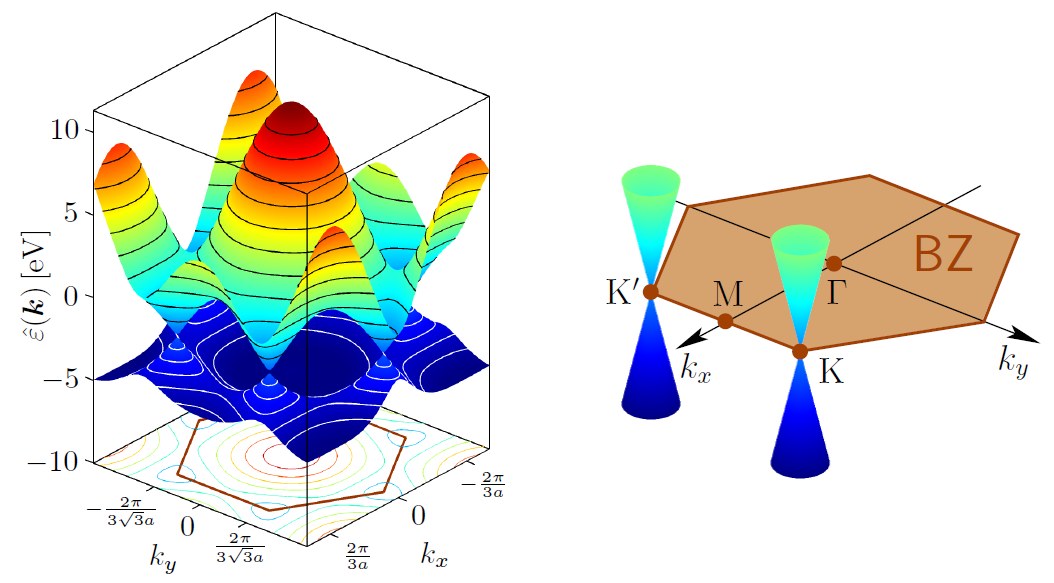}
\caption{A sinistra, grafico della relazione di dispersione delle bande energetiche del grafene. A destra, rappresentazione della prima zona di Brillouin dove sono mostrati i doppi coni che approssimano le bande.}
\label{FIG:CAP2:band_structure}
\end{figure}
\FloatBarrier


\section{Il modello semiclassico nel caso del grafene}
Si vuole applicare il modello semiclassico per il trasporto di cariche al caso del grafene. Allo scopo, occorre determinare esplicitamente il termine di collisione considerando i meccanismi di scattering che intervengono maggiormente. 

I paragrafi successivi sono tratti da \citep{ART:RoMaCo_JCP} e \citep{THESIS:Lichtenberger}.

\subsection{L'equazione di Boltzmann}
Nel modello semiclassico il trasporto di cariche nel grafene è descritto da quattro equazioni di Boltzmann, una per gli elettroni nella banda di valenza $\pi$ e una per gli elettroni nella banda di conduzione $\pi^*$ in ciascuna delle valli $\mathbf{K}$ e $\mathbf{K}'$. Sinteticamente le quattro equazioni possono essere scritte
\begin{equation}\label{EQ:CAP2:Eq_Boltz_grafene}
\frac{\partial f_{l,s}(t,\mathbf{x},\mathbf{k})}{\partial t} + \mathbf{v}_{l,s} \cdot \nabla_{\mathbf{x}} f_{l,s}(t,\mathbf{x},\mathbf{k}) - \frac{e}{\hbar} \mathbf{E} \cdot \nabla_{\mathbf{k}}f_{l,s}(t,\mathbf{x},\mathbf{k}) = \left( \frac{\diff f_{l,s}}{\diff t} (t,\mathbf{x},\mathbf{k}) \right)\ped{coll},
\end{equation}
dove $f_{l,s}(t,\mathbf{x},\mathbf{k})$ è la funzione di distribuzione dei portatori di carica, relativa alla banda $\pi$ ($s=-1$) o $\pi^*$ ($s=1$) nella valle $l$ ($\mathbf{K}$ o $\mathbf{K}'$), nella posizione $\mathbf{x}$, tempo $t$ e impulso $\mathbf{k}$. La velocità di gruppo $\mathbf{v}_{l,s}$ è legata alla banda di energia $\varepsilon_{l,s}$ da
\begin{equation}
\mathbf{v}_{l,s} = \frac{1}{\hbar}\nabla_{\mathbf{k}} \varepsilon_{l,s}.
\end{equation}
Inoltre la carica elementare è indicata con $e$ ed $\mathbf{E}$ è il campo elettrico, ottenibile mediante un'equazione di Poisson che è possibile accoppiare alle equazioni \eqref{EQ:CAP2:Eq_Boltz_grafene}. Come già scritto nel Paragrafo \ref{PAR:CAP2:Struttura_bande_grafene}, una buona approssimazione della relazione di dispersione per le bande energetiche vicino ai punti di Dirac è quella lineare, cioè
\begin{equation}
\varepsilon_{l,s} = s\hbar v_F \vert \mathbf{k}-\mathbf{k}_l\vert,
\end{equation}
dove $v_F$ è la velocità di Fermi, $\hbar$ è la costante di Planck ridotta e $\mathbf{k}_l$ è la posizione del punto di Dirac $l$. Il secondo membro della \eqref{EQ:CAP2:Eq_Boltz_grafene} è il termine collisionale in cui sono inclusi i meccanismi di scattering tra elettroni e fononi.

\subsection{Il termine di collisione}
Il modello collisionale descritto in questa tesi considera come sorgente di scattering le interazioni tra elettroni e fononi (cfr.~\citep{ART:RoMaCo_JCP}). Nel Paragrafo \ref{PAR:CAP1:Eq_Boltz_semi} si è data la forma esplicita del termine collisionale che nel caso della \eqref{EQ:CAP2:Eq_Boltz_grafene} assume la forma
\begin{equation}
\begin{aligned}
\left( \frac{\diff f_{l,s}}{\diff t} (t,\mathbf{x},\mathbf{k}) \right)\ped{coll} = & \sum_{l',s'} \left[ \int \! S_{l',s',l,s} (\mathbf{k}',\mathbf{k})f_{l',s'}(t,\mathbf{x},\mathbf{k}')\left( 1-f_{l,s}(t,\mathbf{x},\mathbf{k}) \right) \, \diff \mathbf{k}' \right.\\
&- \left. \int \! S_{l,s,l',s'} (\mathbf{k},\mathbf{k}')f_{l,s}(t,\mathbf{x},\mathbf{k})\left( 1-f_{l',s'}(t,\mathbf{x},\mathbf{k}') \right)  \, \diff \mathbf{k}' \right],
\end{aligned}
\end{equation}
dove $S_{l',s',l,s} (\mathbf{k}',\mathbf{k})$ è la distribuzione di probabilità di avere uno scattering dallo stato $\mathbf{k}'$ nella banda $s'$ e nella valle $l'$ allo stato $\mathbf{k}$ nella banda $s$ e nella valle $l$. Applicando quanto visto nel Paragrafo \ref{PAR:CAP1:Nucleo_coll} la distribuzione di probabilità appena menzionata presenta la forma
\begin{equation}\label{EQ:CAP2:Prob_scatter}
\begin{aligned}
S_{l',s',l,s} (\mathbf{k}',\mathbf{k}) = & \sum_{\nu} \left\vert G^{(\nu)}_{l',s',l,s}(\mathbf{k}',\mathbf{k}) \right\vert^2 \times \left[ \left( n_{\mathbf{q}}^{(\nu)}+1 \right)\delta \left( \varepsilon_{l,s}(\mathbf{k})-\varepsilon_{l',s'}(\mathbf{k}')+\hbar \omega_{\mathbf{q}}^{(\nu)} \right)\right. \\
& \left. + n_{\mathbf{q}}^{(\nu)} \delta \left( \varepsilon_{l,s}(\mathbf{k})-\varepsilon_{l',s'}(\mathbf{k}')-\hbar \omega_{\mathbf{q}}^{(\nu)} \right)  \right],
\end{aligned}
\end{equation}
dove $\delta$ indica la delta di Dirac, l'indice $\nu$ rappresenta il $\nu$-esimo modo vibrazionale, $G^{(\nu)}_{l',s',l,s}(\mathbf{k}',\mathbf{k})$ è il fattore di sovrapposizione corrispondente, $\omega_{\mathbf{q}^{(\nu)}}$ è la frequenza del $\nu$-esimo modo vibrazionale e $n_{\mathbf{q}}^{(\nu)}$ è la distribuzione di Bose-Einstein corrispondente, cioè
\begin{equation}
n_{\mathbf{q}}^{(\nu)} = \frac{1}{e^{\hbar\omega_{\mathbf{q}^{(\nu)}} / k_B T} -1 },
\end{equation}
in cui $k_B$ è la costante di Boltzmann e $T$ è la temperatura del reticolo di grafene. Se per un certo modo vibrazionale di tipo $(\nu_*)$ risulta $\hbar\omega_{\mathbf{q}}^{(\nu_*)}\ll k_B T$ allora vale la cosiddetta approssimazione elastica, cioè $n_{\mathbf{q}}^{(\nu)} \approx k_B T / \hbar\omega_{\mathbf{q}}^{(\nu_*)}$ ed è possibile eliminare i termini $\hbar\omega_{\mathbf{q}}^{(\nu_*)}$ negli argomenti delle delta nella \eqref{EQ:CAP2:Prob_scatter}.

\subsection{Fononi e meccanismi di scattering}
Nel grafene i fononi sono descritti da vettori d'onda bidimensionali. Gli atomi del reticolo possono vibrare sia nella direzione del vettore d'onda che perpendicolarmente ad esso, generando così modi vibrazionali longitudinali e trasversali, rispettivamente. I due atomi di carbonio della cella unitaria danno luogo all'ulteriore suddivisione in modi fononici ottici ed acustici (cfr. \citep{THESIS:Lichtenberger}). Inoltre, gli atomi di carbonio possono anche muoversi perpendicolarmente al piano su cui giace il grafene. Di conseguenza ai precedenti quattro modi fononici se ne aggiungono altri due, chiamati \textbf{fononi di tipo Z} e suddivisi a loro volta in ottici e acustici. In definitiva, nel grafene si hanno sei modi vibrazionali: fononi longitudinali ottici (LO) ed acustici (LA), fononi trasversali ottici (TO) ed acustici (TA) e fononi di tipo Z ottici (ZO) ed acustici (ZA).

I modi vibrazionali che intervengono in modo rilevante nel trasporto di carica nel grafene sono tre di tipo ottico e uno di tipo acustico. I tre fononi rilevanti di tipo ottico sono quelli longitudinali e trasversali con vettore d'onda vicino al punto $\Gamma$ e quelli vicini al punto $\mathbf{K}$. Questi tre tipi di fononi sono indicati con $\Gamma$-LO, $\Gamma$-TO e $\mathbf{K}$-phonons. Infine, i fononi acustici rilevanti sono di tipo longitudinale e hanno vettore d'onda vicino al punto $\Gamma$.

Le interazioni con i fononi acustici sono intra-valley e intra-band, invece quelle con i fononi ottici sono solamente intra-vally e possono essere sia intra-band che inter-band. Infine gli scattering con i $\mathbf{K}$-phonons sono inter-valley.

Per i fononi acustici vale l'approssimazione elastica, per cui il loro contributo al termine collisionale è dato da
\begin{equation}
2 n_{\mathbf{q}}\ap{(ac)} \left\vert G\ap{(ac)}(\mathbf{k}',\mathbf{k}) \right\vert^2 = \frac{1}{(2\pi)^2} \frac{\pi D^2\ped{ac} k_B T}{2\hbar\sigma\ped{m}v^2\ped{p}}(1+\cos \vartheta_{\mathbf{k},\mathbf{k}'}),
\end{equation}
dove $D\ped{ac}$ è la costante di accoppiamento dei fononi acustici, $v\ped{p}$ è la velocità del suono nel grafene, $\sigma\ped{m}$ è la densità del grafene e $\vartheta_{\mathbf{k},\mathbf{k}'})$ è l'ampiezza dell'angolo convesso tra $\mathbf{k}$ e $\mathbf{k}'$.

Si procede adesso calcolando le probabilità di transizione per i $\Gamma$-LO, $\Gamma$-TO e $\mathbf{K}$-phonons. Esse sono date rispettivamente da (cfr. \citep{ART:RoMaCo_JCP})
\begin{equation}\label{EQ:CAP2:Gamma_LO}
\left\vert G\ap{(LO)}(\mathbf{k}',\mathbf{k})  \right\vert^2 = \frac{1}{(2\pi)^2} \frac{\pi D\ped{O}^2}{\sigma\ped{m}\omega\ped{O}} (1-\cos(\vartheta_{\mathbf{k},\mathbf{k}'-\mathbf{k}}+\vartheta_{\mathbf{k}',\mathbf{k}'-\mathbf{k}} ))
\end{equation}
\begin{equation}\label{EQ:CAP2:Gamma_TO}
\left\vert G\ap{(TO)}(\mathbf{k}',\mathbf{k})  \right\vert^2 = \frac{1}{(2\pi)^2} \frac{\pi D\ped{O}^2}{\sigma\ped{m}\omega\ped{O}} (1+\cos(\vartheta_{\mathbf{k},\mathbf{k}'-\mathbf{k}}+\vartheta_{\mathbf{k}',\mathbf{k}'-\mathbf{k}} ))
\end{equation}
\begin{equation}
\left\vert G\ap{(K)}(\mathbf{k}',\mathbf{k})  \right\vert^2 = \frac{1}{(2\pi)^2} \frac{2\pi D\ped{K}^2}{\sigma\ped{m}\omega\ped{K}} (1-\cos \vartheta_{\mathbf{k},\mathbf{k}'}),
\end{equation}
dove $D\ped{O}$ è la costante di accoppiamento dei fononi ottici, $\omega\ped{O}$ è la frequenza dei fononi ottici, $D\ped{K}$ è la costante di accoppiamento dei $\mathbf{K}$-phonons e $\omega\ped{K}$ è la frequenza dei $\mathbf{K}$-phonons. Con $\vartheta_{\mathbf{k},\mathbf{k}'-\mathbf{k}}$ e $\vartheta_{\mathbf{k}',\mathbf{k}'-\mathbf{k}}$ si indicano le ampiezze degli angoli convessi rispettivamente tra $\mathbf{k}$ e $\mathbf{k}'-\mathbf{k}$ e tra $\mathbf{k}'$ e $\mathbf{k}'-\mathbf{k}$.

\subsection{Scattering rates}\label{PAR:CAP3:Scatt_rates}
Le frequenze di collisione associate ai vari meccanismi di scattering sono calcolate integrando le probabilità di transizione su tutti i possibili valori di $\mathbf{k}'$, cioè si ha
\begin{equation}
\Gamma (\mathbf{k}) = \int \! S_{l',s',l,s} (\mathbf{k},\mathbf{k}') \, \diff \mathbf{k}'.
\end{equation}
In realtà $\Gamma (\mathbf{k})$ dipende da $\mathbf{k}$ tramite l'energia, cioè si ha $\Gamma (\mathbf{k})=\Gamma(\varepsilon)$. Allora, per i fononi acustici, la frequenza di collisione è data da
\begin{equation}
\Gamma\ped{ac}(\varepsilon) = \frac{D^2\ped{ac} k_B T}{4\hbar^3 v_F^2 \sigma\ped{m}v^2\ped{p}}\varepsilon.
\end{equation}
Invece la frequenza di collisione totale per i fononi ottici è data dalla somma dei contributi longitudinali e trasversali,
\begin{equation}
\Gamma\ped{op}(\varepsilon) = \frac{D^2\ped{O}}{\sigma\ped{m}\omega\ped{O}\hbar^2 v_F^2} \left[ (\varepsilon-\hbar\omega\ped{O})\left( n_{\mathbf{q}}\ap{(O)}+1 \right)H(\varepsilon+\hbar\omega\ped{O})+(\varepsilon+\hbar\omega\ped{O})n_{\mathbf{q}}\ap{(O)} \right],
\end{equation}
e la frequenza di collisione per i $\mathbf{K}$-phonons è data da
\begin{equation}
\Gamma\ped{K}(\varepsilon) = \frac{D^2\ped{K}}{\sigma\ped{m}\omega\ped{K}\hbar^2 v_F^2} \left[ (\varepsilon-\hbar\omega\ped{K})\left( n_{\mathbf{q}}\ap{(K)}+1 \right)H(\varepsilon-\hbar\omega\ped{K})+(\varepsilon+\hbar\omega\ped{K})n_{\mathbf{q}}\ap{(K)} \right],
\end{equation}
dove $H$ è la funzione di Heaviside.


\section{Il metodo Monte Carlo}
In questa sezione si darà una introduzione alle tecniche Monte Carlo, grazie alle quali è possibile simulare le distribuzioni di probabilità. Si illustrerà poi il cosiddetto metodo di Simulazione Diretta Monte Carlo, seguendo l'approccio descritto in \citep{ART:CoMajRo_Ric_mat_2016}, che permette alla simulazione di essere in accordo con il principio di esclusione di Pauli. Si concluderà con la questione del costo computazionale e il confronto dei tempi di esecuzione delle simulazioni.

Il materiale di riferimento per i successivi paragrafi è \citep{BOOK:Jacoboni_Lugli}, \citep{ART:CoMajRo_Ric_mat_2016} e \citep{BOOK:Jacoboni}.

\subsection{Introduzione}
Quando si parla di \emph{metodi Monte Carlo} (MC) si fa riferimento ad un insieme di tecniche che si servono di variabili aleatorie, generate artificialmente al calcolatore, per risolvere problemi matematici. Approcci ai problemi di questo tipo sono conosciuti da molto tempo ed il primo esempio di utilizzo dei numeri casuali per risolvere integrali definiti risale al 1777, formulato da Georges-Luis Leclerc. Per molto tempo questi metodi sono stati usati solo sporadicamente. Il primo utilizzo vero e proprio dei metodi MC è dovuto al gruppo di Enrico Fermi negli anni '40, durante lo sviluppo del progetto della bomba atomica. Con il diffondersi di calcolatori sempre più potenti, i metodi MC hanno preso piede e oggi sono utilizzati in molti settori della ricerca scientifica \cite{BOOK:Rotondi}.

Il nome ``Monte Carlo'' apparve per la prima volta in un articolo dal titolo ``The Monte Carlo Method'' dei matematici Nicholas C. Metropolis e Stanislaw Ulam pubblicato nel 1949 \cite{ART:Metropolis}. Il nome fa riferimento alla nota città del Principato di Monaco, famosa per il suo casinò. L'idea dell'appellativo nasce dal fatto che la roulette è un semplice mezzo meccanico con cui è possibile generare variabili aleatorie \cite{BOOK:Jacoboni}.

L'utilizzo di questi metodi per la risoluzione di un problema matematico porta ad un risultato che inevitabilmente è affetto da errore statistico. Tuttavia spesso può essere molto complicato o addirittura impossibile risolvere certi problemi con l'utilizzo di metodi analitici o numerici e di conseguenza un approccio di tipo MC diventa molto utile.

\subsection{La generazione di numeri random}
Le tecniche Monte Carlo sono basate sulla generazione di numeri random. In particolare, è necessario generare numeri random che seguono distribuzioni di probabilità predefinite. I più comuni linguaggi di programmazione hanno funzioni che generano numeri random uniformemente distribuiti tra 0 e 1.

A partire dalla distribuzione uniforme tra 0 e 1, $U([0,1])$, è possibile generare numeri random che seguono distribuzioni di probabilità volute. Esistono varie tecniche che permettono di fare questo (cfr.~\citep{BOOK:Jacoboni_Lugli}).

\textbf{La tecnica diretta}\\
Sia $f(x)$ la distribuzione di probabilità di una variabile aleatoria continua definita in un intervallo $(a,b)$, limitato o illimitato. Posto
\begin{equation}
A=\int_a^b \! f(x) \, \diff x,
\end{equation}
si indichi con $F(x)$ la funzione integrale di $f$. Se $A=1$, dato un numero $r\sim U([0,1])$, si può determinare in corrispondenza di esso un numero $x_r$ tale che
\begin{equation}\label{EQ:CAP3:tec_dir}
r=F(x_r)=\int_a^{x_r} \! f(x) \, \diff x.
\end{equation}
In questo modo $x_r$ è una nuova variabile aleatoria che ha come distribuzione proprio $f(x)$. Se $A\neq 1$, si sostituisce la \eqref{EQ:CAP3:tec_dir} con
\begin{equation}
r=\frac{1}{A}\int_a^{x_r} \! f(x) \, \diff x
\end{equation}
e la variabile aleatoria $x_r$ ha ancora distribuzione $f(x)$.

Nel caso in cui la variabile aleatoria sia discreta, sia $S=\left\lbrace s_0, s_1, s_2 \ldots \right\rbrace$ l'insieme (finito o numerabile) dei possibili valori che essa può assumere. Ricordando che la funzione di ripartizione $F(x)$ di una variabile aleatoria discreta è costante a tratti e i valori assunti da $F$ sono quindi dati da
\begin{equation}
y_i = F(x_i), \qquad \forall x_i\in [s_i, s_{i+1}[, \qquad i=0,1,2, \ldots,
\end{equation}
con $0\leq y_i\leq 1$ e $y_i < y_{i+1}$ per $i=0,1,2,\ldots$. Sia dato un numero $r\sim U([0,1])$. Allora esiste un indice $\overline{i}$ tale che $y_{\overline{i}}< r < y_{\overline{i}+1}$. Si osservi che in corrispondenza di $s_{\overline{i}}$ si ha $y_{\overline{i}}=F(s_{\overline{i}})$. Allora $s_{\overline{i}}$ è distribuito secondo la funzione di ripartizione $F$.

\textbf{Il metodo del rigetto}\\
Sia $X$ una variabile aleatoria continua e sia $f(x)$ la sua distribuzione di probabilità, definita in un intervallo finito $(a,b)$. Supponendo che $f$ sia limitata, sia $C$ un numero reale positivo tale che $C\geq f(x)$, per ogni $x\in (a,b)$ e siano $r_1$ ed  $r_1'$ due numeri random tali che $r_1,r_1'\sim U((0,1))$. Allora, posto
\begin{equation}\label{EQ:CAP3:Cond_rigetto}
x_1 = a+(b-a)r_1, \qquad f_1 = r_1' C
\end{equation}
si ha che $x_1\sim U((a,b))$ e $f_1\sim U((0,C))$.

Se $f_1\leq f(x_1)$ si prende $x_1$ come una buona scelta di $X$, altrimenti si rigetta $x_1$ e si ripete la procedura con una nuova coppia di numeri random $r_2$ e $r_2'$, fino a che la \eqref{EQ:CAP3:Cond_rigetto} non è soddisfatta.

\subsection{Il metodo di Simulazione Diretta Monte Carlo (DSMC)}
Il metodo consiste nella simulazione del moto degli elettroni nel cristallo, soggetti all'azione di un campo elettrico esterno $\mathbf{E}$ e dopo aver assegnato i meccanismi di scattering. Vengono determinati in modo stocastico il free-flight, cioè il tempo che intercorre tra due collisioni successive, e gli eventi di scattering le cui probabilità sono determinate dalla meccanica quantistica.

La simulazione inizia considerando un elettrone che si trova in una data condizione iniziale con vettore d'onda $\mathbf{k}_0$. Viene scelta la durata del free-flight con una certa distribuzione di probabilità determinata dalle probabilità di transizione. Durante il free-flight, il campo esterno agisce in accordo con la dinamica semiclassica $\hbar\dot{\mathbf{k}}=e\mathbf{E}$. Vengono calcolate le grandezze fisiche di interesse come velocità, energia, ecc. \`{E} scelto un meccanismo di scattering come responsabile della fine del free-flight, in accordo con le probabilità relative rispetto a tutti i meccanismi di scattering. Dopo lo scattering è scelto un nuovo stato che viene posto come stato iniziale di un nuovo free-flight. A questo punto la procedura si ripete iterativamente. Mano a mano che la simulazione va avanti i risultati diventano più accurati fino a quando non viene ottenuta la precisione desiderata. A quel punto la simulazione si arresta. 

L'implementazione utilizzata per questa tesi riguarda il caso spazialmente omogeneo, cioè in cui la distribuzione dei portatori di carica non dipende da $\mathbf{x}$, cioè è del tipo $f(t,\mathbf{k})$. In questo caso l'equazione di Boltzmann relativa alla generica valle $\mathbf{K}$ diviene
\begin{equation}\label{EQ:CAP3:Eq_Bolz_omo}
\begin{aligned}
\frac{\partial f(t,\mathbf{k})}{\partial t} -\frac{e}{\hbar}\mathbf{E}\cdot\nabla_{\mathbf{k}}f(t,\mathbf{k}) = & \int \! S(\mathbf{k}',\mathbf{k})f(t,\mathbf{k}')(1-f(t,\mathbf{k})) \, \diff \mathbf{k}'\\
& - \int \! S(\mathbf{k},\mathbf{k}')f(t,\mathbf{k})(1-f(t,\mathbf{k}')) \, \diff \mathbf{k}'.
\end{aligned}
\end{equation}
Un'equazione simile vale per la valle $\mathbf{K}'$ (cfr.~\citep{ART:RoMaCo_JCP}). Per quanto riguarda la condizione iniziale, supponendo che l'insieme degli elettroni considerati sia in equilibrio termico, si assume che la distribuzione del numero di occupazione degli stati con energia $\varepsilon(\mathbf{k})$ segue la statistica di Fermi-Dirac, introdotta nel Paragrafo \ref{PAR:CAP1:Densita_stati}, cioè si ha
\begin{equation}
f(0,\mathbf{k})=\frac{1}{1+\exp\left( \frac{\varepsilon(\mathbf{k})-\varepsilon_F}{k_B T} \right)},
\end{equation}
dove $\varepsilon_F$ è il livello di Fermi. Lo spazio degli impulsi è approssimato dall'insieme $[-k_{\mbox{x max}},k_{\mbox{x max}}] \times [-k_{\mbox{y max}},k_{\mbox{y max}}]$, dove $k_{\mbox{x max}}$ e $k_{\mbox{y max}}$ sono scelti in modo tale che il numero di elettroni aventi vettore d'onda fuori da tale insieme sia trascurabile. Si osservi che elettroni e lacune popolano maggiormente gli stati vicini alle valli $\mathbf{K}$ e $\mathbf{K}'$. Inoltre, considerare un valore alto del livello di Fermi è equivalente ad effettuare un drogaggio di tipo $n$. Allora gli elettroni che stanno nella banda di conduzione non raggiungeranno la banda di valenza e viceversa. Quindi la dinamica delle lacune può essere trascurata.

Si vuole descrivere dettagliatamente il metodo DSMC per questo problema. Si supponga da ora in avanti che il campo elettrico applicato $\mathbf{E}$ sia costante. Come scritto sopra, il free-flight avviene seguendo la dinamica semiclassica. Per cui la variazione dell'impulso in un intervallo di tempo $\Delta t$ in cui avviene il free-flight è dato da
\begin{equation}
\hbar\Delta\mathbf{k}=-e\mathbf{E}\Delta t.
\end{equation}
L'intervallo di tempo $\Delta t$ è scelto in modo random per ogni particella mediante la relazione
\begin{equation}\label{EQ:CAP3:Delta_t}
\Delta t = -\frac{\ln\xi}{\Gamma\ped{tot}},
\end{equation}
dove $\xi$ è un numero random che segue una distribuzione uniforme nell'intervallo $[0,1]$ e $\Gamma\ped{tot}$ proviene dagli scattering rates determinati per il grafene nel Paragrafo~\ref{PAR:CAP3:Scatt_rates}, cioè è definita da
\begin{equation}
\Gamma\ped{tot} = \Gamma\ped{ac}(\varepsilon(t))+\Gamma\ped{op}(\varepsilon(t))+\Gamma\ped{K}(\varepsilon(t))+\Gamma\ped{ss}(\varepsilon(t)),
\end{equation}
dove $\Gamma\ped{ss}(\varepsilon(t))$ è detto \textbf{self-scattering rate} ed è lo scattering rate associato a uno scattering fittizio che non cambia lo stato dell'elettrone, aggiunto con lo scopo di rendere costante $\Gamma\ped{tot}$ e calcolare la \eqref{EQ:CAP3:Delta_t}. Per fare ciò a livello pratico, si potrebbe prendere il valore massimo $\Gamma_M$ di $ \Gamma\ped{ac}(\varepsilon(t))+\Gamma\ped{op}(\varepsilon(t))+\Gamma\ped{K}(\varepsilon(t))$ e porre $\Gamma\ped{tot}=\alpha\Gamma_M$, con $\alpha>1$ (ad esempio $\alpha=1.1$). In questo modo si può ottenere $\Gamma\ped{ss}(\varepsilon(t))$ per differenza. Questo metodo non è molto adatto al problema in esame in quanto gli scattering rates variano in un intervallo piuttosto ampio e con ordini di grandezza diversi, come si vede dalla Figura \ref{FIG:CAP3:Scatt_rates}. 
\begin{figure}[ht]
	\centering
	\includegraphics[width=9cm]{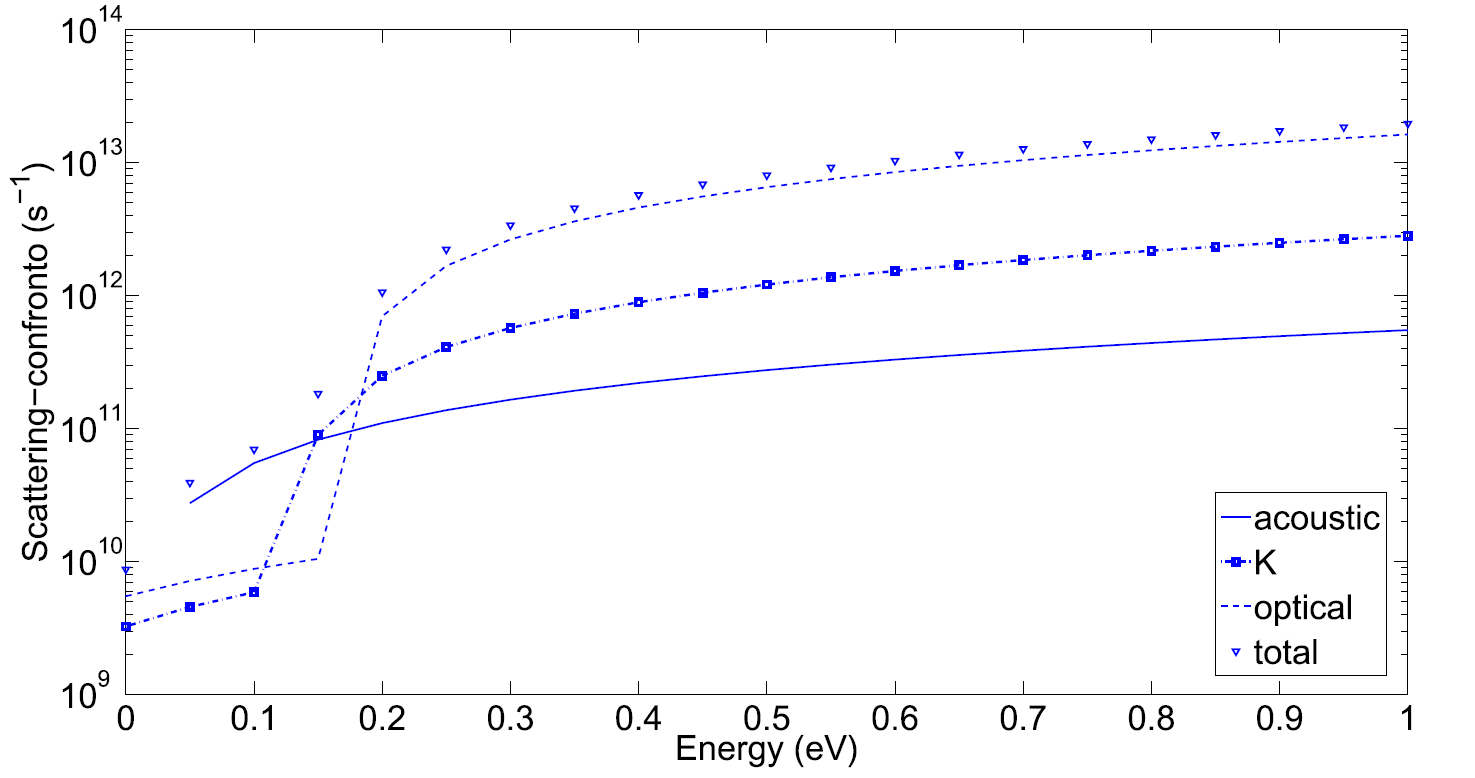}
 	\caption{Grafico degli scatering rates.}\label{FIG:CAP3:Scatt_rates}
\end{figure}
Questo comporta il verificarsi di un gran numero di self-scattering e conseguentemente il costo computazionale aumenta (cfr.~\citep{ART:RoMaCo_JCP}). Una soluzione a questo problema consiste nel considerare un valore di $\Gamma\ped{tot}$ variabile e di volta in volta scelto opportunamente per ogni particella e nell'istante di tempo considerato, cioè si pone
\begin{equation}
\Gamma\ped{tot} = \alpha \left( \Gamma\ped{ac}(\varepsilon(t))+\Gamma\ped{op}(\varepsilon(t))+\Gamma\ped{K}(\varepsilon(t)) \right).
\end{equation}
Dopo il free-flight viene scelto in modo random un tipo di scattering in accordo con la distribuzione data dagli scattering-rates. In pratica si calcola la funzione di ripartizione della distribuzione delle probabilità di transizione e si stabilisce la tipologia di scattering con il metodo di inversione. A questo punto viene determinato il nuovo stato in accordo con il tipo di scattering scelto. Questa scelta deve rispettare la distribuzione di probabilità della tipologia di scattering considerata. Per fare ciò si utilizza il metodo del rigetto.

Gli scattering rates sono calcolati tenendo conto del principio di esclusione di Pauli. Tuttavia può accadere che una particella alla fine del free-flight raggiunga una cella nello spazio degli impulsi già completamente occupata. Di conseguenza la distribuzione $f$ supera il valore 1. Il processo risulta quindi errato dal punto di vista fisico anche se i valori ottenuti delle quantità medie di energia e velocità sono accettabili. Questo problema si presenta più frequentemente alzando il livello di Fermi (cfr.~\citep{ART:RoMaCo_JCP}).

\subsection{Approccio DSMC in accordo con il principio di Pauli}
Lo spunto principale per questa tesi nasce da un algoritmo proposto da V.~Romano, A.~Majorana e M.~Coco nel 2015 e descritto in \citep{ART:RoMaCo_JCP} in cui si utilizza un approccio che consente il rispetto del principio di esclusione di Pauli anche durante il free-flight. L'idea consiste nel suddividere la procedura in due parti: prima si effettua il free-flight per tutte le particelle e poi le collisioni. Più dettagliatamente, si consideri l'equazione \eqref{EQ:CAP3:Eq_Bolz_omo}. Dopo aver fissato una griglia temporale uniforme e quindi un passo costante $\Delta t$, si vuole trovare la distribuzione degli stati al tempo $t+\Delta t$, nota quella al tempo $t$. Allora si risolve la parte dell'equazione che coinvolge il free-flight, cioè
\begin{equation}\label{EQ:CAP3:DSMC_Pauli_1}
\frac{\partial f(t,\mathbf{k})}{\partial t}- \frac{e}{\hbar}\mathbf{E}\cdot\nabla_{\mathbf{k}}f(t,\mathbf{k})=0,
\end{equation}
prendendo come condizione iniziale la distribuzione al tempo $t$. Successivamente si effettuano le collisioni risolvendo
\begin{equation}
\frac{\partial f(t,\mathbf{k})}{\partial t} = \left( \frac{\diff f}{\diff t}(t,\mathbf{k}) \right)\ped{coll},
\end{equation}
prendendo come condizione iniziale la soluzione della \eqref{EQ:CAP3:DSMC_Pauli_1}. Questa procedura fornisce un'approssimazione della $f(t+\Delta t,\mathbf{k})$ al primo ordine in $\Delta t$. Inoltre il principio di esclusione di Pauli è rispettato completamente. Si osservi infine che in \citep{ART:RoMaCo_JCP} il metodo è stato convalidato mediante un confronto con una soluzione numerica ottenuta attraverso il metodo \emph{discontinuous Galerkin}.

\subsection{Risultati numerici nel caso sospeso}
I risultati delle simulazioni sono stati ottenuti con differenti valori del campo elettrico applicato e del livello di Fermi considerato. Il valore del livello di Fermi è scelto in modo sufficientemente grande da rendere trascurabili le interazioni intra-band e quindi la dinamica delle lacune. La temperatura del reticolo è fissata in $\si{\num{300}.\kelvin}$, il numero di particelle è pari a $\si{\num{e5}}$. Per quanto riguarda le costanti fisiche sono stati considerati i valori riportati nella Tabella~\ref{TAB:CAP3:Cost_fis_sospeso}.

Si è scelto di considerare campi elettrici aventi soltanto la componente $x$ diversa da zero. Di conseguenza, soltanto la componente $x$ della velocità media $\mathbf{v}$ è rilevante. Essa è definita da
\begin{equation}
\mathbf{v}(t) = \frac{2}{(2\pi)^2\rho} \int \! f(t,\mathbf{k)}\frac{1}{\hbar}\nabla_{\mathbf{k}}\varepsilon(\mathbf{k}) \, \diff \mathbf{k}.
\end{equation}
La velocità è correlata alla corrente $\mathbf{J}$ tramite la relazione
\begin{equation}
\mathbf{J} = -e\rho\mathbf{v},
\end{equation} 
e inoltre la velocità è legata alla mobilità elettronica $\mu(\mathbf{E})$ da
\begin{equation}
\mathbf{v} = \mu(\mathbf{E})\mathbf{E}.
\end{equation}
\FloatBarrier
\begin{table}[!ht]
\centering
\begin{tabular}{c|c}
\textbf{Parametro fisico} & \textbf{Valore} \\
\hline
$\sigma_m$ & $\si{\num{7.6e-8}.\gram\per\centi\metre\squared}$ \\
$v_F$ & $\si{\num{e6}.\metre\per\second}$ \\
$v_p$ & $\si{\num{2e4}.\metre\per\second}$ \\
$D\ped{ac}$ & $\si{\num{6.8}.\electronvolt}$ \\
$\hbar\omega\ped{O}$ & $\si{\num{164.6}.\milli\electronvolt}$ \\
$D\ped{O}$ & $\si{\num{e9}.\electronvolt\per\centi\metre}$ \\
$\hbar\omega\ped{K}$ & $\si{\num{124}.\milli\electronvolt}$ \\
$D\ped{K}$ & $\si{\num{3.5e8}.\electronvolt\per\centi\metre}$ \\
\end{tabular}
\caption{Costanti fisiche utilizzate nelle simulazioni.}
\label{TAB:CAP3:Cost_fis_sospeso}
\end{table}
\FloatBarrier
I seguenti grafici mostrano i risultati di alcune simulazioni eseguite con il nuovo approccio al metodo DSMC. Nel primo gruppo sono mostrati gli andamenti della velocità e dell'energia nei vari casi. Nel secondo gruppo sono invece riportati i grafici delle distribuzioni, da cui si evince che $0\leq f\leq 1$, pertanto il principio di esclusione di Pauli è rispettato.
\FloatBarrier
\begin{figure}[ht]
	\centering
	\subfigure[]
   		{\includegraphics[width=7.1cm]{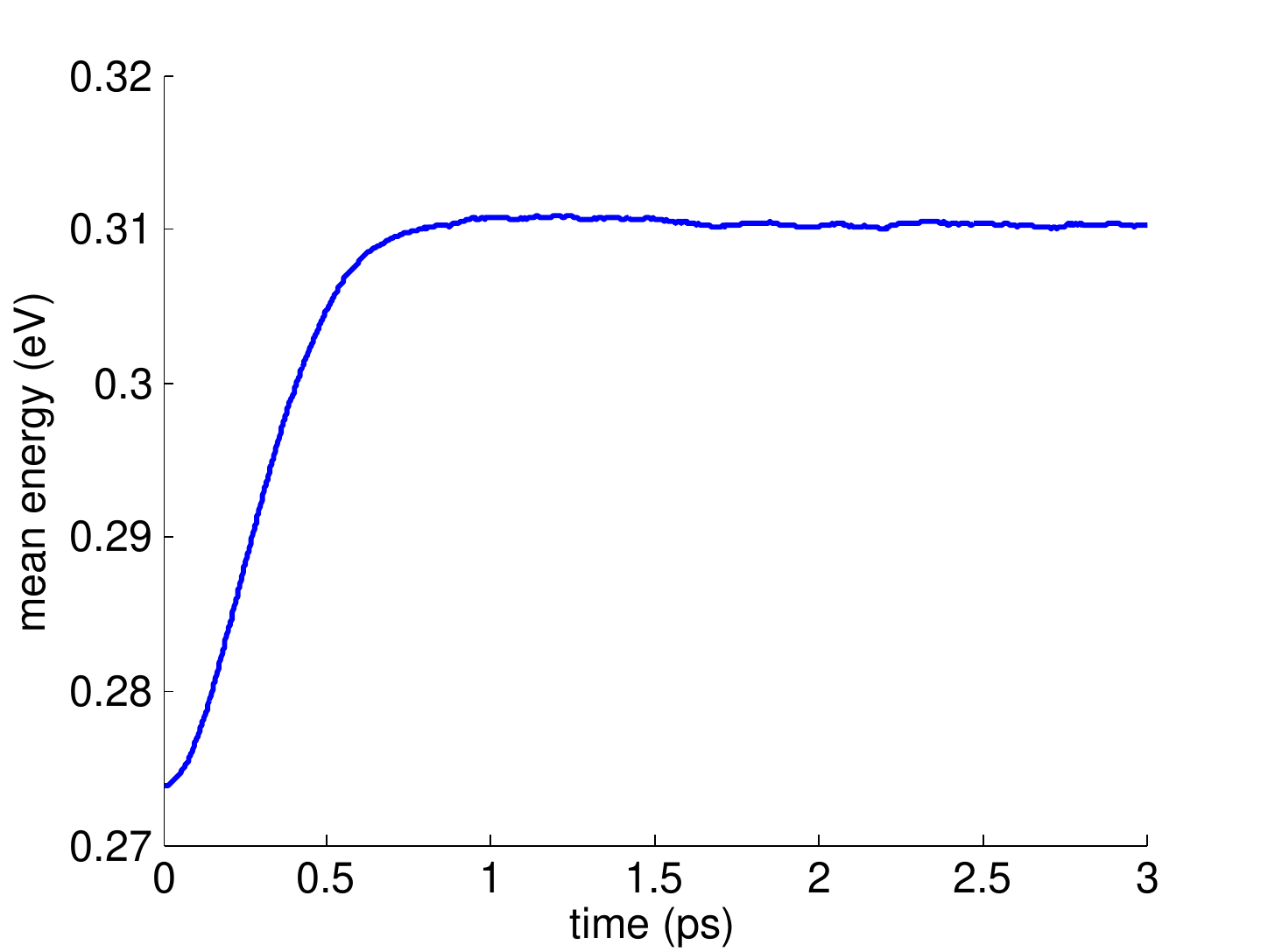}}
 	\,
 	\subfigure[]
   		{\includegraphics[width=7.1cm]{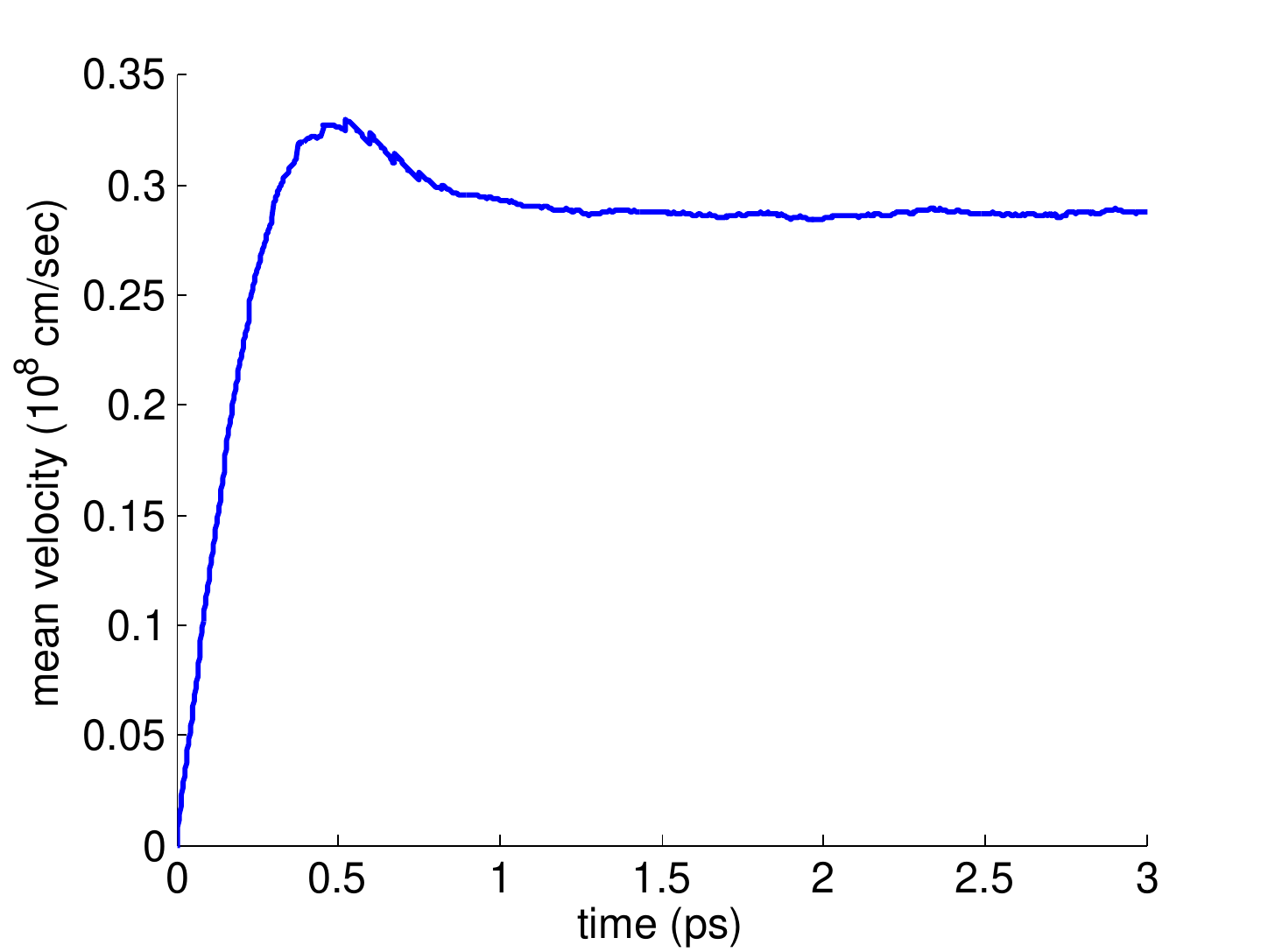}}
 	\caption{Energia (a) e velocità (b) medie ottenute con un campo elettrico applicato di $\si{\num{5}.\kilo\volt\per\centi\metre}$ e un livello di Fermi pari a $\si{\num{0.4}.\electronvolt}$.}\label{FIG:CAP3:DSMC_new_1}
\end{figure}
\begin{figure}[ht]
	\centering
	\subfigure[]
   		{\includegraphics[width=7.1cm]{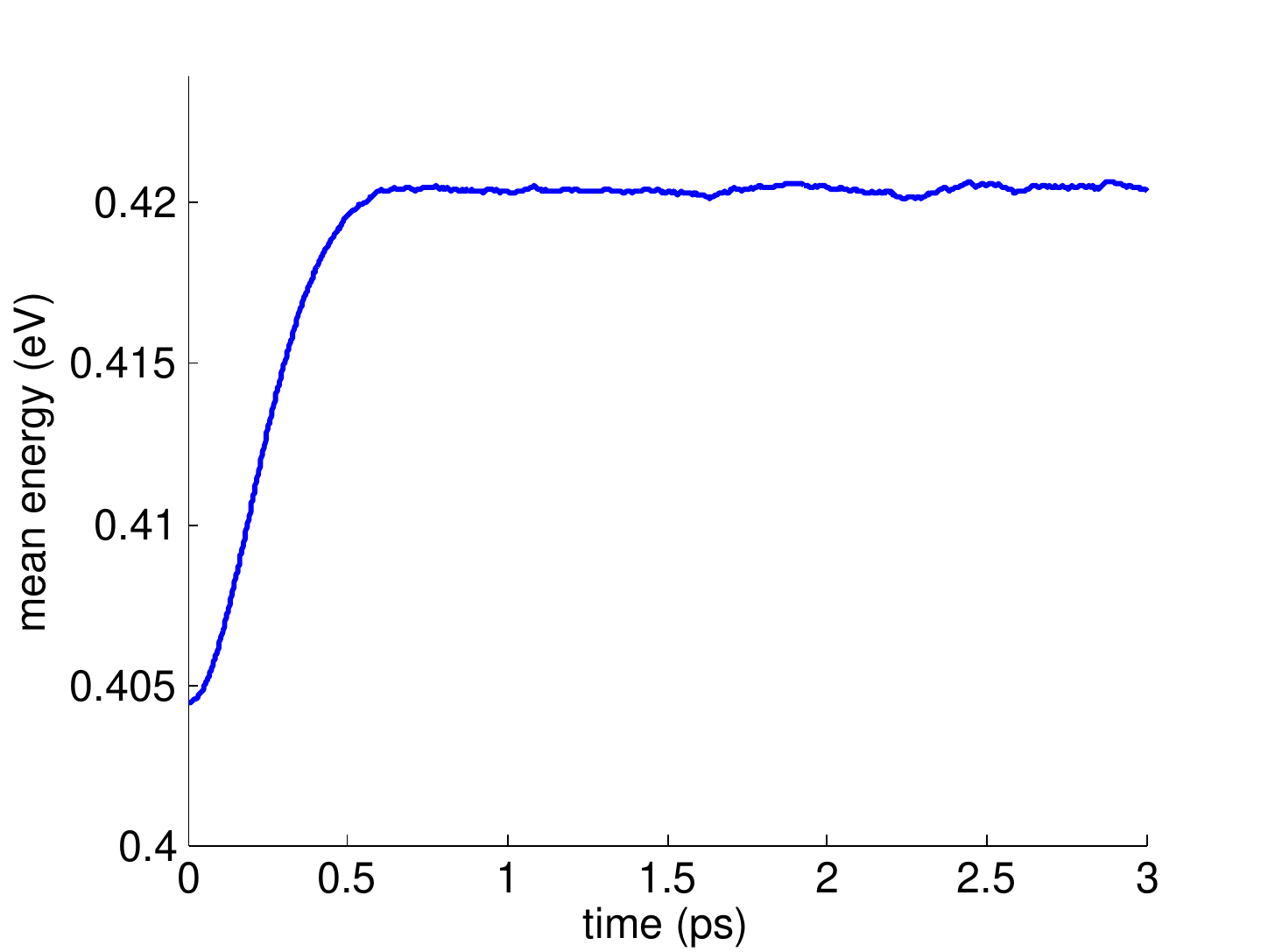}}
 	\,
 	\subfigure[]
   		{\includegraphics[width=7.1cm]{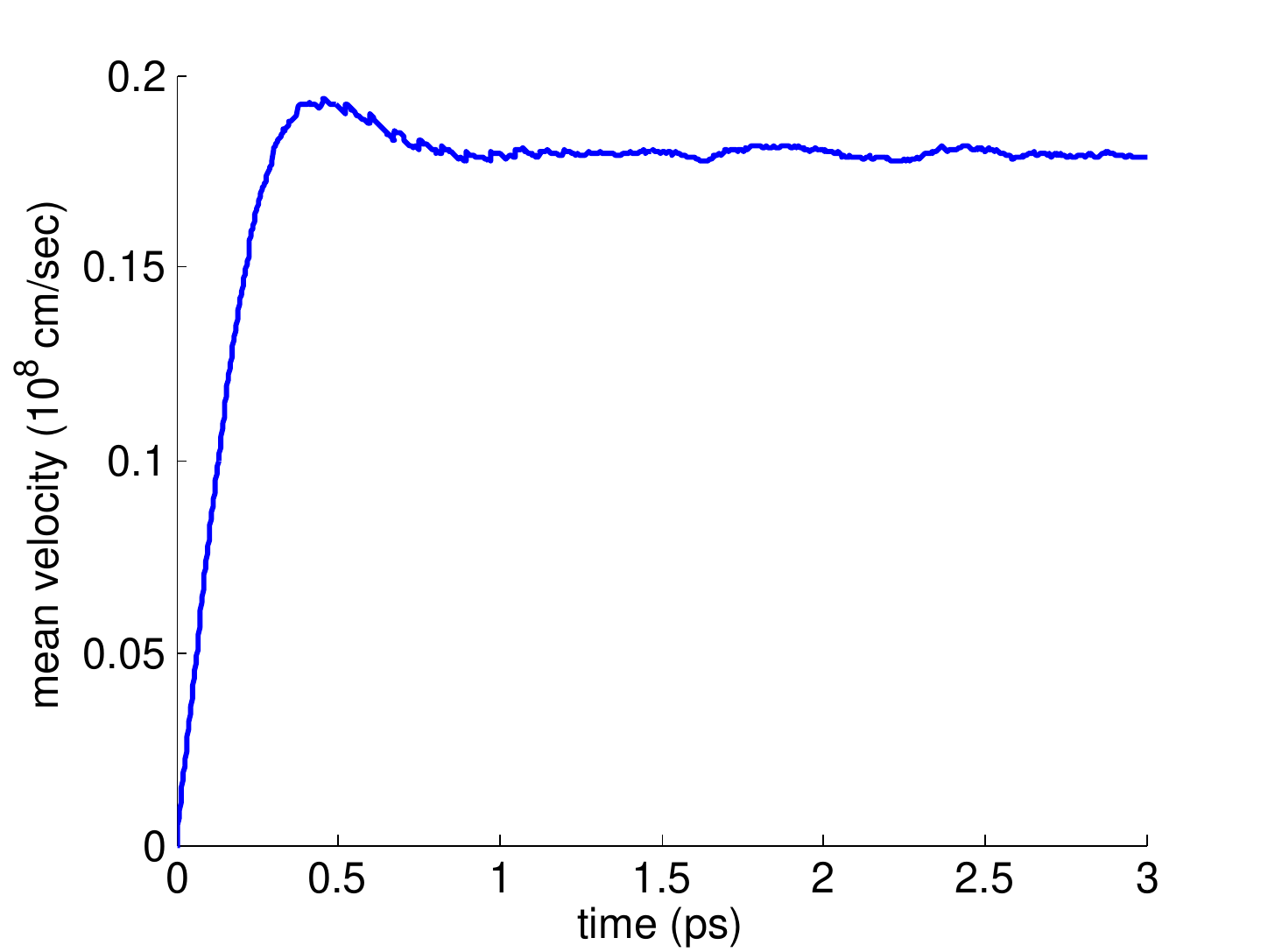}}
 	\caption{Energia (a) e velocità (b) medie ottenute con un campo elettrico applicato di $\si{\num{5}.\kilo\volt\per\centi\metre}$ e un livello di Fermi pari a $\si{\num{0.6}.\electronvolt}$.}\label{FIG:CAP3:DSMC_new_2}
\end{figure}
\begin{figure}[ht]
	\centering
	\subfigure[]
   		{\includegraphics[width=7.1cm]{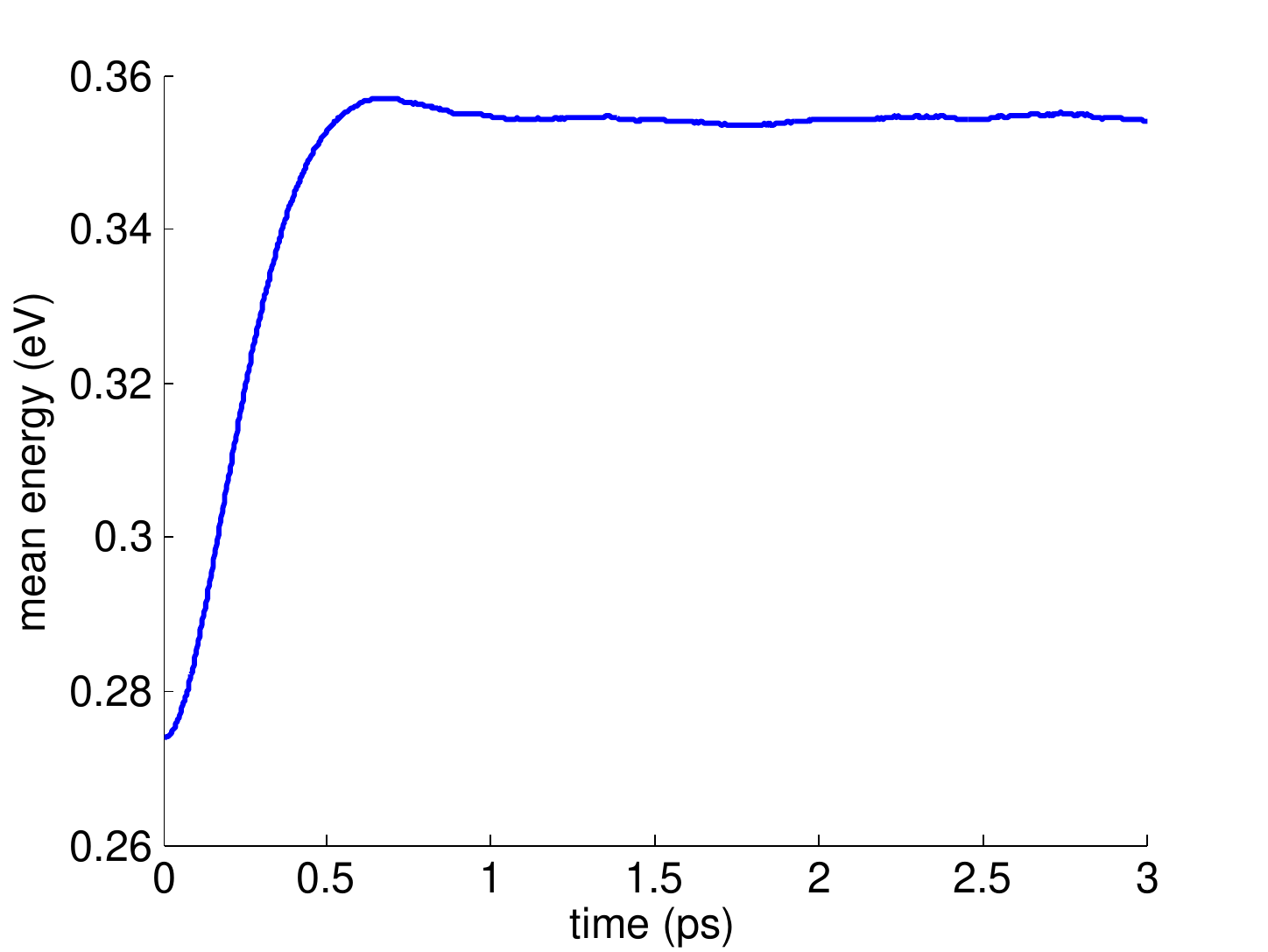}}
 	\,
 	\subfigure[]
   		{\includegraphics[width=7.1cm]{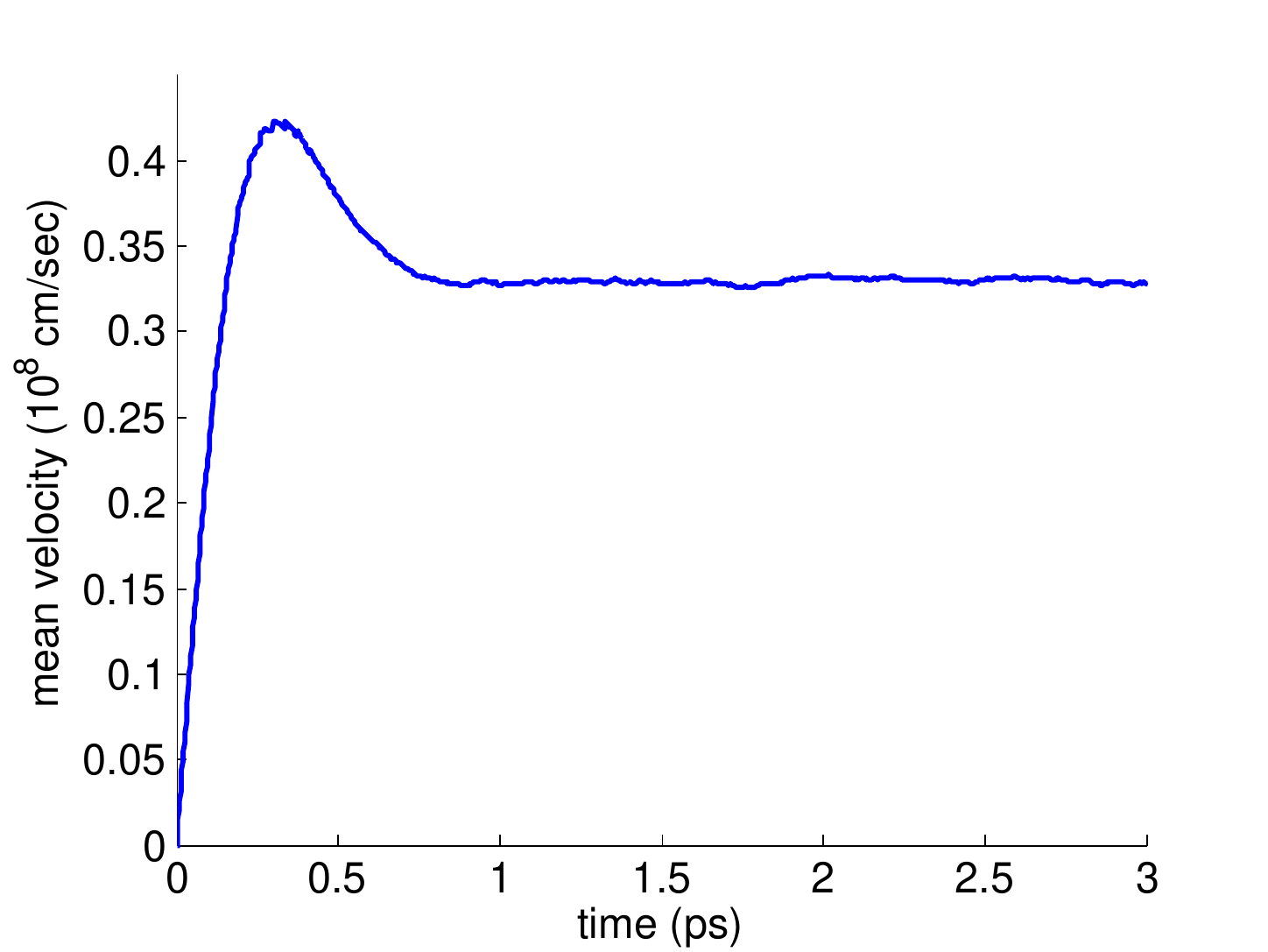}}
 	\caption{Energia (a) e velocità (b) medie ottenute con un campo elettrico applicato di $\si{\num{10}.\kilo\volt\per\centi\metre}$ e un livello di Fermi pari a $\si{\num{0.4}.\electronvolt}$.}\label{FIG:CAP3:DSMC_new_3}
\end{figure}
\begin{figure}[ht]
	\centering
	\includegraphics[width=12cm]{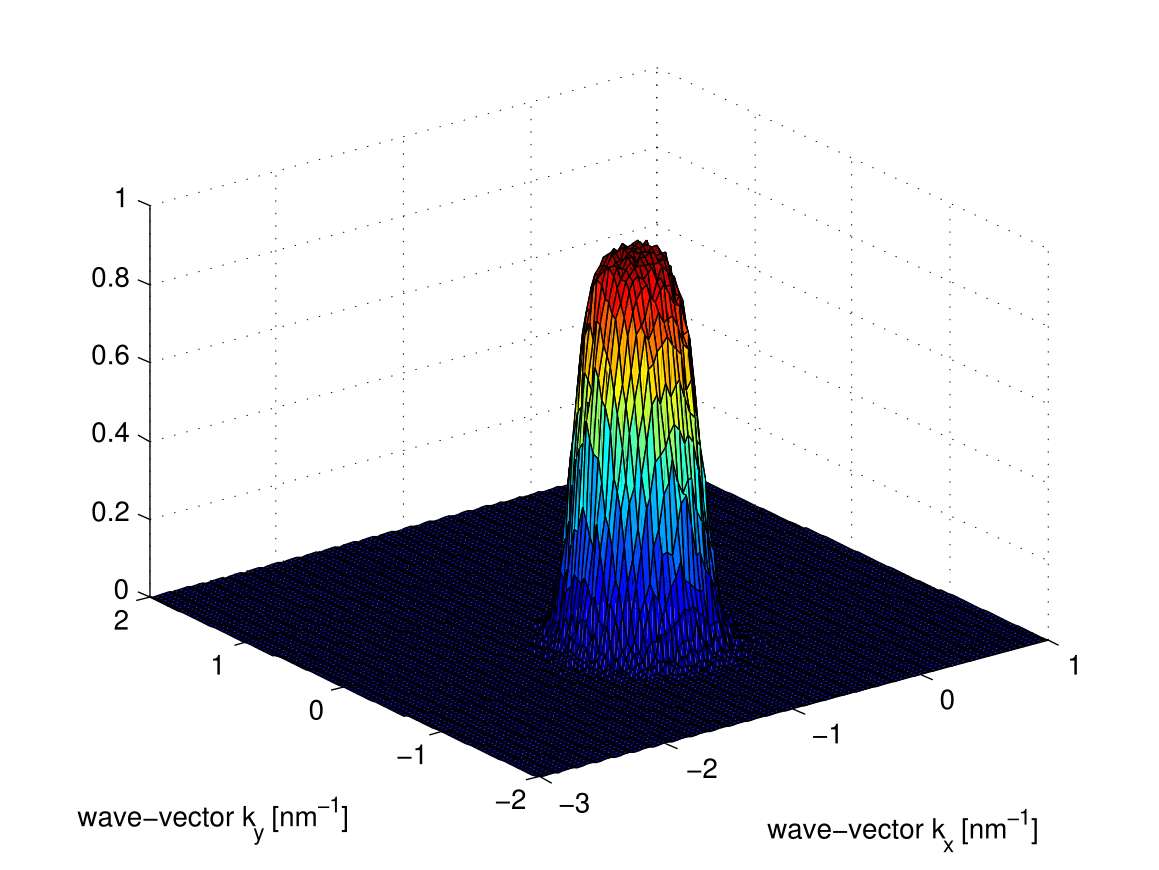}
 	\caption{Grafico della funzione di distribuzione, ottenuto con un campo elettrico applicato di $\si{\num{5}.\kilo\volt\per\centi\metre}$ e un livello di Fermi pari a $\si{\num{0.4}.\electronvolt}$.}
\end{figure}
\begin{figure}[ht]
	\centering
	\includegraphics[width=12cm]{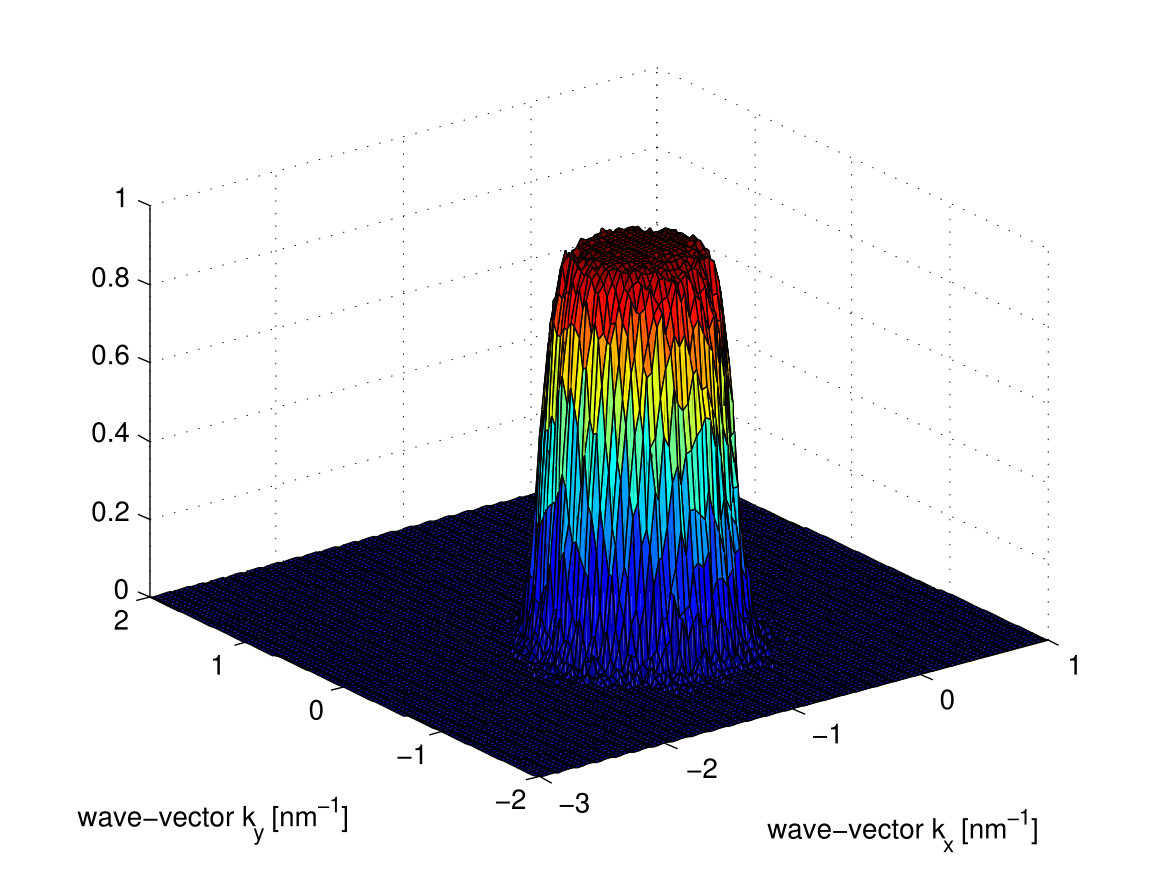}
 	\caption{Grafico della funzione di distribuzione, ottenuto con un campo elettrico applicato di $\si{\num{5}.\kilo\volt\per\centi\metre}$ e un livello di Fermi pari a $\si{\num{0.6}.\electronvolt}$.}
\end{figure}
\FloatBarrier
\begin{figure}[ht]
	\centering
	\includegraphics[width=12cm]{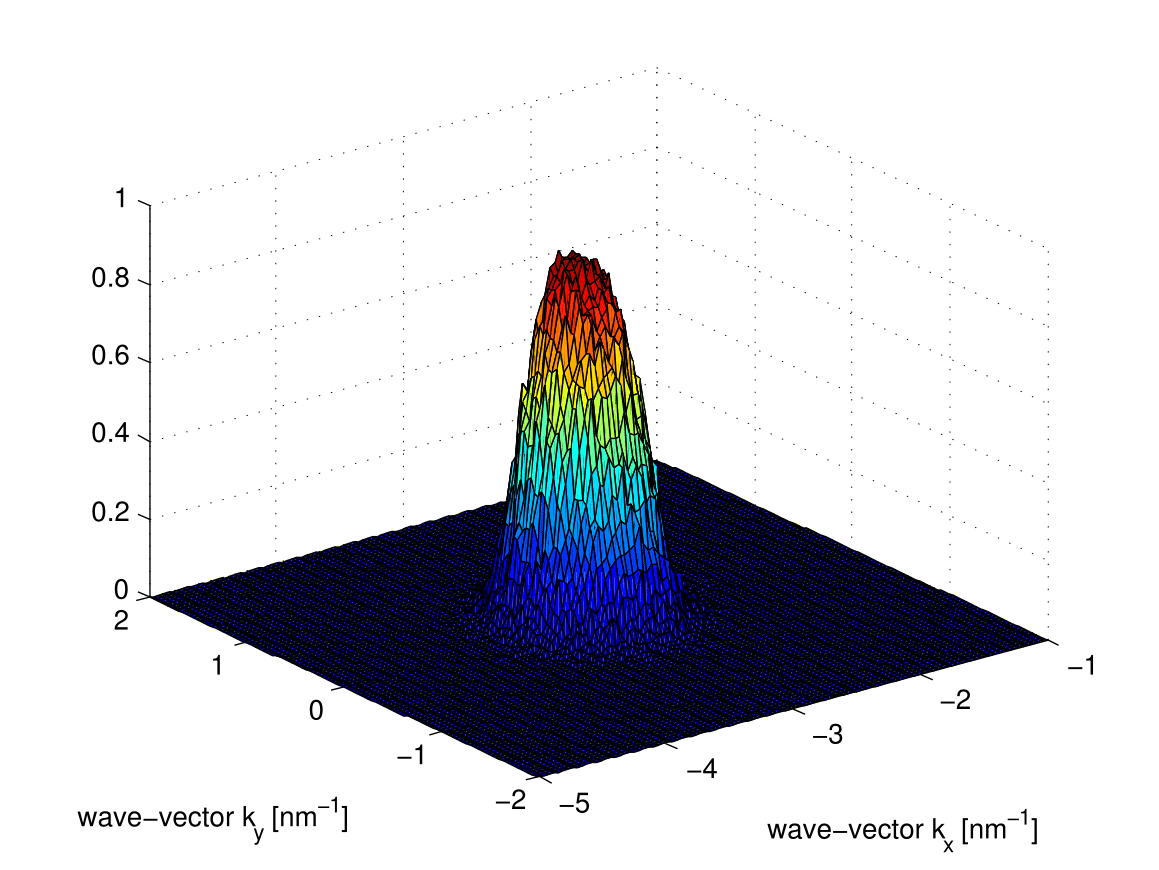}
 	\caption{Grafico della funzione di distribuzione, ottenuto con un campo elettrico applicato di $\si{\num{10}.\kilo\volt\per\centi\metre}$ e un livello di Fermi pari a $\si{\num{0.4}.\electronvolt}$.}
\end{figure}
\FloatBarrier

\subsection{Il problema del costo computazionale}
La prima implementazione dell'algoritmo discusso nei paragrafi precedenti è stata eseguita utilizzando il software di calcolo scientifico \textbf{MATLAB}. I codici scritti in linguaggio MATLAB hanno il vantaggio di essere piuttosto snelli, in quanto spesso non è richiesto di definire i tipi di variabili utilizzate e inoltre vi è un'ottima gestione di vettori e matrici, che permette la scrittura delle istruzioni in forma vettoriale ove necessario. MATLAB ha anche un gran numero di funzioni statistiche già implementate, il che è molto utile nell'ambito della programmazione Monte Carlo. Un altro punto di forza è la gestione dei numeri random. MATLAB ha già implementati diversi algoritmi per la generazione di numeri pseudocasuali che seguono le distribuzioni di probabilità più importanti. Tuttavia questa implementazione ha un costo computazionale piuttosto elevato. Questo riscontro ha fatto scaturire l'inizio di questo lavoro di tesi che è partito dalla ricerca di una implementazione più efficiente. Poiché nel codice non è richiesto un utilizzo avanzato di vettori e matrici né di funzioni particolarmente ostiche da implementare, si è optato per utilizzare un linguaggio di più basso livello. La scelta è ricaduta sul linguaggio \textbf{Fortran~90}, in quanto è ampiamente utilizzato per la programmazione scientifica e non è difficile trovare pacchetti che permettono di utilizzare funzioni la cui implementazione non è affatto banale. Inoltre la sintassi e lo stile di programmazione in certe circostanze è simile a MATLAB, il che ne ha resto più semplice il porting. La Tabella \ref{TAB:CAP3:Costo_comp} mostra il confronto tra i tempi di esecuzione dei due algoritmi al variare di alcuni parametri.
\FloatBarrier
\begin{table}[!ht]
\centering
\begin{tabular}{c|cc}
\textbf{N. particelle} & \textbf{MATLAB} & \textbf{Fortran 90} \\
\hline
$\si{\num{e3}}$ & $\si{\num{165}.\second}$ & $\si{\num{1.75}.\second}$ \\
$\si{\num{e4}}$ & $\si{\num{1853}.\second}$ & $\si{\num{6.59}.\second}$ \\
$\si{\num{e5}}$ & $>\si{\num{e4}.\second}$ & $\si{\num{55}.\second}$ \\
\end{tabular}
\caption{Confronto del tempo di esecuzioni tra il codice scritto in MATLAB e quello scritto in Fortran~90.}
\label{TAB:CAP3:Costo_comp}
\end{table}
\FloatBarrier
Il problema che si è presentato durante l'implementazione è stato quello di stabilire con certezza la validità del codice. Infatti, l'utilizzo di numeri pseudocasuali porta a risultati che sono inevitabilmente affetti da errore statistico. Di conseguenza, i profili risultanti dai vari test non risultano mai perfettamente identici. Anche se è possibile stabilire numericamente la bontà del codice, si è preferito adottare un confronto brutale forzando i due codici a fornire lo stesso output. L'algoritmo scelto per la generazione dei numeri pseudocasuali nell'implementazione MATLAB è quello di default e si chiama \textbf{Mersenne Twister}, come recita la documentazione del software\footnote{http://it.mathworks.com/help/matlab/ref/randstream.list.html}. Questo algoritmo è stato proposto da M.~Matsumoto e T.~Nishimura nel 1998 (cfr. \citep{ART:Matsumoto_Nishimura}) e fornisce un periodo elevatissimo pari a $2^{19937}-1$. Per questa ragione risulta molto efficace nelle tecniche di simulazione Monte Carlo. Di questo algoritmo esistono molte implementazioni nei linguaggi di programmazione più diffusi, tra cui Fortran~90, e vengono fornite dagli stessi autori\footnote{http://www.math.sci.hiroshima-u.ac.jp/~m-mat/MT/VERSIONS/eversions.html}. Inoltre occorre sincronizzare i due generatori sullo stesso seme. Per una buona inizializzazione, gli autori suggeriscono di utilizzare il valore 5489, come risulta da una implementazione fornita dagli autori in linguaggio C\footnote{http://www.math.sci.hiroshima-u.ac.jp/~m-mat/MT/MT2002/emt19937ar.html}. Tale valore è anche quello utilizzato di default da MATLAB\footnote{http://it.mathworks.com/help/matlab/math/updating-your-random-number-generator-syntax.html}. Pertanto, nell'implementazione in Fortran~90 si è utilizzato il generatore \texttt{mt95.f90}, reperibile sul sito web degli sviluppatori di Mersenne Twister\footnote{http://www.math.sci.hiroshima-u.ac.jp/~m-mat/MT/VERSIONS/FORTRAN/mt95.f90}. Mediante questa tecnica le due implementazioni utilizzano esattamente la stessa sequenza di numeri random, pertanto i valori in output delle simulazioni presentano una massima discrepanza tra l'implementazione in MATLAB e quella in Fortran~90 che è inferiore a $\si{\num{e-7}}$.


\section{Il trasporto di cariche nel grafene su substrato}
La descrizione fin qui trattata è relativa al caso del grafene sospeso. Lo studio del grafene in questa forma è utile per comprendere a fondo le proprietà di trasporto di questo materiale. Nelle situazioni reali, il foglio di grafene è posizionato su un substrato ossido che si comporta come una sorgente aggiuntiva di scattering tra gli elettroni del grafene responsabili della conduzione e le impurità del substrato (cfr.~\citep{ART:MajMaRo}).  

In questa sezione si descriverà dal punto di vista modellistico il caso del grafene su substrato. Prendendo come punto di partenza i risultati ottenuti con il metodo DSMC (cfr.~\citep{ART:CoMajRo_Ric_mat_2016}), si è voluto utilizzare il codice prodotto per questa tesi effettuando delle simulazioni in cui sono state fatte delle ipotesi sulla distribuzione delle impurezze. Infine, le stesse simulazioni sono state ripetute facendo variare il tipo di substrato.

\subsection{Descrizione}
Al fine di migliorare le proprietà di conduzione dei semiconduttori, una tecnica molto diffusa consiste nell'aggiungere nel materiale una piccola percentuale di atomi di natura diversa, dette \textbf{impurezze}. L'inserimento di impurezze in un semiconduttore viene detto \textbf{drogaggio} e si può effettuare, ad esempio, mediante impiantazione ionica, che consiste nell'introduzione di atomi carichi ad elevata energia in un substrato, o mediante diffusione nel campione di droganti allo stato gassoso.

Più in dettaglio, il \textbf{drogaggio di tipo n} consiste nell'utilizzare come impurezze degli atomi donori. In questo modo si crea un eccesso di elettroni liberi nel materiale che contribuiscono alla banda di conduzione. Invece, il \textbf{drogaggio di tipo p} consiste nell'aggiungere atomi accettori. Si crea così un eccesso di lacune nella banda di valenza.

Di conseguenza, nel caso del grafene su substrato, oltre ai meccanismi di scattering già presentanti nel caso sospeso, è necessario includere gli effetti dovuti ai fononi e alle impurità del substrato. Si può vedere una schematizzazione della situazione fisica in Figura~\ref{FIG:CAP2:graphene_substrate}.
\begin{figure}[ht]
\centering
\includegraphics[width=0.5\columnwidth]{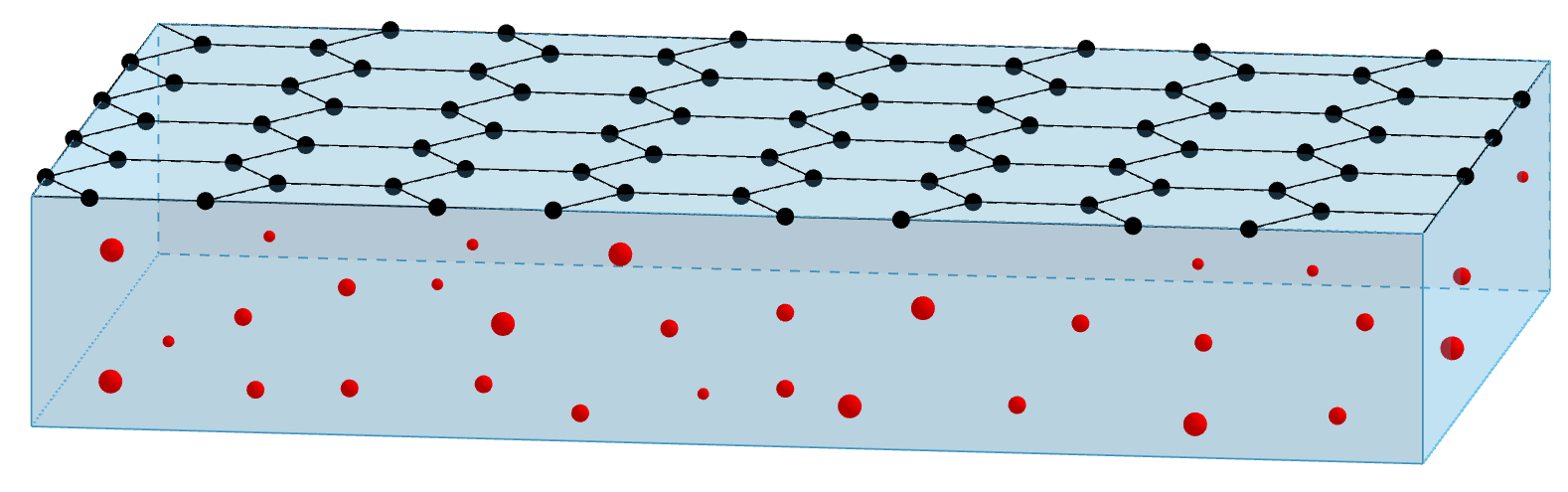}
\caption{Foglio di grafene su substrato.}
\label{FIG:CAP2:graphene_substrate}
\end{figure}
\FloatBarrier
Il campione è realizzato posizionando un foglio di grafene su un substrato e, dopo averlo portato a temperatura elevata, si osserva una riorganizzazione del materiale. Gli elettroni del foglio di grafene interagiscono con i fononi e con le impurità del substrato. Le interazioni tra gli elettroni del grafene e i fononi ottici del substrato avvengono seguendo le probabilità di transizione della forma \eqref{EQ:CAP2:Gamma_LO} e \eqref{EQ:CAP2:Gamma_TO}. Per quanto riguarda le interazioni tra elettroni del grafene e impurezze del substrato, si assume che esse si trovino su un piano a distanza $d$ dal foglio di grafene. Definire le probabilità di transizione in questo caso è piuttosto complicato. Senza entrare nel dettaglio, si adotta la seguente forma per le frequenze di collisione,
\begin{equation}
S\ap{(imp)}(\mathbf{k},\mathbf{k}')=\frac{2\pi}{\hbar}\frac{n_i}{(2\pi)^2} \left\vert \frac{V_i(\left\vert \mathbf{k}-\mathbf{k}',d \right\vert)}{\varepsilon (\left\vert \mathbf{k}-\mathbf{k}' \right\vert)} \right\vert^2 \frac{1+\cos \theta_{\mathbf{k},\mathbf{k}'}}{2}\delta \left( \varepsilon(\mathbf{k}')-\varepsilon(\mathbf{k}) \right),
\end{equation}
dove $n_i$ è il numero di impurità per unità di area. $V_i$ è dato da
\begin{equation}
V_i(\left\vert \mathbf{k}-\mathbf{k}',d \right\vert) = 2\pi e^2 \frac{\exp (-d \left\vert \mathbf{k}-\mathbf{k}' \right\vert )}{\tilde{\kappa} \left\vert \mathbf{k}-\mathbf{k}' \right\vert},
\end{equation}
in cui $d$ è la distanza a cui si trovano le impurità e $\tilde{\kappa}$ è la costante dielettrica effettiva, definita da $4\pi\varepsilon_0 (\kappa\ped{top}+\kappa\ped{bottom})/2$, dove $\varepsilon_0$ è la costante dielettrica nel vuoto, $\kappa\ped{top}$ e $\kappa\ped{bottom}$ sono le costanti dielettriche relative del mezzo al di sopra e al di sotto dello strato di grafene. Infine la funzione dielettrica $\varepsilon (\left\vert \mathbf{k}-\mathbf{k}' \right\vert)$ è data da
\begin{equation}
\varepsilon (\left\vert \mathbf{k}-\mathbf{k}' \right\vert)= 
\begin{cases}
1+\frac{q_s}{\left\vert \mathbf{k}-\mathbf{k}' \right\vert}-\frac{\pi q_s}{8 k_F} & \mbox{se } \left\vert \mathbf{k}-\mathbf{k}' \right\vert<2k_F\\
1+\frac{q_s}{\left\vert \mathbf{k}-\mathbf{k}' \right\vert} - \frac{q_s \sqrt{\left\vert \mathbf{k}-\mathbf{k}' \right\vert^2 -4k_F^2}}{2\left\vert \mathbf{k}-\mathbf{k}' \right\vert^2}-\frac{q_s}{4k_F}\arcsin \left( \frac{2k_F}{\left\vert \mathbf{k}-\mathbf{k}' \right\vert} \right) & \mbox{altrimenti}
\end{cases}
\end{equation}
Inoltre
\begin{equation}
q_s = \frac{4e^2 k_F}{\tilde{\kappa}\hbar v_F}
\end{equation}
è detto vettore d'onda di Thomas-Fermi per il grafene e
\begin{equation}
k_F = \frac{\varepsilon_F}{\hbar v_F}
\end{equation}
è detto vettore d'onda di Fermi.

La frequenza di collisione per le impurità deve essere valutata numericamente calcolando il seguente integrale,
\begin{equation}
\Gamma\ped{imp}(\varepsilon) = \int \! S\ap{(imp)}(\mathbf{k},\mathbf{k}')(1-\cos\theta_{\mathbf{k},\mathbf{k}'}) \, \diff \mathbf{k}'.
\end{equation}
Per quanto riguarda la mobilità elettronica, dall'analisi della velocità media è possibile stimare gli effetti delle impurezze sulla mobilità. Ci si aspetta che gli scattering con tali impurezza conducano a un decadimento della mobilità più o meno accentuato in base al materiale.

\subsection{Risultati numerici su \texorpdfstring{$\ce{SiO2}$}{SiO2}}\label{PAR:CAP3:Subs_1}
Il substrato utilizzato per le simulazioni in \citep{ART:CoMajRo_Ric_mat_2016} è il diossido di silicio ($\ce{SiO2}$). I valori delle costanti fisiche utilizzate sono, oltre a quelli riportati nella Tabella \ref{TAB:CAP3:Cost_fis_sospeso}, il valore dell'energia delle impurezze pari a $\hbar\omega\ped{ac-op}=\si{\num{55}.\milli\electronvolt}$ e il potenziale di deformazione pari a $D_f = \si{\num{5.14e7}.\electronvolt\per\centi\metre}$. Inoltre la temperatura scelta è quella ambiente pari a $\si{\num{300}.\kelvin}$, la densità superficiale delle impurezze è $n_i=\si{\num{2.5e11}.\centi\metre^{-2}}$ ed il numero di particelle è pari a $n_P=\si{\num{e4}}$. Le simulazioni seguenti sono eseguite al variare del campo e elettrico e del livello di Fermi. Inoltre è effettuato un confronto fissando di volta in volta dei valori costanti per la distanza delle impurezze $d$ pari a $\si{\num{0}.\nano\meter}$, $\si{\num{0.5}.\nano\meter}$ e $\si{\num{1}.\nano\meter}$.
\FloatBarrier
\begin{figure}[ht]
	\centering
	\subfigure[]
   		{\includegraphics[width=7.1cm]{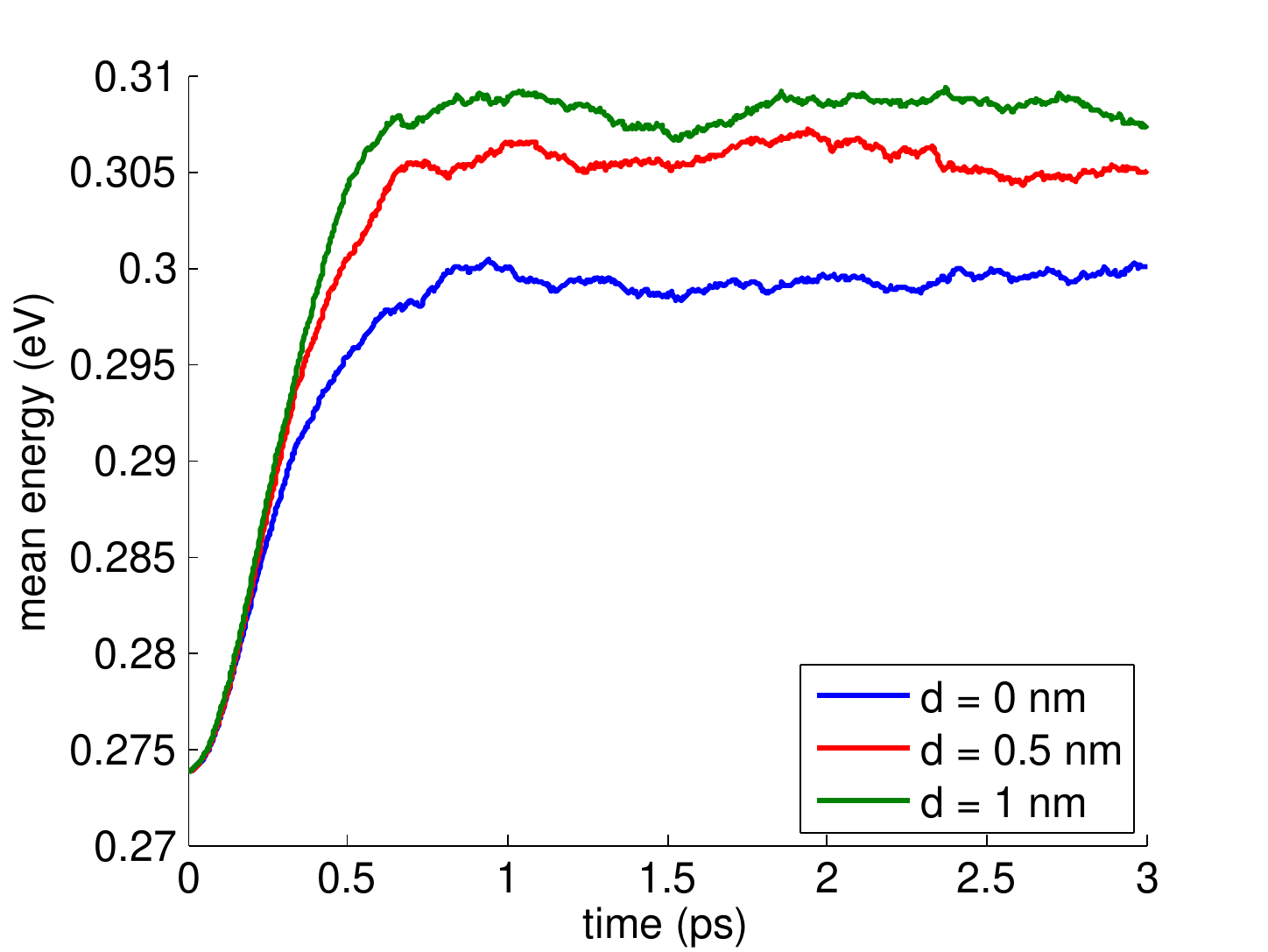}}
 	\,
 	\subfigure[]
   		{\includegraphics[width=7.1cm]{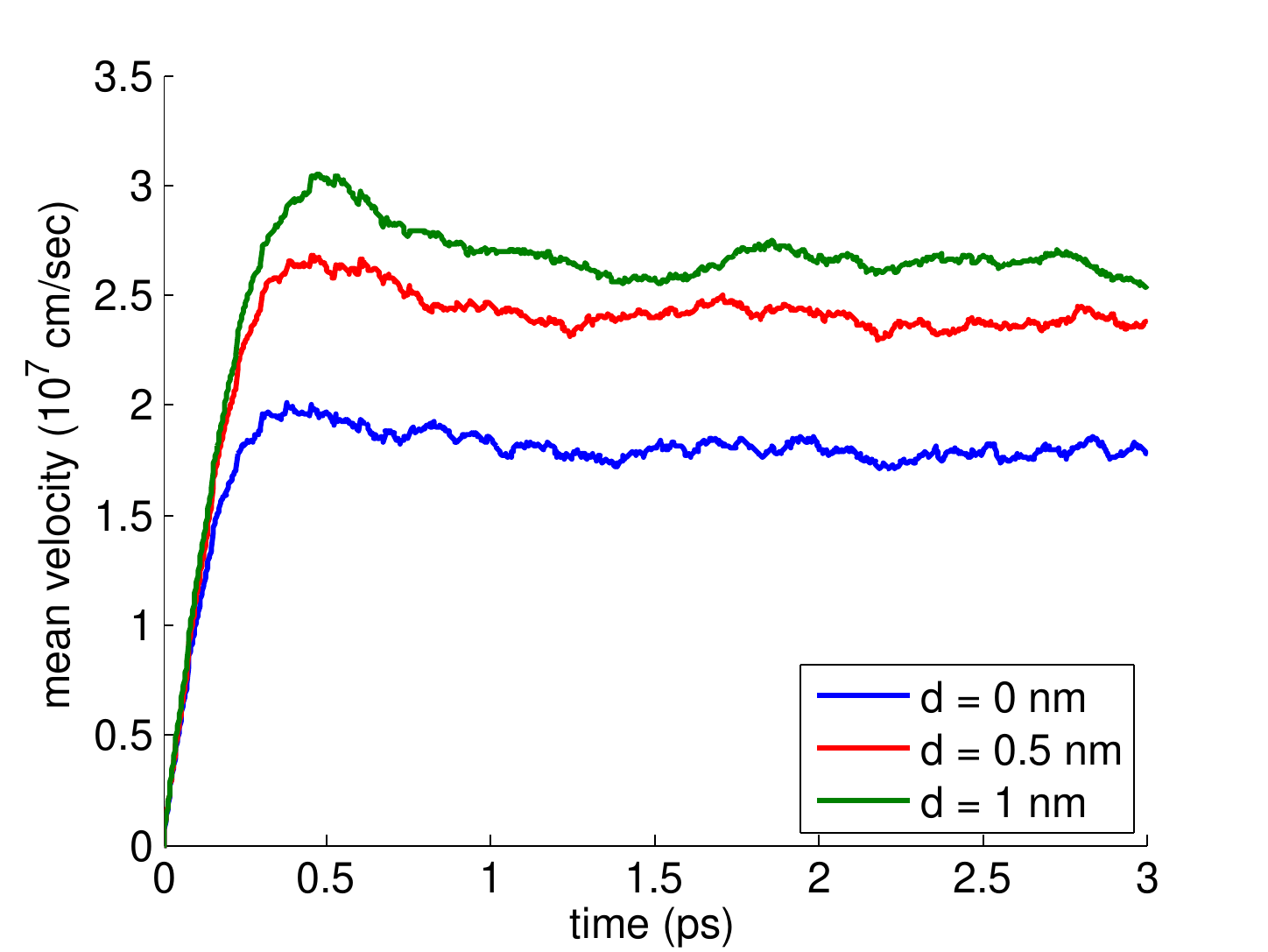}}
 	\caption{Energia (a) e velocità (b) medie ottenute con un campo elettrico applicato di $\si{\num{5}.\kilo\volt\per\centi\metre}$ e un livello di Fermi pari a $\si{\num{0.4}.\electronvolt}$.}\label{FIG:CAP3:DSMC_new_imp_1}
\end{figure}
\begin{figure}[ht]
	\centering
	\subfigure[]
   		{\includegraphics[width=7.1cm]{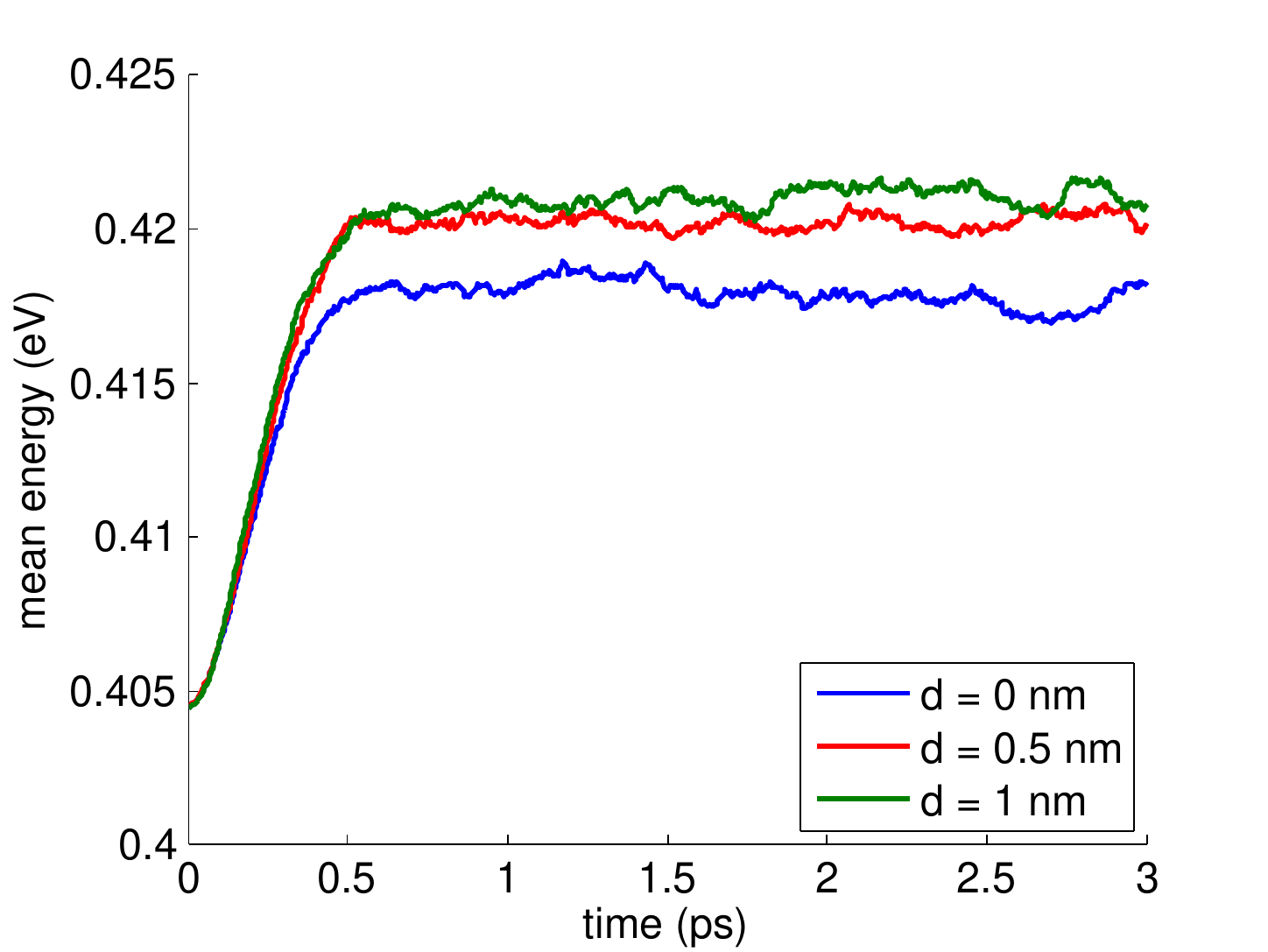}}
 	\,
 	\subfigure[]
   		{\includegraphics[width=7.1cm]{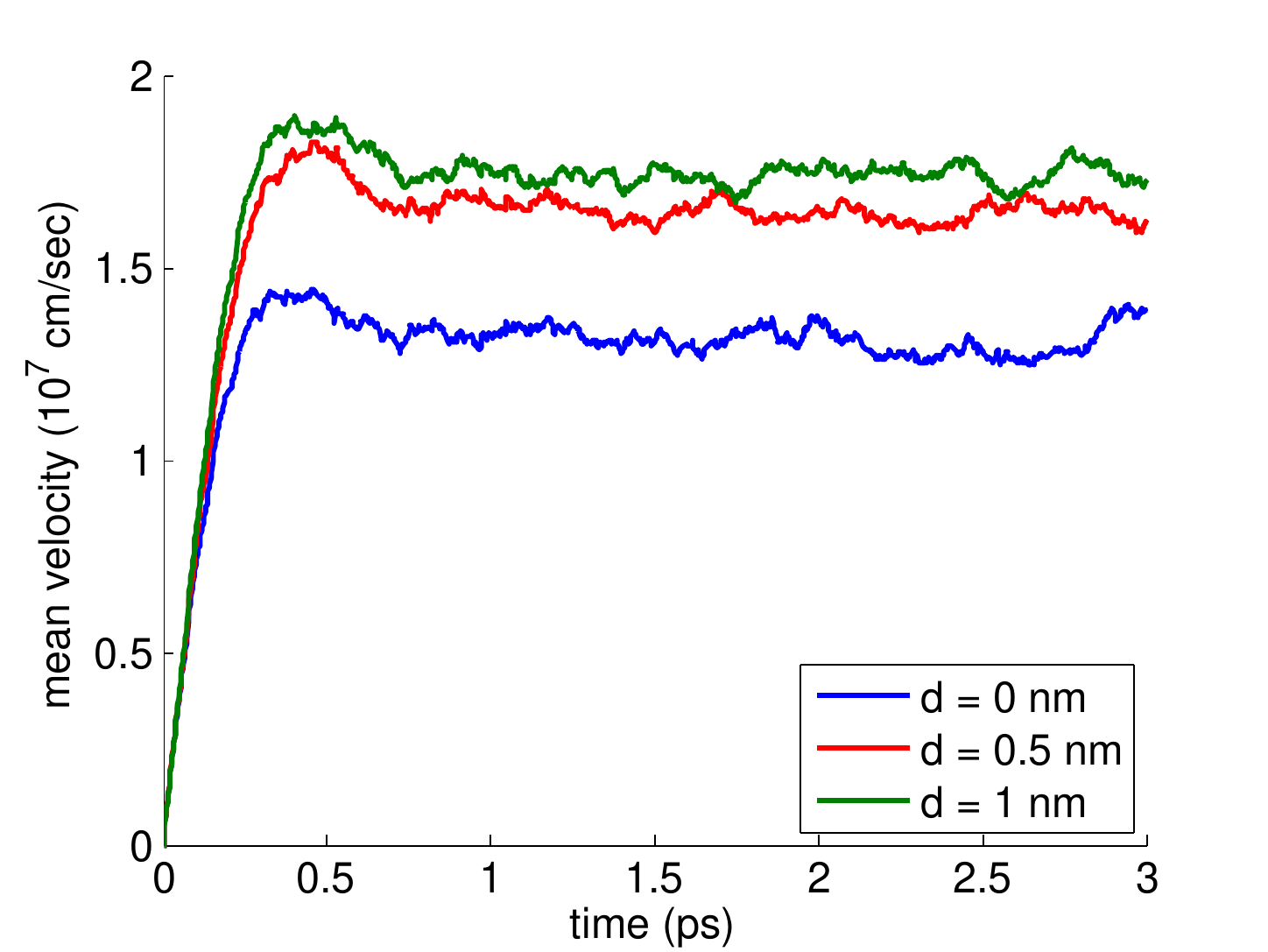}}
 	\caption{Energia (a) e velocità (b) medie ottenute con un campo elettrico applicato di $\si{\num{5}.\kilo\volt\per\centi\metre}$ e un livello di Fermi pari a $\si{\num{0.6}.\electronvolt}$.}\label{FIG:CAP3:DSMC_new_imp_2}
\end{figure}
\begin{figure}[ht]
	\centering
	\subfigure[]
   		{\includegraphics[width=7.1cm]{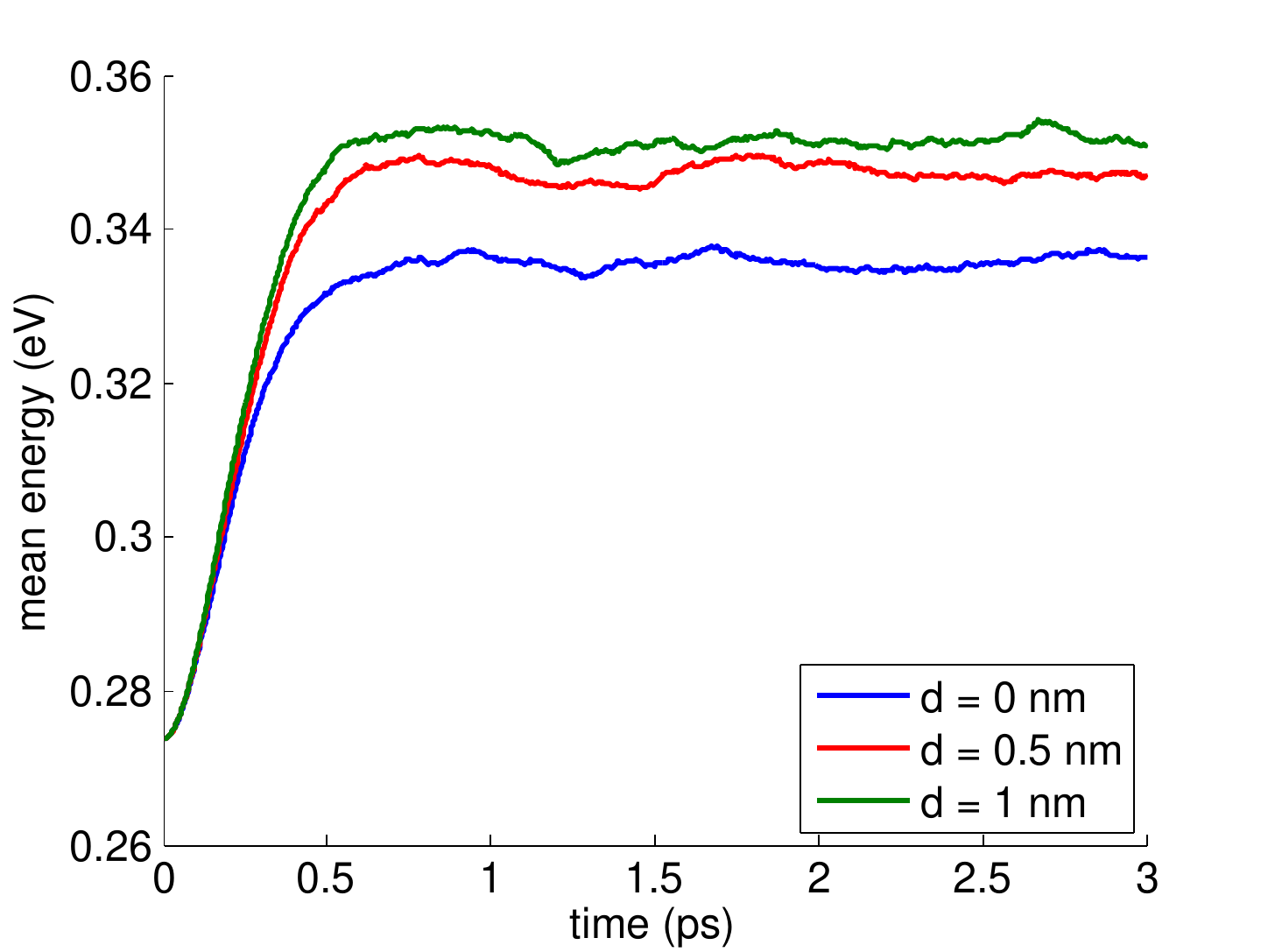}}
 	\,
 	\subfigure[]
   		{\includegraphics[width=7.1cm]{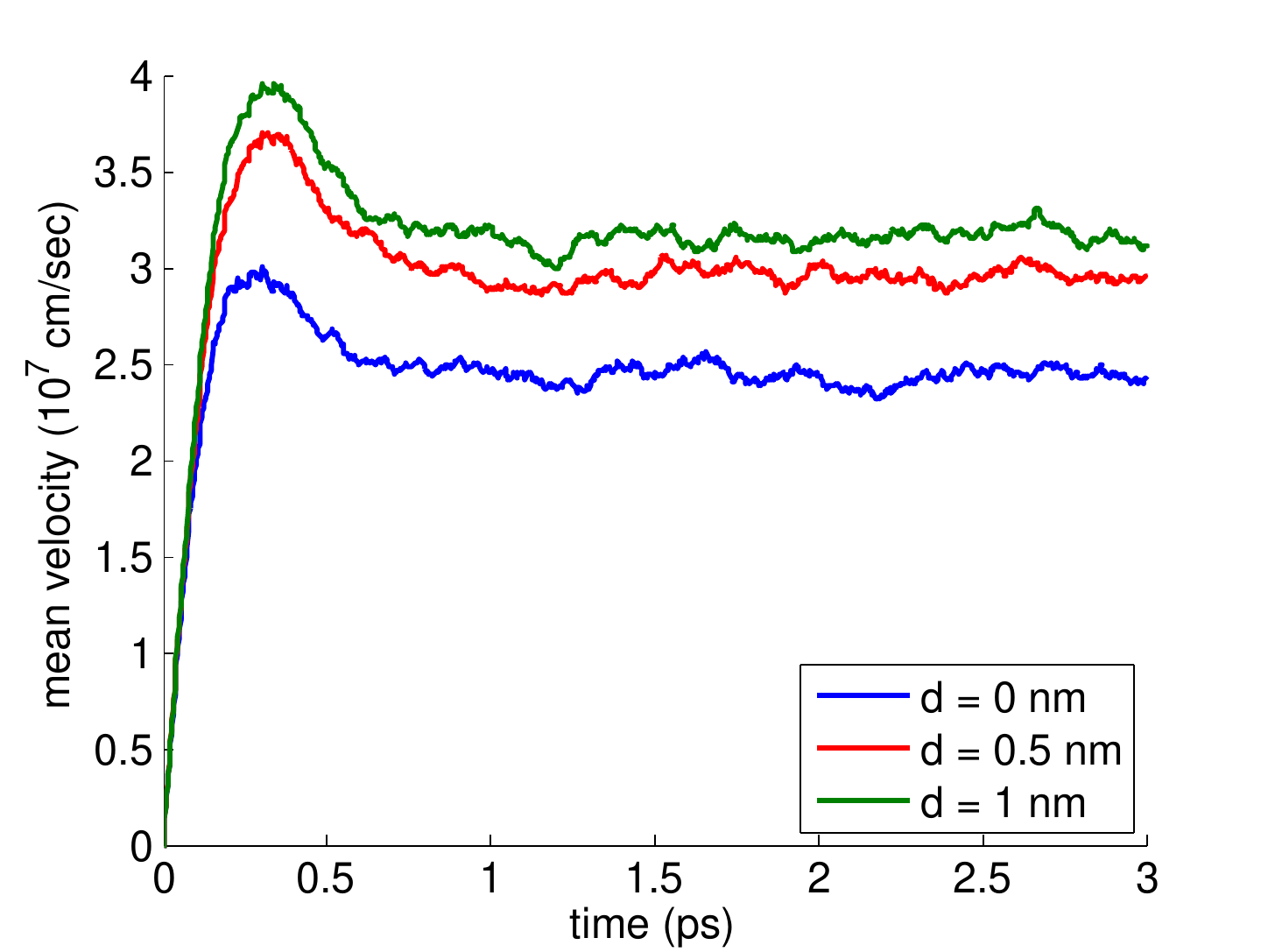}}
 	\caption{Energia (a) e velocità (b) medie ottenute con un campo elettrico applicato di $\si{\num{10}.\kilo\volt\per\centi\metre}$ e un livello di Fermi pari a $\si{\num{0.4}.\electronvolt}$.}\label{FIG:CAP3:DSMC_new_imp_3}
\end{figure}
\FloatBarrier

\subsection{Ipotesi modellistiche sulla distanza delle impurezze}\label{PAR:CAP3:Subs_2}
Il parametro più importante è la distanza $d$ a cui si trovano le impurità del substrato rispetto al foglio di grafene. Essa è dell'ordine di alcuni angstrom ma il valore esatto può variare da un campione all'altro. Per questa ragione si è voluto complicare il modello, facendo delle ipotesi sulla distribuzione di queste distanze.

Si è visto che i valori di $d$ scelti nel caso costante sono pari a $\si{\num{0}.\nano\meter}$, $\si{\num{0.5}.\nano\meter}$ e $\si{\num{1}.\nano\meter}$. Si è pensato allora in primo luogo di sostituire la distanza costante con una variabile aleatoria distribuita uniformemente tra $\si{\num{0}.\nano\meter}$ e $\si{\num{1}.\nano\meter}$. Cioè, si assume che $d$ abbia la forma,
\begin{equation}
d = \tilde{d}\cdot d^*, \qquad \mbox{con } \tilde{d}\sim U([0,1]), \quad d^*=\si{\num{1}.\nano\meter}.
\end{equation}
Il valore atteso per la variabile aleatoria $\tilde{d}$ è pari a $1/2$. Pertanto ci si aspetta che i risultati non siano molto distanti da quelli ottenuti con $d=\si{\num{0.5}.\nano\meter}$, come in effetti avviene. Infatti, i grafici riportati di seguito mostrano il confronto tra i valori medi dell'energia e della velocità ottenuti con $d$ costante e quelli ottenuti con quest'ultima forma per $d$.
\FloatBarrier
\begin{figure}[ht]
	\centering
	\subfigure[]
   		{\includegraphics[width=7.1cm]{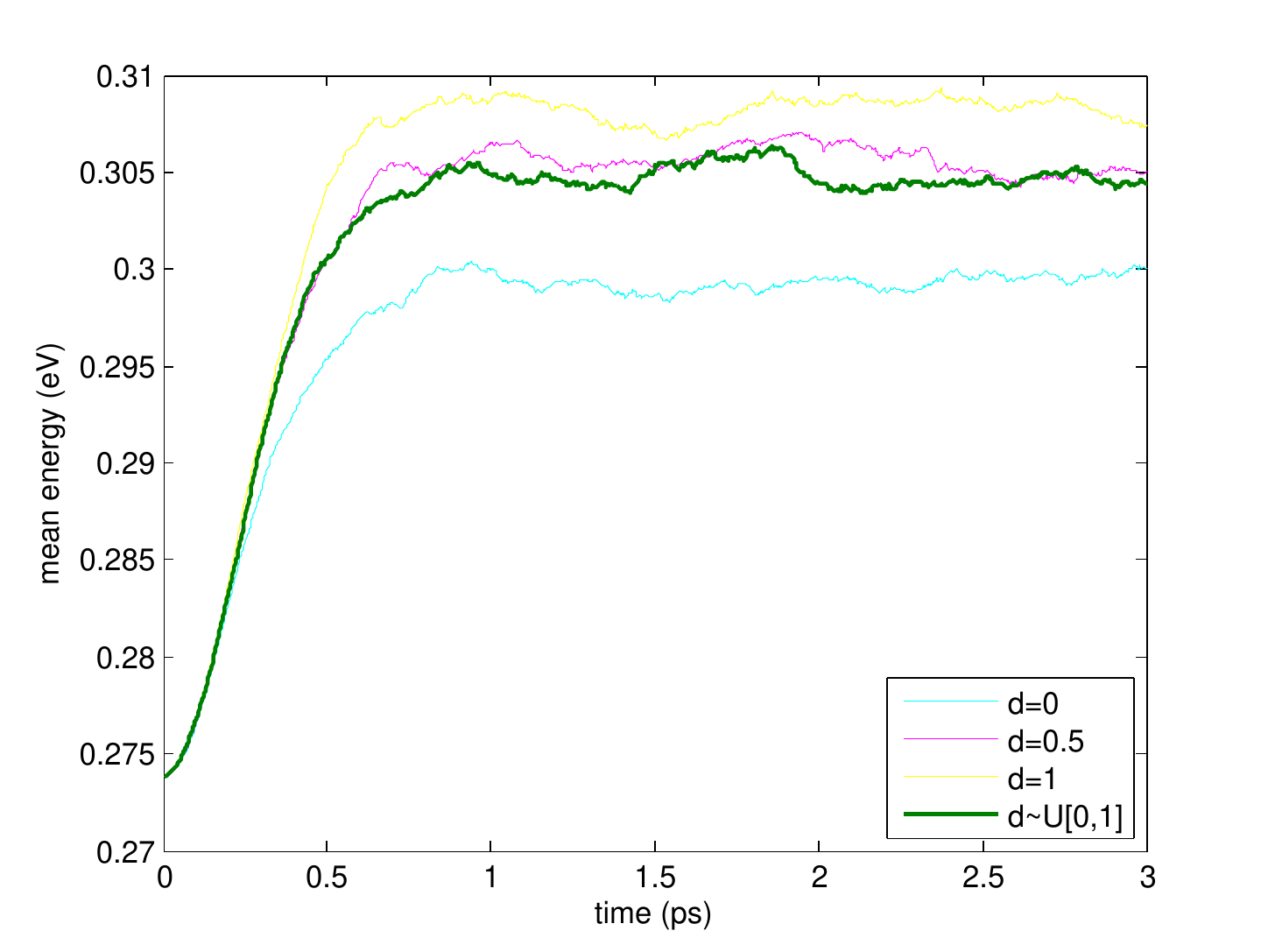}}
 	\,
 	\subfigure[]
   		{\includegraphics[width=7.1cm]{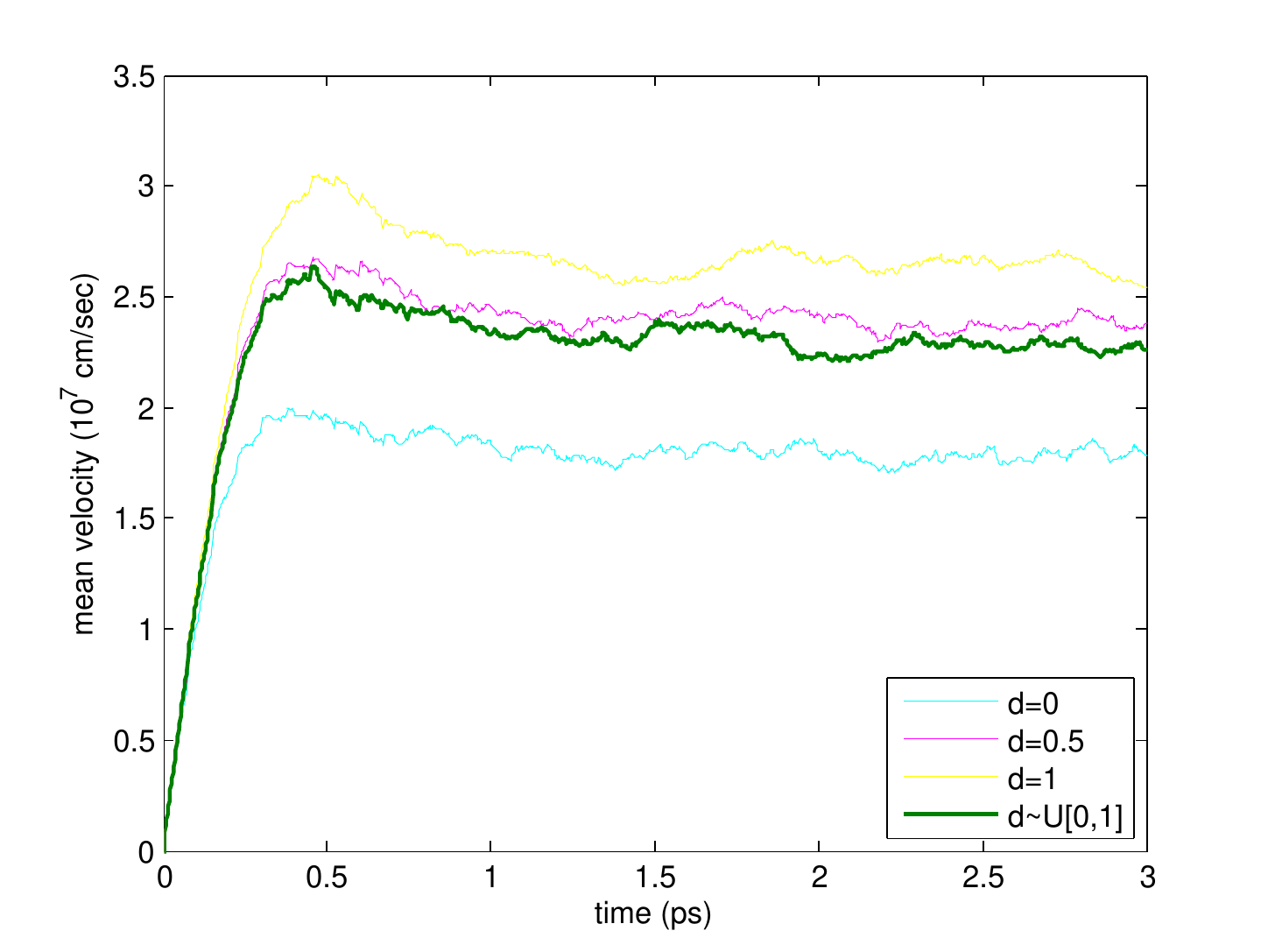}}
 	\caption{Energia (a) e velocità (b) medie ottenute con un campo elettrico applicato di $\si{\num{5}.\kilo\volt\per\centi\metre}$, un livello di Fermi pari a $\si{\num{0.4}.\electronvolt}$ e $d\sim U([0,1])$.}\label{FIG:CAP3:DSMC_new_dunif_1}
\end{figure}
\begin{figure}[ht]
	\centering
	\subfigure[]
   		{\includegraphics[width=7.1cm]{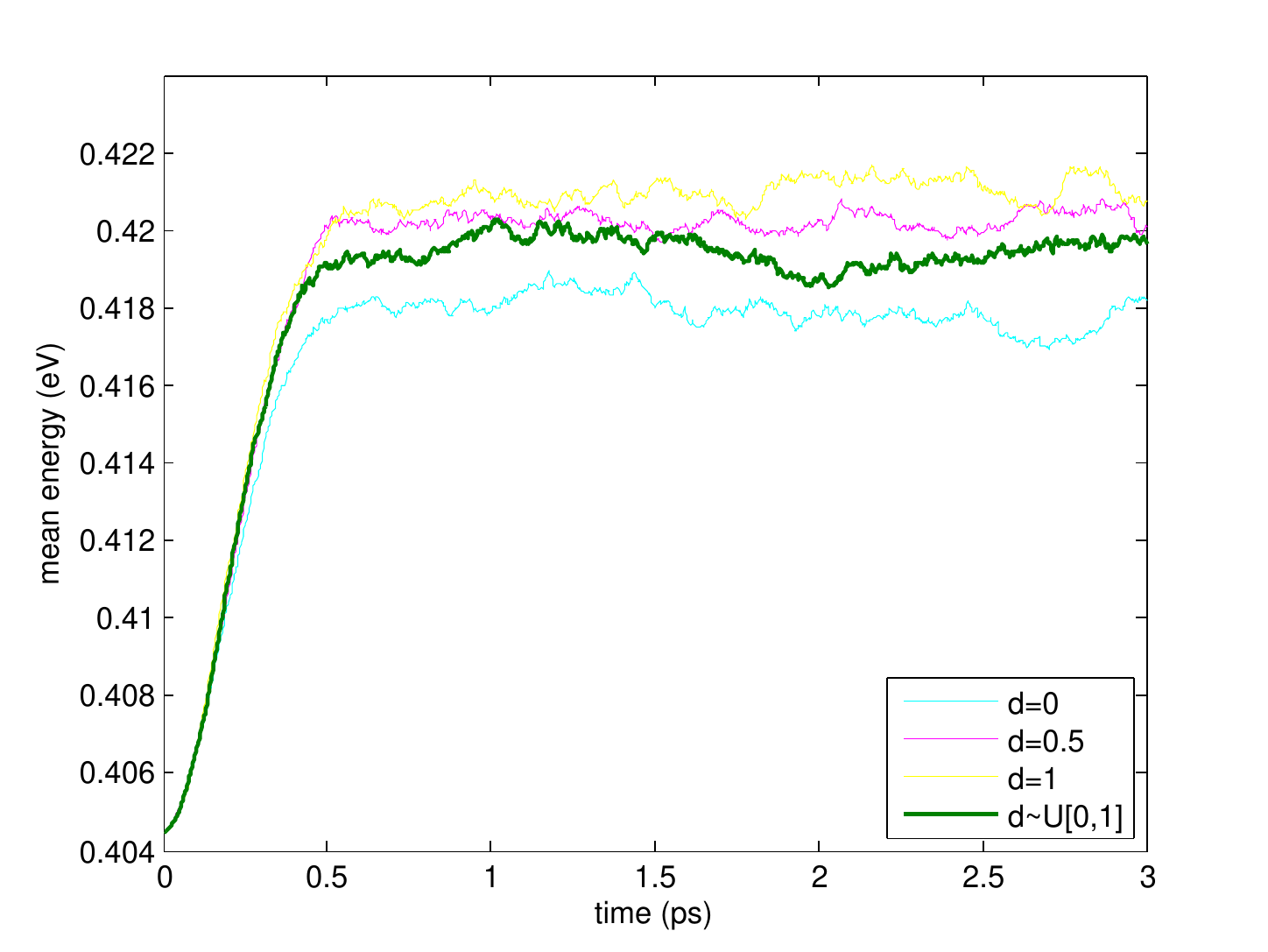}}
 	\,
 	\subfigure[]
   		{\includegraphics[width=7.1cm]{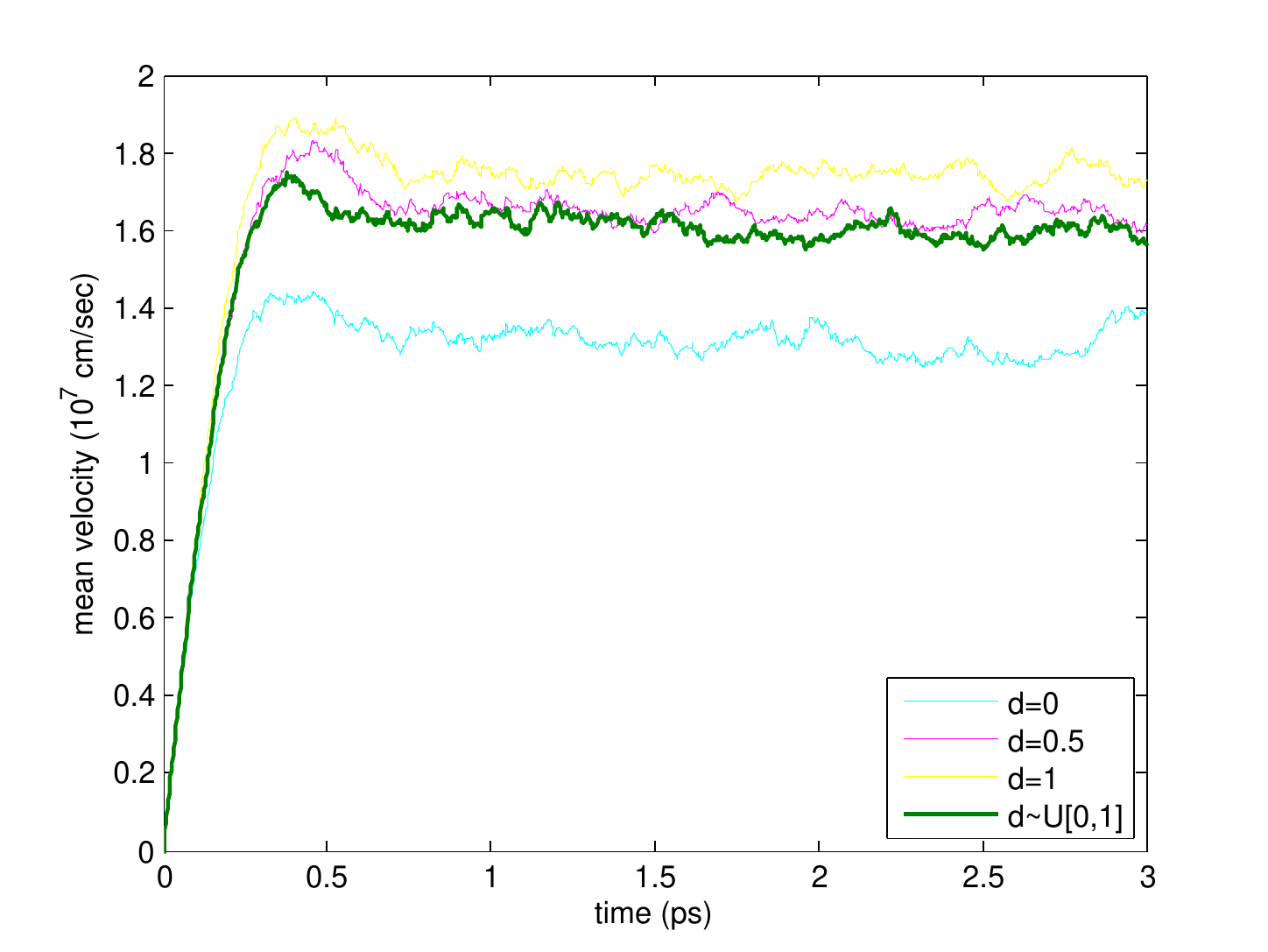}}
 	\caption{Energia (a) e velocità (b) medie ottenute con un campo elettrico applicato di $\si{\num{5}.\kilo\volt\per\centi\metre}$, un livello di Fermi pari a $\si{\num{0.6}.\electronvolt}$ e $d\sim U([0,1])$.}\label{FIG:CAP3:DSMC_new_dunif_2}
\end{figure}
\begin{figure}[ht]
	\centering
	\subfigure[]
   		{\includegraphics[width=7.1cm]{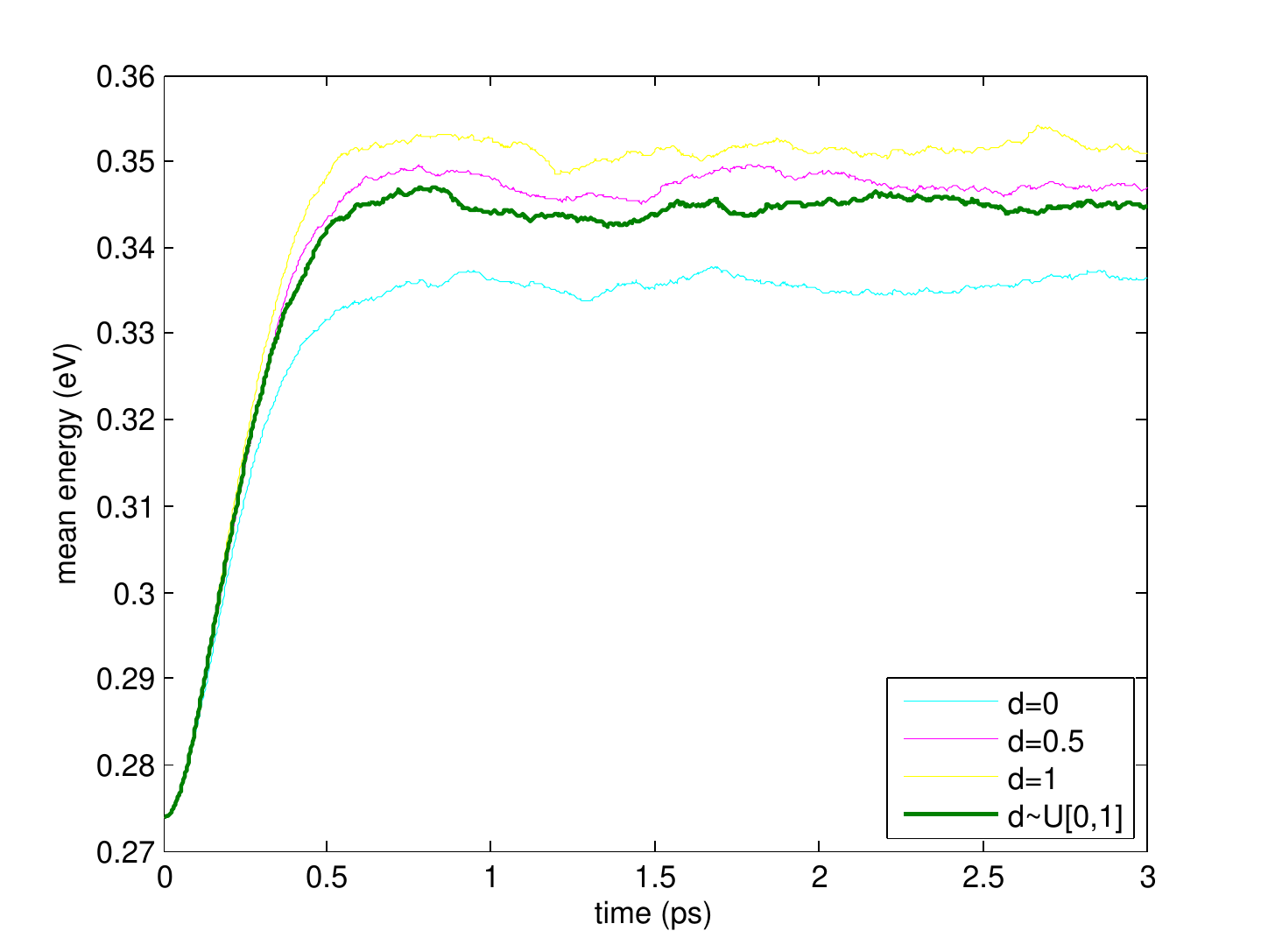}}
 	\,
 	\subfigure[]
   		{\includegraphics[width=7.1cm]{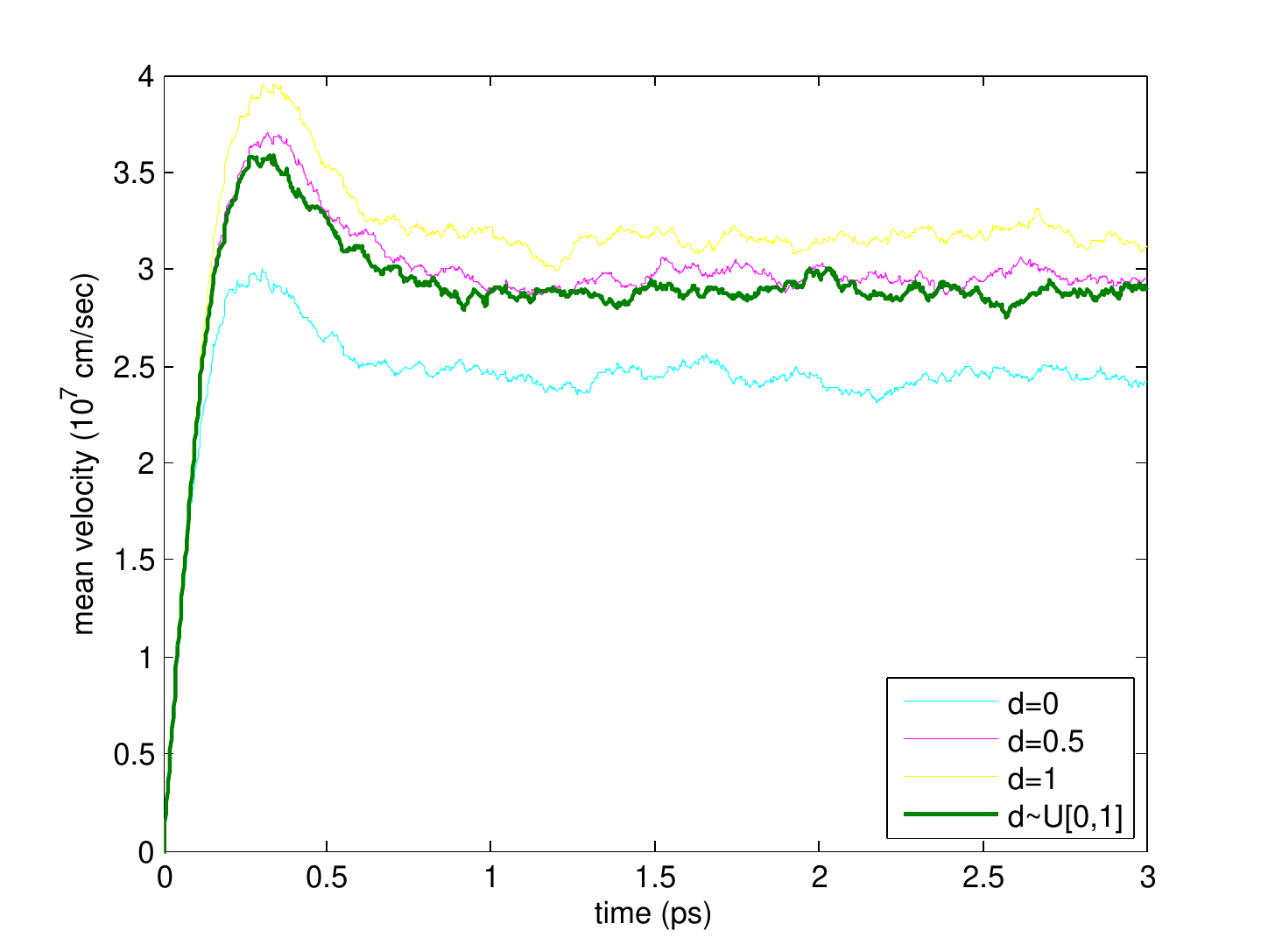}}
 	\caption{Energia (a) e velocità (b) medie ottenute con un campo elettrico applicato di $\si{\num{10}.\kilo\volt\per\centi\metre}$, un livello di Fermi pari a $\si{\num{0.4}.\electronvolt}$ e $d\sim U([0,1])$.}\label{FIG:CAP3:DSMC_new_dunif_3}
\end{figure}
\FloatBarrier
Successivamente si è pensato di adottare una scelta più sofisticata per la distanza $d$. Si è voluta adottare una distribuzione che abbia una densità di valori grossomodo tra $\si{\num{0}.\nano\meter}$ e $\si{\num{1}.\nano\meter}$ con un picco in prossimità dello strato di grafene. La scelta è ricaduta nella famiglia delle distribuzioni \textbf{Gamma}, di cui si dà una breve descrizione.

Si dice che una variabile aleatoria $X$ segue una distribuzione $\Gamma(k,\theta)$, dove $k>0$ è detto \textbf{parametro di forma} e $\theta>0$ è detto \textbf{parametro di scala}, se ha densità
\begin{equation}
f(x) = \frac{1}{\Gamma(k)\theta^k}x^{k-1}e^{-\frac{x}{\theta}}, \qquad \mbox{per }x>0
\end{equation}
e $f(x) =0$ per $x\leq 0$, dove $\Gamma(k)$ è data da
\begin{equation}
\Gamma(k) = \int_0^{+\infty} \! x^{k-1}e^{-x} \, \diff x.
\end{equation}
Inoltre se $k\in\mathbb{N}^{+}$ si ha $\Gamma(k)=(k-1)!$.



Nel problema in esame si sono presi per il parametro $k$ i valori $k=2,3,4$ e per il parametro $\theta$ il valore 2. Si osservi che con questa scelta vi è una probabilità non nulla di ottenere valori per la variabile aleatoria per ogni $x>0$. Tuttavia, si può calcolare che la probabilità di trovarne al di fuori dell'intervallo $[0,5]$ è di circa l'$1\%$, come si evince dalla Figura \ref{FIG:CAP2:Funz_rip_Gamma}. Pertanto, al fine di scegliere come scala per la distanza il valore di $\si{\num{1}\nano\metre}$, si è scelto di prendere
\begin{equation}
d = \tilde{d}\cdot d^*/5, \qquad \mbox{con } \tilde{d}\sim\Gamma(k,2)/5, \quad k=2,3,4, \quad d^*=\si{\num{1}.\nano\meter}.
\end{equation}
\FloatBarrier
\begin{figure}[ht]
\centering
\includegraphics[width=0.5\columnwidth]{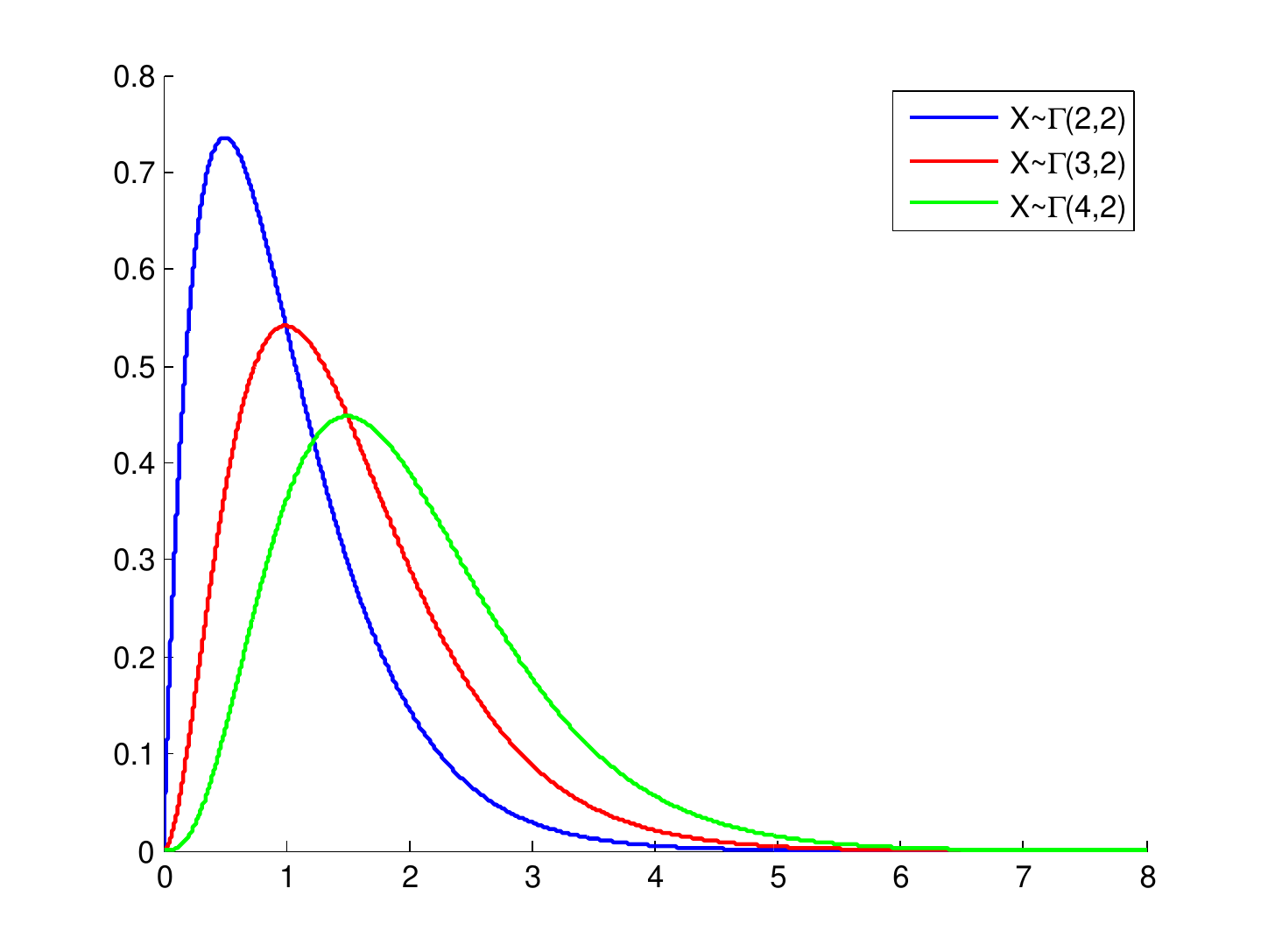}
\caption{Grafico delle densità delle distribuzioni $\Gamma(k,2)$ con $k=2,3,4$.}
\label{FIG:CAP2:Funz_rip_Gamma}
\end{figure}
\FloatBarrier
I grafici riportati di seguito mostrano il confronto tra i valori medi dell'energia e della velocità ottenuti con $d$ costante e quelli ottenuti utilizzando le distribuzioni appena menzionate. Infine, nella Figura \ref{FIG:CAP3:DSMC_new_SiO2_distrib} si evince ancora una volta che le simulazioni rispettano il principio di esclusione di Pauli.
\FloatBarrier
\begin{figure}[ht]
	\centering
	\subfigure[]
   		{\includegraphics[width=7.1cm]{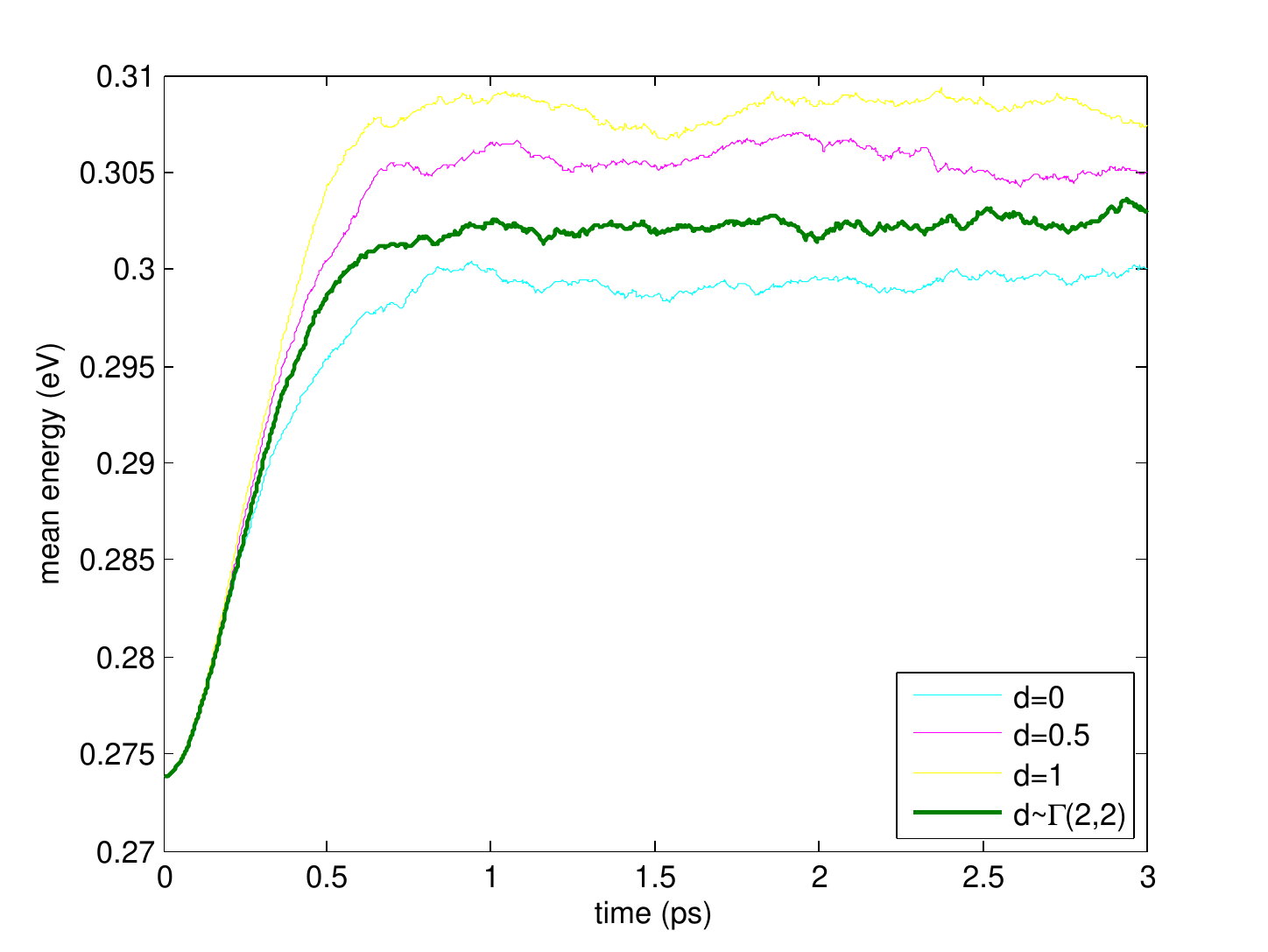}}
 	\,
 	\subfigure[]
   		{\includegraphics[width=7.1cm]{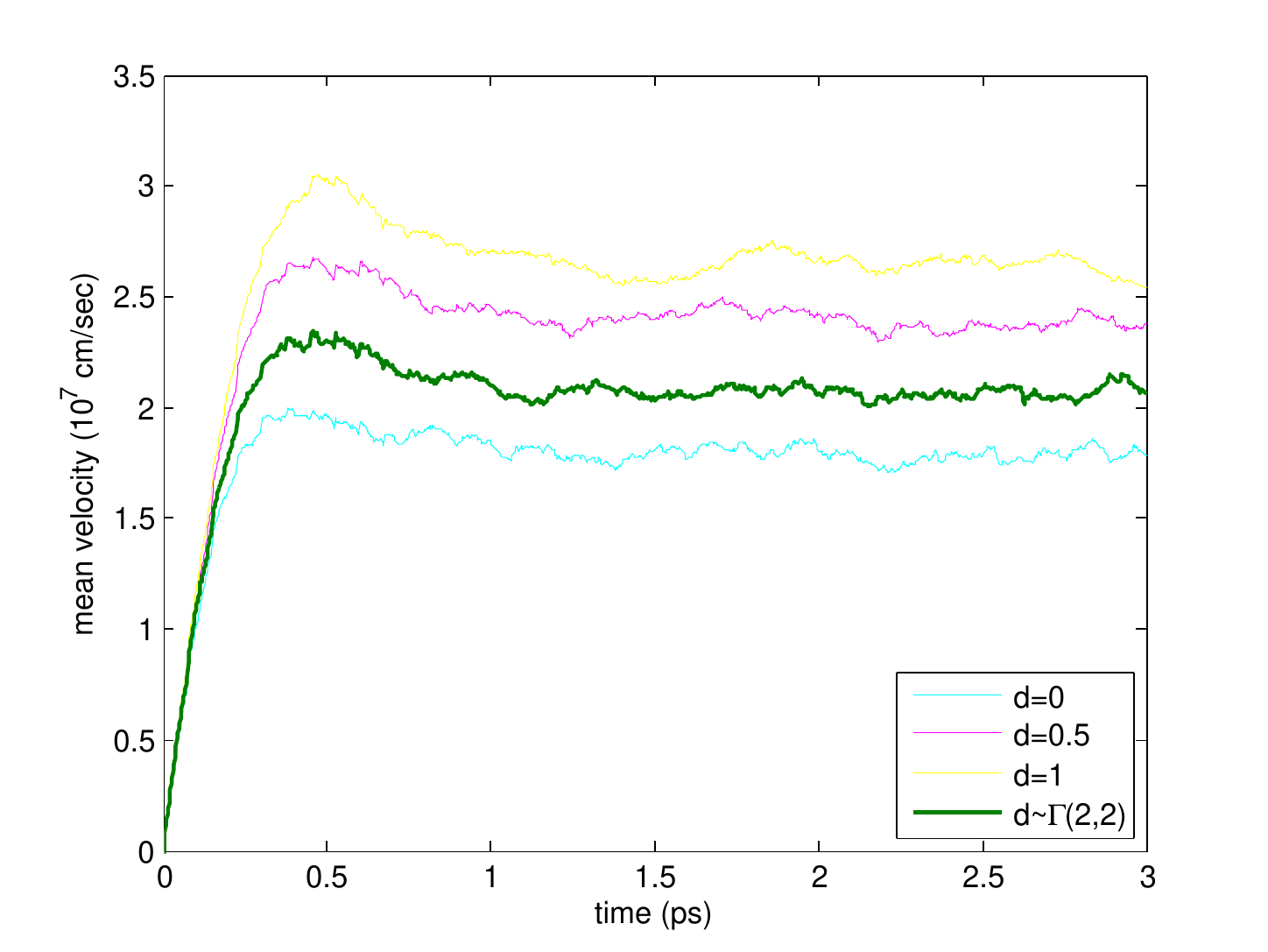}}
 	\caption{Energia (a) e velocità (b) medie ottenute con un campo elettrico applicato di $\si{\num{5}.\kilo\volt\per\centi\metre}$, un livello di Fermi pari a $\si{\num{0.4}.\electronvolt}$ e $d\sim\Gamma(2,2)$.}
\end{figure}
\begin{figure}[ht]
	\centering
	\subfigure[]
   		{\includegraphics[width=7.1cm]{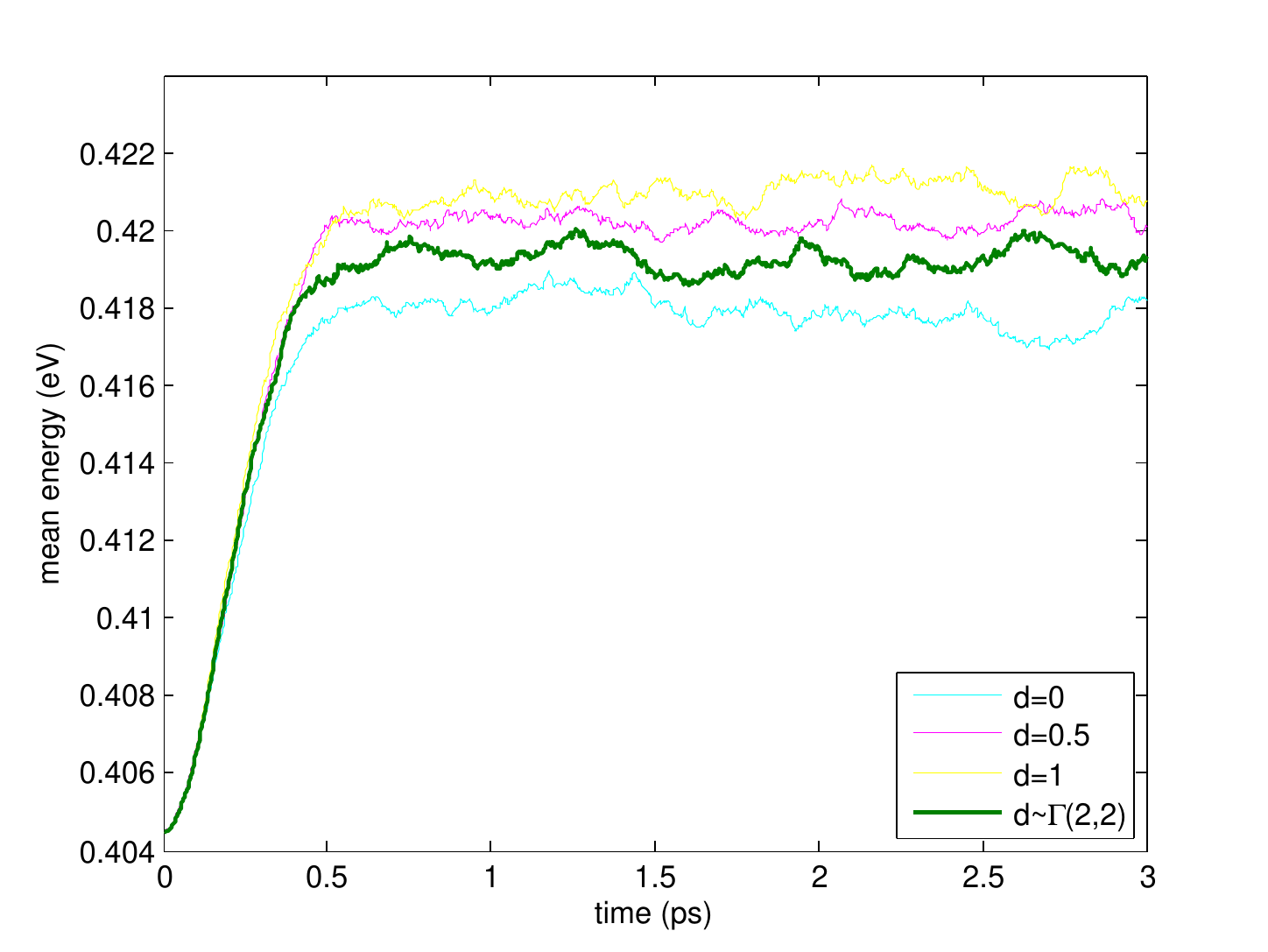}}
 	\,
 	\subfigure[]
   		{\includegraphics[width=7.1cm]{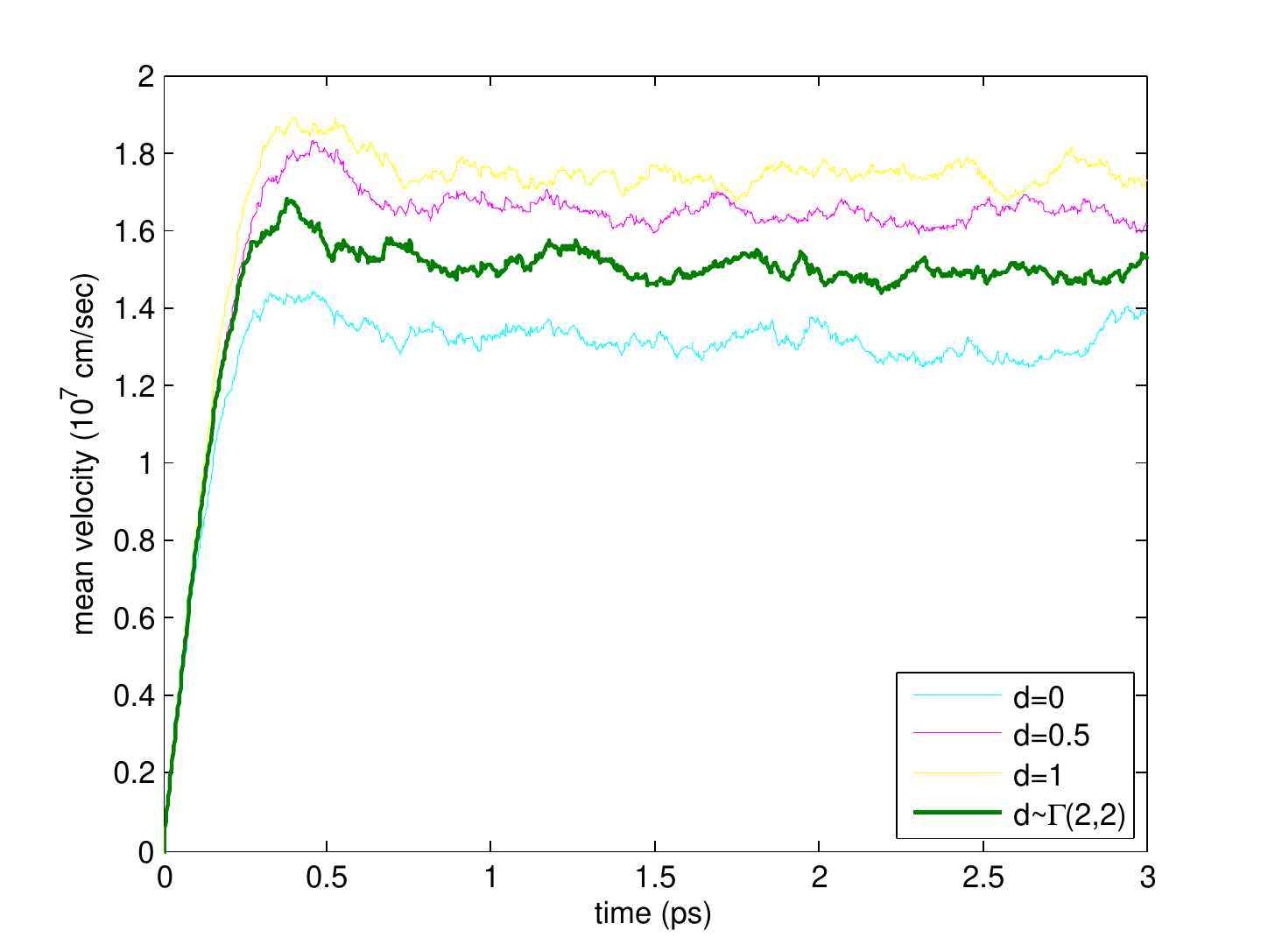}}
 	\caption{Energia (a) e velocità (b) medie ottenute con un campo elettrico applicato di $\si{\num{5}.\kilo\volt\per\centi\metre}$, un livello di Fermi pari a $\si{\num{0.6}.\electronvolt}$ e $d\sim\Gamma(2,2)$.}
\end{figure}
\begin{figure}[ht]
	\centering
	\subfigure[]
   		{\includegraphics[width=7.1cm]{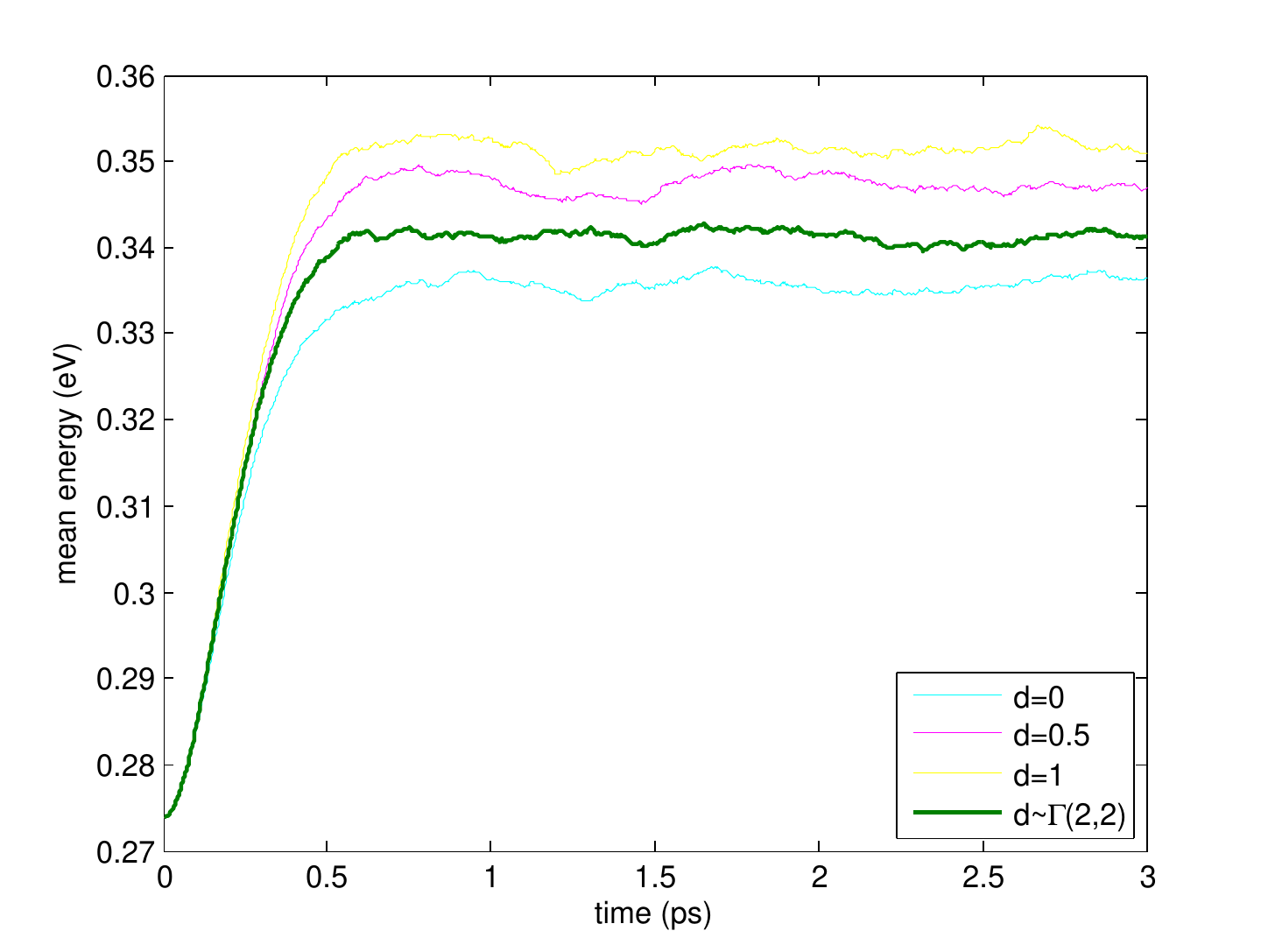}}
 	\,
 	\subfigure[]
   		{\includegraphics[width=7.1cm]{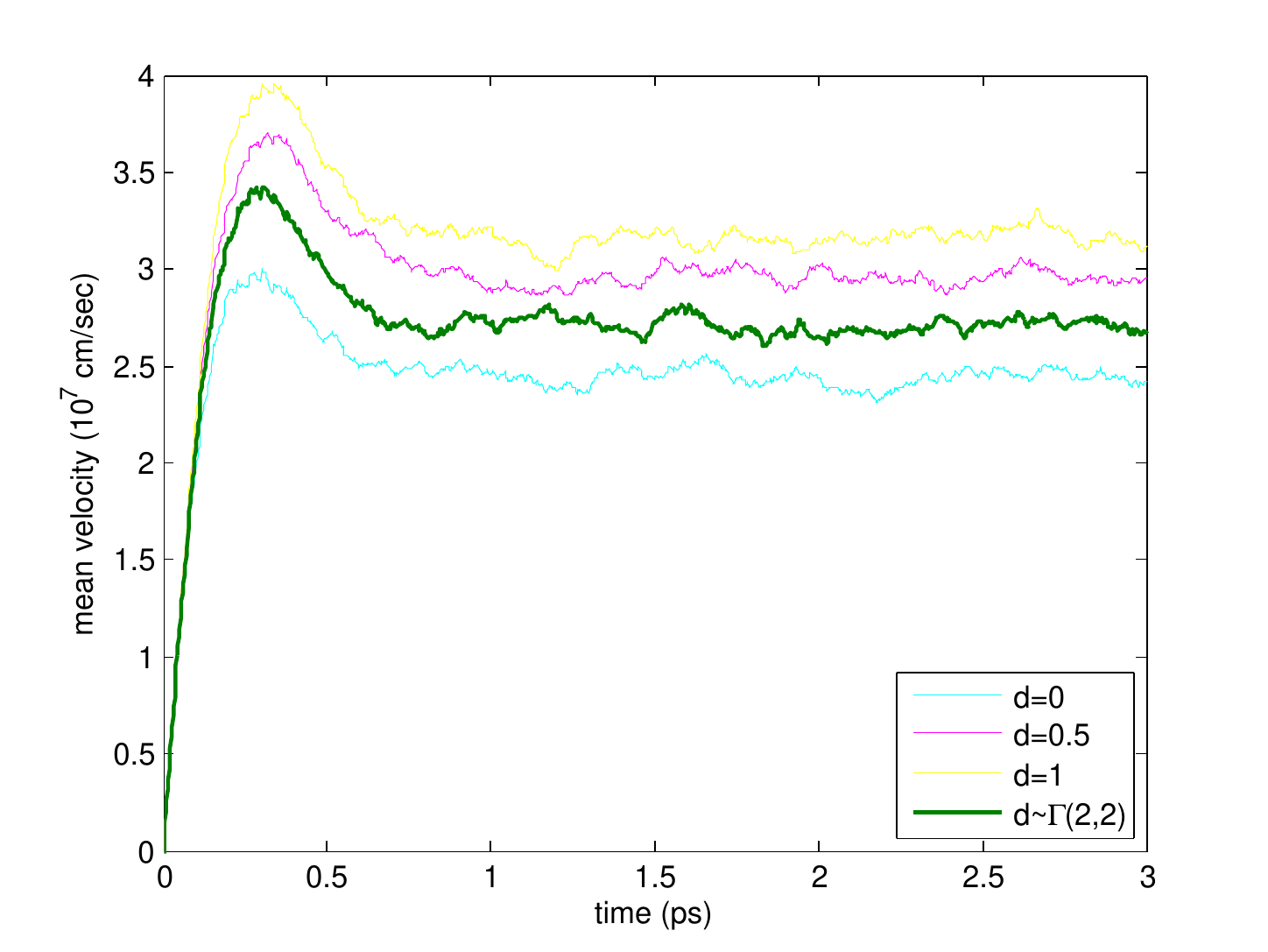}}
 	\caption{Energia (a) e velocità (b) medie ottenute con un campo elettrico applicato di $\si{\num{10}.\kilo\volt\per\centi\metre}$, un livello di Fermi pari a $\si{\num{0.4}.\electronvolt}$ e $d\sim\Gamma(2,2)$.}
\end{figure}
\begin{figure}[ht]
	\centering
	\subfigure[]
   		{\includegraphics[width=7.1cm]{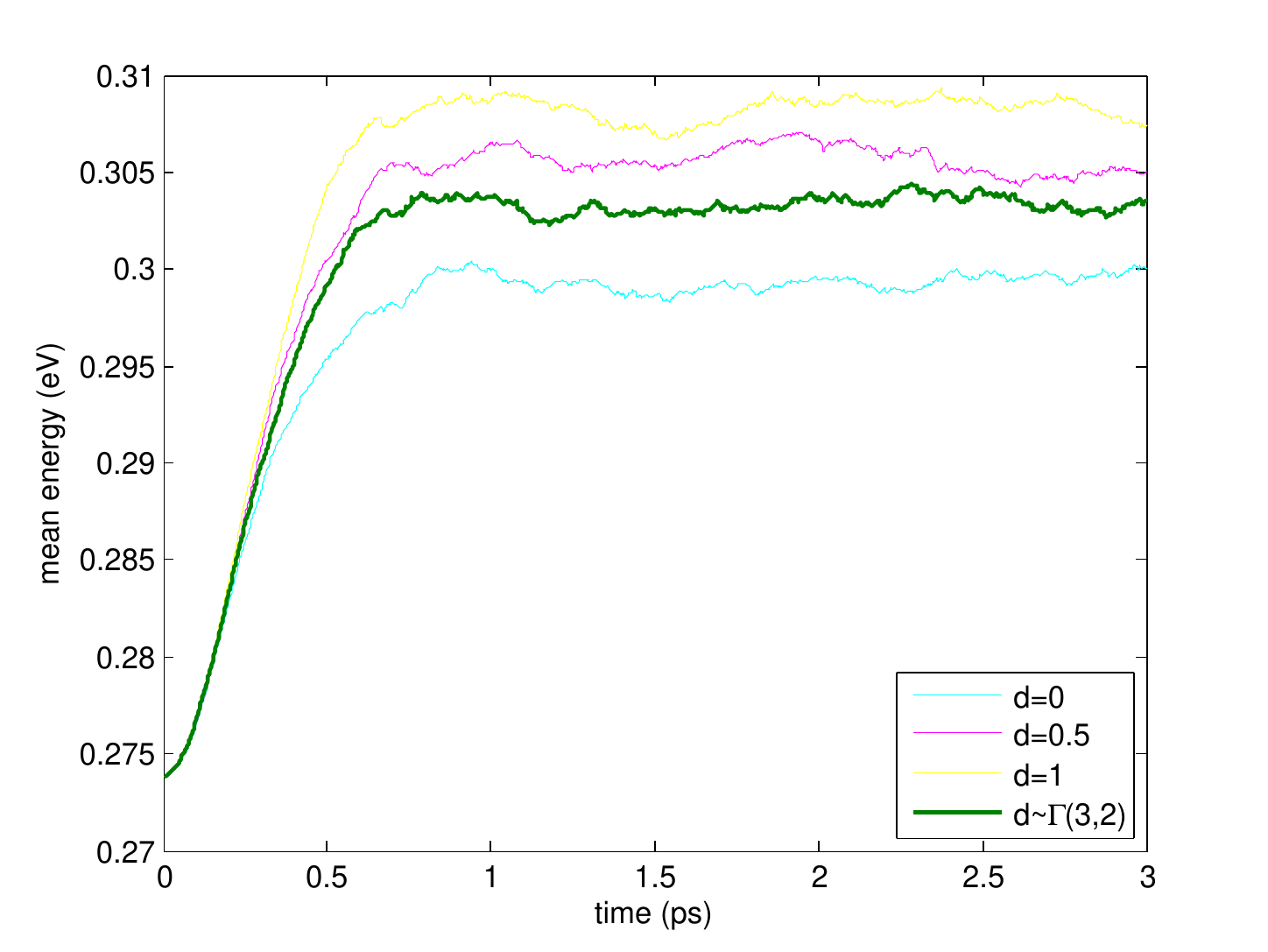}}
 	\,
 	\subfigure[]
   		{\includegraphics[width=7.1cm]{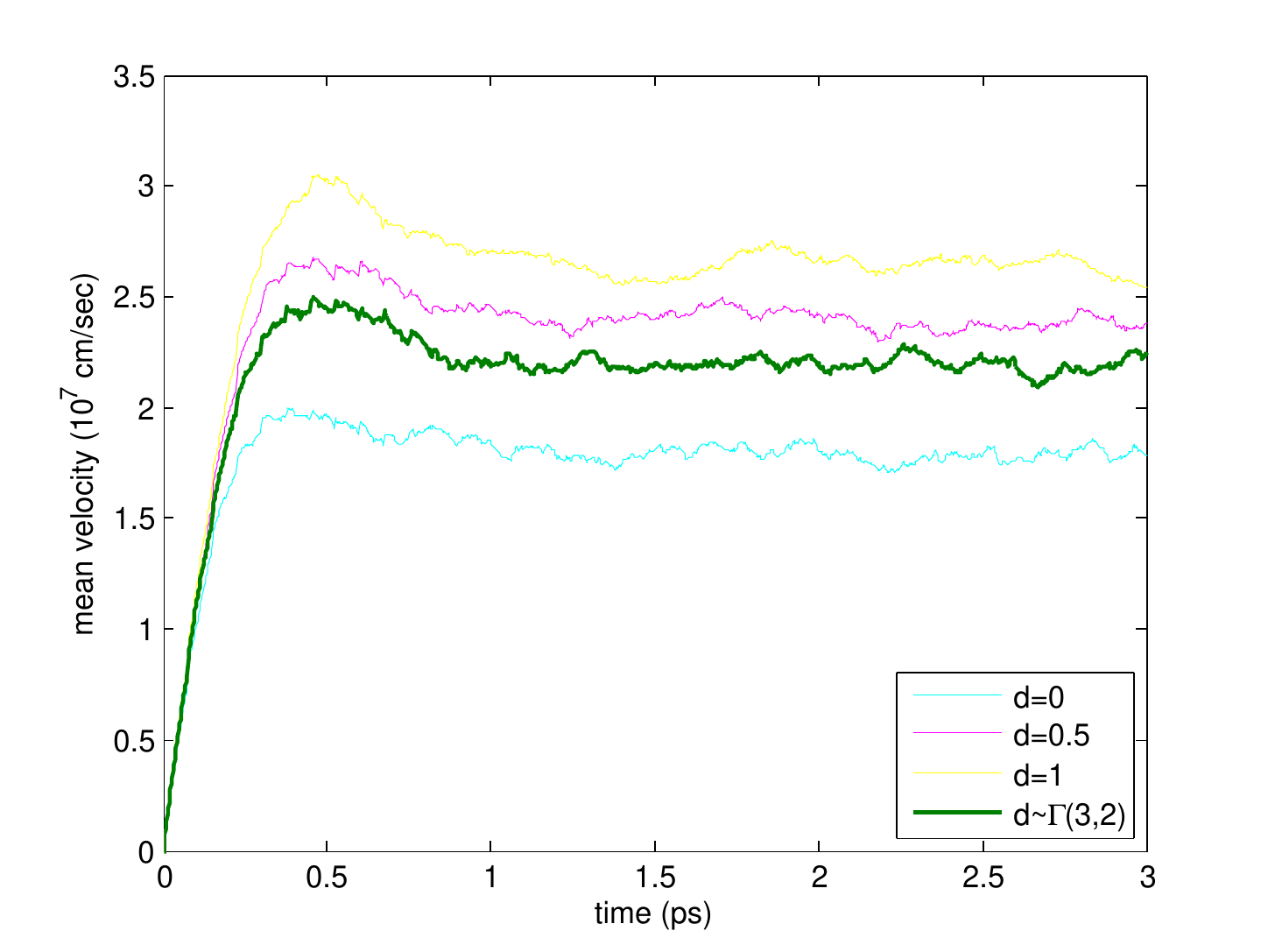}}
 	\caption{Energia (a) e velocità (b) medie ottenute con un campo elettrico applicato di $\si{\num{5}.\kilo\volt\per\centi\metre}$, un livello di Fermi pari a $\si{\num{0.4}.\electronvolt}$ e $d\sim\Gamma(3,2)$.}
\end{figure}
\begin{figure}[ht]
	\centering
	\subfigure[]
   		{\includegraphics[width=7.1cm]{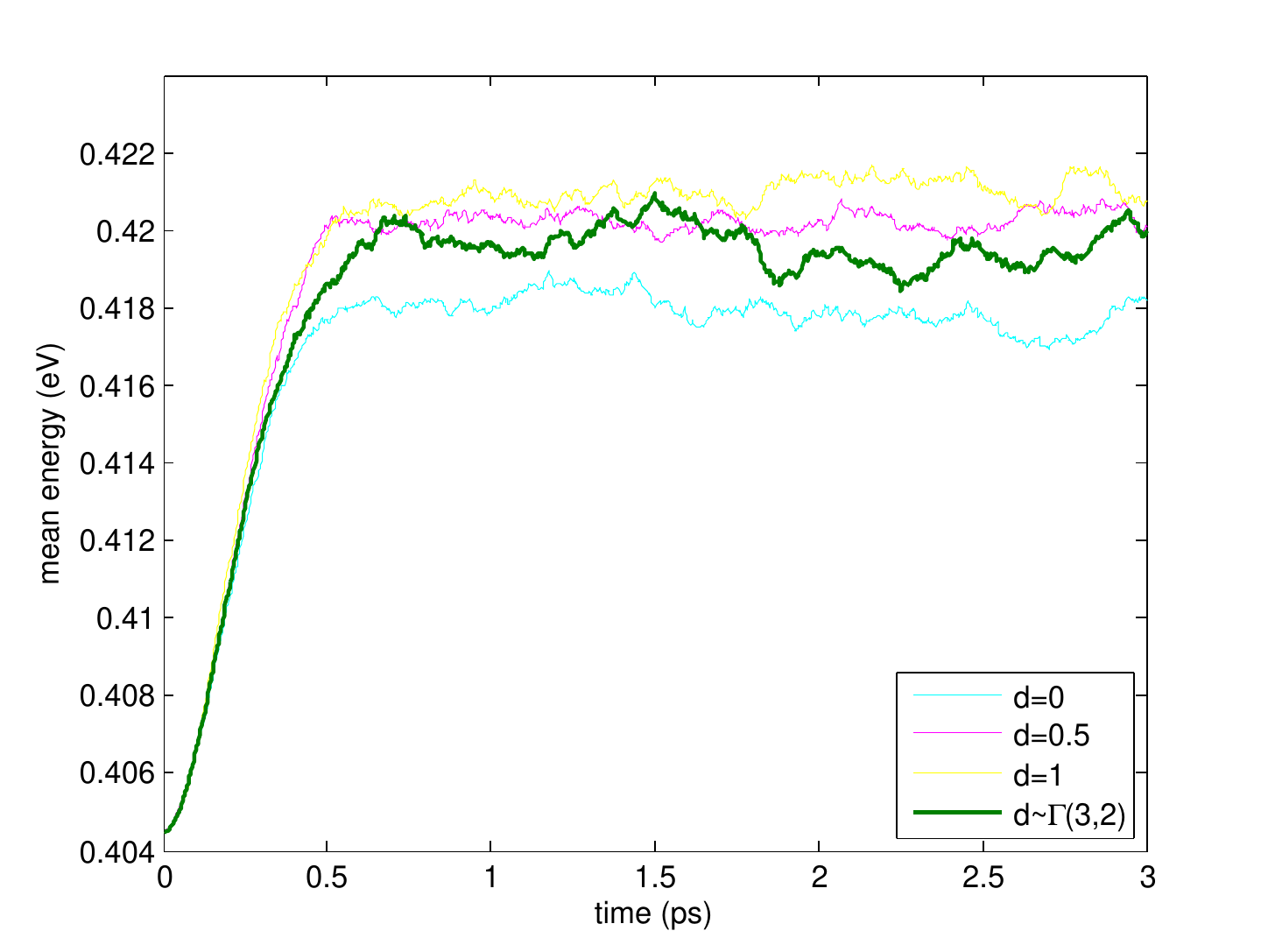}}
 	\,
 	\subfigure[]
   		{\includegraphics[width=7.1cm]{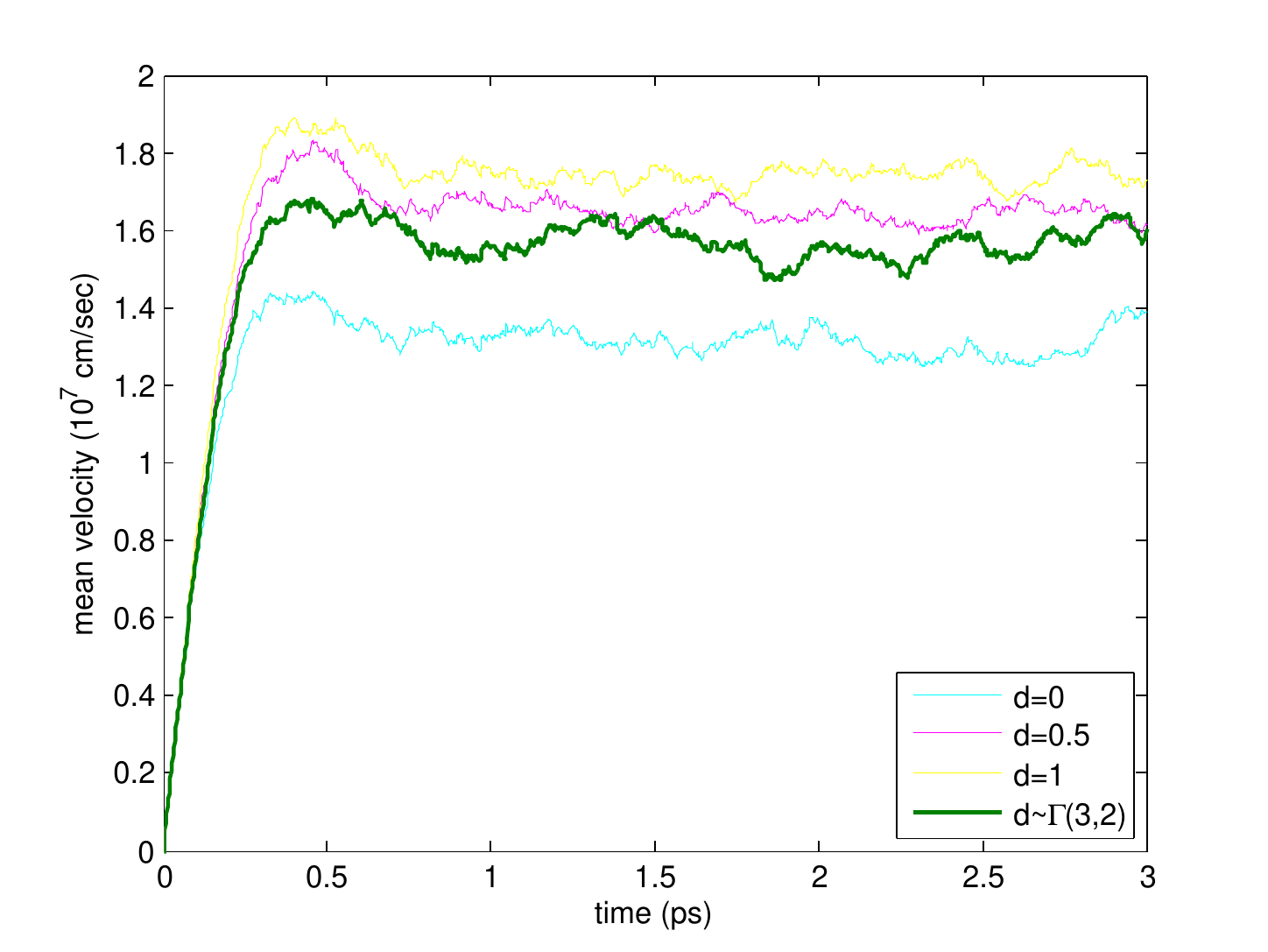}}
 	\caption{Energia (a) e velocità (b) medie ottenute con un campo elettrico applicato di $\si{\num{5}.\kilo\volt\per\centi\metre}$, un livello di Fermi pari a $\si{\num{0.6}.\electronvolt}$ e $d\sim\Gamma(3,2)$.}
\end{figure}
\begin{figure}[ht]
	\centering
	\subfigure[]
   		{\includegraphics[width=7.1cm]{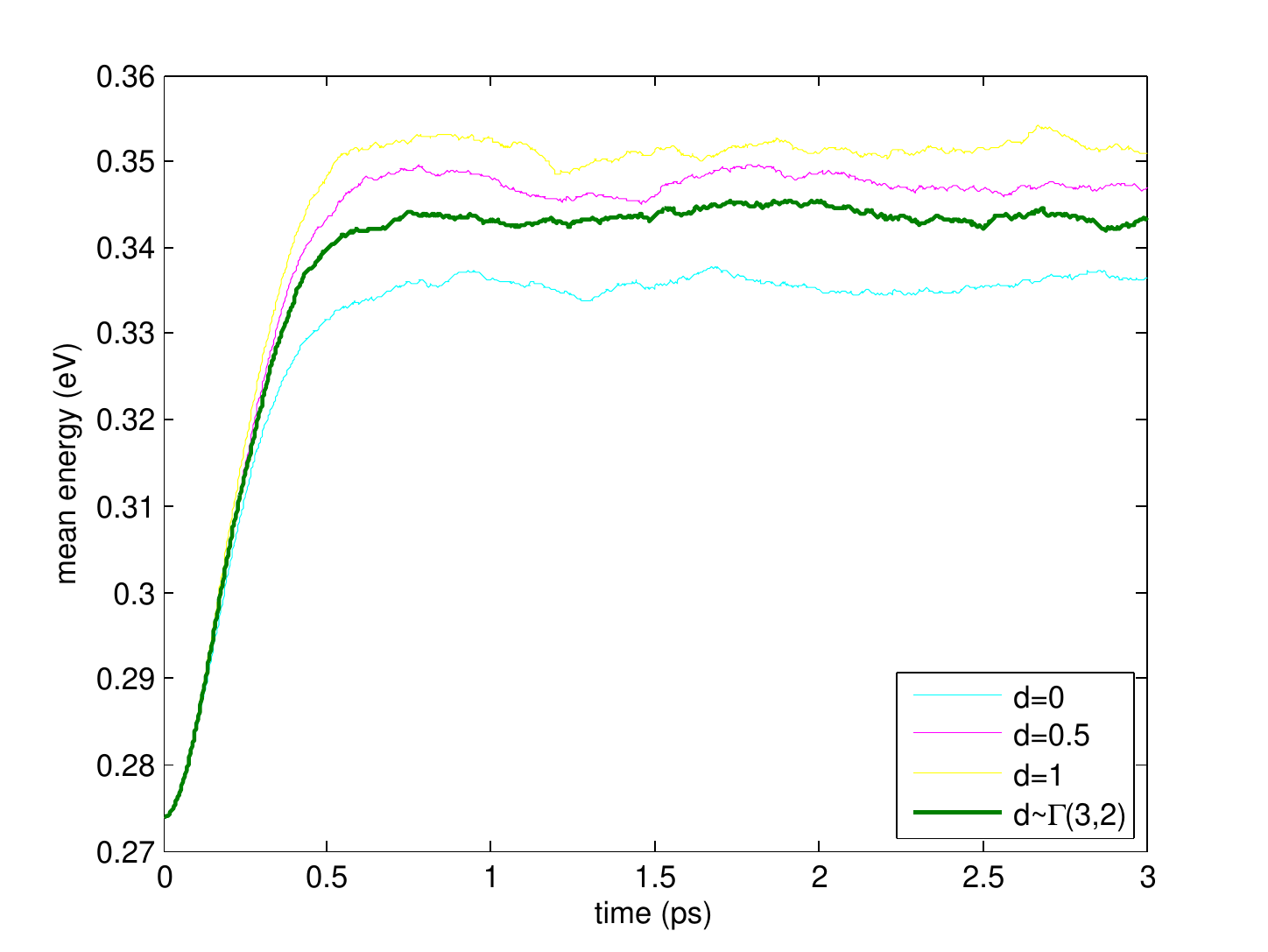}}
 	\,
 	\subfigure[]
   		{\includegraphics[width=7.1cm]{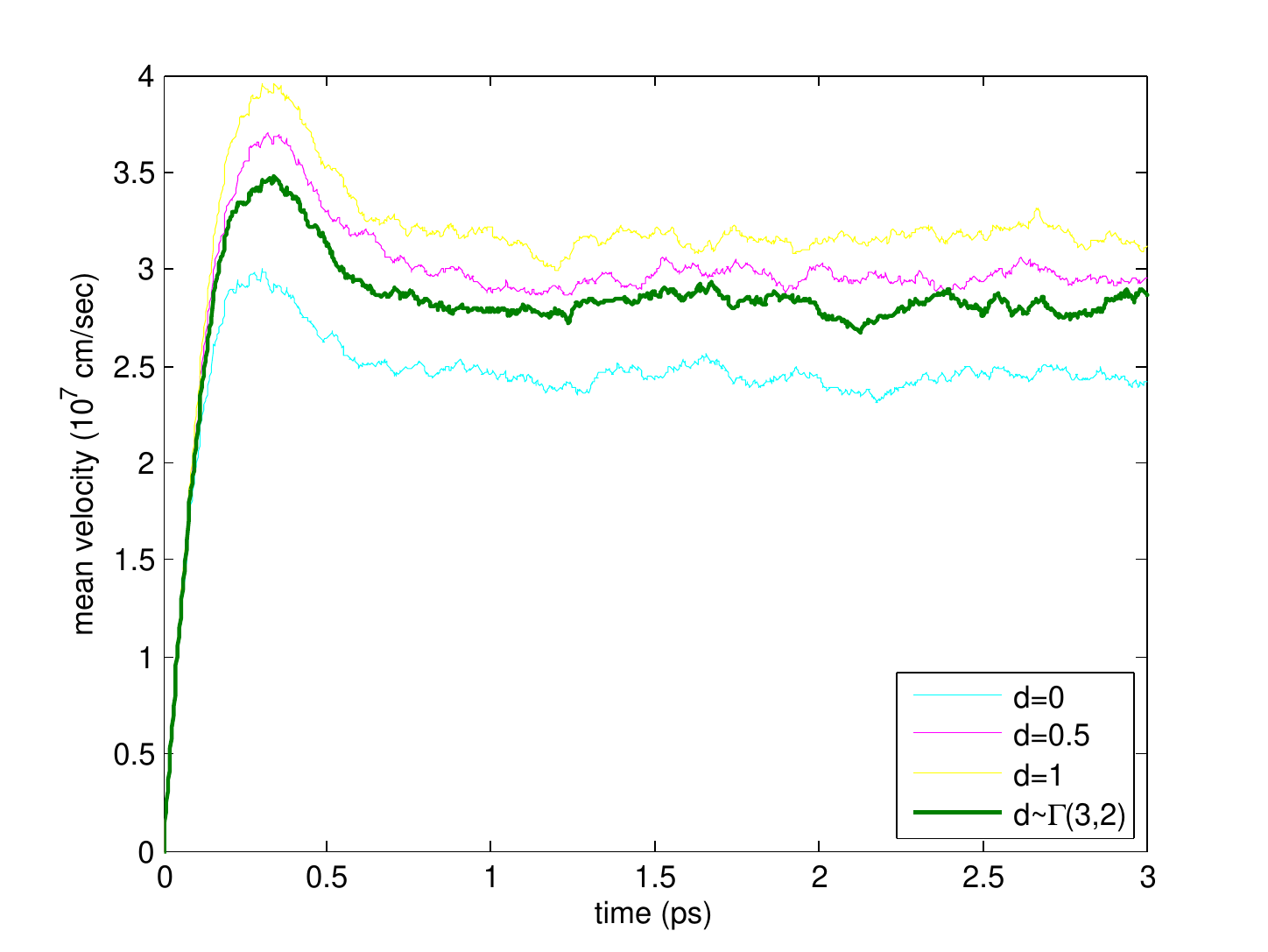}}
 	\caption{Energia (a) e velocità (b) medie ottenute con un campo elettrico applicato di $\si{\num{10}.\kilo\volt\per\centi\metre}$, un livello di Fermi pari a $\si{\num{0.4}.\electronvolt}$ e $d\sim\Gamma(3,2)$.}
\end{figure}
\begin{figure}[ht]
	\centering
	\subfigure[]
   		{\includegraphics[width=7.1cm]{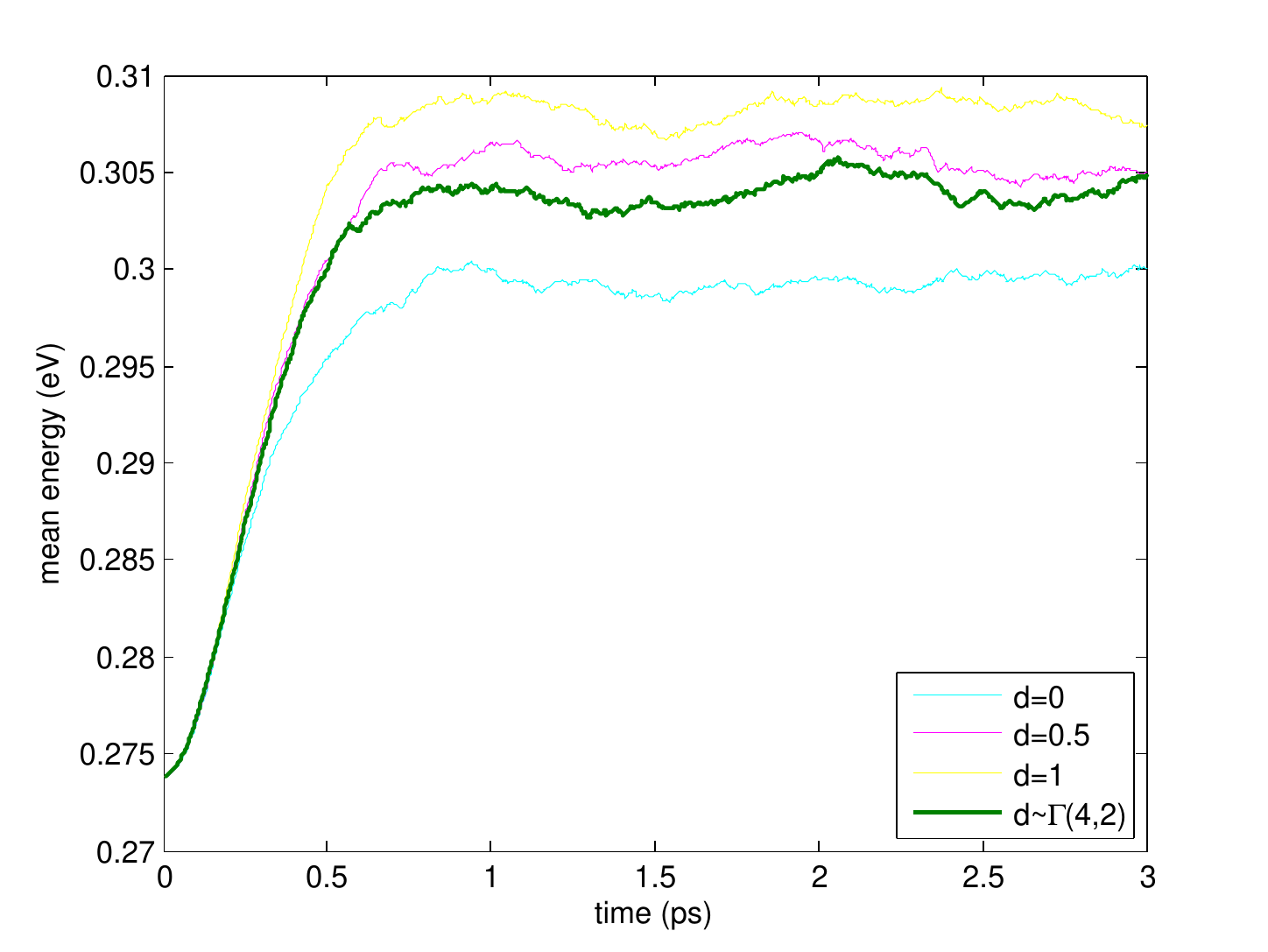}}
 	\,
 	\subfigure[]
   		{\includegraphics[width=7.1cm]{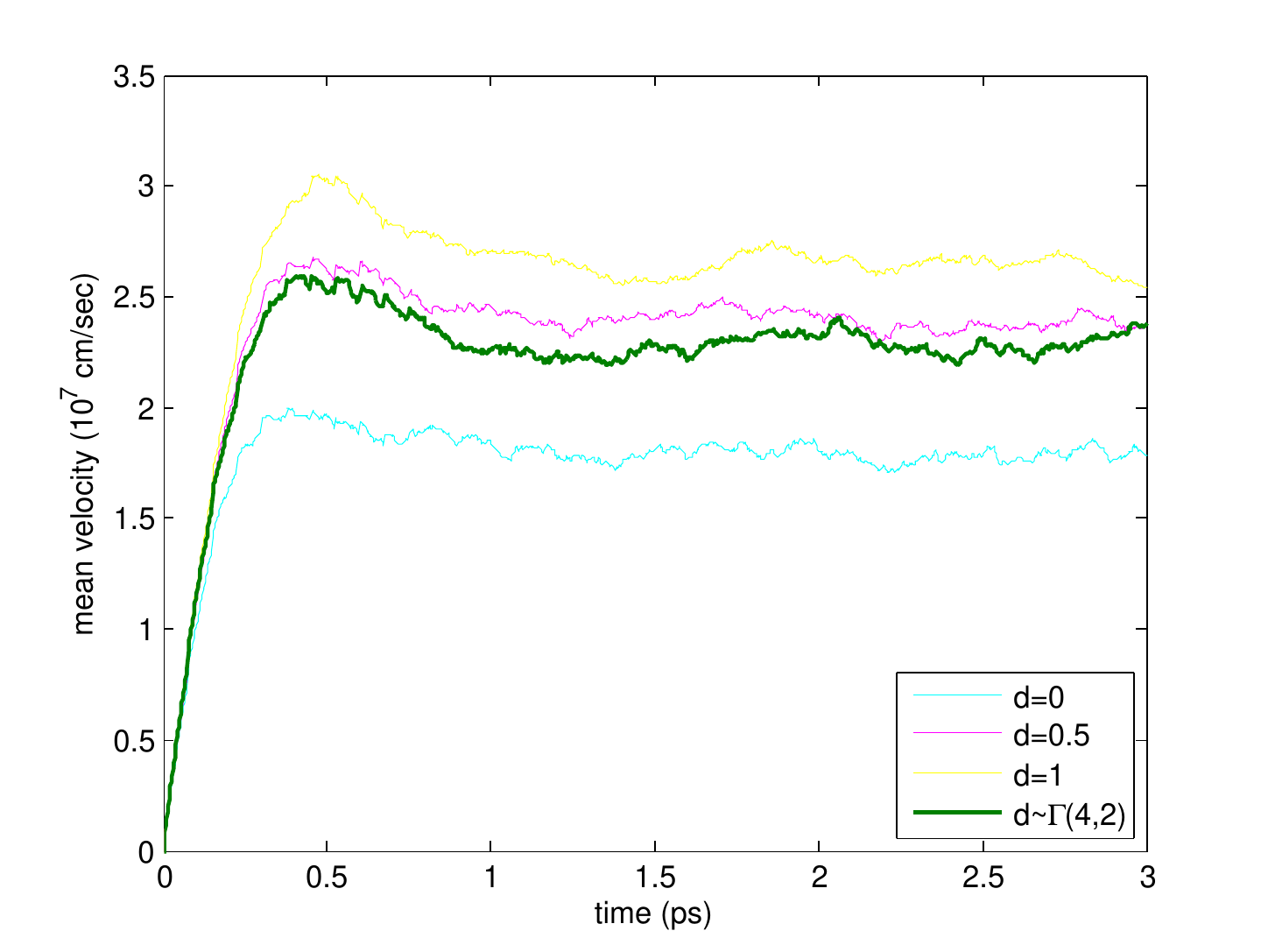}}
 	\caption{Energia (a) e velocità (b) medie ottenute con un campo elettrico applicato di $\si{\num{5}.\kilo\volt\per\centi\metre}$, un livello di Fermi pari a $\si{\num{0.4}.\electronvolt}$ e $d\sim\Gamma(4,2)$.}
\end{figure}
\begin{figure}[ht]
	\centering
	\subfigure[]
   		{\includegraphics[width=7.1cm]{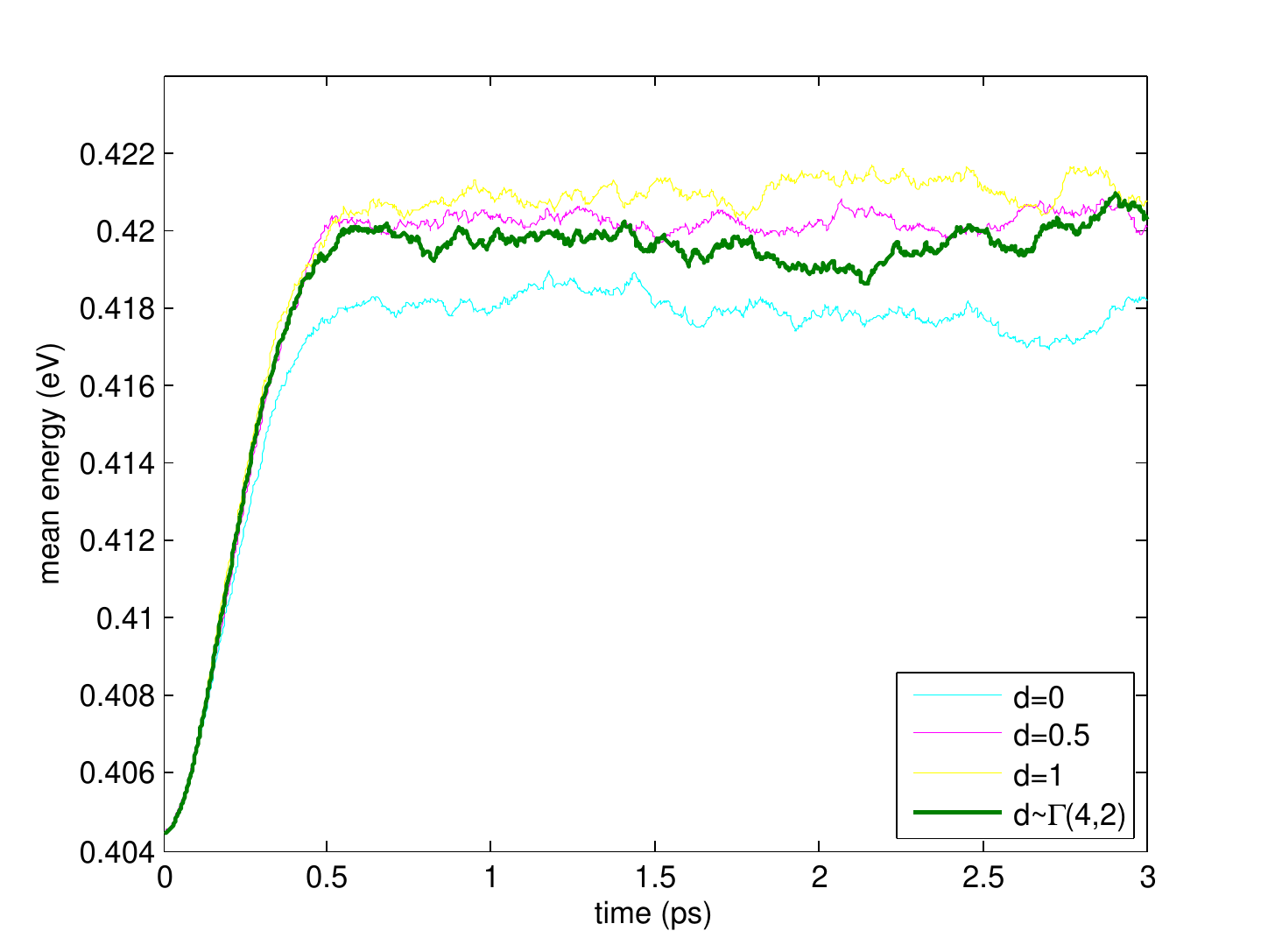}}
 	\,
 	\subfigure[]
   		{\includegraphics[width=7.1cm]{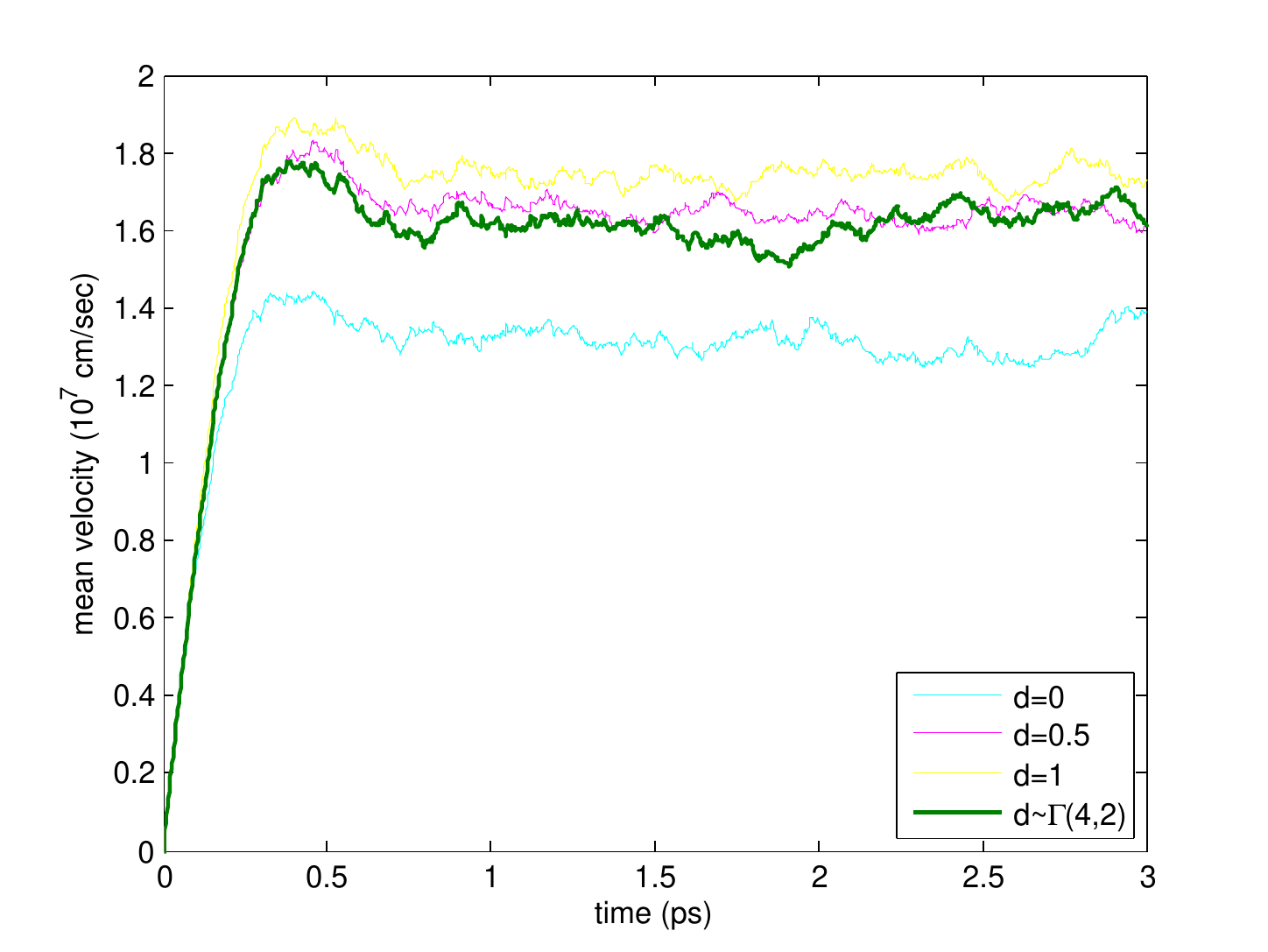}}
 	\caption{Energia (a) e velocità (b) medie ottenute con un campo elettrico applicato di $\si{\num{5}.\kilo\volt\per\centi\metre}$, un livello di Fermi pari a $\si{\num{0.6}.\electronvolt}$ e $d\sim\Gamma(4,2)$.}
\end{figure}
\begin{figure}[ht]
	\centering
	\subfigure[]
   		{\includegraphics[width=7.1cm]{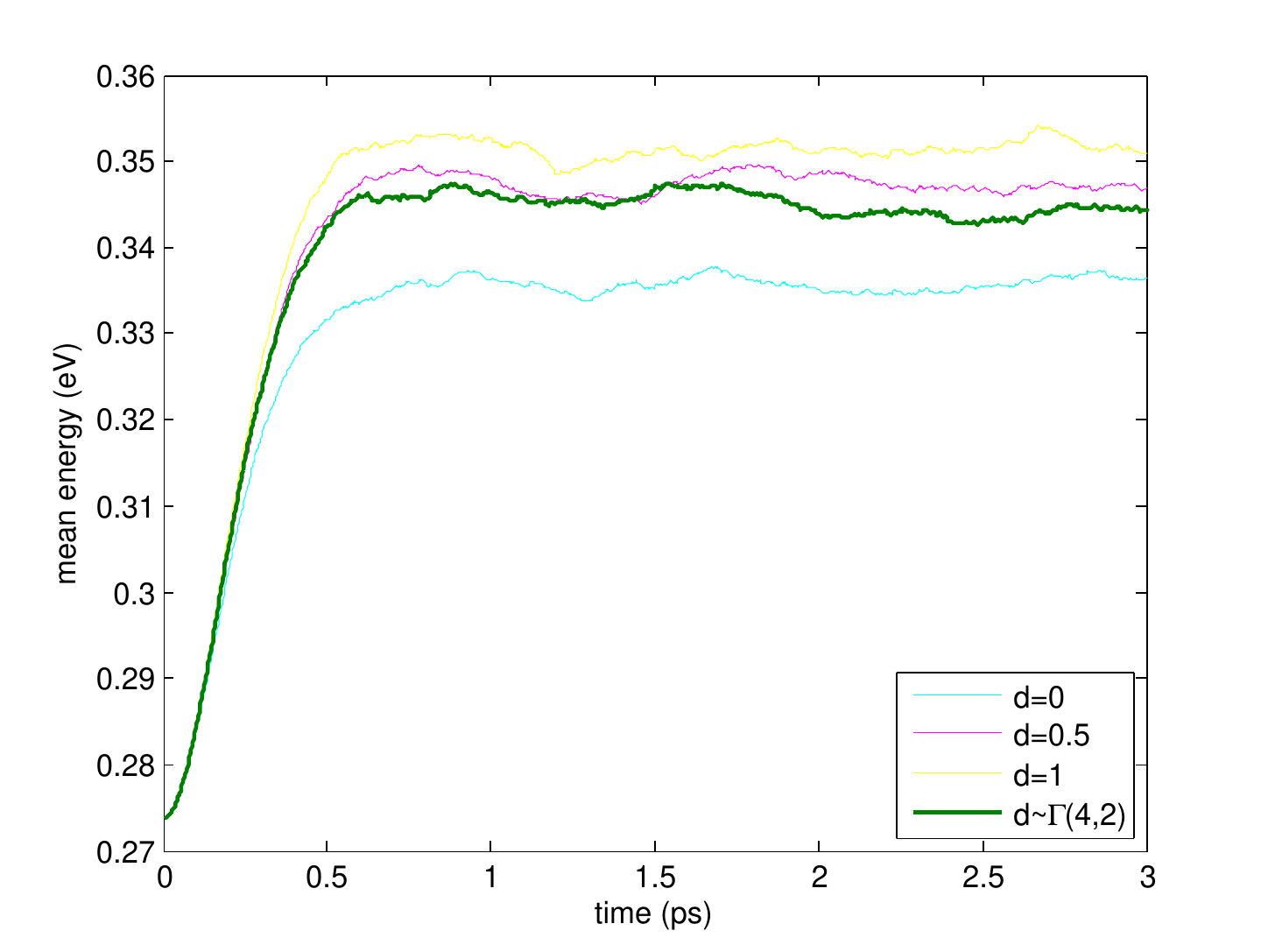}}
 	\,
 	\subfigure[]
   		{\includegraphics[width=7.1cm]{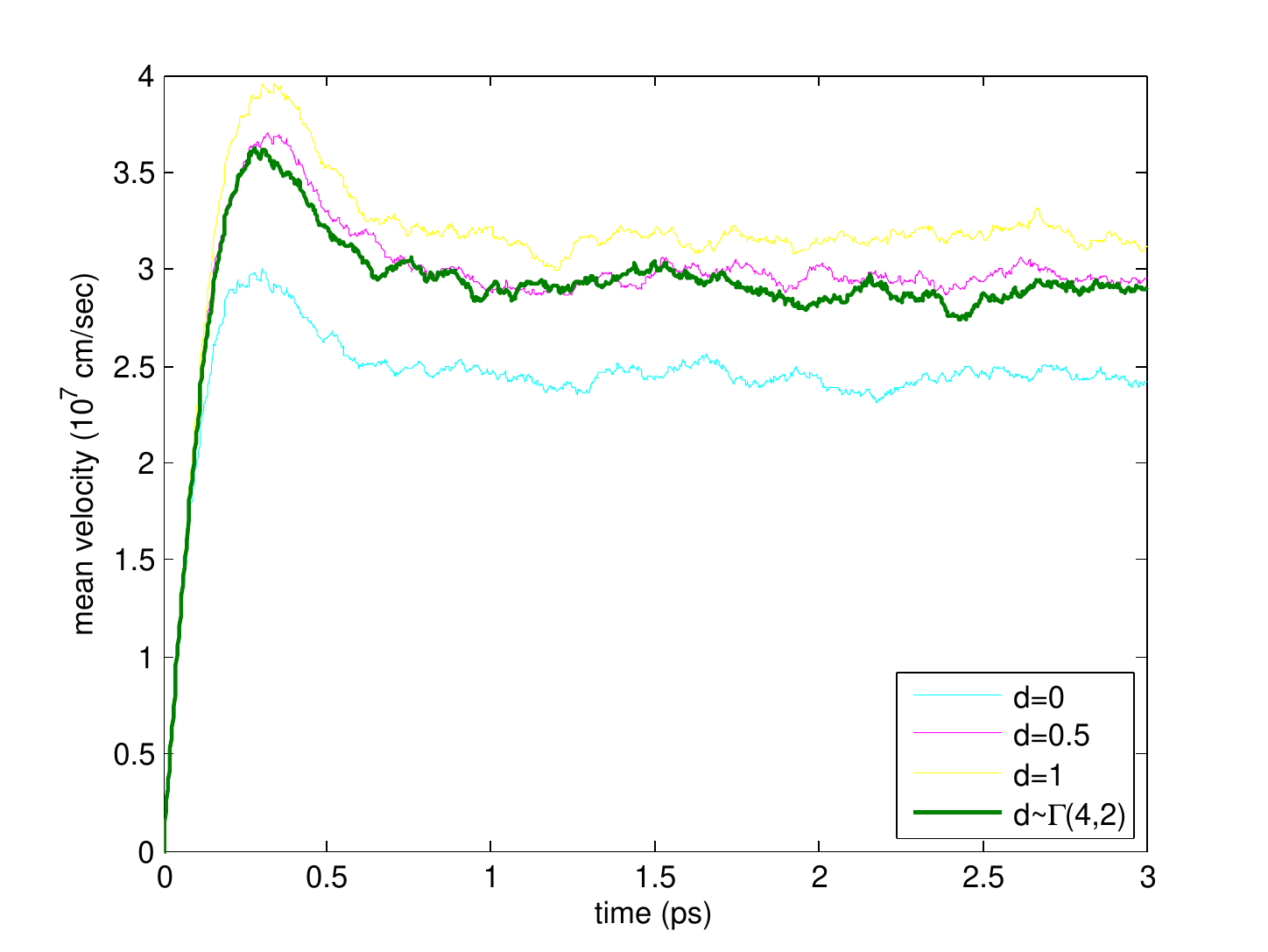}}
 	\caption{Energia (a) e velocità (b) medie ottenute con un campo elettrico applicato di $\si{\num{10}.\kilo\volt\per\centi\metre}$, un livello di Fermi pari a $\si{\num{0.4}.\electronvolt}$ e $d\sim\Gamma(4,2)$.}
\end{figure}
\FloatBarrier
\FloatBarrier
\begin{figure}[ht]
	\centering
	\subfigure[]
   		{\includegraphics[width=7.1cm]{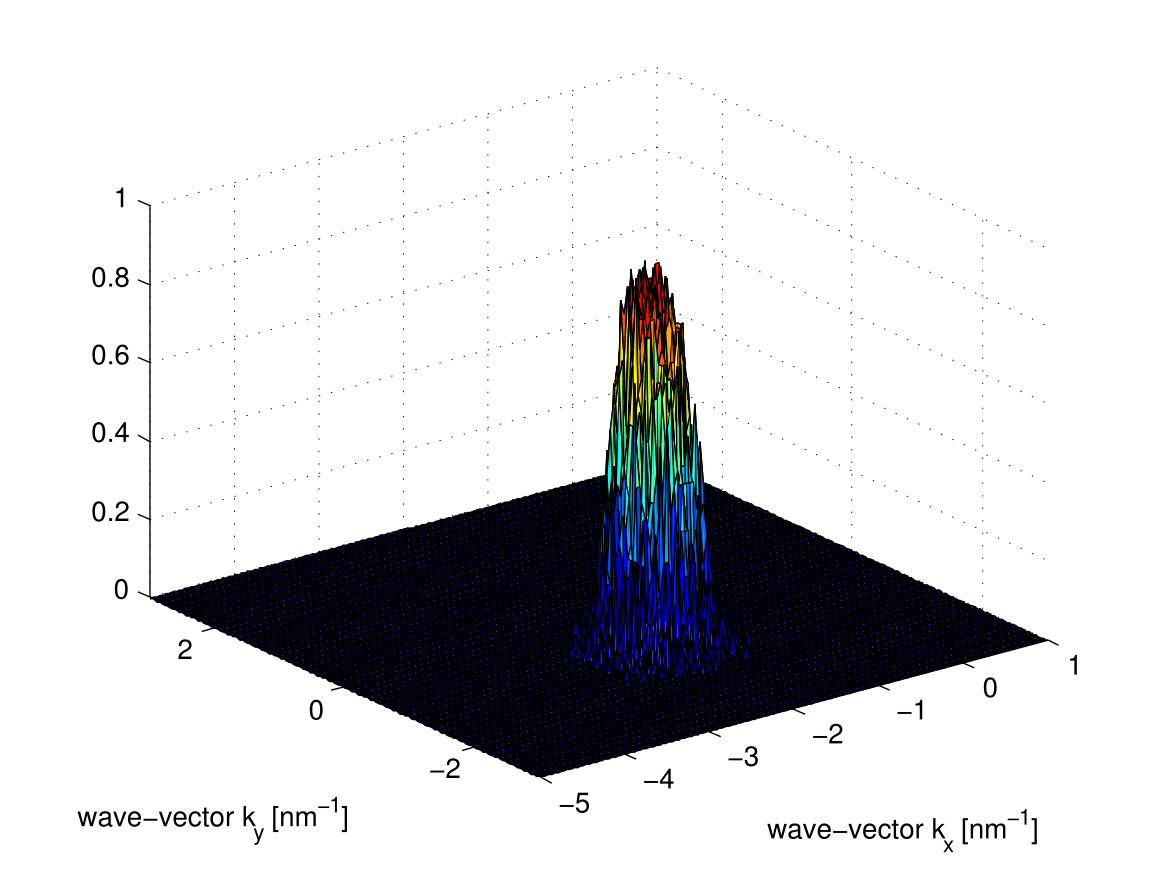}}
 	\,
 	\subfigure[]
   		{\includegraphics[width=7.1cm]{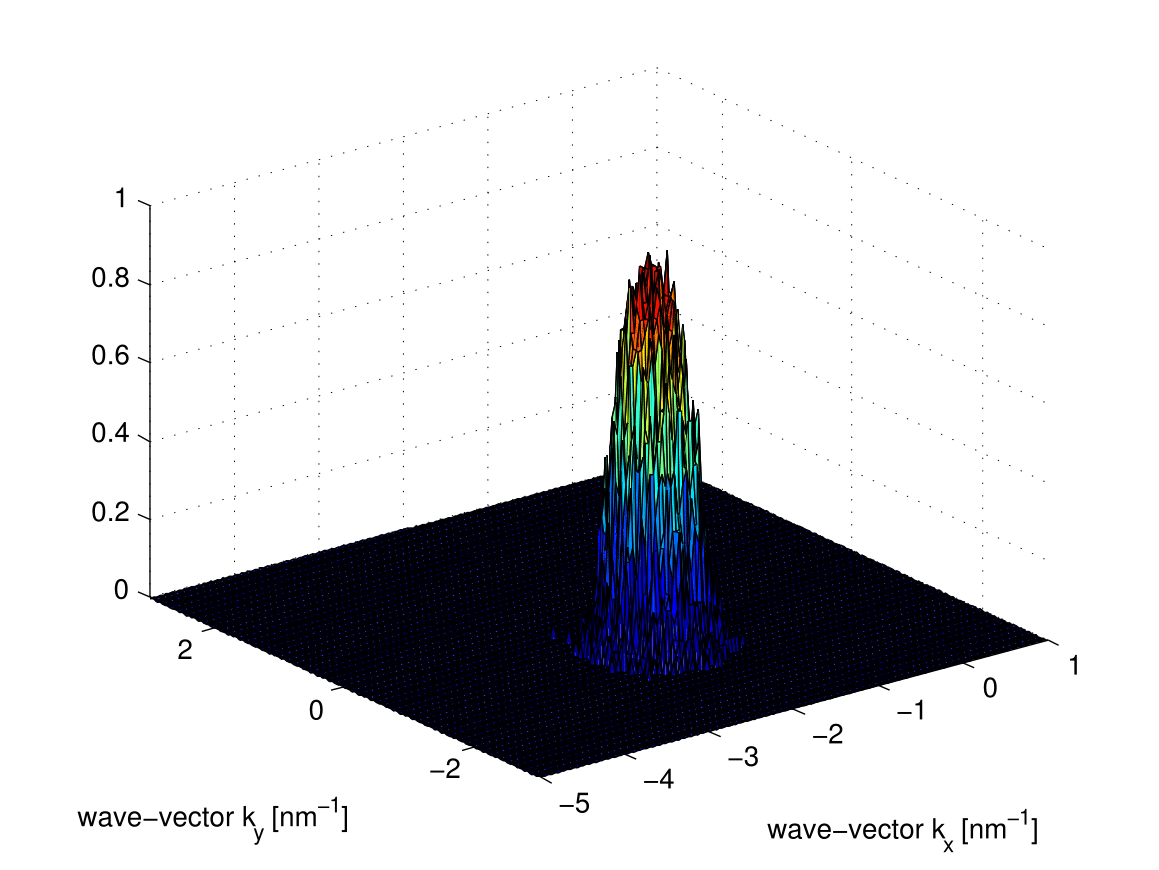}}
	\\
   	\subfigure[]
   		{\includegraphics[width=7.1cm]{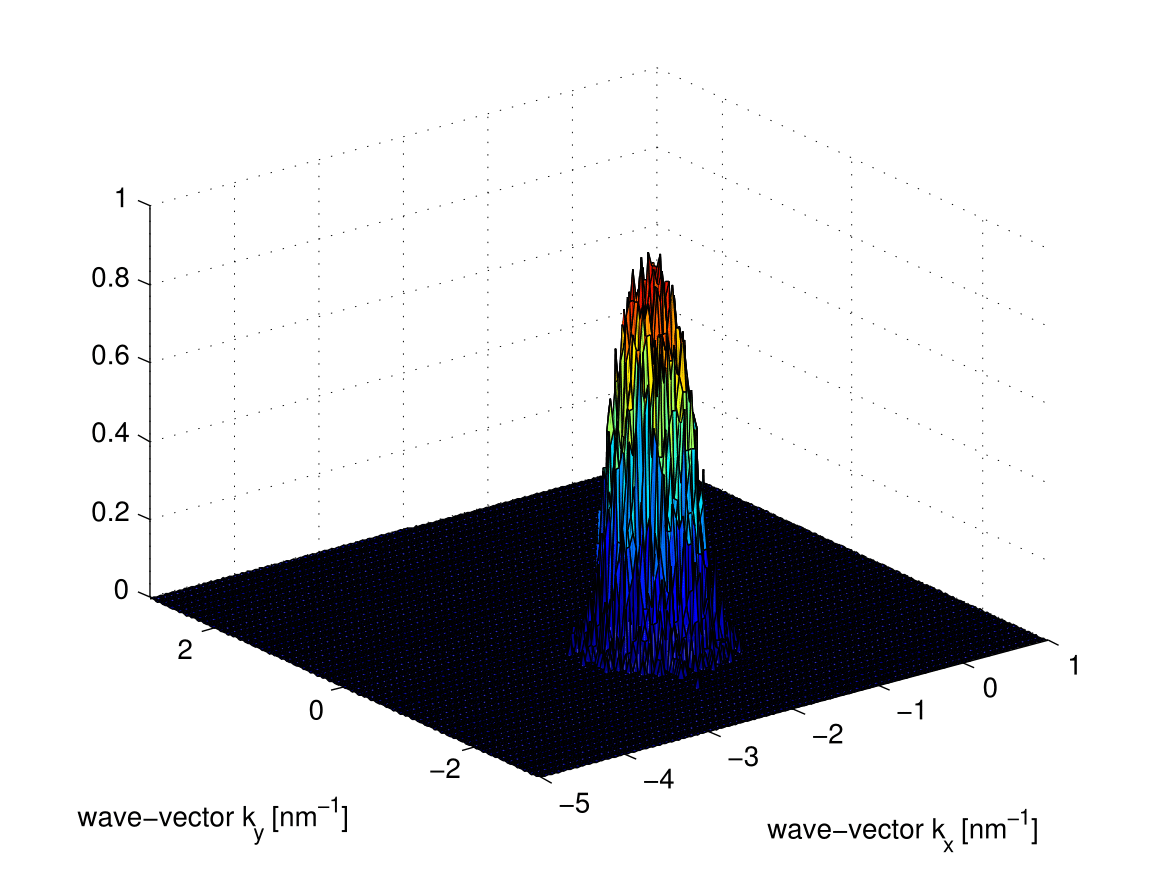}}
   	\,
   	\subfigure[]
   		{\includegraphics[width=7.1cm]{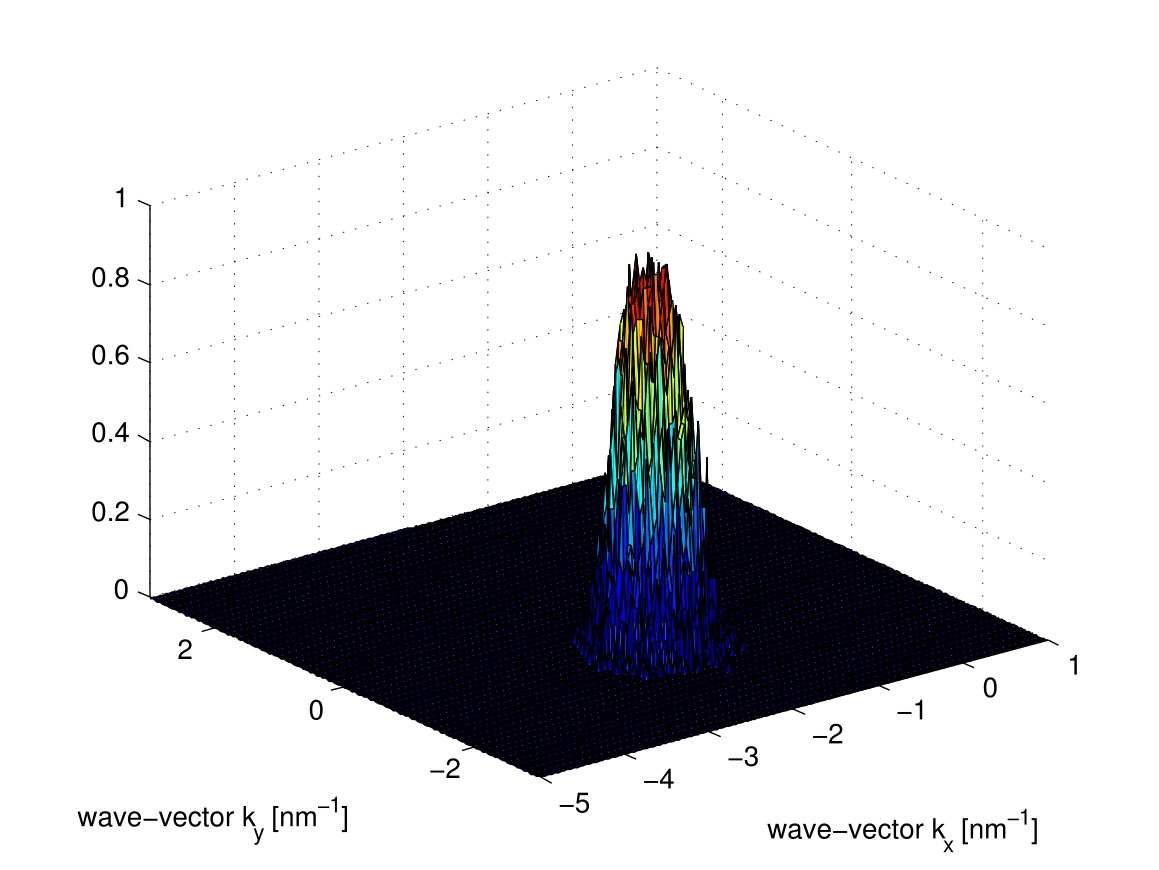}}
 	\caption{Grafici delle funzioni di distribuzione degli elettroni, ottenuti con un campo elettrico applicato di $\si{\num{10}.\kilo\volt\per\centi\metre}$ e un livello di Fermi pari a $\si{\num{0.4}.\electronvolt}$. Inoltre (a) è stato ottenuto con la scelta uniforme per $d$, (b), (c) e (d) sono stati ottenuti scegliendo $\Gamma(k,2)$ con $k=2,3,4$, rispettivamente.}
\label{FIG:CAP3:DSMC_new_SiO2_distrib}
\end{figure}
\FloatBarrier

\subsection{Risultati numerici su \texorpdfstring{$\ce{h-BN}$}{h-BN} e \texorpdfstring{$\ce{HfO2}$}{HfO2}}
Si è voluta ripetere l'analisi illustrata nei Paragrafi \ref{PAR:CAP3:Subs_1} e \ref{PAR:CAP3:Subs_2} utilizzando come substrati, al posto del $\ce{SiO2}$, il nitruro di boro esagonale ($\ce{h-BN}$) e il biossido di afnio ($\ce{HfO2}$). Studiare questi due substrati è importante soprattutto per le applicazioni tecnologiche (cfr.~\citep{ART:Hirai2014}). Il confronto dei risultati si basa sull'adattamento del modello per il $\ce{SiO2}$ agli altri materiali, considerando i diversi parametri fisici che entrano in gioco. Le costanti comuni ai tre materiali sono riportati nella Tabella \ref{TAB:CAP3:Cost_fis_sospeso}. Per quanto riguarda i parametri relativi agli scattering con il substrato, si osserva che il valore del potenziale di deformazione per il $\ce{h-BN}$ è incognito ed è noto che varia da $\si{\num{e6}.\electronvolt\per\centi\metre}$ a $\si{\num{1.29e9}.\electronvolt\per\centi\metre}$. Comunque la mobilità degli elettroni risulta indipendente da questo valore (cfr.~\citep{ART:Hirai2014}), di conseguenza nei risultati mostrati si è scelto di utilizzare il valore più grande. Il grafene su $\ce{HfO2}$ presenta una mobilità che è di circa tre ordini di grandezza inferiore rispetto agli altri materiali (cfr.~\citep{ART:Hirai2014}). Tale riduzione della mobilità è dovuta alla bassa energia dei fononi ottici, pari a $\si{\num{12.4}.\milli\electronvolt}$, e al suo grande potenziale di deformazione, pari a $\si{\num{1.29e9}.\electronvolt\centi\metre^{-1}}$. Tutte le costanti fisiche che caratterizzano i tre substrati sono riportati nella Tabella~\ref{TAB:CAP3:Cost_fis_subs}.
\FloatBarrier
\begin{table}[!ht]
\centering
\begin{tabular}{c|ccc}
\multicolumn{1}{c|}{\textbf{Parametro fisico}} & \multicolumn{3}{c}{\textbf{Materiale}} \\
\hline
\multicolumn{1}{c|}{} & \multicolumn{1}{c}{$\ce{SiO2}$} & \multicolumn{1}{c}{$\ce{HfO2}$} & \multicolumn{1}{c}{$\ce{h-BN}$} \\
\cline{2-4}
$\hbar\omega\ped{op}$ ($\si{\milli\electronvolt}$) & $\si{\num{55}}$ & $\si{\num{12.4}}$ & $\si{\num{200}}$ \\
$D_f$ ($\si{\electronvolt\per\centi\metre}$) & $\si{\num{5.14e7}}$ & $\si{\num{1.29e9}}$ & $\si{\num{1.0e6}-\num{1.29e9}}$ \\
$n\ped{imp}$ ($\si{\centi\metre^{-2}}$) & $\si{\num{2.5e11}}$ & $\ldots$ & $\si{\num{2.5e10}}$ \\
\end{tabular}
\caption{Costanti fisiche dei vari substrati utilizzati per le simulazioni.}
\label{TAB:CAP3:Cost_fis_subs}
\end{table}
\FloatBarrier
In \citep{ART:CoMajNaRo} sono state studiate in primo luogo le prestazioni generali nei tre materiali, confrontando la velocità media per tre diversi valori di $d$, considerati costanti. \`{E} evidente dalla Figura~\ref{FIG:CAP3:DSMC_comp_3_subs} che nel caso del $\ce{HfO2}$ si ha una forte riduzione della velocità media con un conseguente decadimento della mobilità. Ciò rende questo materiale non adatto ad essere usato come substrato per i dispositivi elettronici. Pertanto nel seguito si è voluto scartare il $\ce{HfO2}$ dall'analisi.
\FloatBarrier
\begin{figure}[ht]
	\centering
	\subfigure[]
   		{\includegraphics[width=7.1cm]{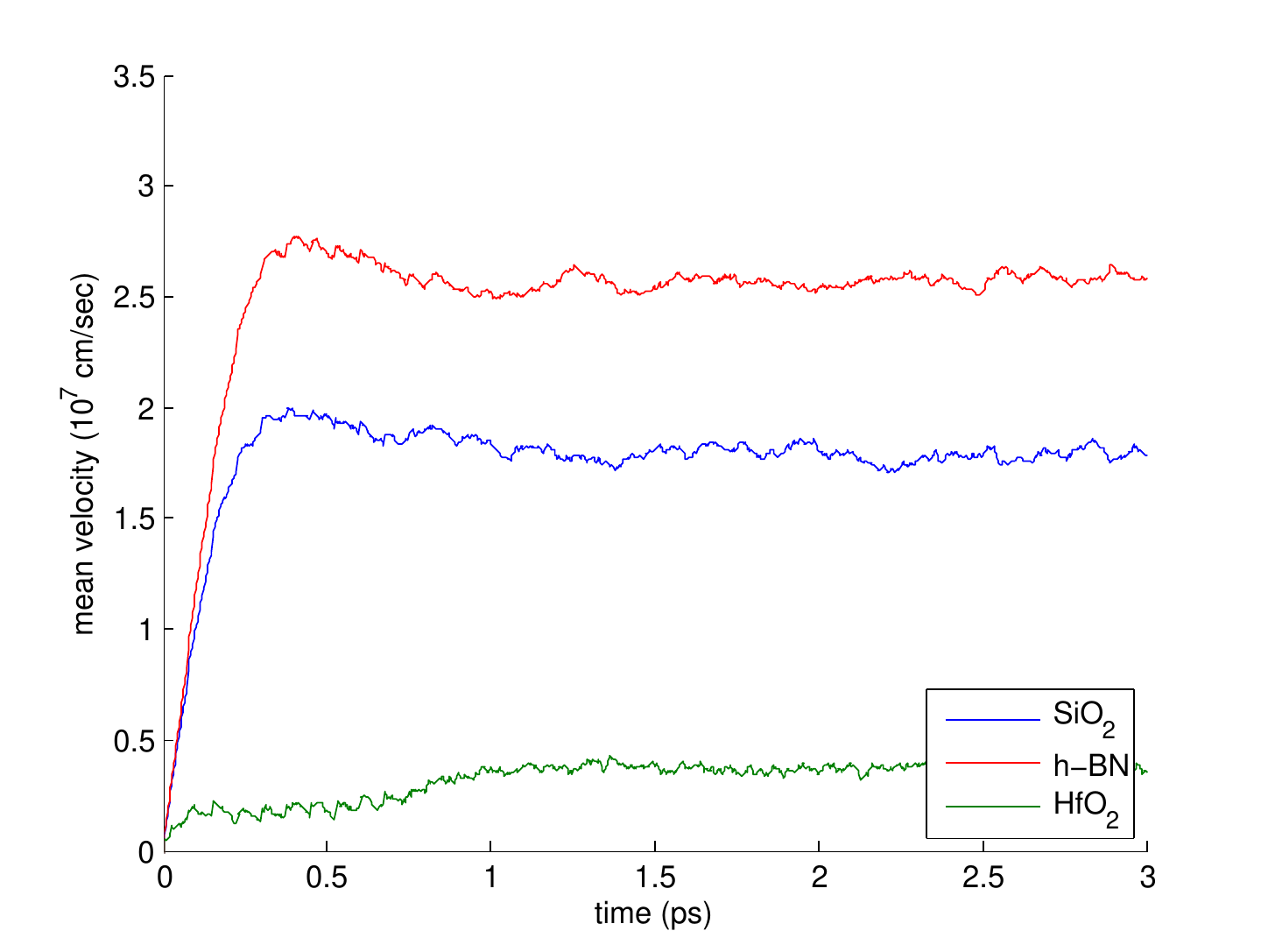}}
 	\,
 	\subfigure[]
   		{\includegraphics[width=7.1cm]{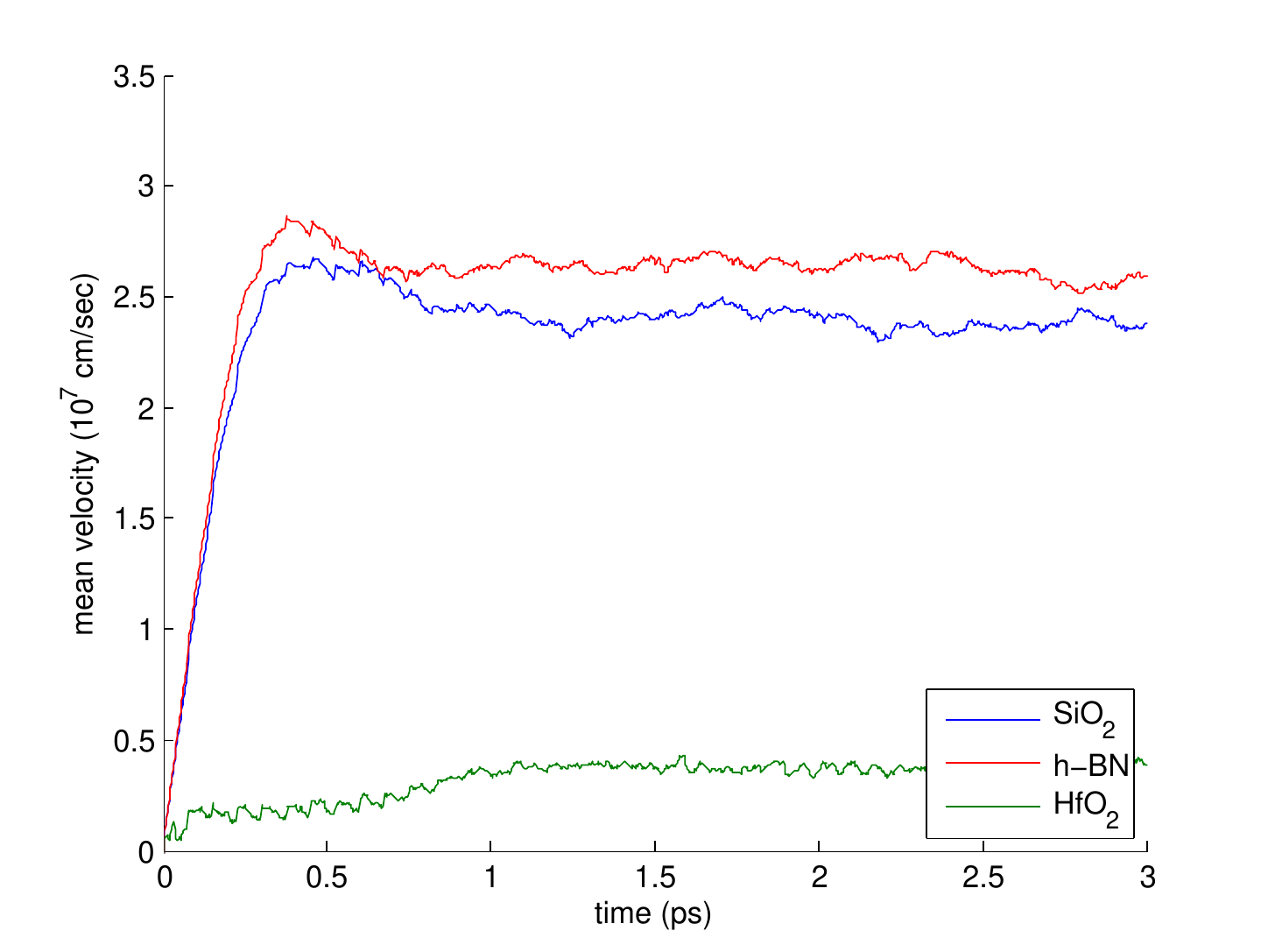}}
	\\
   	\subfigure[]
   		{\includegraphics[width=7.1cm]{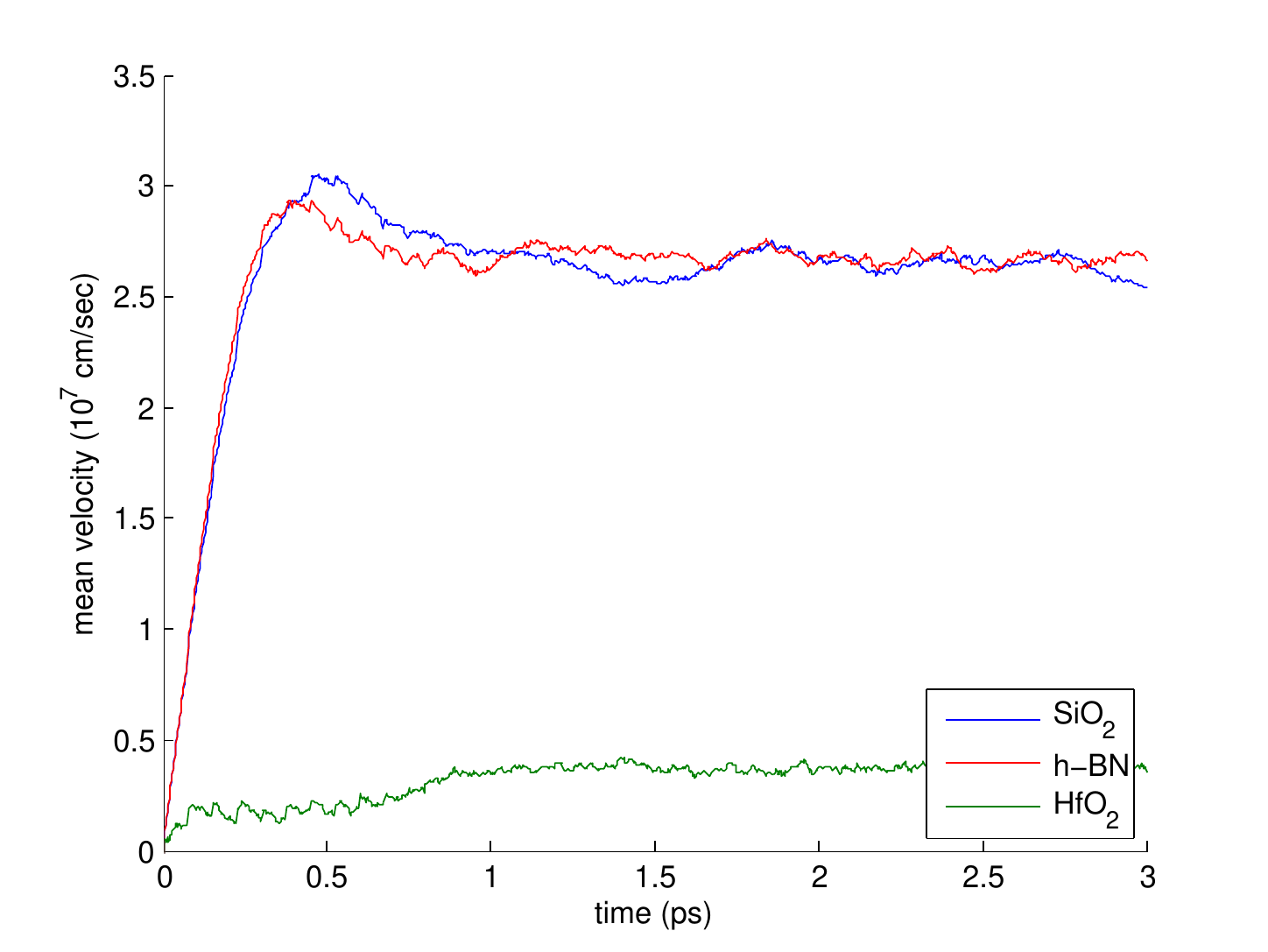}}
 	\caption{Confronto delle velocità medie rispetto al tempo per $d=\si{\num{0}.\nano\metre}$ (a), $d=\si{\num{0.5}.\nano\metre}$ (b) e $d=\si{\num{1}.\nano\metre}$ (c), nel caso di un campo elettrico applicato di $\si{\num{5}.\kilo\volt\per\centi\metre}$ e un livello di Fermi pari a $\si{\num{0.4}.\electronvolt}$.}
\label{FIG:CAP3:DSMC_comp_3_subs}
\end{figure}
\FloatBarrier
\`{E} possibile osservare che i valori della velocità media diminuiscono riducendo la distanza $d$ dalle impurezze del materiale, confermando il decadimento della mobilità dovuta al substrato come conseguenza diretta degli scattering aggiuntivi con le impurezze. Per valori elevati di $d$ la velocità nei due casi, ottenuti utilizzando il $\ce{SiO2}$ e il $\ce{h-BN}$, è confrontabile. Per valori intermedi di $d$ il $\ce{h-BN}$ si comporta meglio rispetto al $\ce{SiO2}$. I precedenti risultati non tengono conto del fatto che la distanza dalle impurezze è in realtà variabile nel substrato. Per ovviare a questo inconveniente si è operato come nel Paragrafo \ref{PAR:CAP3:Subs_1} introducendo per la distanza $d$ sia il modello uniforme che il modello in cui sono usate le leggi Gamma. I grafici seguenti mostrano i risultati ottenuti.
\FloatBarrier
\begin{figure}[ht]
	\centering
	\subfigure[]
   		{\includegraphics[width=7.1cm]{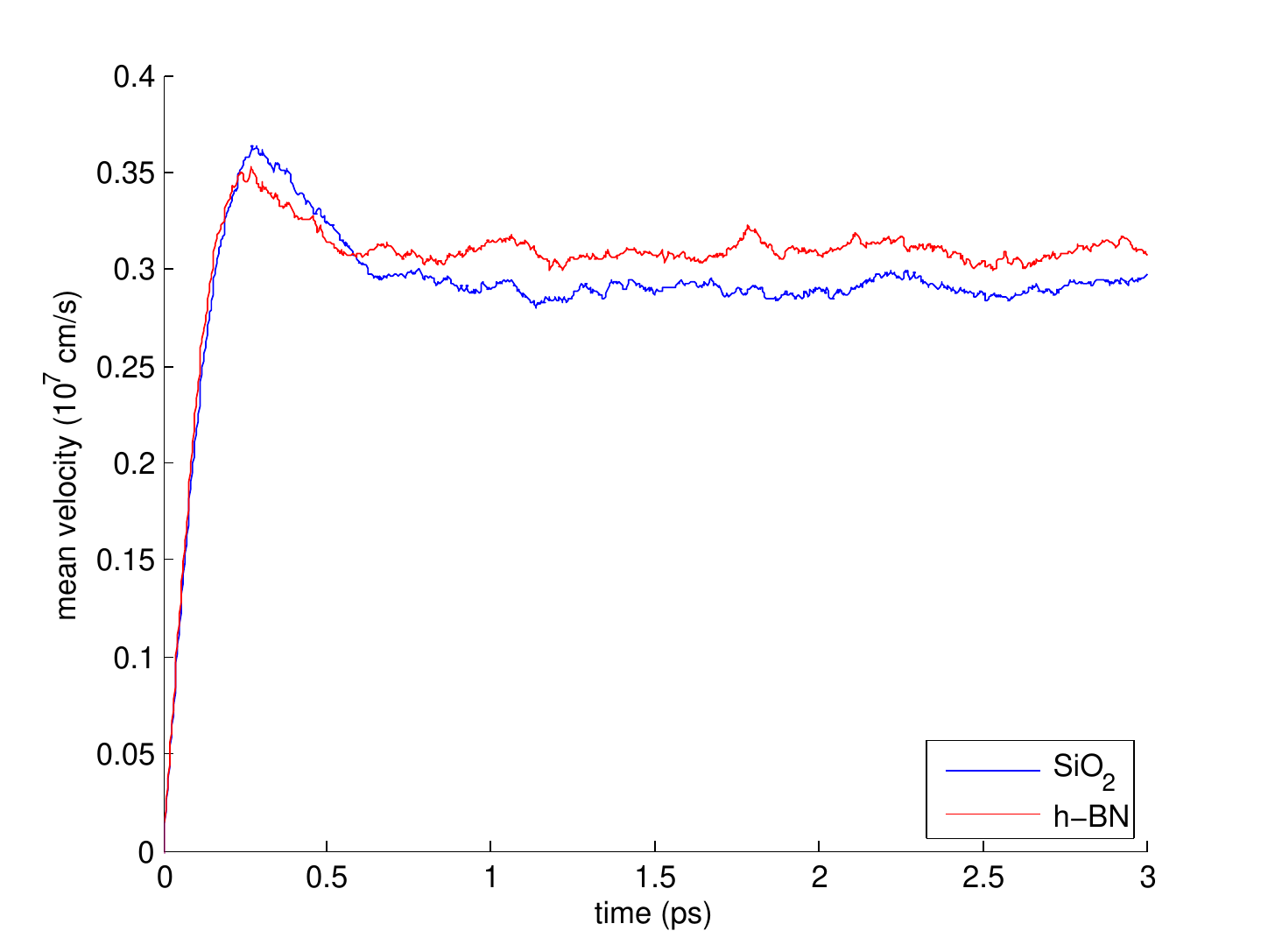}}
 	\,
 	\subfigure[]
   		{\includegraphics[width=7.1cm]{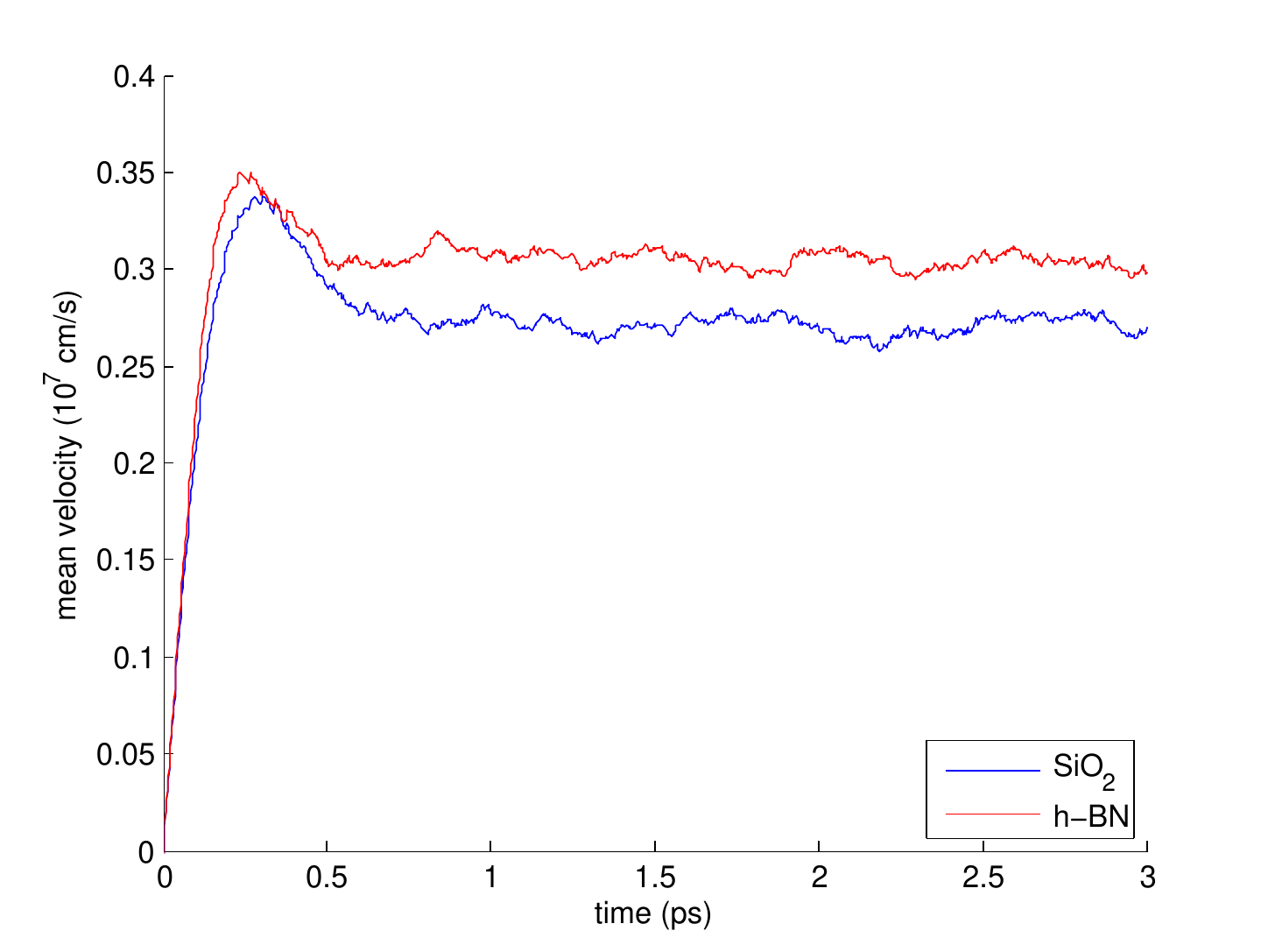}}
	\\
   	\subfigure[]
   		{\includegraphics[width=7.1cm]{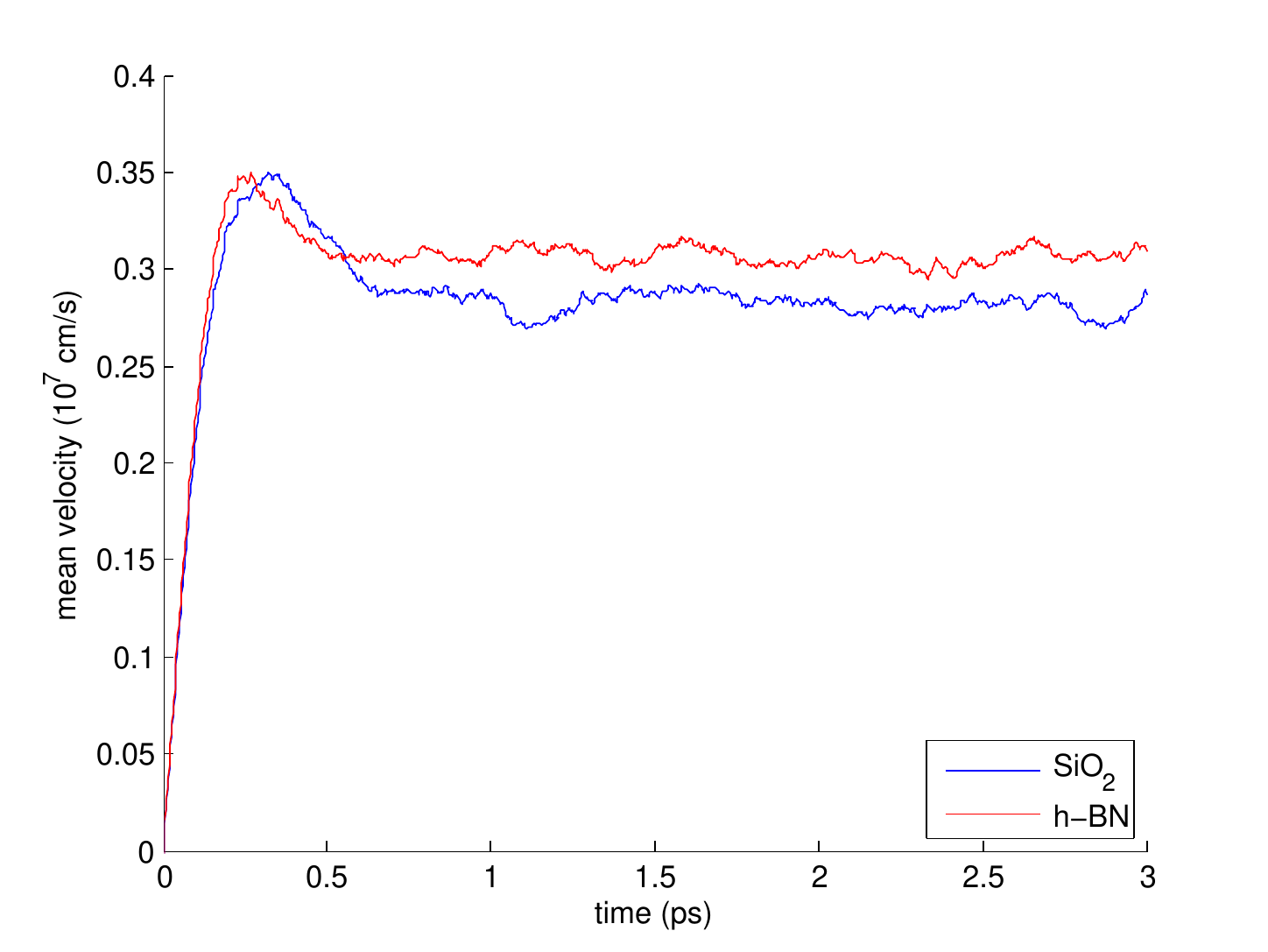}}
   	\,
   	\subfigure[]
   		{\includegraphics[width=7.1cm]{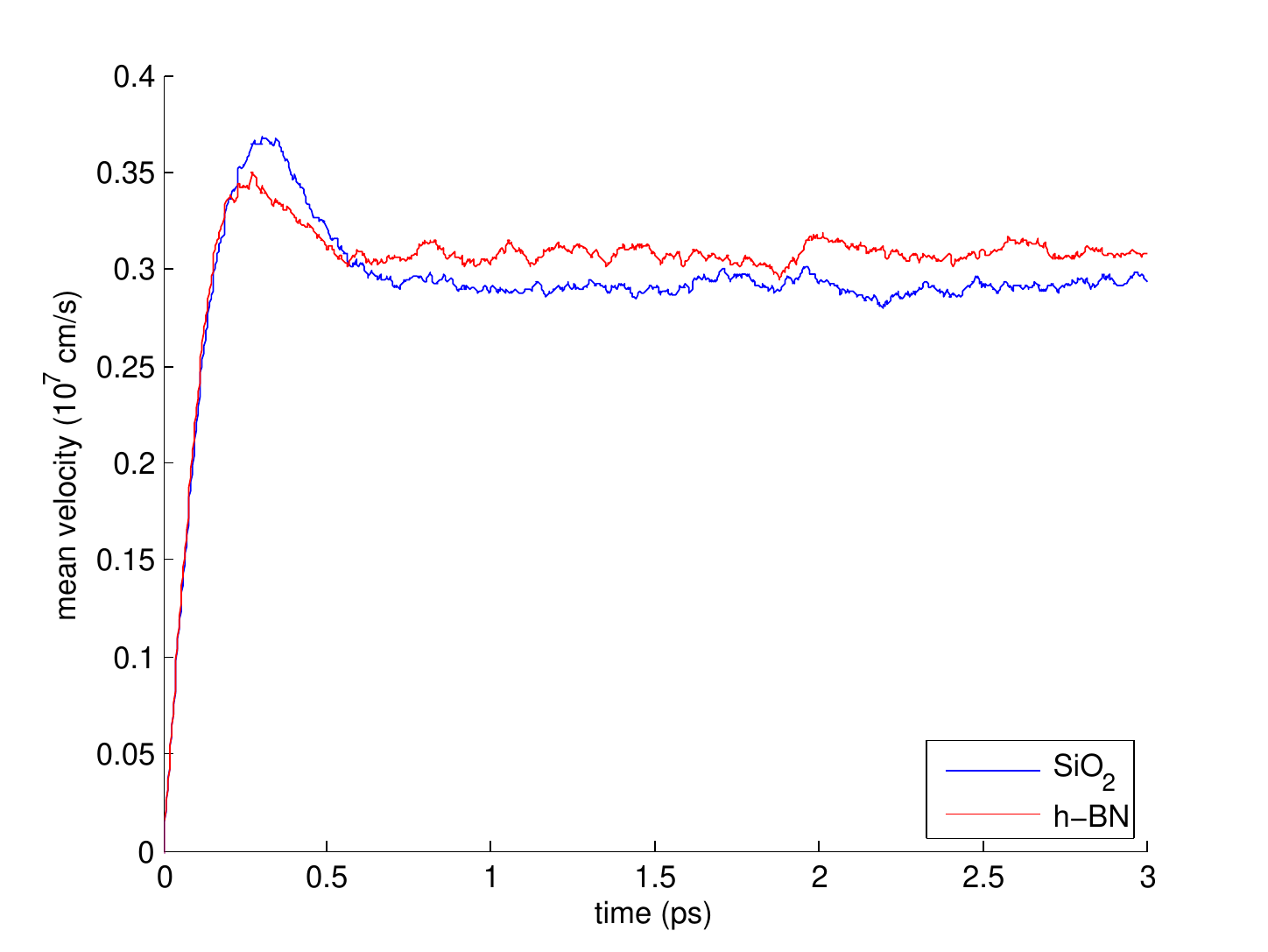}}
 	\caption{Confronto delle velocità medie rispetto al tempo nel caso di un campo elettrico applicato di $\si{\num{10}.\kilo\volt\per\centi\metre}$ e un livello di Fermi pari a $\si{\num{0.4}.\electronvolt}$, considerando diverse distribuzioni per la distanza $d$: $U([0,1])$ (a), $\Gamma(2,2)$ (b), $\Gamma(3,2)$ (c) e $\Gamma(4,2)$ (d).}
\label{FIG:CAP3:DSMC_comp_2_subs}
\end{figure}
\FloatBarrier
Si osservi che, all'aumentare del parametro $k$ delle leggi $\Gamma(k,2)$, le soluzioni di entrambi i substrati tendono a quelle del caso intrinseco ma per $k=2$ il $\ce{h-BN}$ ha un andamento migliore rispetto al caso con $d$ costante. Nei grafici seguenti è mostrato ancora una volta che il principio di esclusione di Pauli viene rispettato. Viene rappresentata la distribuzione degli elettroni considerando un campo elettrico applicato di $\si{\num{10}.\kilo\volt\per\centi\metre}$ e un livello di Fermi pari a $\si{\num{0.4}.\electronvolt}$.
\FloatBarrier
\begin{figure}[ht]
	\centering
	\subfigure[]
   		{\includegraphics[width=7.1cm]{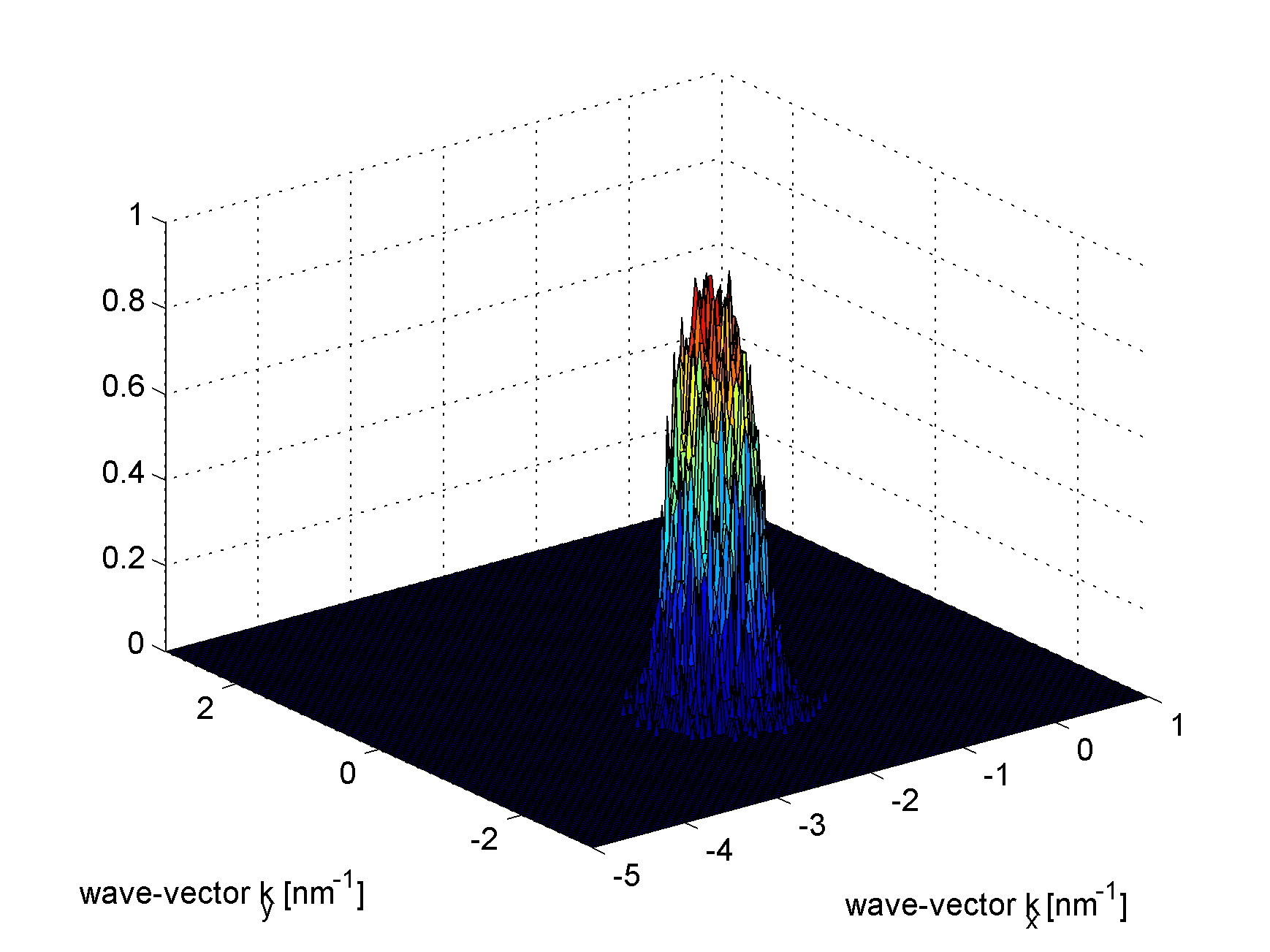}}
 	\,
 	\subfigure[]
   		{\includegraphics[width=7.1cm]{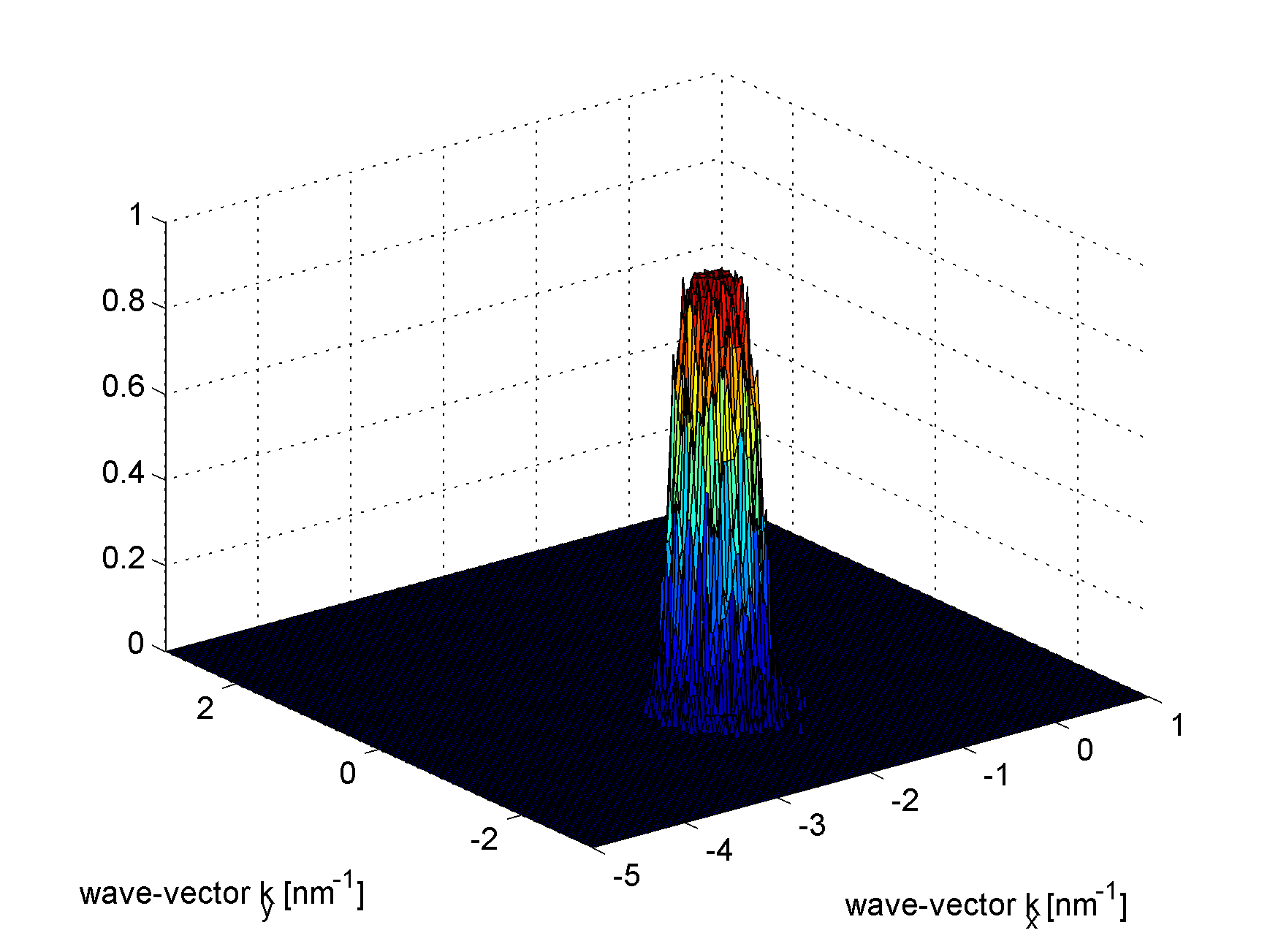}}
 	\caption{Grafici delle distribuzioni degli elettroni nel $\ce{SiO2}$ (a) e nel $\ce{h-BN}$ (b), con $d\sim U([0,1])$.}
\end{figure}
\begin{figure}[ht]
	\centering
	\subfigure[]
   		{\includegraphics[width=7.1cm]{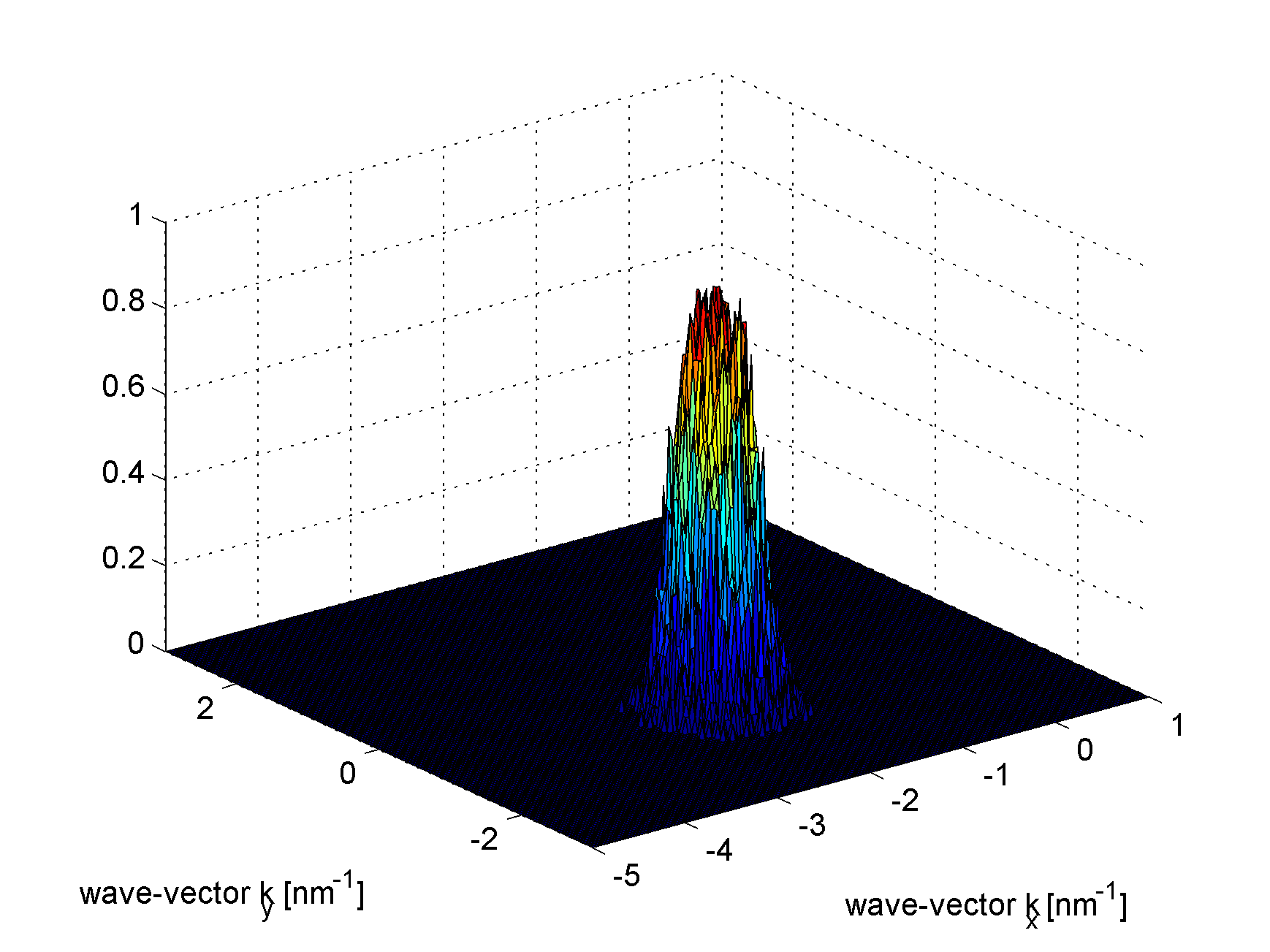}}
 	\,
 	\subfigure[]
   		{\includegraphics[width=7.1cm]{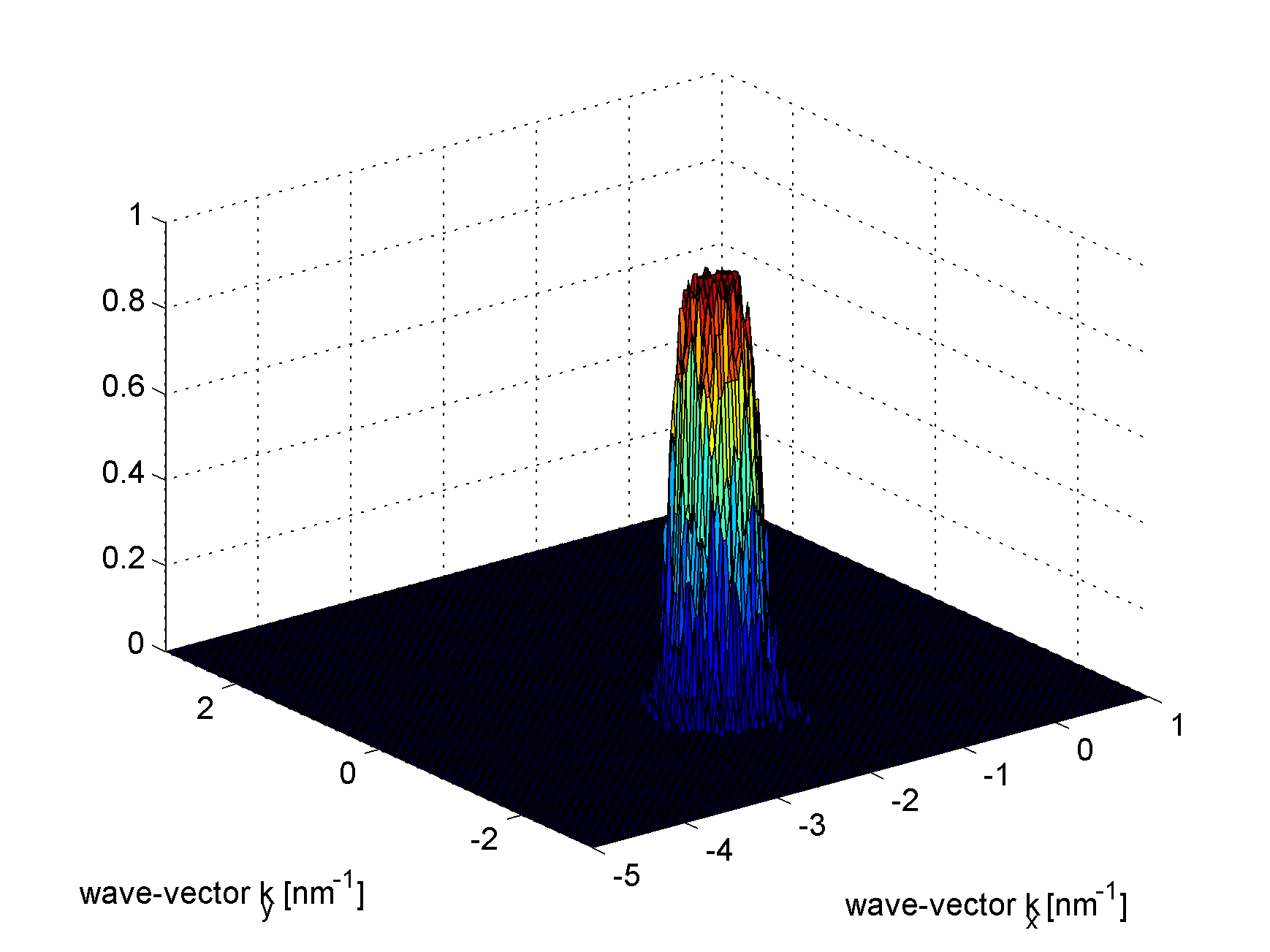}}
 	\caption{Grafici delle distribuzioni degli elettroni nel $\ce{SiO2}$ (a) e nel $\ce{h-BN}$ (b), con $d\sim \Gamma(2,2)$.}
\end{figure}
\begin{figure}[ht]
	\centering
	\subfigure[]
   		{\includegraphics[width=7.1cm]{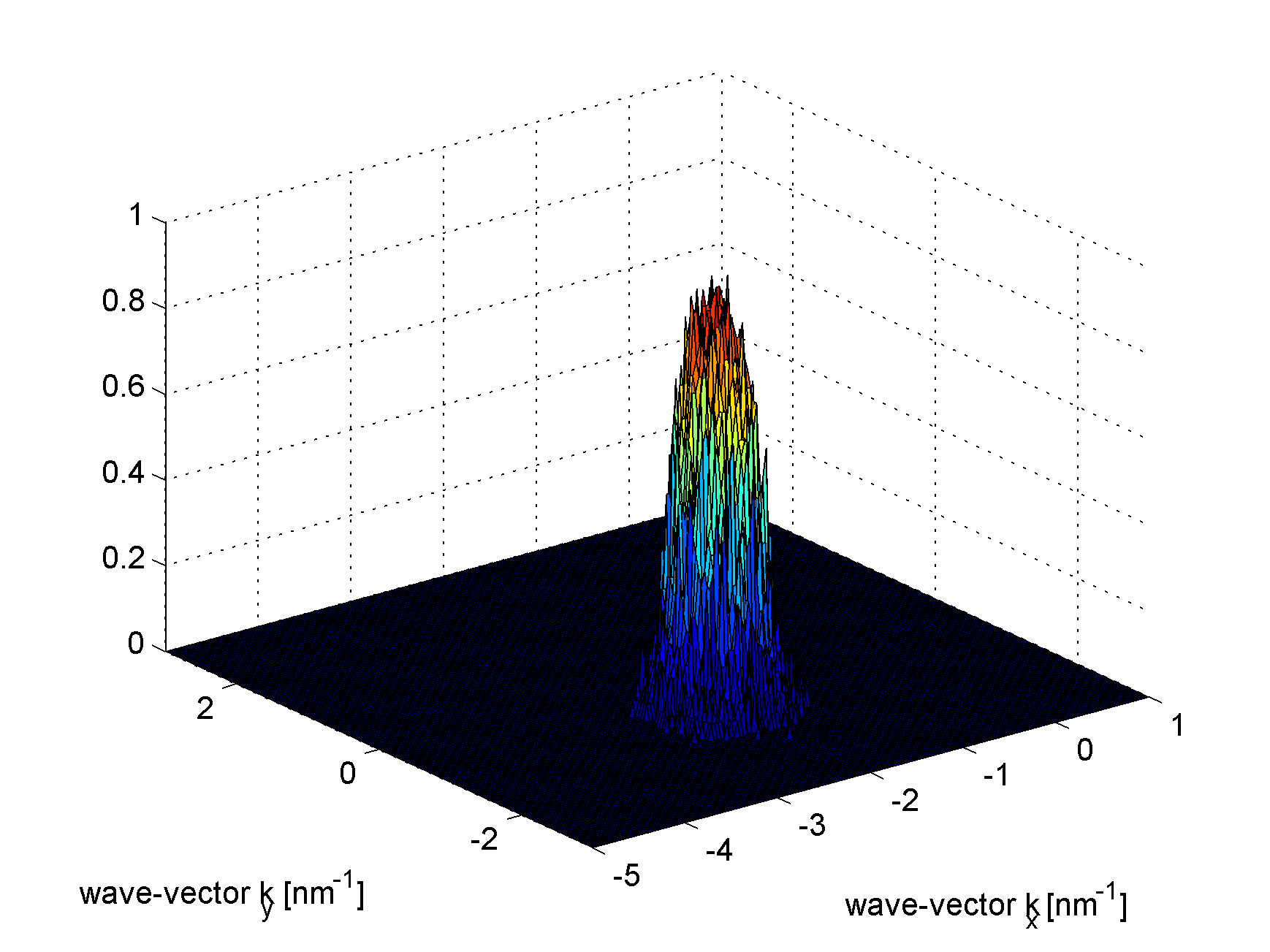}}
 	\,
 	\subfigure[]
   		{\includegraphics[width=7.1cm]{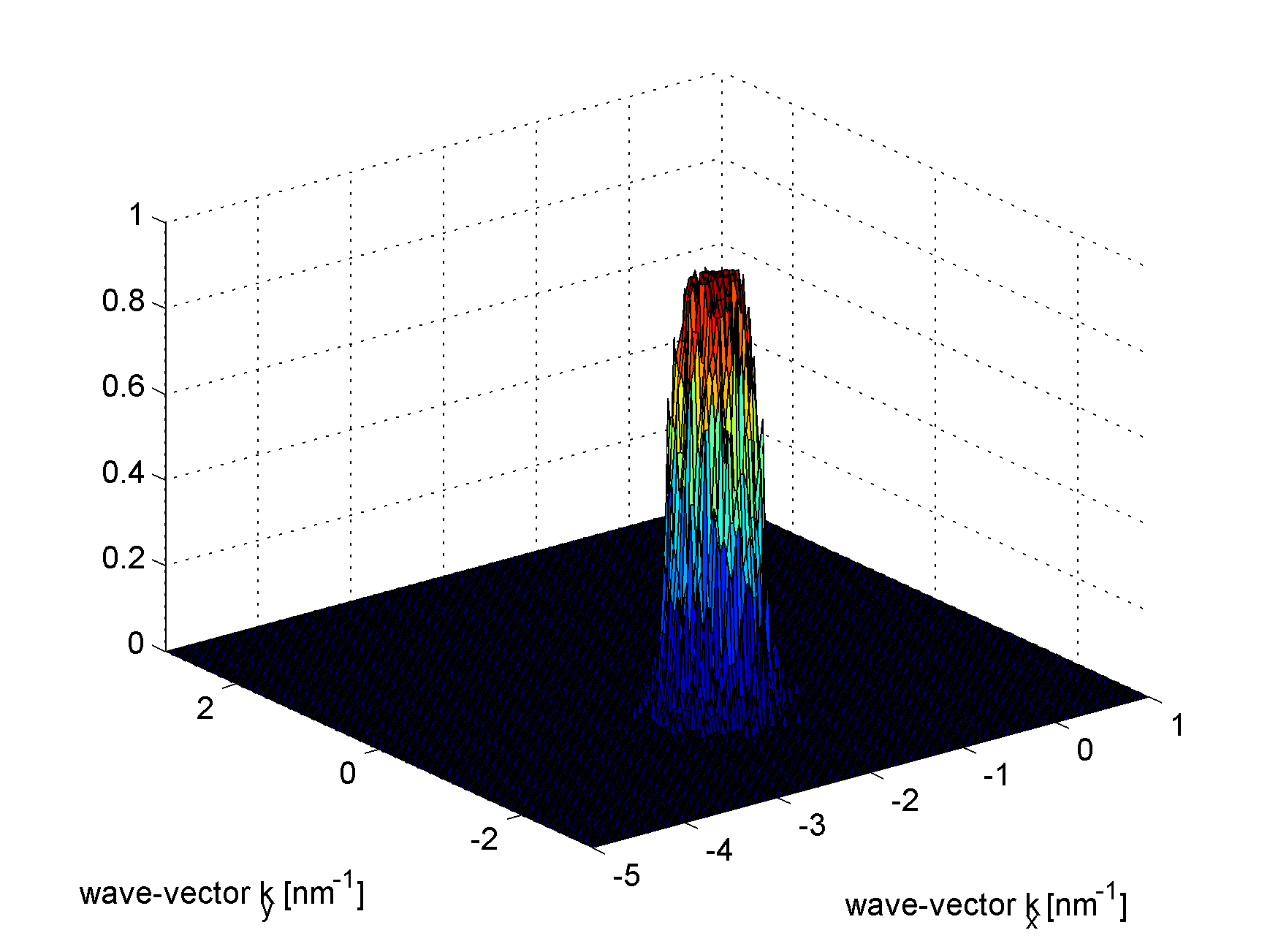}}
 	\caption{Grafici delle distribuzioni degli elettroni nel $\ce{SiO2}$ (a) e nel $\ce{h-BN}$ (b), con $d\sim \Gamma(3,2)$.}
\end{figure}
\FloatBarrier
\begin{figure}[ht]
	\centering
	\subfigure[]
   		{\includegraphics[width=7.1cm]{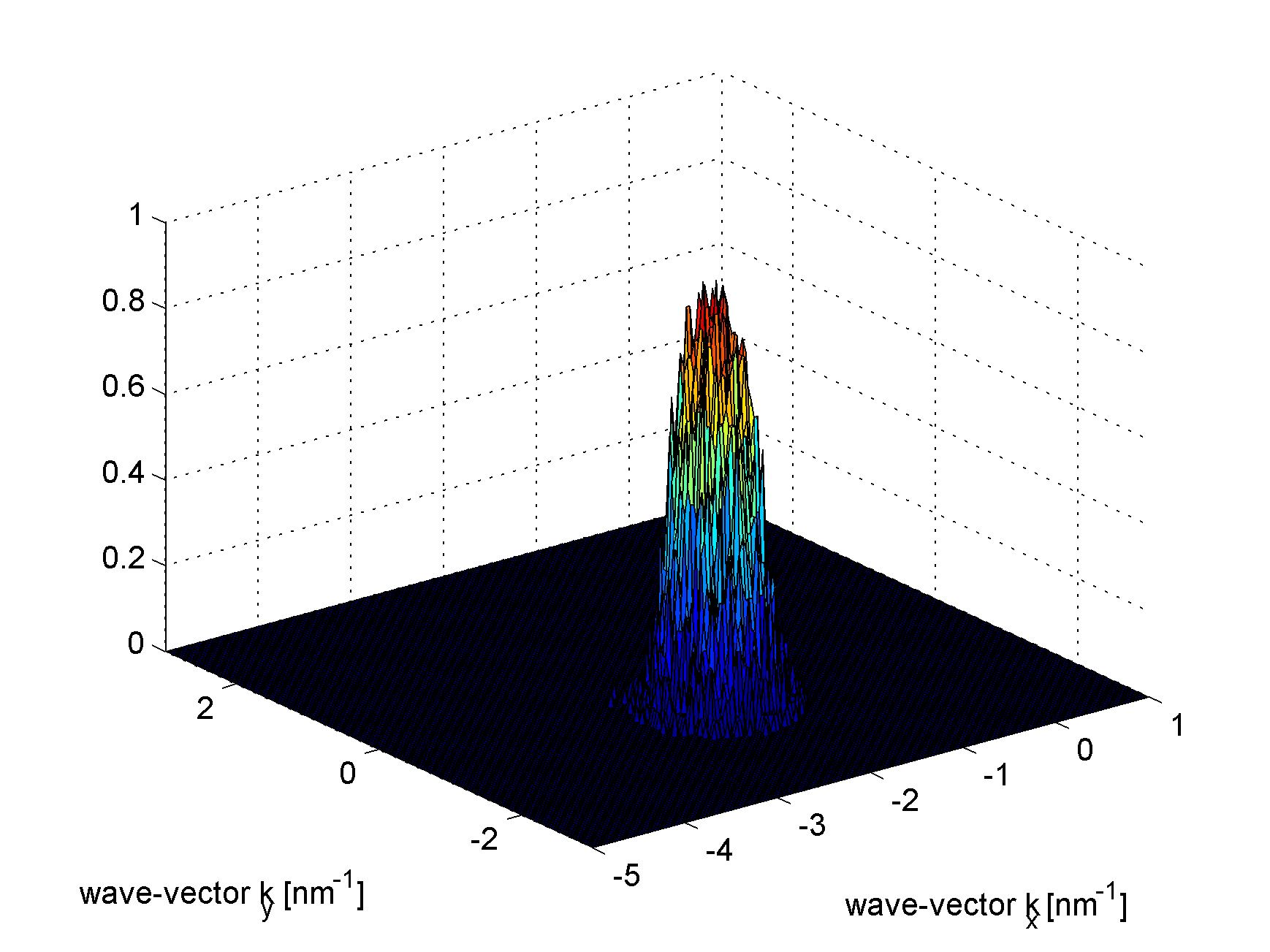}}
 	\,
 	\subfigure[]
   		{\includegraphics[width=7.1cm]{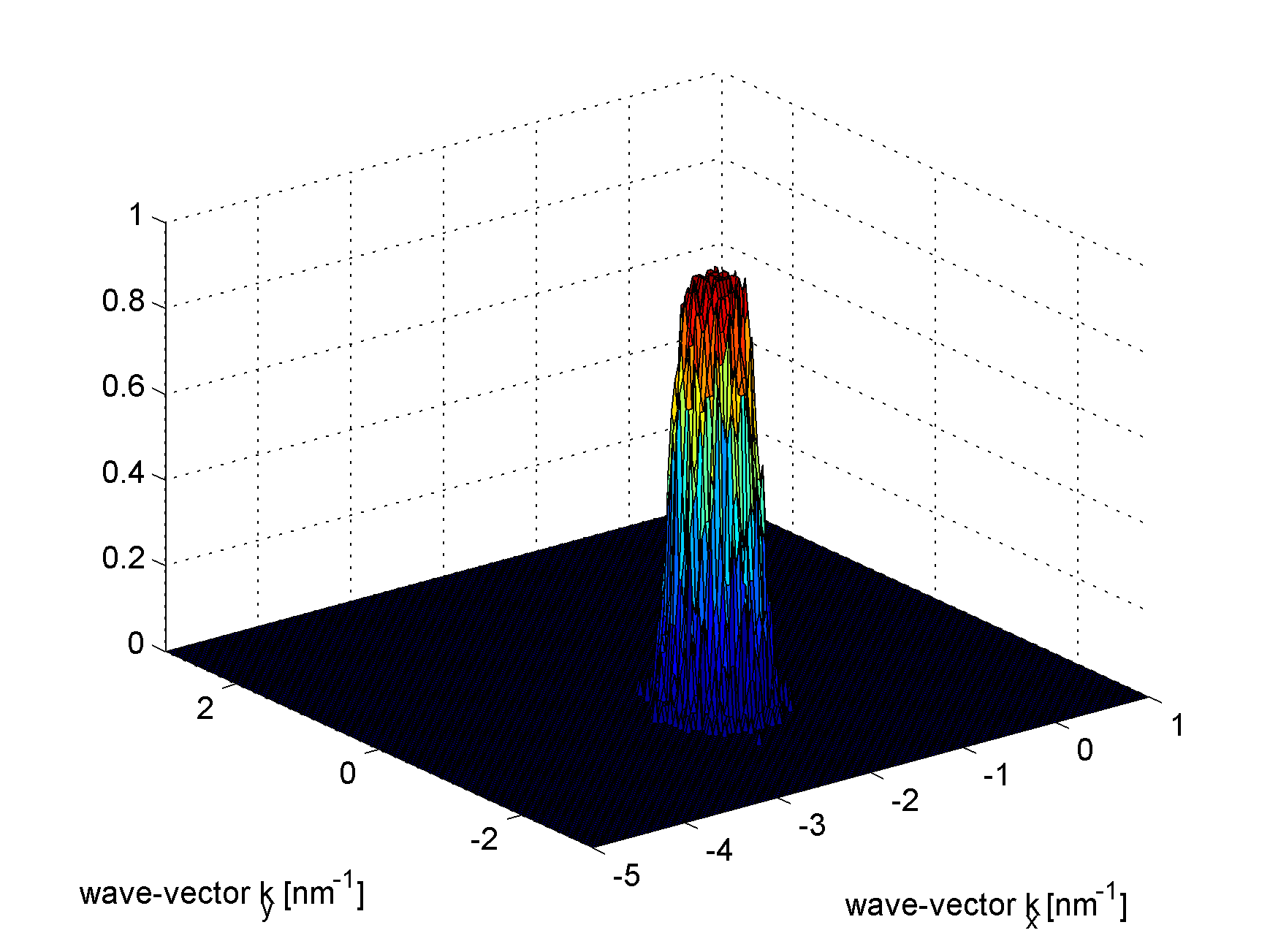}}
 	\caption{Grafici delle distribuzioni degli elettroni nel $\ce{SiO2}$ (a) e nel $\ce{h-BN}$ (b), con $d\sim \Gamma(4,2)$.}
\end{figure}
\FloatBarrier
In conclusione, le differenze tra le velocità medie per i substrati considerati sono in accordo con gli effetti attesi e confermano un decadimento della mobilità. Come già trovato in \citep{ART:Hirai2014}, riguardo alla la mobilità per campi bassi il $\ce{h-BN}$ si rivela essere un materiale migliore del $\ce{SiO2}$ in quanto produce un decadimento inferiore e ciò si manifesta anche nelle simulazioni per campi alti, riportate in questa tesi e in corso di pubblicazione in \citep{ART:CoMajNaRo}.



\addtocontents{toc}{\vspace{2em}} 

\appendix 



\chapter*{Conclusioni}
\addcontentsline{toc}{chapter}{Conclusioni}
\chaptermark{Conclusioni.}

\lhead{Conclusioni}


Questo lavoro di tesi è stato incentrato sul modello semiclassico per il trasporto di cariche nei semiconduttori, con particolare riguardo al grafene. Ciò ha richiesto degli approfondimenti di meccanica classica, meccanica quantistica e modelli cinetici, nonché nozioni di fisica atomica e dei materiali. Per risolvere l'equazione di Boltzmann semiclassica, scritta per la distribuzione degli elettroni, si è utilizzato il metodo di Simulazione Diretta Monte Carlo, nella versione descritta in \citep{ART:RoMaCo_JCP}. 

Poiché l'alto costo computazionale può rappresentare un ostacolo nel proseguo delle ricerche su questo nuovo approccio, il punto di partenza è stato sviluppare un codice più rapido. Effettivamente, cambiando linguaggio di programmazione da MATLAB a Fortran 90 il tempo di esecuzione delle simulazioni si abbassa in maniera considerevole.

Visti gli ottimi risultati conseguiti nel caso del grafene sospeso, si è optato per costruire un codice in Fortran 90 anche nel caso del grafene su substrato, già presentato in \citep{ART:CoMajRo_Ric_mat_2016}. Questo è stato il punto di partenza per la parte sperimentale della tesi. Essa ha riguardato in primo luogo l'uso di modelli che descrivessero l'intrinseca oscillazione della distanza tra gli atomi di grafene e le impurezze, che è dell'ordine di qualche angstrom.

Si è agito considerando il caso del $\ce{SiO2}$ e confrontando i valori ottenuti per i valori costanti pari a $\si{\num{0}.\nano\metre}$, $\si{\num{0.5}.\nano\metre}$ e $\si{\num{1}.\nano\metre}$ con distanze calcolate casualmente in accordo con le distribuzioni di probabilità scelte: $U([0,1])$, $\Gamma(k,2)$ con $k=2,3,4$.

In termini di risultati sui valori medi di energia e velocità degli elettroni, si è ottenuto che la scelta uniforme è molto vicina alla scelta costante per la distanza $d$ pari a $\si{\num{0.5}.\nano\metre}$. Invece, con la scelta Gamma, partendo da un numero di gradi di libertà pari a 2, si ottengono dei valori intermedi tra la scelta $d=0\si{\nano\metre}$ e la scelta $d=\si{\num{0.5}.\nano\metre}$. All'aumentare del valore di $k$ i risultati sono sempre più prossimi a quelli ottenuti scegliendo $d=\si{\num{0.5}.\nano\metre}$.

Il secondo aspetto riguardante la parte sperimentale è consistito nel confrontare i valori della velocità media facendo variare il tipo di substrato. Seguendo l'analisi già trattata in \citep{ART:Hirai2014} per valori bassi di campo elettrico, si è passati a simulare il $\ce{h-BN}$ e il $\ce{HfO2}$ per valori di campo più elevati. La distanza dalle impurezze è un parametro cruciale che influisce sul decadimento della mobilità elettronica. Questo decadimento è risultato molto accentuato nel caso del $\ce{HfO2}$. Per quanto riguarda gli altri substrati, si è osservato che per valori elevati di $d$ le velocità simulate nei due materiali sono confrontabili, per valori intermedi di $d$ invece il $\ce{h-BN}$ si comporta meglio rispetto al $\ce{SiO2}$.

In tutti i casi simulati si è potuto constatare che la distribuzione degli elettroni non eccede mai il valore 1, in accordo con il principio di esclusione di Pauli che è dunque pienamente rispettato.

I risultati presentati in questa tesi saranno pubblicati in \citep{ART:CoMajNaRo}.

In conclusione, per quanto riguarda i possibili sviluppi futuri di questo lavoro di tesi, si può pensare di approfondire il caso del grafene su substrato, determinando le curve di mobilità, seguendo l'approccio utilizzato in \citep{ART:MajMaRo}; studiare ed implementare il caso bipolare, in quanto, fino a questo momento, il moto delle lacune è trascurato imponendo un livello di Fermi sufficientemente elevato (cfr.~\citep{ART:RoMaCo_JCP}); implementare il modello applicando campi elettrici non costanti, come estensione del metodo proposto in \cite{ART:RoMaCo_JCP}; studiare ed implementare il modello considerando anche gli scattering di tipo elettrone-elettrone.

\addtocontents{toc}{\vspace{2em}} 

\backmatter


\label{Bibliography}

\lhead{\emph{Bibliografia}} 

\bibliographystyle{unsrtnat-ita} 

\bibliography{Bibliography} 

\end{document}